\pdfoutput=1
\RequirePackage{ifpdf}
\ifpdf 
\documentclass[pdftex]{sigma}
\else
\documentclass{sigma}
\fi

\usepackage{stackengine}
\newcommand\oast{\stackMath\mathbin{\stackinset{c}{0ex}{c}{0ex}{\ast}{\bigcirc}}}

\numberwithin{equation}{section}

\newtheorem{Theorem}{Theorem}[section]
\newtheorem*{Theorem*}{Theorem}
\newtheorem{Corollary}[Theorem]{Corollary}
\newtheorem{Lemma}[Theorem]{Lemma}
\newtheorem{Proposition}[Theorem]{Proposition}
 { \theoremstyle{definition}
\newtheorem{Definition}[Theorem]{Definition}

\newtheorem{Remark}[Theorem]{Remark} }

\begin{document}
\allowdisplaybreaks

\newcommand{\arXivNumber}{2203.08249}

\renewcommand{\PaperNumber}{009}

\FirstPageHeading

\ShortArticleName{Quantum Curves, Resurgence and Exact WKB}

\ArticleName{Quantum Curves, Resurgence and Exact WKB}

\Author{Murad ALIM~$^{\rm a}$, Lotte HOLLANDS~$^{\rm b}$ and Iv\'an TULLI~$^{\rm a}$}

\AuthorNameForHeading{M.~Alim, L.~Hollands and I.~Tulli}

\Address{$^{\rm a)}$~Fachbereich Mathematik, Universit\"at Hamburg, Bundesstr.~55, 20146 Hamburg, Germany}
\EmailD{\href{mailto:murad.alim@uni-hamburg.de}{murad.alim@uni-hamburg.de}, \href{mailto:ivan.tulli@uni-hamburg.de}{ivan.tulli@uni-hamburg.de}}

\Address{$^{\rm b)}$~Department of Mathematics at Heriot-Watt University,\\
\hphantom{$^{\rm b)}$}~Maxwell Institute for Mathematical Sciences, Edinburgh EH14 4AS, UK}
\EmailD{\href{mailto:l.hollands@hw.ac.uk}{l.hollands@hw.ac.uk}}

\ArticleDates{Received June 09, 2022, in final form February 08, 2023; Published online March 06, 2023}

\Abstract{We study the non-perturbative quantum geometry of the open and closed topological string on the resolved conifold and its mirror. Our tools are finite difference equations in the open and closed string moduli and the resurgence analysis of their formal power series solutions. In the closed setting, we derive new finite difference equations for the refined partition function as well as its Nekrasov--Shatashvili (NS) limit. We write down a distinguished analytic solution for the refined difference equation that reproduces the expected non-perturbative content of the refined topological string. We compare this solution to the Borel analysis of the free energy in the NS limit. We find that the singularities of the Borel transform lie on infinitely many rays in the Borel plane and that the Stokes jumps across these rays encode the associated Donaldson--Thomas invariants of the underlying Calabi--Yau geometry. In the open setting, the finite difference equation corresponds to a canonical quantization of the mirror curve. We analyze this difference equation using Borel analysis and exact WKB techniques and identify the 5d BPS states in the corresponding exponential spectral networks. We furthermore relate the resurgence analysis in the open and closed setting. This guides us to a five-dimensional extension of the Nekrasov--Rosly--Shatashvili proposal, in which the NS free energy is computed as a generating function of $q$-difference opers in terms of a special set of spectral coordinates. Finally, we examine two spectral problems describing the corresponding quantum integrable system. }


\Classification{40G10; 39A70; 81T30}

\setcounter{tocdepth}{2}
\tableofcontents

\section{Introduction}

A recurring fascinating insight of the interaction of mathematics and physics is that mathematical invariants and structures are most naturally encoded in terms of the data of physical theories. This is in particular the case for supersymmetric quantum field theories and string theo\-ries whose partition functions often have interpretations as generating functions of invariants of manifolds and whose parameter spaces become mathematical moduli spaces.

The physical partition functions which encode the mathematical invariants often correspond however to asymptotic formal power series with zero radius of convergence which is a common feature of series obtained in quantum field and string theories using a perturbative approach. A~theme which has generated a considerable amount of excitement recently is the realization that the systematic mathematical treatment of the asymptotic series using the ideas of resurgence leads to the uncovering of further deep mathematical structures as well as physical insights, see, e.g.,~\cite{Aniceto:2011nu,Marinolecture} and references therein.

Mathematically, when the Borel resummation of the asymptotic series is considered it turns out that other sets of invariants are encoded in the Stokes jumps of different Borel resummations, see for example \cite{Couso-Santamaria:2016vcc,Garoufalidis:2020nut,Garoufalidis:2020xec,Gu:2021ize}, see also \cite{garoufalidis2020resurgence} for a study of resurgence for the quantum dilogarithm function, which is a building block of many interesting objects in mathematical physics and in particular in Chern--Simons theories. In~\cite{ASTT21}, the techniques of \cite{garoufalidis2020resurgence} were used to study the Borel resummation of the Gromov--Witten potential $F^{\mathrm{top}}(\lambda,t)$ for the resolved conifold. Earlier results on the Borel resummation for the resolved conifold with different techniques and scope were obtained in \cite{Hatsuda:2015owa,Pasquetti:2009jg}. The Borel transform has infinitely many singularities organized along rays coinciding with the rays $\pm \mathbb{R}_{<0}Z_{\gamma}$, where $Z_{\gamma}$ denotes the central charge of a BPS state of charge~$\gamma$. Different Borel resummations were defined along rays which avoid the singularities, and it was found that they experience Stokes jumps across the rays $\pm \mathbb{R}_{<0}Z_{\gamma}$, with the BPS charge $\gamma \in \Gamma$ contributing to the jump by \cite{ASTT21}:
\begin{equation}\label{eq:topstringjumps}
 \Delta_{\gamma} F^{\rm top}_{\textrm{Borel}} (\lambda,t) = \frac{\Omega(\gamma)}{2\pi \mathrm{i}} \partial_{\check{\lambda}} \big( \check{\lambda} {\rm Li}_2\big({\rm e}^{Z_{\gamma}/\check{\lambda}}\big)\big), \qquad \check{\lambda}=\frac{\lambda}{2\pi},
\end{equation}
where $\Omega(\gamma)$ correspond to the Donaldson--Thomas (DT) invariant. The identification of the DT invariants was established by providing the link to a Riemann--Hilbert problem put forward by Bridgeland in~\cite{Bridgeland1} and applied to the resolved conifold in~\cite{BridgelandCon}.

The Borel analysis furthermore allowed to connect to previous proposals for definitions of non-perturbative topological string theory and elucidate their overlaps of validity. The Borel summation along a distinguished ray for instance gave an expression previously proposed in \cite{Hatsuda:2013oxa,Hatsuda:2015owa}, while a limiting expression obtained from the latter through infinitely many jumps gave the Gopakumar--Vafa expression for the resummation of the free energies. In the work of Bridgeland it was suggested that a tau-function, obtained as a solution of a~Riemann--Hilbert problem defined from the wall-crossing structure of Donaldson--Thomas invariants, provides a~non-perturbative completion of the Gromov--Witten potential. In further developments, a~difference equation was obtained from the asymptotic expansion of the Gromov--Witten free energy of the resolved conifold in \cite{alim2020difference}, while it was found in~\cite{Alim:2021lld} that this difference equation admits an analytic solution in~$\lambda$ related to Bridgeland's tau-function, whose non-perturbative content was shown in~\cite{Alim:2021ukq} to match earlier expectations of \cite{Hatsuda:2013oxa,Hatsuda:2015owa}.

Moreover, following the ideas of \cite{Coman:2018uwk,CLT20}, it was found in \cite{ASTT21} that the exponentials of the Borel summations $\exp\big(F^{\mathrm{top}}_{\rho}\big)$ along the ray $\rho$ form a collection of local sections of the conformal limit of a certain hyperholomorphic line bundle previously considered in~\cite{APP, Neitzke_hyperhol}, whose transition functions are the exponentials of the above Stokes factors, leading to a new geometrical, non-perturbative picture of the topological string.

In parallel developments in a different context, it was realized in \cite{Hollands:Heun,Hollands:recipe} that Borel resummation plays a central role in the geometric formulation of the effective twisted superpotential~$\mathcal{W}^{\mathrm{eff}}$ of a four-dimensional $\mathcal{N}=2$ theory $T_{\mathrm{4d}}$ of class S in the $\frac{1}{2} \Omega$-background $\mathbb{R}^{2}_\epsilon \times \mathbb{R}^{2}$.
This was motivated by a conjecture of Nekrasov, Rosly and Shatashvili \cite{Nekrasov:NRS}, which says that the superpotential~$\mathcal{W}^{\mathrm{eff}}$ may be obtained as a generating function of opers in terms of a special kind of Darboux coordinates on the associated moduli space of complexified flat connections.
It was found that the NRS Darboux coordinates~$\{x_i,y^i\}$ can be expressed in terms of the Borel summation $\mathcal{B}_\theta$ of the quantum periods $\oint_{\gamma}\lambda^{\mathrm{qu}}(\epsilon)$ in a certain critical direction with phase~$\vartheta$.\footnote{To be more precise, this critical direction with phase $\vartheta$ corresponds to a ray in the Borel plane with infinitely many singularities. The Borel summation $\mathcal{B}_{\vartheta}$ is not defined precisely in this critical direction. Instead we refer to the median of the so-called lateral Borel summation, which averages the result of the two Borel summations $\mathcal{B}_{\vartheta^+}$ and $\mathcal{B}_{\vartheta^-}$.} That~is,~\looseness=1
\begin{equation}\label{eqn:NRSBorel}
x_i = \frac{\epsilon}{\pi \mathrm{i}} \log \bigg( \mathcal{B}_\vartheta \oint_{A^i} \lambda^{\mathrm{qu}}(\epsilon) \bigg) \qquad \mathrm{and} \qquad y^i = \frac{1}{2 \epsilon} \log \bigg( \mathcal{B}_\vartheta \oint_{B_i} \lambda^{\mathrm{qu}}(\epsilon) \bigg).
\end{equation}
In terms of these Borel sums, the superpotential $\mathcal{W}^{\mathrm{eff}}$ may be obtained as the generating function%
\begin{equation}\label{eqn:Weff_generatingfunction}
y^i = \frac{1}{\epsilon} \frac{\partial \mathcal{W}^{\mathrm{eff}}(x,\epsilon)}{\partial x_i}.
\end{equation}

For a four-dimensional gauge theory $T_{\mathrm{4d}}$ it was observed that the critical phase $\vartheta$ agrees with the phase of the central charge of the W-boson. In this case the corresponding NRS Darboux coordinates are a complexified version of the well-known Fenchel--Nielsen coordinates (see also~\cite{Hollands:2013qza}). For non-Lagrangian theories $T_{\mathrm{4d}}$ of class S, such as the $E_{6}$ Minahan--Nemeschansky theory, there is a discrete set of critical phases $\vartheta_{\gamma}$ corresponding to 4d BPS particles whose electro-magnetic charge is a multiple of~$\gamma$~\cite{Hollands:t3abel}.

Similar to \cite{ASTT21}, it was noted that there exists a natural generalization $\mathcal{W}^{\mathrm{eff}}_{\vartheta'}$ of the effective twisted superpotential $\mathcal{W}^{\mathrm{eff}}$. The superpotential $\mathcal{W}^{\mathrm{eff}}_{\vartheta'}(x,\epsilon)$ is defined in the same way as in equations (\ref{eqn:NRSBorel}) and (\ref{eqn:Weff_generatingfunction}), but now in terms of quantum periods that are Borel summed in an arbitrary direction~$\vartheta'$. As for $F^{\mathrm{top}}_{\rho}$, the generalized superpotential $\mathcal{W}^{\mathrm{eff}}_{\vartheta'}$ is piece-wise constant in~$\vartheta'$ and jumps along a discrete set of Stokes rays. In the $\mathcal{N}=2$ context these Stokes rays are naturally labeled by 4d BPS states in the theory $T$. A BPS hypermultiplet with mass~$x$ would for instance induce a Stokes factor of the form
\begin{equation*}
 \Delta \mathcal{W}^{\mathrm{eff}}_{\vartheta'}(\epsilon) =
 \frac{\epsilon}{2\pi\mathrm{i}} \mathrm{Li}_2\big({\rm e}^{\frac{\pi \mathrm{i} x}{\epsilon}}\big).
\end{equation*}
Again, the exponentials of the Borel summations $\mathcal{W}^{\mathrm{eff}}_{\vartheta'}$ can be interpreted as a collection of local sections of a distinguished line bundle.

It was proposed in \cite{Hollands:t3abel, Hollands:recipe} that the generalized superpotential $\mathcal{W}^{\mathrm{eff}}_{\vartheta'}$ has a natural physical interpretation in terms of the 4d $\mathcal{N}=2$ theory $T$. It was argued that each Stokes sector corresponds to a certain IR boundary condition of the 4d theory $T$ in the $\frac{1}{2}\Omega$-background $\mathbb{R}_\epsilon^{2} \times \mathbb{R}^{2}$ labeled by $\vartheta'$, which can be explicitly described by coupling a 3d $\mathcal{N}=2$ theory of class~R to the boundary $S^{1} \times \mathbb{R}^{2}$ of the $\frac{1}{2}\Omega$-background.

Our goal in this work is to establish the relation between the two occurrences of Borel resummation and the almost identical Stokes jumps and transition functions for $F^{\rm top}$ and $\mathcal{W}^{\mathrm{eff}}$. To do so, we note that both the topological string partition function as well as the effective twisted superpotential can be obtained as two distinguished limits of a refinement of topological string theory, which has two deformation parameters $\epsilon_1$ and $\epsilon_2$. The refinement of the topological string partition function is motivated by Nekrasov's computation of the instanton partition function in the $\Omega$-background for four-dimensional $\mathcal{N}=2$ gauge theories as well as their five-dimensional lifts \cite{Nekrasov:2002qd}, and the fact that the K-theoretic instanton partition function agrees with the topological string partition function on a class of non-compact Calabi--Yau threefolds that geometrically engineer the afore-mentioned gauge theories \cite{Katz:1997eq,KKV}, when the limit $\epsilon_1=-\epsilon_2=\lambda$ is taken.

Topological string theory on any Calabi--Yau manifold $\mathbf{X}$ has a connection to quantization, since the topological string partition function may be interpreted as a wave function obtained by quantizing the space of real-valued 3-forms on $\mathbf{X}$ in a complex polarization \cite{Witten:1993ed}. In this interpretation the holomorphic anomaly equations \cite{Bershadsky:1993cx} correspond to a projectively flat connection on the associated bundle of Hilbert spaces that describes the independence of infinitesimal changes in the polarization. In the context of topological string theory on non-compact CY manifolds, the authors of \cite{ADKMV} provided a further link to quantization and moreover to integrability. It was proposed that the mirror curves of non-compact Calabi--Yau manifolds can be intepreted as the analogs of Hamiltonians of a quantum mechanical system and the closed string moduli as the analogs of energy levels in the quantum mechanical phase space. This interpretation has in particular led to the notion of a quantum curve, which depending on whether the curve variables are in $\mathbb{C}$ or $\mathbb{C}^*$, are described by a differential or finite difference equation.

A relation between the two occurrences of a quantum mechanical system was expected, but could not be made precise. One reason for this is that the quantum mechanical setup of the flatness equation of \cite{Witten:1993ed} contains derivatives with respect to the closed string moduli, while the quantum curves contains derivatives with respect to the open string moduli. In~\cite{ACDKV}, it was furthermore realized that the quantum curves are naturally associated to the refined topological string partition function in the Nekrasov--Shatashvili (NS) limit~\cite{NS}. More insights on the relation between the two occurrences of quantization have been obtained through the study of spectral properties of the difference operators describing the mirror curves of non-compact toric Calabi--Yau manifolds \cite{GHM}. See also \cite{MarinoSpectral} and references therein.

To achieve the goal of this paper, namely relating the resurgence of $F^{\rm top}$ and $\mathcal{W}^{\mathrm{eff}}$, we are naturally led to re-consider the precise relations between appearances of quantum mechanical setups occurring in the open and closed string moduli. This will in particular lead to an understanding of the non-perturbative completion of the quantum geometry considered in~\cite{ACDKV}.\looseness=1

The organization of this work is as follows:
In Section~\ref{sec:freeenergies}, we introduce the main players of our work, such as the topological string free energy (in Gromov--Witten and Gopakumar--Vafa form), its refinement and the Nekrasov--Shatashvili limit. We also recollect the definition of the resolved conifold geometry, its mirror and the associated mirror curve, as well as the relation of topological string theory to five-dimensional gauge theory.

In Section~\ref{sec:qmirrorcurve}, we briefly summarize how to obtain a Schr\"odinger operator $\mathbf{D}$ from quantizing the mirror curve, in particular for the resolved conifold geometry. Using the WKB approximation, we find the all-order-in-$\epsilon$ formal power series solution for the Schr\"odinger operator $\mathbf{D}$. We furthermore identify two particular exact solutions. The first exact solution is a ratio of quantum dilogarithm functions that was previously considered in for example \cite{ACDKV}. The other exact solution is given in terms of the Faddeev quantum dilogarithm. In Section \ref{sec:Borelsums}, it will become clear that these two exact solutions correspond to two distinguished Borel summations of the formal WKB solution.\footnote{See also the related discussion in \cite[Section 4]{GHN}, where the all-order-in-$\epsilon$ solution and its Borel analysis is discussed as well.}

Generalizing the approaches of \cite{alim2020difference, Iwaki2}, in Section~\ref{sec:diffeq}, we use the perturbative expansion of the refined free energies for the resolved conifold to derive a finite difference equation in the closed string modulus. We will show in Section \ref{sec:nonpertsol} that this difference equation~\eqref{eq:refdiffeq} is solved by a~special function with pleasant analytical properties. This solution provides a natural candidate for the non-perturbative completion of the perturbative expansion that we started with. We show that it is furthermore nicely related to previous proposals for non-perturbative refined partition functions for the resolved conifold, such as the one obtained by Lockhart and Vafa from the relation to superconformal field theory \cite{Lockhart:2012vp}, as well as the one obtained in \cite{Hatsuda:2015oaa} from ABJM theory and spectral analysis and the proposal of \cite{Krefl:2015vna} obtained from universal Chern--Simons theory.
We furthermore find that this non-perturbative solution satisfies an additional difference equation that involves integer shifts of the closed string modulus. This new difference equation (the third equation in \eqref{eps1shift}) is invisible to the perturbative expansion of the refined free energy, which is periodic in the closed modulus. We show in \eqref{eq:intshifttopstring} that this new difference equation encodes the Stokes jumps of the unrefined topological string, obtained previously in~\cite{ASTT21}. In Section~\ref{NSlimitsec}, we write down the analogous statements in the Nekrasov--Shatashvili (NS) limit. We will see in Section~\ref{sec:Borelsums} that the NS difference equation~\eqref{tdiffeqNS} corresponding to integer shifts in the closed string modulus encodes the Stokes phenomena of the Borel summations. In Section~\ref{sec:relfreeenergies}, we establish a difference equation relating the refined topological free energy to the NS free energy, and in Section~\ref{sec:qcclosedm}, we interpret this difference equation as the quantization of an algebraic curve in the closed string modulus.

In Section \ref{sec:Borelsums}, we complete a systematic study of the Borel summation of the NS free energy of the resolved conifold, as well as the formal WKB solution to the Schr\"odinger operator $\mathbf{D}$ from Section \ref{sec:qmirrorcurve}, and establish a relation between the two. Similar to \cite{ASTT21}, we prove in Theorem~\ref{maintheorem} that the Borel transform of the NS free energy has infinitely many singularities organized along rays that may be labeled by the charges of 5d BPS states, and that accumulate along the imaginary axis in the $\epsilon$-plane. We find that the non-perturbative solution of the finite difference equation found in Section~\ref{sec:diffeq} agrees with the Borel summation along the ray $\mathbb{R}_{>0}$ in the $\epsilon$-plane. We compute the Stokes jumps along these rays in equation~\eqref{eq:StokesjumpsW} and the limit of the Borel summed free energy along the imaginary axis in equation~\eqref{eq:limitImW}. We find that this last limit agrees with the Gopakumar--Vafa formulation of the NS free energy.
We complete a similar analysis in Theorem~\ref{th:BorelS} for the Borel sums of the formal solution of the Schr\"odinger operator~$\mathbf{D}$. In particular, we find that the two non-perturbative solutions found in Section~\ref{sec:qmirrorcurve} correspond to the Borel sums along $\mathrm{i}\mathbb{R}_{>0}$ and $\mathbb{R}_{>0}$. We finish the section with Theorem~\ref{theorem3}, relating the Borel sums of the free energy with the Borel sums of the formal solution of the Schr\"odinger operator.\looseness=1

In Section~\ref{sec:DT}, we extract the BPS structure from the previous analysis. We furthermore use this structure to formulate the Borel summation results geometrically as defining a~section of a~certain holomorphic line bundle $\mathcal{L}\to \mathbb{C}^{\times}\times M$, where $M$ denotes the K\"{a}hler parameter space, and $\mathbb{C}^{\times}$ the parameter space for $\epsilon$. We furthermore show that $\mathcal{L}$ is actually the same line bundle as defined by $\exp\big(F^{{\rm top}}_{\rho}\big)$ and studied in~\cite{ASTT21}.

In Section~\ref{5dWKBsec}, we interpret the Borel analysis of the previous sections in terms of 5d gauge theory and quantum integrable systems. In Section~\ref{sec:spectralnetwork}, we study the spectral networks (or Stokes graphs) corresponding to the Schr\"odinger operator $\mathbf{D}$ for the resolved conifold geometry in the spirit of the exact WKB analysis, and identify the above 5d BPS states as degenerate trajectories. We identify two distinguished spectral networks and define the associated spectral coordinates (on the moduli space of the corresponding multiplicative Hitchin system) in the language of~\cite{Hollands:2013qza}. In Section~\ref{sec:NRS}, we interpret the relation between the Borel sum of the formal solution to the Schr\"odinger operator $\mathbf{D}$ and the Borel sum of the NS free energy as a non-perturbative completion of the relation between the open and closed topological string found in~\cite{ACDKV}. We furthermore interpret this relation along the ray $\rho = \mathrm{i} \mathbb{R}_{>0}$ as a lift of the Nekrasov--Rosly--Shatashvili conjecture~\cite{Nekrasov:NRS} to five dimensions. That is, we argue that the NS free energy is equal to the generating function of the space of $q$-difference opers in the relevant spectral coordinates. We conclude this section with a description of two associated spectral problems.\looseness=1

Finally, we speculate about a possible generalization of our results to mirror curves of higher genus in the discussion Section~\ref{sec:discussion}.

\section{Topological string free energies}\label{sec:freeenergies}

Let $\mathbf{X}$ and $\check{\mathbf{X}}$ be a mirror pair of Calabi--Yau threefolds with $h^{1,1}(\mathbf{X})=h^{2,1}(\check{\mathbf{X}})$ and $h^{2,1}(\mathbf{X})=h^{1,1}(\check{\mathbf{X}})$. Consider a distinguished set $\mathbf{t}=\big(t^1,\dots,t^n\big)$ of local coordinates on the moduli space $\mathcal{M}_\mathrm{symp}(\mathbf{X})$ of symplectic structures on~$\mathbf{X}$. These coordinates are related by the mirror map~$\mathbf{t}(z)$ to a set of complex structure coordinates~$z$ on the isomorphic moduli space $\mathcal{M}_\mathrm{comp}(\check{\mathbf{X}})$ of complex structures on $\check{\mathbf{X}}$.

Then topological string theory associates a topological string partition function to the mirror family $\mathbf{X}_\mathbf{t}$ and $\check{\mathbf{X}}_{\mathbf{t}(z)}$, which is defined as an asymptotic series in the topological string coupling~$\lambda$ and sums over the free energies associated to the world-sheets of genus~$g$:
\begin{equation*}
Z^{\mathrm{top}} (\lambda,\mathbf{t})= \exp \bigg(\sum_{g=0}^{\infty} \lambda^{2g-2} \mathcal{F}_{g}(\mathbf{t})\bigg).
\end{equation*}
This topological string partition function may either be computed in terms of the topological A-model on $\mathbf{X}_\mathbf{t}$, or equivalently in terms of the topological B-model on $\check{\mathbf{X}}_{\mathbf{t}(z)}$, which are related by mirror symmetry.

\subsection{GW and GV expansions}
 In an expansion around a distinguished ``large volume" point in the moduli space $\mathcal{M}_\mathrm{symp}(\mathbf{X})$, the topological string free energies become the generating functions of higher genus Gromov--Witten (GW) invariants on the A-model side of mirror symmetry. The GW potential of $\mathbf{X}$ is the formal power series
\begin{equation*}
F^\mathrm{top}(\lambda,\mathbf{t}) = \sum_{g\ge 0} \lambda^{2g-2} F_g(\mathbf{t})= \sum_{g\ge 0} \lambda^{2g-2} \sum_{\mathbf{\beta}\in H_2(X,\mathbb{Z})} N^g_{\mathbf{\beta}} \mathbf{Q}^{\mathbf{\beta}},
\end{equation*}
where $\mathbf{Q}^{\mathbf{\beta}} := \exp (2\pi \mathrm{i} \langle \mathbf{t},\mathbf{\beta} \rangle)$ is a formal variable living in a suitable completion of the effective cone in the group ring of $H_2(\mathbf{X},\mathbb{Z})$.

The GW potential can be furthermore written as
\begin{equation*}
F^\mathrm{top}=F^\mathrm{top}_{\mathbf{\beta}=0} + \tilde{F}^\mathrm{top},
\end{equation*}
where $F_{\beta=0}$ denotes the contribution from constant maps and $ \tilde{F}$ the contribution from non-constant maps. The constant map contribution at genus 0 and 1 are $\mathbf{t}$-dependent and the higher genus constant map contributions take the universal form \cite{Faber}
\begin{equation*}
F^\mathrm{top}_{g,\beta=0} = \frac{\chi(\mathbf{X})(-1)^{g-1} B_{2g} B_{2g-2}}{4g (2g-2) (2g-2)!},
\end{equation*}
for $g\ge2$, where $\chi(\mathbf{X})$ is the Euler characteristic of $\mathbf{X}$ and the Bernoulli numbers $B_n$ are generated by
\begin{equation*}
\frac{w}{{\rm e}^w-1} = \sum_{n=0}^{\infty} B_n \frac{w^n}{n!}.
\end{equation*}

The non-constant part of the GW potential may be resummed in terms of the Gopakumar--Vafa invariants $n_g^{\beta} \in \mathbb{Z}$ that ``count" electrically charged $M2$ branes wrapping two-cycles of degree $\beta$ in an M-theory compactification on the Calabi--Yau $\mathbf{X}$ \cite{Gopakumar:1998ii,Gopakumar:1998jq}. The GW potential in GV form reads
\begin{equation*}
\tilde{F}^\mathrm{top}(\lambda,\mathbf{t})= \sum_{\beta>0}\sum_{g\ge 0} n_g^{\beta} \sum_{k\ge 1} \frac{1}{k} \left( 2 \sin \left( \frac{k\lambda}{2}\right)\right)^{2g-2} \mathbf{Q}^{k\mathbf{\beta}},
\end{equation*}
with in particular:
\begin{equation*}
\tilde{F}_0(\mathbf{t})=\sum_{\mathbf{\beta}>0} n^{\beta}_0 {\rm Li}_3\big(\mathbf{Q}^{\mathbf{\beta}}\big).
\end{equation*}

\subsection{Refinement of topological string theory}

The refinement of topological string theory is motivated by Nekrasov's instanton calculation in the $\Omega$-background \cite{Nekrasov:2002qd}. A topological string interpretation of the refinement can be given in terms of the refined Gopakumar--Vafa invariants $N_{j_L,j_R}^{\beta}$. These GV invariants compute an index of five-dimensional BPS states with charge $\beta$ and spin $(j_L,j_R)$ with respect to the rotation group ${\rm SU}(2)_L\times {\rm SU}(2)_R$ of $\mathbb{R}^4$ \cite{IKV}. This index is an invariant of non-compact CY manifolds and can be used to define refined topological string theory.

The refined topological string free energy is an expression in the moduli $\mathbf{t}$ as well as the two parameters~$\epsilon_1$,~$\epsilon_2$. It can be written as (see, e.g.,~\cite{Hatsuda:2013oxa, Huang:2010kf})
\begin{equation*}
F^{{\rm ref}}(\epsilon_1,\epsilon_2,\mathbf{t})= \sum_{\beta} \sum_{j_L,j_R\ge 0} N_{j_L,j_R}^{\beta} \sum_{k\ge 1} \frac{1}{k} \frac{\chi_{j_L}\big(q^k_L\big) \chi_{j_R}\big(q^k_R\big) }{\big(q_1^{k/2}-q_1^{-k/2}\big)\big(q_2^{k/2}-q_2^{-k/2}\big)} \mathbf{Q}^{k\beta},
\end{equation*}
in terms of the ${\rm SU}(2)$ characters
\begin{equation*}
\chi_j(q) =\frac{q^{2j-1}-q^{-2j-1}}{q-q^{-1}},
\end{equation*}
where $q_j = {\rm e}^{\mathrm{i}\epsilon_j}$ for $j=1,2,L,R$ and
\begin{equation*}
 \epsilon_{L} = \frac{\epsilon_1-\epsilon_2}{2}, \qquad \epsilon_R= \frac{\epsilon_1+\epsilon_2}{2}.
 \end{equation*}

The unrefined topological string may be obtained from the refined one by setting
\begin{equation*}
 \epsilon_1=-\epsilon_2= \lambda.
 \end{equation*}
Another limit of refined topological string theory, which was put forward in~\cite{NS}, is given by sending one of the parameters $\epsilon_1$, $\epsilon_2$ to zero while the other one is kept finite. E.g.,
\begin{equation*}
 \epsilon_1=\epsilon,\qquad \epsilon_2\rightarrow 0.
 \end{equation*}
Since the refined topological string free energy has a simple pole in this limit, the free energy in this Nekrasov--Shatashvili (NS) limit is defined as
\begin{equation*}
F^{\mathrm{NS}} (\epsilon,{\bf{Q}}) := \lim_{\epsilon_2\rightarrow 0} \epsilon_2 F^{\rm ref}(\epsilon_1,\epsilon_2,{\bf{Q}}) |_{\epsilon_1=\epsilon}.
\end{equation*}


\subsection{Topological string free energies of the resolved conifold}\label{sec:freeenergyrescon}

The CY threefold given by the total space of the rank two bundle
\begin{equation*}
\mathbf{X} := \mathcal{O}(-1) \oplus \mathcal{O}(-1) \rightarrow \mathbb{P}^1,
\end{equation*}
over the projective line
corresponds to the resolution of the conifold singularity in $\mathbb{C}^4$ and is known as the resolved conifold. This geometry is defined on the A-model side of mirror symmetry and $t$ corresponds to
\begin{equation*}
t= \int_{C} B+ \mathrm{i} \omega,
\end{equation*}
where $B\in H^{2}(\mathbf{X},\mathbb{R})/H^{2}(\mathbf{X},\mathbb{Z})$ is the B-field, while $\omega$ is the K\"ahler form and $C$ corresponds to the $\mathbb{P}^1$ class in this example. The GW potential for this geometry was determined, in physics \cite{Gopakumar:1998ii,GV} as well as in mathematics \cite{Faber}, with the outcome\footnote{See also \cite{MM} for the determination of $F^g$ from a string theory duality and the explicit appearance of the polylogarithm expressions.}
\begin{equation}\label{resconfree}
\tilde{F}^\mathrm{top}(\lambda,t)= \sum_{g=0}^{\infty} \lambda^{2g-2} \tilde{F}_g(t)= \frac{1}{\lambda^2} {\rm Li}_{3}(Q)+ \sum_{g=1}^{\infty} \lambda^{2g-2} \frac{(-1)^{g-1}B_{2g}}{2g (2g-2)!} {\rm Li}_{3-2g} (Q)
\end{equation}
for the non-constant maps, where $Q=\exp(2\pi \mathrm{i} t)$ and the polylogarithm is defined by
\begin{equation*}
{\rm Li}_s(z) = \sum_{n=0}^{\infty} \frac{z^n}{n^s}
\end{equation*}
for $s\in \mathbb{C}$ and $|z|<1$.

The refined topological string free energy for the resolved conifold geometry is given by~\cite{IKV}
\begin{equation*}
\begin{split}
F^{{\rm ref}}(\epsilon_1,\epsilon_2,t) &=
 -\sum_{k=1}^{\infty} \frac{Q^k}{k(2 \sin (k\epsilon_1/2)) (2\sin (k\epsilon_2/2))}.
\end{split}
\end{equation*}
From this expression it is easy to find formal expansions in $\epsilon_1$ and $\epsilon_2$ provided we consider shifts of the form $F^{\rm ref}\big(\epsilon_1,\epsilon_2,t\pm \frac{\check\epsilon_1}{2} \pm \frac{\check\epsilon_2}{2}\big)$, where $\check\epsilon_i=\epsilon_i/2\pi$. Indeed, using the expression for the generating function of the Bernoulli numbers
\begin{equation}\label{eq:Bernoulligen}
\frac{w}{{\rm e}^w-1} = \sum_{n=0}^{\infty} B_n \frac{w^n}{n!},
\end{equation}
we can write the expansion of, for example, $F^{\rm ref}\big(\epsilon_1,\epsilon_2,t- \frac{\check\epsilon_1+\check\epsilon_2}{2}\big)$ as
\begin{align}
 F^{\rm ref}\left(\epsilon_1,\epsilon_2,t- \frac{\check\epsilon_1+\check\epsilon_2}{2}\right)&= -\sum_{k=1}^{\infty} \frac{{\rm e}^{2\pi \mathrm{i} k (t-\frac{\check\epsilon_1+\check\epsilon_2}{2} )}}{k(2 \sin (k\epsilon_1/2)) (2\sin (k\epsilon_2/2))} \nonumber\\
 &=\sum_{k=1}^{\infty}\frac{Q^k}{k({\rm e}^{\mathrm{i} k\epsilon_1} -1)({\rm e}^{\mathrm{i} k\epsilon_2} -1)}\nonumber\\
 &= \sum_{k=1}^{\infty} \sum_{m,n=0}^{\infty}\frac{Q^k}{k^{3-m-n}} \frac{B_n B_m}{n! m!} \mathrm{i}^{m+n-2} \epsilon_1^{m-1} \epsilon_2^{n-1} \nonumber\\
 &=-\sum_{m,n=0}^{\infty} {\rm Li}_{3-m-n}(Q) \frac{B_m B_n}{m! n!} \mathrm{i}^{m+n} \epsilon_1^{m-1} \epsilon_2^{n-1}.\label{refshifexp1}
\end{align}
A similar computation shows that
\begin{align}
 F^{\rm ref}\left(\epsilon_1,\epsilon_2,t+ \frac{\check\epsilon_1+\check\epsilon_2}{2}\right)&=\sum_{m,n=0}^{\infty} {\rm Li}_{3-m-n}(Q) \frac{B_m B_n}{m! n!} (-1)^{m+n-1}\mathrm{i}^{m+n} \epsilon_1^{m-1} \epsilon_2^{n-1},\nonumber\\
 F^{\rm ref}\left(\epsilon_1,\epsilon_2,t+ \frac{\check\epsilon_1-\check\epsilon_2}{2} \right)&=\sum_{m,n=0}^{\infty} {\rm Li}_{3-m-n}(Q) \frac{B_m B_n}{m! n!} (-1)^{m-1}\mathrm{i}^{m+n} \epsilon_1^{m-1} \epsilon_2^{n-1},\nonumber\\
 F^{\rm ref}\left(\epsilon_1,\epsilon_2,t- \frac{\check\epsilon_1-\check\epsilon_2}{2}\right)&=\sum_{m,n=0}^{\infty} {\rm Li}_{3-m-n}(Q) \frac{B_m B_n}{m! n!} (-1)^{n-1}\mathrm{i}^{m+n} \epsilon_1^{m-1} \epsilon_2^{n-1}.\label{refshifexp2}
 \end{align}

On the other hand, the NS limit of $F^{\rm ref}(\epsilon_1,\epsilon_2,t)$ in the case of the resolved conifold is given~by
\begin{align}
 F^{\rm NS}(\epsilon,t)&:= \lim_{\epsilon_2 \rightarrow 0} \epsilon_2\cdot F^{\rm ref}(\epsilon_1,\epsilon_2,t) \big{|}_{\epsilon_1=\epsilon}\nonumber\\
 &=\lim_{\epsilon_2 \rightarrow 0} -\epsilon_2 \sum_{k=1}^{\infty} \frac{Q^k}{k(2 \sin (k\epsilon_1/2)) (2\sin (k\epsilon_2/2))} \bigg{|}_{\epsilon_1=\epsilon}\nonumber\\
 &=- \frac{1}{2} \sum_{k=1}^{\infty} \frac{Q^k}{k^2 \sin(k\epsilon/2)}.\label{def:FNS}
\end{align}

As before with $F^{\rm ref}(\epsilon_1,\epsilon_2,t)$, one can easily find a formal expansion in $\epsilon$ for the shifts $F^{\rm NS}(\epsilon,t\pm \check\epsilon/2)$ using the generating function \eqref{eq:Bernoulligen}. That is,
\begin{align}
 F^{\rm NS}(\epsilon,t-\check{\epsilon}/2)&= -\mathrm{i} \sum_{n=0}^{\infty} {\rm Li}_{3-n}(Q) \frac{B_n}{n!} (\mathrm{i}\epsilon)^{n-1}=-\frac{1}{2\pi} \sum_{n=0}^{\infty} \partial_t^n {\rm Li}_{3}(Q) \frac{B_n}{n!} \check{\epsilon}^{n-1}\nonumber\\
 F^{\rm NS}(\epsilon,t+\check{\epsilon}/2)&=\mathrm{i}\sum_{n=0}^{\infty}\mathrm{Li}_{3-n}(Q)\frac{B_n}{n!}(-\mathrm{i}\epsilon)^{n-1}=\frac{1}{2\pi} \sum_{n=0}^{\infty} \partial_t^n {\rm Li}_{3}(Q) \frac{B_n}{n!} (-1)^{n-1}\check{\epsilon}^{n-1}.\label{eq:FNSasymp}
\end{align}

\subsection{Mirror geometry}\label{appendix:mirrorsymmetry}
Mirrors of non-compact CY threefolds are described in \cite{Chianglocal,Hori:2000kt, Katz:1997eq}, see also \cite{Hosonolocal}. We focus here on toric examples. The non-compact CY threefolds in these cases are given~by
\begin{equation*}
 \mathbf{X}= \frac{\mathbb{C}^{3+k}\setminus S}{(\mathbb{C}^*)^k},
\end{equation*}
where the $k$ algebraic tori $\mathbb{C}^*$ act on the space by
 \begin{equation*}
 (\mathbb{C}^*)^a \colon \ (z_1,\dots,z_j,\dots, z_{3+k}) \mapsto \big(\lambda^{l^{(a)}_1} z_1,\dots,\lambda^{l^{(a)}_j} z_j,\dots,\lambda^{l^{(a)}_{3+k}} z_{3+k}\big), \qquad a=1,\dots,k.
\end{equation*}
 Here, $\lambda \in \mathbb{C}^*$, whereas $l^{(a)}_{i} \in \mathbb{Z}$ are the toric charges and $S$ is a~subset which is fixed by a~subgroup of $(\mathbb{C}^*)^k$. The resolved conifold geometry corresponds to the toric variety associated to the toric charge vector
\begin{equation*}
l=(1,1,-1,-1).
\end{equation*}

To define the mirror of $\mathbf{X}$ we consider the variables $w_i \in \mathbb{C}^*$, for $i=0,\dots,3$, subject to the constraint
\begin{equation*}
 \prod_{i=0}^{3} w_i^{l_i}= \frac{w_0 w_1}{w_2 w_3}=1,
\end{equation*}
and the polynomial
\begin{equation*}
 P(a,w)=\sum_{i=0}^3 a_i w_i,
\end{equation*}
with $a_i \in \mathbb{C}$ for $i=0,\ldots,3$. This polynomial enters the definition of the Landau--Ginzburg potential of the mirror, which is given by
\begin{equation}\label{eq:mirrorLG}
 W= -v^+ v^- + P(a,w).
\end{equation}
The additional variables $v^+,v^- \in \mathbb{C}$ are an artifact of local mirror symmetry, see, e.g.,~\cite{Hori:2000kt}. There is a freedom to rescale $W$ by a nonzero complex number, which we can use to set one of the~$w_i$ variables to~1. Without loss of generality we set $w_0=1$.

The mirror $\check{\mathbf{X}}$ of the resolved conifold geometry is then defined as the variety
\begin{equation*}
 \check{\mathbf{X}}= \big\{ (v^+,v^-,w_1,w_2) \in \mathbb{C}^2\times (\mathbb{C}^*)^2 \mid v^+ v^- = a_0 + a_1 w_1 + a_2 w_2 + a_3 w_1 w_2^{-1} \big\}.
\end{equation*}
The $a_i$ are complex parameters which determine the complex structure of $\check{\mathbf{X}}$. The rescaling of~$W$ and $w_1$, $w_2$ can be further used to show that the complex structure of $\check{\mathbf{X}}$ only depends on the combination
\begin{equation}\label{Qa}
Q=\frac{a_2 a_3}{a_0 a_1}.
\end{equation}

The equation $W=0$ can be written in the form\footnote{The term $v^+ v^-$ has been rescaled here compared to \eqref{eq:mirrorLG}}
\begin{equation*}
 v^+ v^- = {\rm e}^{2 \pi \mathrm{i} x} + {\rm e}^{2 \pi \mathrm{i} y} + 1 + Q {\rm e}^{2 \pi \mathrm{i} (y-x)},
\end{equation*}
with $w_1=\exp(2 \pi \mathrm{i} x)$, $w_2=\exp(2 \pi \mathrm{i} y)$ and $Q = \exp( 2 \pi \mathrm{i} t)$.
The equation
\begin{equation}\label{eqn:mirrorcurve1}
\Sigma\colon \ {\rm e}^{2 \pi \mathrm{i} x} + {\rm e}^{2 \pi \mathrm{i} y} + 1 + Q {\rm e}^{2 \pi \mathrm{i} (y-x)} = 0,
\end{equation}
defines the so-called mirror curve $\Sigma$.

It turns out to be convenient to redefine the variables $x$ and $y$ by mapping $x \mapsto x+ t + 1/2$ and $y \mapsto y + x $, so that $\Sigma$ is now parametrized by the equation
\begin{equation}\label{eqn:mirrorcurve1p}
 \Sigma\colon \ (1-X ) Y - (1 - Q X ) = 0 \subset \mathbb{C}^*_X \times \mathbb{C}^*_Y,
\end{equation}
in terms of the $\mathbb{C}^*$-variables $X={\rm e}^{2 \pi \mathrm{i} x}$ and $Y={\rm e}^{2 \pi \mathrm{i} y}$. Topologically, $\Sigma$ is a four-punctured sphere, with punctures at
\begin{align*}
(X,Y) \in \big\{ (0, 1), (1, \infty), \big(Q^{-1}, 0\big), (\infty, Q) \big\}.
\end{align*}
That is, its punctures are located at the points where the curve intersects the lines $X=0,\infty$ or $Y=0,\infty$ (see Figure~\ref{fig:plotY}).

\begin{figure}[t]\centering
\includegraphics[scale=0.8]{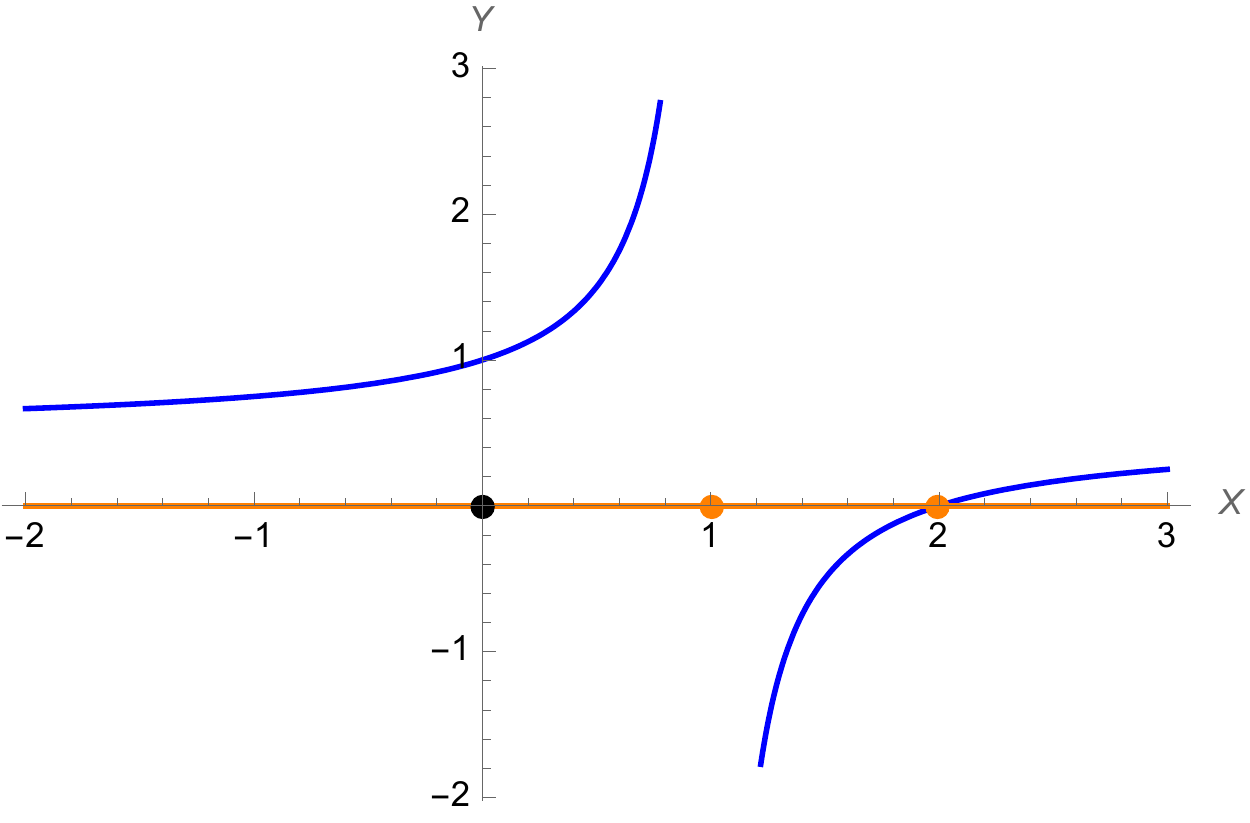}
\caption{Plot of a cross-section of the mirror curve $\Sigma$ (in blue), with $Q=1/2$, as well as its projection to $C = \mathbb{C}^*_X$ (in orange), and its logarithmic punctures at $X=1$ and $X=Q^{-1}$ (in orange) and regular punctures at $X=0$ and $X=\infty$ (in black).}\label{fig:plotY}
\end{figure}

The mirror curve $\Sigma$ comes equipped with the tautological 1-form
\begin{equation}\label{eq:SWform}
 \lambda^\mathrm{cl} = y \, {\rm d}x = - \frac{1}{4 \pi^2} \log Y \,{\rm d} \log X,
\end{equation}
which is a reduction from the holomorphic three-form on $\check{\mathbf{X}}$ to $\Sigma$. The genus zero free energy~$F^\mathrm{top}_0(\mathbf{t})$ at large radius is then simply determined by the relations
\begin{equation}\label{eqn:classical_periods}
t_i = \oint_{\gamma_{A^i}} \lambda^\mathrm{cl}, \qquad \frac{1}{2 \pi \mathrm{i}} \frac{\partial F^\mathrm{top}_0}{\partial t_i} = \frac{4 \pi^2}{\epsilon} \oint_{\gamma_{B_i}} \lambda^\mathrm{cl},
\end{equation}
where $\gamma_{A^i}$ and $\gamma_{B_i}$ are a suitable basis of 1-cycles on $\Sigma$, illustrated in Figure~\ref{fig:cyclesSigma} for the resolved conifold geometry.\footnote{Since the B-cycle $\gamma_B$ is non-compact in this example, the corresponding period requires a regularization. This will be discussed in Section \ref{sec:classicalPF}.}
Note that the periods of $\lambda^\mathrm{cl}$ are invariant under ${\rm SL}(2,\mathbb{Z})$-transformations acting on the variables $x$ and $y$. This is known as the framing ambiguity.

\begin{figure}[t]\centering
\includegraphics[scale=0.8]{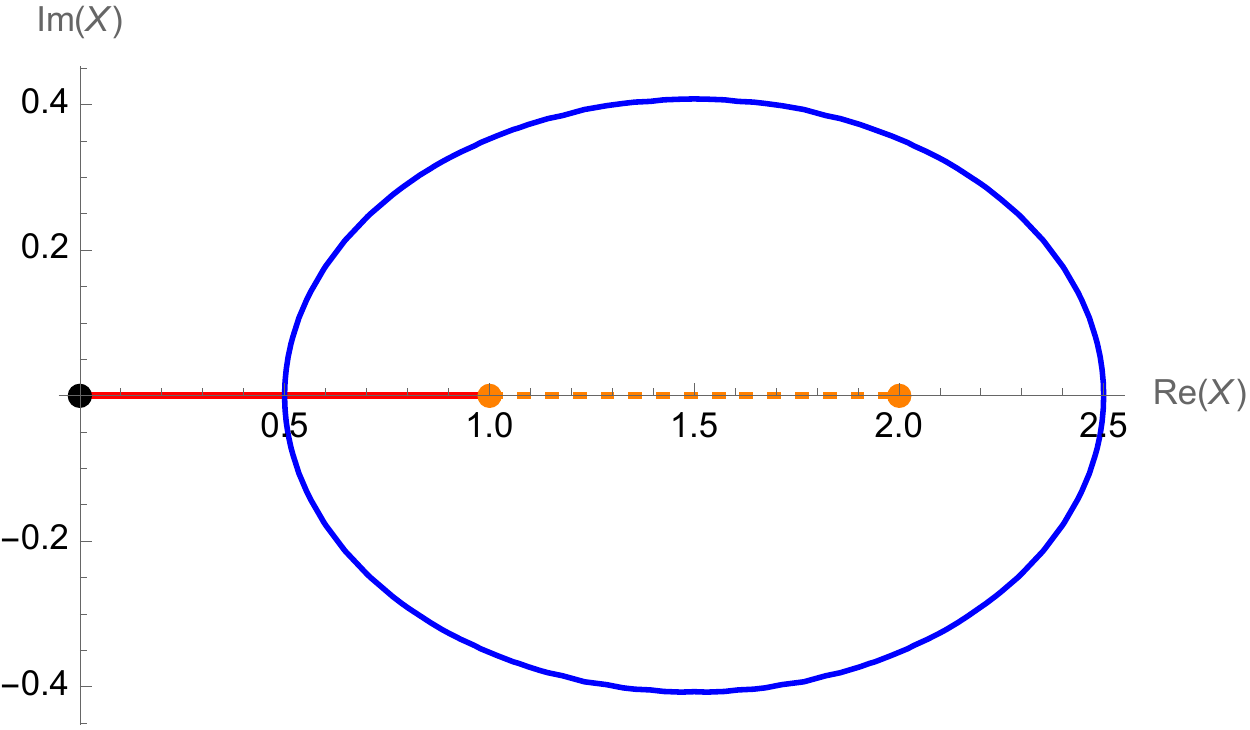}
\caption{Plot of 1-cycles $\gamma_A$ (in blue) and $\gamma_B$ (in red) on the mirror curve $\Sigma$ with $Q=1/2$. The 1-cycle~$\gamma_A$ is compact and oriented in the clockwise direction. The 1-cycle $\gamma_B$ runs from the puncture at $X=1$ to the puncture at $X=0$ and is non-compact.}\label{fig:cyclesSigma}
\end{figure}

Because of the logarithms in the definition~\eqref{eq:SWform} of $\lambda^\mathrm{cl}$, it is sometimes needed to consider a~$\mathbb{Z}$-covering $\widetilde{\Sigma}$ of $\Sigma$ that resolves the logarithmic ambiguities of $\lambda^\mathrm{cl}$. This covering has additional branching at the punctures at $X=1$ and $X=Q^{-1}$ where $\lambda^\mathrm{cl}$ is multi-valued. We choose a trivialization of the covering by connecting these ``logarithmic" punctures by a logarithmic branch-cut and labeling the sheets away from the cut by an integer $N$.

The $\mathbb{Z}$-covering $\widetilde{\Sigma}$ has an additional 1-cycle~$\gamma_0$ compared to~$\Sigma$. In contrast to the 1-cycles~$\gamma_A$ and~$\gamma_B$, the 1-cycle $\gamma_0$ crosses the logarithmic cut twice (with opposite orientations). This 1-cycle is for instance illustrated in~\cite[Figure~2]{BLR}.

\subsection{Picard--Fuchs equation and classical terms}\label{sec:classicalPF}
The periods of the classical differential $\lambda^{\text{cl}}$ can be computed explicitly as outlined in the previous subsection. Alternatively they can be determined as solutions of a differential equation, called the Picard--Fuchs equation. The latter stems from the fact that the differential $\lambda^{\text{cl}}$ descends from the holomorphic three-form on the mirror CY threefolds. The holomorphic three-form can be used to capture the variation of Hodge structure in the middle dimensional cohomology of the threefold. The flatness of the associated Gauss-Manin connection leads to differential equations which annihilate the periods of the three-form and correspondingly the periods of the meromorphic differential $\lambda^{\text{cl}}$ on the curve. To derive the Picard--Fuchs equation, rescalings of the defining equation of the mirror curve can be used which we review for the case of the mirror of the resolved conifold.

As mentioned in the previous section, the $a_i$ appearing in \eqref{eq:mirrorLG} are complex parameters which determine the complex structure of $\check{\mathbf{X}}$, and the rescaling of $W$ and $w_1$, $w_2$ can be used to show that the complex structure of $\check{\mathbf{X}}$ only depends on the combination~\eqref{Qa}.
Keeping the explicit dependence on the $a_i$ is however more convenient for the derivation of the Picard--Fuchs equations from a GKZ~\cite{GKZ}
system of differential equations annihilating periods of the unique holomorphic $(3,0)$ form on $\check{\mathbf{X}}$. The latter is given by
\begin{equation*}
 \Omega = \textrm{Res}_{W=0} \frac{1}{W} \,{\rm d}v^+ \,{\rm d}v^- \frac{{\rm d}w_1}{w_1} \frac{{\rm d}w_2}{w_2}.
\end{equation*}

The periods of $\Omega$ are annihilated by the GKZ operator
\begin{equation*}
 \frac{\partial}{\partial a_0} \frac{\partial}{\partial a_1}- \frac{\partial}{\partial a_2} \frac{\partial}{\partial a_3},
\end{equation*}
which translates into the Picard--Fuchs operator expressed in $Q$, namely
\begin{equation*}
 L= (1-Q)\theta^2, \qquad \theta=Q\frac{\mathrm{d}}{\mathrm{d}Q}.
\end{equation*}
This operator has the following solutions:
\begin{equation*}
 \varpi^0=1, \qquad \varpi^1=\frac{1}{2\pi \mathrm{i}}\log (Q).
\end{equation*}
These correspond to periods of $\Omega$ over appropriately defined compact three-cycles in $H_3(\check{\mathbf{X}},\mathbb{Z})$. The mirror map is identified as
\begin{equation*}
 t= \frac{1}{2\pi \mathrm{i}} \log (Q),
\end{equation*}
we note that the mirror of the resolved conifold is a special case where the parameter appearing in the definition of the curve equation $Q$ is given by $\exp(2\pi \mathrm{i} t )$ without further corrections, so that $Q$ is the exponentiated mirror map.

A familiar phenomenon of mirror symmetry for local CY is that the Picard--Fuchs system of the mirror does not have enough solutions to recover the expected ingredients of a special K\"ahler geometry. Generically it is missing expressions for periods of non-compact three-cycles. One way to recover these is to carefully define non-compact three-cycles on the geometry as was done in \cite{Hosonolocal}. Alternatively one may extend the PF operators, guided by the expectation of its general form in compact CY, when it is formulated in terms of the distinguished coordinates corresponding to the mirror map. This was done in \cite{Forbes}, which we will outline here.

The guiding principle is the expected form of the PF operator in terms of the special (flat) coordinate $t$ in the case when the moduli space is complex one-dimensional. It is given by
\begin{equation*}
 L = \partial_t^2 C_{ttt}^{-1} \partial_t^2,
\end{equation*}
where $C_{ttt}:=\frac{\partial^3}{\partial t^3} F_0$, see, e.g.,~\cite{Ceresole:1992su,coxkatz}. This leads to the extended PF operator
\begin{equation}\label{eq:classicalextPF}
 L= \theta^2 (1-Q) \theta^2
\end{equation}
in the $Q$ coordinate \cite{Forbes}. This operator has the solutions
\begin{gather*}
\varpi^0 = 1,\\
\varpi^1 =\frac{1}{2\pi \mathrm{i}} \log (Q),\\
\varpi^2 =\frac{1}{(2\pi \mathrm{i})^2} \left( \frac{1}{2} (\log (Q))^2 + {\rm Li}_2(Q) \right),\\
\varpi^3 =\frac{1}{(2\pi \mathrm{i})^3}\left(- \frac{1}{6} (\log (Q))^3 -\log (Q) {\rm Li}_2(Q) +2 {\rm Li}_3(Q)\right).
\end{gather*}

We can identify the additional solutions with
\begin{equation*}
 \varpi^2=: F_t,\qquad \varpi^3= 2F_0 - t F_t,
\end{equation*}
where $F_t:=\partial_t F_0$ and where the prepotential $F_0$ reads
\begin{equation*}
F_0 =\frac{1}{(2\pi \mathrm{i})^3}\left( \frac{1}{3!} (\log Q)^3+ {\rm Li}_3(Q)\right).
\end{equation*}
This matches with the expected generating function of the GW invariants of the resolved conifold.

Note that the first term of the prepotential was not included in the discussion of the topological string free energies in the previous subsections and in particular in equation~\eqref{resconfree}. This term is known as the classical contribution to prepotential and should reflect the triple intersection number of three divisors in the CY threefold $\mathbf{X}$. These intersection numbers are however ill-defined geometrically for non-compact CY manifolds. The term at hand, and more generally the intersection numbers for any non-compact CY threefold, can be thought of as being obtained through a careful regularization of an appropriate decompactification limit, in which the non-compact CY is considered as a local geometry that is embedded in a compact geometry, see, e.g.,~\cite{Chianglocal}.

We will later in particular need the expression for the period for $F_t$, including the classical term, which is given by
\begin{equation}\label{eq:classicatFt}
 F_t= \frac{t^2}{2}+\frac{1}{(2\pi \mathrm{i})^2} {\rm Li}_2(Q).
\end{equation}

\subsection{Geometric engineering}\label{sec:geomeng}

Non-compact toric Calabi--Yau threefolds $\mathbf{X}$ (or more generally dot diagrams) are related to five-dimensional $\mathcal{N}=1$ field theories on $\mathbb{R}^4 \times S_R^1$ by geometric engineering \cite{Benini:2009gi, KKV}. In this correspondence the mirror curve $\Sigma$ gets interpreted as the Seiberg--Witten curve of the five-dimensional field theory. In most of this paper we have set $R=1$. If we want to re-introduce $R$ we roughly need to scale all $\mathbb{C}$-valued variables and parameters by a factor of~$1/R$.

The field theory may be interpreted as a theory of class~S once we choose a projection
$\Sigma \subset \mathbb{C}_X^* \times \mathbb{C}_Y^*$ to $C = \mathbb{C}_X^*$. For instance, the resolved conifold geometry with mirror curve~$\Sigma$ as in equation~\eqref{eqn:mirrorcurve1p}, projected to $\mathbb{C}^*_X$, engineers the five-dimensional pure ${\rm U}(1)$ theory on $\mathbb{R}^4 \times S_R^1$ with $Q = R^2 \Lambda^2$. On the other hand, the projection of $\Sigma$, as parametrized in \eqref{eqn:mirrorcurve1}, to $\mathbb{C}^*_Y$ geometrically engineers a massive five-dimensional ${\rm SU}(2)$ hypermultiplet with $Q = \exp(2 \pi \mathrm{i} m)$. In the limit $R \to 0$ the five-dimensional $\mathcal{N}=1$ field theory reduces to a four-dimensional $\mathcal{N}=2$ field theory whose Seiberg--Witten curve is given by the $R\to 0$ limit of $\Sigma$.

The refined partition function $Z^{\rm ref} = \exp \big(F^{\rm ref}\big)$ of the non-compact Calabi--Yau threefold $\mathbf{X}$ corresponds to the K-theoretic version of the Nekrasov instanton partition function of the five-dimensional field theory. For instance, the refined partition function for the resolved conifold equals\footnote{Remark that by using~\eqref{eq:productinvpow}, we can equivalently write $Z^{\rm ref}( \epsilon_1,\epsilon_2,t) = \prod_{k,l=1}^\infty \big(1- Q q_1^{k-1/2} q_2^{-l+1/2} \big)$. Going to the unrefined limit where $\epsilon_1=-\epsilon_2=\lambda$ sets $q=q_1=q_2^{-1}$, so that
 $Z^{\rm ref}(\lambda,-\lambda,t)=\prod_{k,l=1}\big(1-Q q^{k+l-1}\big)=\prod_{s=1}\big(1-Q q^s\big)^s $,
recovering the usual expression for the partition function of the resolved conifold found in \cite{IKV}.}
\begin{align}\label{eqn:one-loop}
 Z^{\rm ref}( \epsilon_1,\epsilon_2,t) = \exp F_\mathrm{GV}^{\rm ref}( \epsilon_1,\epsilon_2,t) = \prod_{k,l=1}^\infty \big(1- Q q_1^{k-1/2} q_2^{l-1/2} \big)^{-1},
\end{align}
(where we recall that $q_j={\rm e}^{\mathrm{i}\epsilon_j}$ for $j=1,2$)
and either computes the instanton partition function for the pure ${\rm U}(1)$ gauge theory or the massive hypermultiplet, which both only receive a 1-loop contribution.

Finally, there are various kinds of BPS states corresponding to the topological string set-up. In particular, in the context of Type IIA string theory one may wrap D0-branes around the 0-cycle mirror to the 1-cycle $\gamma_0 \subset \widetilde{\Sigma}$, D2-branes around 2-cycles mirror to the 1-cycles $\gamma_{A^i} \subset \Sigma$, D4-branes around 4-cycles mirror to the 1-cycles $\gamma_{B_i} \subset \Sigma$ and D6-branes on $X$. These may be interpreted as BPS states in the corresponding five-dimensional field theory through the technique of exponential spectral networks \cite{Banerjee:2018syt, Eager:2016yxd}. Particularly relevant for this paper is the study of exponential spectral networks for the resolved conifold geometry in~\cite{BLR}.

\section{WKB analysis of quantum mirror curves}\label{sec:qmirrorcurve}

In this section we study the WKB asymptotics of the quantum mirror curve of the resolved conifold, in preparation for a more advanced exact WKB analysis in the forthcoming sections. For general background in exact WKB analysis for differential operators we refer to \cite{NI, Takeireview}. These techniques were first applied to difference operators in \cite{DingleMorganI,DingleMorganII}. The WKB analysis in relation to quantum mirror curves was studied in many places including \cite{ACDKV,Kashani-Poor:2016edc}.
Our goal in this section is to derive an all-order expression for the asymptotic expansion of local solutions to $q$-difference equations corresponding to quantum mirror curves. This will be useful later in the study of the Borel resummation of these solutions, as well as in relation of these solutions to closed partition functions.

\subsection{The quantum mirror curve}\label{sec:quantummirrorcurve}

The period integral relations~(\ref{eqn:classical_periods})
\begin{equation*}
t_i = \oint_{\gamma_{A^i}} \lambda^\mathrm{cl}, \qquad \frac{1}{2 \pi \mathrm{i}} \frac{\partial F^\mathrm{top}_0}{\partial t_i} = \frac{4 \pi^2}{\epsilon} \oint_{\gamma_{B_i}} \lambda^\mathrm{cl},
\end{equation*}
with $\lambda^{\mathrm{cl}}=y(x) \mathrm{d}x $
reveal that the ``classical'' topological string has the structure of an classical integrable system, with spectral curve $\Sigma$ and Liouville form $\lambda^{\mathrm{cl}}$. This is a structure that is argued to persists to all orders~\cite{ADKMV}, perhaps most transparent in the NS limit~\cite{NS}.

The mirror curve is given by an algebraic equation in two $\mathbb{C}^*$ variables $X=\exp(2\pi \mathrm{i} x)$ and $Y=\exp(2\pi \mathrm{i} y)$, and has the form
\begin{equation*}
 P\big({\rm e}^{2\pi \mathrm{i} x},{\rm e}^{2\pi \mathrm{i} y},\mathbf{t}\big)=0,
 \end{equation*}
where $\mathbf{t}$ are the closed string moduli.

The quantum mirror curve is a quantization of the mirror curve obtained by promoting its variables to operators which act on states in a corresponding Hilbert space. In our case we will replace $x$ and $y$ with the operators
\begin{equation*}
\hat{x} \Psi(\epsilon,x,\mathbf{t}) = x \Psi (\epsilon,x,\mathbf{t}), \qquad \hat{y} \Psi(\epsilon,x,\mathbf{t}) = \frac{\epsilon}{4\pi^2\mathrm{i}} \partial_x \Psi (\epsilon,x,\mathbf{t}),
\end{equation*}
that obey the canonical quantization relation \begin{equation*}\left[\hat{x},\hat{y} \right]= \frac{\mathrm{i} \epsilon}{4\pi^2}.\end{equation*}
For instance, quantizing the curve $\Sigma$ defined in \eqref{eqn:mirrorcurve1p} by
\begin{equation*}
 (1-\exp(2\pi \mathrm{i} x) ) \exp(2\pi \mathrm{i} y) - (1 - Q \exp(2\pi \mathrm{i} x) ) = 0,
\end{equation*}
we obtain the Schr\"odinger operator
\begin{equation}\label{eqn:diffDX1}
 \mathbf{D}_\Sigma (\epsilon,x,t) = \big(1-{\rm e}^{2 \pi \mathrm{i} x}\big) {\rm e}^{\check{\epsilon} \partial_x} - \big(1 - Q {\rm e}^{2 \pi \mathrm{i} x}\big).
\end{equation}
acting on wave-functions $\Psi(\epsilon,x,t)$ in the relevant Hilbert space. Since
\begin{equation*}
 {\rm e}^{\check{\epsilon} \partial_x} \Psi(X) = \Psi(q X)
\end{equation*}
with $q = \exp (\mathrm{i} \epsilon)$, the Schr\"odinger operator $\mathbf{D}_\Sigma$ is also known as a $q$-difference operator.

Note that there is no unique difference operator $\mathbf{D}_{\Sigma}$ associated to $\Sigma$, because of the ${\rm SL}(2,\mathbb{Z})$ symmetry mentioned above. Up to some ambiguities in quantization, that may be absorbed in the redefinition of some of the parameters, there is a unique difference operator~$\mathbf{D}_{\Sigma}$ for~$\Sigma$ together with a choice of projection to~$\mathbb{C}^*$. As we explained in Section~\ref{sec:geomeng}, this choice of projection is relevant in determining which five-dimensional gauge theory the Calabi--Yau threefold $\mathbf{X}$ engineers. For instance, the projection $\Sigma \to \mathbb{C}^*_{X}$, leading to~(\ref{eqn:diffDX1}), geometrically engineers the pure ${\rm U}(1)$ theory in five dimensions.

Consider a B-brane extended in the $v^+$-direction of the mirror $\check{\mathbf{X}}$, while at $v^-=0$, and concentrated at a point $x \in \Sigma$. If $y(x)$ is the local function determined by the equation $(\ref{eqn:mirrorcurve1p})$ for $\Sigma$, then it is known that the classical vev of this brane is given by~\cite{Aganagic:2000gs}
\begin{equation*}
 \Psi^{\text{cl}} (\epsilon,x,t) = \exp \left( \frac{4 \pi^2 \mathrm{i} }{\epsilon}\int^z \lambda^{\text{cl}} \right),
\end{equation*}
whereas its quantum vev satisfies the time-independent Schr\"odinger equation \cite{ADKMV}
\begin{equation}\label{eqn:NS_Diff_Eqn}
 \mathbf{D}_\Sigma(\epsilon,x,t) \Psi(\epsilon,x,t) = 0.
\end{equation}
Hence, while the curve $\Sigma$ has the physical interpretation as the moduli space of a particular class of branes in the open topological string theory, the solutions $\Psi(\epsilon,x,t)$ to the difference equation~(\ref{eqn:NS_Diff_Eqn})
have the physical interpretation as vevs for these branes \cite{ACDKV, NS}. In the corresponding five-dimensional gauge theory these have the interpretation as surface defect vevs~\cite{Dimofte:2010tz}. For example, the brane corresponding to~(\ref{eqn:diffDX1}) is ending on the compact toric leg of the toric diagram of the resolved conifold. In the corresponding ${\rm U}(1)$ theory this engineers a simple abelian surface defect~\cite{Dimofte:2010tz}.

It is also known that the NS partition function $Z^{\mathrm{NS}}(\epsilon,t)$, in a series expansion in $\epsilon$, may be recovered from an all-order WKB analysis of $\mathbf{D}_\Sigma(\epsilon,x,t)$ \cite{ACDKV, NS}. If we write the general solution $\Psi(\epsilon,x,t)$ to the difference equation~(\ref{eqn:NS_Diff_Eqn}) in the form
\begin{equation*}
\Psi(\epsilon,x,t) = \exp \left( \frac{4 \pi^2 \mathrm{i} }{ \epsilon} \int^x \mathrm{d}z \, \lambda^{\mathrm{qu}}(\epsilon,z,t) \right),
\end{equation*}
then it follows that the quantum Liouville form $\lambda^{\mathrm{qu}}(\epsilon,x,t)$ has an $\epsilon$-expansion
\begin{equation*}
\lambda^{\mathrm{qu}}(\epsilon) = \lambda^\mathrm{cl} + \sum_{n=1}^{\infty} \lambda^{\mathrm{qu}}_n \epsilon^n,
\end{equation*}
starting with the classical Liouville form $\lambda^\mathrm{cl} = y \,\mathrm{d} x$. The all-order-in-$\epsilon$-expansion of the NS free energy $F^{\mathrm{NS}}(\epsilon,t)$ may then be obtained from the relations
\begin{equation}\label{eqn:quantum_periods}
t^{\mathrm{qu}}_i := \oint_{\gamma_{A^i}} \lambda^{\mathrm{qu}}, \qquad \frac{1}{2 \pi \mathrm{i}} \frac{\partial F^{\rm NS}\big(t^{\mathrm{qu}}_i;\epsilon\big)}{\partial t^{\mathrm{qu}}_i} = \frac{4 \pi^2}{\epsilon} \oint_{\gamma_{B_i}} \lambda^{\mathrm{qu}}.
\end{equation}

Note that, somewhat confusingly, the relations on the left-hand side of (\ref{eqn:quantum_periods}) define new parameters $t^{\mathrm{qu}}_i$ such that the function $F^{\rm NS}\big(t^{\mathrm{qu}}_i,\epsilon\big)$, as defined by the relations on the right-hand side of (\ref{eqn:quantum_periods}), equals the NS free energy $F^{\rm NS}(\epsilon,t)$ when simply replacing the $t^{\mathrm{qu}}_i$ by the symbols~$t_i$. This was first verified in the four-dimensional pure~${\rm SU}(2)$ gauge theory in~\cite{Mironov:2009uv}.

\subsection{All-order WKB solution}
In the following we will apply the WKB analysis to the solution of the quantum curve.
The outcome of the analysis will provide the asymptotic series whose non-perturbative completion will be studied in Section \ref{sec:Borelsums} using Borel analysis.\footnote{We remark that the all-order WKB solution has also appeared in the recent work \cite[Section 4]{GHN}.}

We will derive a formal power series solution, to all orders in $\epsilon$, of the Schr\"odinger equation~\eqref{eqn:NS_Diff_Eqn}. We could do this slightly more generally, namely for the cases in which the Schr\"odinger operator is a $q$-difference operator of the form
\begin{equation}\label{eqn:generalD}
\mathbf{D}_\Sigma (\epsilon,x,t) = \prod_{j=1}^{K} (\exp(\check{\epsilon}\partial_x)-\exp(2\pi \mathrm{i} y_j(x,t)) ).
\end{equation}
Note that, in general, the $q$-difference operator $\mathbf{D}_\Sigma$ requires a specification of an operator ordering.\footnote{Which choice of operator ordering one chooses is not very relevant, since the operator ordering ambiguities may be absorbed in a redefinition of parameters.} Below we specialize to the case $K=1$ where no such ordering ambiguity occurs.

The operator $\mathbf{D}_\Sigma$ in equation \eqref{eqn:generalD} may be obtained from a general mirror curve $\Sigma$ given by an equation
\begin{equation*} P\big({\rm e}^{2\pi \mathrm{i} x},{\rm e}^{2\pi \mathrm{i} y},t\big)=0\end{equation*}
of degree $K$ in $Y={\rm e}^{2\pi \mathrm{i} y}$, and thus defining a (ramified) covering of degree $K$ over $\mathbb{C}^*_X$. Indeed, we obtain the difference operator \eqref{eqn:generalD} after expressing $y$ in terms of $x$ and $t$ and then writing the equation $P$ in the form
\begin{equation*}
 P\big({\rm e}^{2\pi \mathrm{i} x},{\rm e}^{2\pi \mathrm{i} y},t\big)= \prod_{k=1}^K (\exp(y)-\exp(2\pi \mathrm{i} y_j(x,t))),
\end{equation*}
where $K$ is the degree of the covering and $y_j(x,t)$ the solutions to the curve equation.

In the following proposition will restrict ourselves to $K=1$. This will however also be the building block of the solution in the more general case.
\begin{Proposition}\label{prop:Sasymp}
A formal power series solution of the equation
\begin{equation}\label{eq:openwaveeq}
 (\exp(\check{\epsilon}\partial_x)-\exp(2\pi \mathrm{i} y(x,t))) \Psi(\epsilon,x,t)=0
\end{equation}
is given by the formal power series
\begin{equation*}
 \Psi(\epsilon,x,t)= \exp S(\epsilon,x,t) \qquad \mathrm{with} \quad S(\epsilon,x,t)=\sum_{n=0}^{\infty} \epsilon^{n-1} S_n(x,t),
\end{equation*}
where
\begin{equation}\label{eqn:Sformal}
 S(\epsilon,x,t)= \frac{4\pi^2 \mathrm{i}}{\epsilon} \int^x y(z,t) \,{\rm d}z + 2\pi \mathrm{i} \sum_{n=1}^{\infty} \frac{B_n}{n!} \check{\epsilon}^{n-1} \partial_x^{n-1} y(x,t),
\end{equation}
and as before $\check{\epsilon}=\epsilon/2\pi$.
\end{Proposition}

\begin{proof}
We act on \eqref{eq:openwaveeq} on the left with $\partial_x$ to obtain
\begin{gather*}
 (\exp(\check{\epsilon}\partial_x)-\exp(2\pi \mathrm{i} y(x,t)))\cdot \partial_x S\cdot \Psi - 2\pi \mathrm{i} \partial_x y(x,t)\cdot \exp(2\pi \mathrm{i} y(x,t)) \cdot \Psi =0.
\end{gather*}
Then we use that
\begin{equation*}
\exp(\check{\epsilon}\partial_x) (\partial_x S\cdot \Psi)= (\exp(\check{\epsilon}\partial_x) \partial_x S) \cdot (\exp(\check{\epsilon}\partial_x)\Psi),
\end{equation*}
and the original equation \eqref{eq:openwaveeq} to obtain
\begin{equation}\label{eq:dSdiff}
 \big({\rm e}^{\check{\epsilon}\partial_x}-1\big) \partial_x S = 2\pi \mathrm{i} \partial_x y(x,t),
\end{equation}
assuming that $\exp(2\pi \mathrm{i} y(x,t)) \cdot \Psi \ne 0$.

From the generating function
\begin{equation*}
\frac{w}{{\rm e}^w-1}=\sum_{n=0}^{\infty}\frac{B_n}{n!} w^n
\end{equation*}
of Bernoulli numbers, after rearranging and replacing $w\rightarrow \check{\epsilon} \partial_x$, we obtain the identity
\begin{equation*}
 \big({\rm e}^{\check{\epsilon}\partial_x}-1\big) \sum_{n=0}^{\infty}\frac{B_n}{n!} (\check{\epsilon}\partial_x)^n y(x,t)= \check{\epsilon} \partial_x y(x,t),
\end{equation*}
when acting on $y(x,t)$ with both sides. Comparing this equation with \eqref{eq:dSdiff} yields
\begin{equation*}
 \partial_x S= \frac{2\pi \mathrm{i}}{\check{\epsilon}} \sum_{n=0}^{\infty} \frac{B_n}{n!}\check{\epsilon}^n \partial_x^n y(x,t),
\end{equation*}
which can be integrated to
\begin{equation*}
S(\epsilon,x,t)= \frac{4\pi^2 \mathrm{i}}{\epsilon} \int^x y(z,t)\,{\rm d}z + 2\pi \mathrm{i} \sum_{n=1}^{\infty} \frac{B_n}{n!} \check{\epsilon}^{n-1} \partial_x^{n-1} y(x,t). \tag*{\qed}
\end{equation*}
\renewcommand{\qed}{}
\end{proof}

\begin{Corollary}\label{col:formalS}
The formal power series solution for the Schr\"odinger equation
\begin{equation}\label{scheq}
 \mathbf{D}_\Sigma (\epsilon,x,t) \Psi(\epsilon,x,t)= \big(\big(1-{\rm e}^{2 \pi \mathrm{i} x}\big) {\rm e}^{\check{\epsilon} \partial_x} - \big(1 - Q {\rm e}^{2 \pi \mathrm{i} x}\big)\big)\Psi(\epsilon,x,t)=0
\end{equation}
is given by
$\Psi(\epsilon,x,t)=\exp(S(\epsilon,x,t)) $
with
\begin{equation}\label{eq:Sasymp}
 S(\epsilon,x,t)= - \sum_{n=0}^{\infty} \frac{B_n}{n!} (\mathrm{i}\epsilon)^{n-1} \left( {\rm Li}_{2-n}(Q X)-{\rm Li}_{2-n}(X)\right)
\end{equation}
and $X=\exp(2\pi \mathrm{i} x)$. In other words,
\begin{equation*}
\lambda^{\text{qu}}(\epsilon,x,t)=\left(- \frac{1}{2\pi \mathrm{i}}\sum_{n=0}^{\infty} \frac{B_n}{n!} (\mathrm{i}\epsilon)^{n} \left( {\rm Li}_{1-n}(Q X)-{\rm Li}_{1-n}(X)\right) \right) \mathrm{d}x.
\end{equation*}
\end{Corollary}
\begin{proof}
The operator \eqref{eqn:diffDX1} corresponds to
\begin{equation*} y(x,t)= \frac{1}{2\pi \mathrm{i}} \log \left( \frac{1-Q{\rm e}^{2\pi \mathrm{i} x}}{1-{\rm e}^{2\pi \mathrm{i} x}}\right).\end{equation*}
Using the statement of Proposition \ref{prop:Sasymp} as well as the formulae
\begin{equation*}{\rm Li}_1(z)=-\log(1-z),\qquad {\rm Li}_{s-1}\big({\rm e}^{2 \pi \mathrm{i} x}\big) = \frac{1}{2 \pi \mathrm{i}} \partial_x {\rm Li}_{s}({\rm e}^{2 \pi \mathrm{i} x}),
\end{equation*}
we find
\begin{equation*}
 S= \frac{1}{\check{\epsilon}} \int^x \log \left( \frac{1-Q{\rm e}^{2\pi \mathrm{i} z}}{1-{\rm e}^{2\pi \mathrm{i} z}}\right) {\rm d}z - \sum_{n=1}^{\infty} \frac{B_n}{n!} (\mathrm{i}\epsilon)^{n-1} \left( {\rm Li}_{2-n}(Q X)-{\rm Li}_{2-n}(X)\right).
\end{equation*}
We obtain the statement of the corollary by further integrating the first term, and obtain the expression for $\lambda^{\text{qu}}$ by noting that
\begin{equation*}
 \lambda^{\text{qu}}(\epsilon,x,t)= \frac{\epsilon}{4\pi^2 \mathrm{i}} \partial_x S(\epsilon,x,t) \, \mathrm{d}x.\tag*{\qed}
\end{equation*}
\renewcommand{\qed}{}
\end{proof}

\subsection{Exact solutions}\label{sec:psinp}

Exact solutions to difference operators corresponding to quantum mirror curves have been found in many places, using multiple methods. For instance, from the perspective of 5d instanton calculus \cite{Bullimore:2014awa, Gaiotto:2014ina} and spectral theory \cite{Marino:2016rsq}. Also Borel summation has played an important role in the literature. What is new in this paper is that we consider the Borel summation along any ray with phase $\vartheta$. In this subsection we write down a few known exact solutions for the Schr\"odinger equation~\eqref{eqn:diffDX1}, corresponding to our resolved conifold example, and make the connection with the Borel analysis performed in Section~\ref{sec:Borelsums}.

Reference \cite{ACDKV} wrote down the exact solution
\begin{equation} \label{eq:psiopenGV}
 \Psi_{{\rm GV}}(\epsilon,x,t)= \frac{L(x,\epsilon)}{L(x+t,\epsilon)},
\end{equation}
for $\operatorname{Im}(\epsilon)>0$, which expresses the quantum vev as a quotient of quantum dilogarithms
\begin{equation}\label{eqn:quandilog}
 L(x,\epsilon)=\prod_{j=0}^{\infty} (1-\exp(2\pi \mathrm{i} (x + j \check\epsilon))).
\end{equation}
 We have given this expression for $\Psi$ the subscript GV since it is the open topological string expectation value for the brane corresponding to $\Psi$ in Gopakumar--Vafa form. It is also the quantum vev that appears in the gauge theory context, computing the expectation value of an abelian surface defect in the ${\rm U}(1)$ gauge theory (see for instance \cite{Dimofte:2010tz}).

In Section~\ref{sec:Borelsums}, we find that \eqref{eq:psiopenGV} may be obtained by Borel summing \eqref{eq:Sasymp} along the positive imaginary axis.
For this it is convenient to note that $L(x,\epsilon)$, with $\operatorname{Im}(x)>0$ and $\operatorname{Im}(\epsilon)>0$, may be rewritten as
\begin{align}
 \log(L(x,\epsilon))&=\sum_{j=0}^{\infty}\log(1-\exp(2\pi \mathrm{i}(x+j\check\epsilon)))\nonumber\\
 &=-\sum_{l=1}^{\infty}\sum_{j=0}^{\infty}\frac{{\rm e}^{2\pi \mathrm{i} l(x+j\check\epsilon)}}{l} =\sum_{l=1}^{\infty}\frac{{\rm e}^{2\pi \mathrm{i} lx}}{l({\rm e}^{2\pi \mathrm{i} l\check\epsilon}-1)},\label{LGV}
 \end{align}
where the last expression is well-defined in the larger domain $\operatorname{Im}(x)>0$ and $\operatorname{Im}(\epsilon)\neq 0$.

Another solution with better analytical properties is given by the ratio
\begin{equation} \label{eq:psiopennp}
 \Psi_{{\rm np}}(\epsilon,x,t)= \frac{\mathcal{S}_2(x|\check{\epsilon},1)}{\mathcal{S}_2(x+t|\check{\epsilon},1)}
\end{equation}
of Faddeev's quantum dilogarithms $\mathcal{S}_2$, in the conventions reviewed in Appendix~\ref{appendix:qdilog}.
Using the difference equations (\ref{qddiffeq}) satisfied by $\mathcal{S}_2$, it can be verified that $\Psi_{{\rm np}}$ is a solution of the Schr\"odinger equation \eqref{eqn:diffDX1}, and from equation (\ref{S2asymp}) it can be checked that $\Psi_{{\rm np}}$ has an asymptotic expansion given by (\ref{eq:Sasymp}). We will see in Section~\ref{sec:Borelsums} that this solution corresponds to Borel summing \eqref{eq:Sasymp} along the positive real axis.

Using (\ref{S2prod}), the expression~\eqref{eq:psiopennp} can also be written as the product
\begin{gather*}
 \Psi_{{\rm np}}(\epsilon,x,t) = \prod_{k=0}^{\infty} \frac{(1-\exp(2\pi \mathrm{i} (x + k \check{\epsilon}))}{(1-\exp(2\pi \mathrm{i} (x+t+k \check{\epsilon}))} \frac{(1-\exp (2\pi \mathrm{i} ((x+t)/\check{\epsilon}-(k+1)/\check{\epsilon}))}{(1-\exp (2\pi \mathrm{i} (x/\check{\epsilon}-(k+1)/\check{\epsilon}))},
\end{gather*}
which is valid for $\operatorname{Im}(\epsilon)>0$.

\section{Quantum curves of closed string moduli}\label{sec:diffeq}

In this section we study difference equations obeyed by the refined topological string free energy on the resolved conifold, its limits and the connections between the resulting objects. The approach is similar to~\cite{alim2020difference, Iwaki2} where difference equations were derived starting from the asymptotic expansion of the free energy. We will then show that these difference equations obtained from the asymptotic expansion also admit natural analytic functions in the perturbative expansion parameter as in \cite{Alim:2021ukq, Alim:2021lld}, thus providing candidates for the non-perturbative completion. From the analytic solutions we will be able to moreover derive new difference equations which are invisible to the asymptotic expansion, these correspond to shifts of the closed string moduli by integers.

\subsection{The refined difference equation}
We first derive a refined version of the difference equations for the free energies of the resolved conifold starting from its asymptotic expansion.

\begin{Proposition} The refined free energy for the resolved conifold geometry satisfies
\begin{align}
F^{{\rm ref}}&\left(\epsilon_1,\epsilon_2,t+\frac{\check{\epsilon}_1}{2}-\frac{\check{\epsilon}_2}{2}\right) + F^{{\rm ref}}\left(\epsilon_1,\epsilon_2,t-\frac{\check{\epsilon}_1}{2}+\frac{\check{\epsilon}_2}{2} \right) \nonumber\\
&-F^{{\rm ref}}\left(\epsilon_1,\epsilon_2,t+\frac{\check{\epsilon}_1}{2}+\frac{\check{\epsilon}_2}{2}\right)-F^{{\rm ref}}\left(\epsilon_1,\epsilon_2,t-\frac{\check{\epsilon}_1}{2}-\frac{\check{\epsilon}_2}{2} \right)=-{\rm Li}_1(Q),\label{eq:refdiffeq}
\end{align}
with
\begin{equation*}
\check{\epsilon}^a=\frac{\epsilon^a}{2\pi}, \qquad a=1,2.
\end{equation*}
\end{Proposition}
\begin{proof}

The above proposition can be verified by an explicit computation, since the above shifts of $F^{\rm ref}$ can all be expanded in $\epsilon_1$ and $\epsilon_2$ with $Q$-dependent coefficients as in (\ref{refshifexp1}) and (\ref{refshifexp2}). We then find
\begin{gather*}
 F^{{\rm ref}}\left(\epsilon_1,\epsilon_2,t+\frac{\check{\epsilon}_1}{2}-\frac{\check{\epsilon}_2}{2}\right)+ F^{{\rm ref}}\left(\epsilon_1,\epsilon_2,t-\frac{\check{\epsilon}_1}{2}+\frac{\check{\epsilon}_2}{2}\right) \\
 \qquad\quad{} -F^{{\rm ref}}\left(\epsilon_1,\epsilon_2,t+\frac{\check{\epsilon}_1}{2}+\frac{\check{\epsilon}_2}{2}\right)-F^{{\rm ref}}\left(\epsilon_1,\epsilon_2,t-\frac{\check{\epsilon}_1}{2}-\frac{\check{\epsilon}_2}{2}\right) \\
\qquad{}=\sum_{m,n=0}^{\infty}{\rm Li}_{3-m-n}(Q) \frac{B_m B_n}{m! n!} \mathrm{i}^{m+n} \epsilon_1^{m-1} \epsilon_2^{n-1}\left(1+(-1)^{m+n}+(-1)^{m-1}+(-1)^{n-1}\right).
 \end{gather*}
When $m$ and $n$ are either even or of different parity, the last factor in the summands gives $0$, while for $m$ and $n$ odd, only $m=n=1$ survives (due to the Bernoulli numbers vanishing), giving $4 \mathrm{i}^2 (B_1)^2 \mathrm{Li}_1(Q)=-\mathrm{Li}_1(Q)$. The desired result then follows.

In the following we also supplement a proof based on the techniques of \cite{Iwaki2}, which were used in~\cite{alim2020difference}. Consider
\begin{equation*}
 \frac{w_1 w_2}{({\rm e}^{w_1}-1)({\rm e}^{w_2}-1)}= \sum_{m,n=0}^{\infty} \frac{B_n B_m}{n! m!} w_1^m w_2^n,
\end{equation*}
obtained from \eqref{eq:Bernoulligen}. This gives
\begin{equation}\label{eq:diffeqrep}
 \big({\rm e}^{w_1}-1\big)\big({\rm e}^{w_2}-1\big) \sum_{m,n=0}^{\infty} \frac{B_n \cdot B_m}{n! \cdot m!} w_1^{m-1} w_2^{n-1}=1.
\end{equation}

Next, note that the asymptotic expansion of the refined topological string free energy can be written as
\begin{align*}
F^{{\rm ref}}\left(\epsilon_1,\epsilon_2,t-\frac{\epsilon_1+\epsilon_2}{4\pi}\right) &= -\sum_{m,n=0}^{\infty} {\rm Li}_{3-m-n}(Q) \frac{B_m B_n}{m! n!} i^{m+n} \epsilon_1^{m-1} \epsilon_2^{n-1} \\
&= \sum_{m,n=0}^{\infty} \frac{B_m B_n}{m! n!} \left(\partial_t^{m+n-2} {\rm Li}_{1}(Q)\right) \check{\epsilon}_1^{m-1} \check{\epsilon}_2^{n-1},
\end{align*}
where we used
\begin{equation*}
 2 \pi \mathrm{i} {\rm Li}_{s-1} (Q) = \partial_t {\rm Li}_s(Q) \qquad \mathrm{and} \qquad \check{\epsilon}_i =\frac{\epsilon_i}{2\pi}.
\end{equation*}
The negative powers of $\partial_t$ correspond to the indefinite integration while setting the integration constant to zero, i.e.,
\begin{equation*} 2 \pi \mathrm{i} \partial_t^{-1} {\rm Li}_{s-1}(Q) ={\rm Li}_{s}(Q).\end{equation*}

The difference equation then follows after turning \eqref{eq:diffeqrep} into an operator identity, where we replace $w_i$ on the left-hand side with the derivative $\rightarrow \check{\epsilon}_i \partial_t$ for $i=1,2$, while the right-hand side becomes the identity operator. Subsequently, we act with both sides on ${\rm Li}_1(Q)$.
\end{proof}

\subsection{A non-perturbative solution} \label{sec:nonpertsol}

In this section we write down a particularly interesting solution of the difference equation~(\ref{eq:refdiffeq}) and describe its non-perturbative content. The unrefined limit of this solution equals the Borel sum along the real axis of the GW potential, which was obtained in \cite{ASTT21}, whereas the NS limit of this solution matches the Borel sum of the NS limit of the GW potential along the real axis, as we will find in Theorem \ref{maintheorem}.

Define the function\footnote{We have given this function $F^{{\rm ref}}_{{\rm np}}$ a subscript ``np" because it represents a non-perturbative completion of the free energy $F^{{\rm ref}}$. Later in this paper we will find that there are many other non-perturbative completions of $F^{{\rm ref}}$, which we label using the subscript $\rho$. The function $F^{{\rm ref}}_{{\rm np}}$ with subscript ``np" will turn out to play a special role among them.}
\begin{gather}
F^{{\rm ref}}_{{\rm np}}(\epsilon_1,\epsilon_2,t) := \left(\frac{\pi \mathrm{i}}{6} B_{3,3}\left(t+\frac{\check\epsilon_1-\check\epsilon_2}{2} \mid \check{\epsilon}_1,-\check{\epsilon}_2,1 \right)\right) \nonumber\\
\hphantom{F^{{\rm ref}}_{{\rm np}}(\epsilon_1,\epsilon_2,t) :=}{} + \log\left(\sin_3\left(t+\frac{\check\epsilon_1-\check\epsilon_2}{2} \mid \check{\epsilon}_1,-\check{\epsilon}_2,1\right)\right),\label{refdifeqsol}
\end{gather}
where $B_{3,3}(z\mid \omega_1,\omega_2)$ is a multiple Bernoulli polynomial and $\sin_3(z\mid\omega_1,\omega_2)$ the triple sine~func\-tion, both defined in Appendix~\ref{specialfunctapp}.\footnote{Multiple sine functions in the context of the quantum Riemann--Hilbert problem determined by refined
Donaldson--Thomas theory on the resolved conifold have recently appeared in~\cite{Chuang:2022uey}.}

\begin{Proposition}
$F^{{\rm ref}}_{{\rm np}}(\epsilon_1,\epsilon_2,t)$ has the following properties:
\begin{enumerate}\itemsep=0pt
 \item[$1.$] $F^{{\rm ref}}_{{\rm np}}(\epsilon_1,\epsilon_2,t)$
satisfies the difference equations
\begin{gather}
 F_{{\rm np}}^{{\rm ref}}(\epsilon_1,\epsilon_2,t+\check{\epsilon}_1)-F_{{\rm np}}^{{\rm ref}}(\epsilon_1,\epsilon_2,t)= - \log\left( \mathcal{S}_2\left(t+\frac{\check\epsilon_1-\check\epsilon_2}{2} \mid {-}\check{\epsilon}_2,1\right) \right), \nonumber \\
 F_{{\rm np}}^{{\rm ref}}(\epsilon_1,\epsilon_2,t-\check{\epsilon}_2)-F_{{\rm np}}^{{\rm ref}}(\epsilon_1,\epsilon_2,t)= - \log\left( \mathcal{S}_2\left(t+\frac{\check\epsilon_1-\check\epsilon_2}{2} \mid \check{\epsilon}_1,1 \right) \right),\nonumber\\
F_{{\rm np}}^{{\rm ref}}(\epsilon_1,\epsilon_2,t+1)-F_{{\rm np}}^{{\rm ref}}(\epsilon_1,\epsilon_2,t)= - \log\left( \mathcal{S}_2\left(t+\frac{\check\epsilon_1-\check\epsilon_2}{2} \mid \check{\epsilon}_1, -\check{\epsilon_2} \right) \right),\label{eps1shift}
\end{gather}
where
\begin{equation*}
 \mathcal{S}_2(z \mid \omega_1,\omega_2) := \exp\left(-\frac{\pi \mathrm{i}}{2} B_{2,2}(z\mid \omega_1,\omega_2)\right) \sin_2(z \mid \omega_1,\omega_2)
\end{equation*}
is Faddeev's quantum dilogarithm in slightly different conventions, following~{\rm \cite{Bridgeland1,Narukawa}}, which are reviewed in Appendix~{\rm \ref{specialfunctapp}}.

\item[$2.$] $F^{{\rm ref}}_{{\rm np}}(\epsilon_1,\epsilon_2,t)$ has the integral representation
\begin{equation}\label{eq:intreprref}
F^{{\rm ref}}_{{\rm np}}(\epsilon_1,\epsilon_2,t) = - \int_{\mathbb{R}+\mathrm{i} 0^{+}} \frac{{\rm e}^{\left(t+\frac{\check{\epsilon}_1 - \check{\epsilon}_2}{2} \right)s}}{({\rm e}^s-1)({\rm e}^{\check{\epsilon}_1s}-1) ({\rm e}^{-\check{\epsilon}_2s}-1) } \frac{{\rm d}s}{s},
\end{equation}
which is valid for $\operatorname{Re} (\check{\epsilon}_1) >0$, $\operatorname{Re}(-\check{\epsilon}_2) >0$ and $-\operatorname{Re}\big(\frac{\check{\epsilon}_1 -\check{\epsilon}_2}{2} \big) < \operatorname{Re} t < \operatorname{Re} \big(\big(\frac{\check{\epsilon}_1 -\check{\epsilon}_2}{2}\big)+1\big)$. The contour $\mathbb{R}+\mathrm{i} 0^{+}$ is following the real axis from $-\infty$ to $\infty$ avoiding $0$ by a small detour in the upper half plane.
\item[$3.$] $F^{{\rm ref}}_{{\rm np}}(\epsilon_1,\epsilon_2,t) $ satisfies the difference equation~\eqref{eq:refdiffeq}.

\item[$4.$] The exponential $Z^{{\rm ref}}_{{\rm np}}(\epsilon_1,\epsilon_2,t):=\exp \big( F^{{\rm ref}}_{{\rm np}}(\epsilon_1,\epsilon_2,t)\big)$ has the product expansion
\begin{gather}
 Z^{{\rm ref}}_{{\rm np}}(\epsilon_1,\epsilon_2,t)= Z^{{\rm ref}}(\epsilon_1,\epsilon_2,t)
 Z^{{\rm ref}}\big(4\pi^2/\epsilon_1,2\pi (\epsilon_2/\epsilon_1+1),2\pi (t-1/2)/\epsilon_1\big)\nonumber\\
\hphantom{Z^{{\rm ref}}_{{\rm np}}(\epsilon_1,\epsilon_2,t)=}{}
\times Z^{{\rm ref}}\big(4\pi^2/\epsilon_2,-2\pi (\epsilon_1/\epsilon_2+1),-2\pi (t-1/2)/\epsilon_2\big)\label{prodrefnp}
\end{gather}
for $\operatorname{Im}(-\epsilon_1/\epsilon_2)>0$, $\operatorname{Im}(\epsilon_1)>0$ and $\operatorname{Im}(-\epsilon_2)>0$, where
\begin{gather*}
 Z^{\rm ref}( \epsilon_1,\epsilon_2,t) = \exp F_\mathrm{GV}^{\rm ref}( \epsilon_1,\epsilon_2,t) = \prod_{k,l=1}^\infty \big(1- Q q_1^{k-1/2} q_2^{l-1/2} \big)^{-1}.
\end{gather*}
\end{enumerate}
\end{Proposition}

\begin{proof}
To prove the first difference equation in the proposition we use the definition \eqref{refdifeqsol} and obtain
\begin{gather*}
 F_{{\rm np}}^{{\rm ref}}(\epsilon_1,\epsilon_2,t+\check{\epsilon}_1)-F_{{\rm np}}^{{\rm ref}}(\epsilon_1,\epsilon_2,t)= \left(\frac{\pi \mathrm{i}}{2} B_{2,2}\left(t+\frac{\epsilon_1-\epsilon_2}{4\pi}\mid {-}\check{\epsilon}_2,1 \right)\right) \\
\hphantom{F_{{\rm np}}^{{\rm ref}}(\epsilon_1,\epsilon_2,t+\check{\epsilon}_1)-F_{{\rm np}}^{{\rm ref}}(\epsilon_1,\epsilon_2,t)=}{}
 - \log\left(\sin_2\left(t+\frac{\epsilon_1-\epsilon_2}{4\pi} \mid {-}\check{\epsilon}_2,1\right)\right),
 \end{gather*}
where we have used equation \eqref{eq:Bernoullipolrec} for the multiple Bernouilli polynomials, as well as relation~\eqref{eq:sinefunc} for the multiple sine functions. We recognize the right-hand side as
\begin{equation*}
 - \log\left( \mathcal{S}_2\left(t+\frac{\epsilon_1-\epsilon_2}{4\pi}\mid {-}\check{\epsilon}_2,1\right) \right).
\end{equation*}
The second and third difference equation follow analogously.

The integral representation in item (ii) is the one given in \cite[Proposition~2]{Narukawa}.

The proof that $F^{{\rm ref}}_{{\rm np}}$ satisfies the difference equation \eqref{eq:refdiffeq} follows by successively using the first two difference equations proved in item (i) and then using the difference equation for the quantum dilogarithm function reviewed in equation~(\ref{qddiffeq}). Indeed, after applying equation \eqref{eps1shift} twice and then substituting equation~\eqref{qddiffeq}, we find
\begin{gather*}
 F^{{\rm ref}} \left(\epsilon_1,\epsilon_2,t+\frac{\check{\epsilon}_1}{2}-\frac{\check{\epsilon}_2}{2}\right) + F^{{\rm ref}}\left(\epsilon_1,\epsilon_2,t-\frac{\check{\epsilon}_1}{2}+\frac{\check{\epsilon}_2}{2} \right) \\
\qquad\quad {}-F^{{\rm ref}}\left(\epsilon_1,\epsilon_2,t+\frac{\check{\epsilon}_1}{2}+\frac{\check{\epsilon}_2}{2}\right)-F^{{\rm ref}}\left(\epsilon_1,\epsilon_2,t-\frac{\check{\epsilon}_1}{2}-\frac{\check{\epsilon}_2}{2} \right)\\
 \qquad{} =-\log\left(\mathcal{S}_2(t-\check\epsilon_2\mid {-}\check\epsilon_2,1)\right) +\log\left(\mathcal{S}_2(t\mid {-}\check\epsilon_2,1)\right)
 =\log(1-Q)=-\mathrm{Li}_1(Q).
 \end{gather*}

 Finally, the proof of the product formula in item (iv) is found using the product formula of \cite[Corollary~6]{Narukawa}, reviewed in equation~\eqref{eq:multiplesineproduct} of Appendix~\ref{specialfunctapp}, and after applying the identity~\eqref{eq:productinvpow} twice. Spelled out in detail, the product formula of~\cite{Narukawa} gives
 \begin{gather*}
 Z_{{\rm np}}^{\rm ref}= \prod_{j=k=0}^{\infty}(1-\exp(2\pi \mathrm{i}(t/\check{\epsilon}_1+1/2-\check{\epsilon}_2/(2\check{\epsilon}_1) +(j+1)\epsilon_2/\epsilon_1 -(k+1)/\check\epsilon_1)))\\
 \hphantom{Z_{{\rm np}}^{\rm ref}=}{}\times
 \prod_{j=k=0}^{\infty} (1-\exp\left(2\pi \mathrm{i}\left(t+(\check{\epsilon}_1+\check{\epsilon}_2)/2+j\check{\epsilon}_1 -k\check{\epsilon}_2\right)\right))\\
\hphantom{Z_{{\rm np}}^{\rm ref}=}{}\times
 \prod_{j=k=0}^{\infty} (1-\exp\left(2\pi \mathrm{i}\left(-t/\check{\epsilon}_2 -\epsilon_1/(2\epsilon_2) +1/2 -j\epsilon_1/\epsilon_2 +(k+1)/\check{\epsilon}_2\right)\right))^{-1},
 \end{gather*}
 under the assumptions of point (iv).
 By applying \eqref{eq:productinvpow} to the first product with $q_j={\rm e}^{-\mathrm{i}/\epsilon_1}$ and the second product with $q_j={\rm e}^{-\mathrm{i}\epsilon_2}$, we obtain
 \begin{gather*}
 Z_{{\rm np}}^{\rm ref}= \prod_{j=k=0}^{\infty}(1-\exp\left(2\pi \mathrm{i}\left(t-1/2\right)/\check{\epsilon}_1+2\pi \mathrm{i}(k+1/2)/\check{\epsilon}_1 +2\pi \mathrm{i}(j+1/2)(\epsilon_2/\epsilon_1+1) \right))^{-1}\\
\hphantom{Z_{{\rm np}}^{\rm ref}=}{}\times
\prod_{j=k=0}^{\infty} (1-\exp\left(2\pi \mathrm{i} t+\mathrm{i}\epsilon_1(j+1/2)+\mathrm{i}\epsilon_2(k+1/2)\right))^{-1}\\
\hphantom{Z_{{\rm np}}^{\rm ref}=}{}\times
\prod_{j=k=0}^{\infty}\! (1-\exp(2\pi \mathrm{i}(-t+1/2)/\check{\epsilon}_2 -2\pi \mathrm{i}(\epsilon_1/\epsilon_2 +1)(j+1/2) +2\pi \mathrm{i}(k+1/2)/\check{\epsilon}_2))^{-1},
 \end{gather*}
 after slightly rewriting the exponents.
 Finally, by shifting $j$ and $k$ by $1$, so that the above products start from $j=k=1$, we obtain equation~\eqref{prodrefnp}.
\end{proof}

Let us remark that we can also recover the product representation~\eqref{prodrefnp} from the integral representation~\eqref{eq:intreprref}. Indeed, under the assumptions of item~(ii), the integrand has three infinite sets of poles in the upper half plane without zero, given by
\begin{equation*}
s= 2\pi \mathrm{i} k, \qquad s=\frac{2\pi \mathrm{i} k }{\check{\epsilon}_1} \qquad \textrm{and} \qquad s=-\frac{2\pi \mathrm{i} k }{\check{\epsilon}_2},
\end{equation*}
for $k \in \mathbb{N}\setminus\{0\}$. These poles are simple if we assume that $\check{\epsilon}_1,\check{\epsilon}_2\notin \mathbb{Q}$, $ \epsilon_1 \ne r \epsilon_2$ with $r \in \mathbb{Q}$. By closing the contour of \eqref{eq:intreprref} in the upper half plane, and further restricting the range of parameters if necessary, we can compute \eqref{eq:intreprref} by summing over the residues. We then find
\begin{gather}
F^{\rm ref}_{{\rm np}}(\epsilon_1,\epsilon_2,t) = F^{{\rm ref}}_{{\rm GV}}(\epsilon_1,\epsilon_2,t) + F_{{\rm GV}}^{{\rm ref}}(4\pi^2/\epsilon_1,2\pi (\epsilon_2/\epsilon_1+1),2\pi (t-1/2)/\epsilon_1)
\nonumber\\
\hphantom{F^{\rm ref}_{{\rm np}}(\epsilon_1,\epsilon_2,t) =}{}
 + F_{{\rm GV}}^{{\rm ref}}(4\pi^2/\epsilon_2,-2\pi (\epsilon_1/\epsilon_2+1),-2\pi (t-1/2)/\epsilon_2),\label{eq:refnpGVform}
\end{gather}
where
\begin{equation*}
F^{{\rm ref}}_{{\rm GV}}(\epsilon_1,\epsilon_2,t)= -\sum_{k=1}^{\infty} \frac{Q^k}{k(2\sin k\epsilon_1/2)(2\sin k\epsilon_2/2)}
\end{equation*}
is the Gopakumar--Vafa resummation of the free energy.

\begin{Remark}\quad
\begin{itemize}\itemsep=0pt
\item
 The expression \eqref{eq:refnpGVform} for the non-perturbative refined free energy as a sum over three perturbative pieces evaluated at different values of the arguments, matches the proposal of Lockhart and Vafa \cite{Lockhart:2012vp} for the resolved conifold, up to small discrepancies due to different conventions for the parameters $\epsilon_1$ and $\epsilon_2$.

\item In the limit $\epsilon_1=-\epsilon_2 = \lambda$, the free energy $F^{{\rm ref}}_{{\rm np}}(\epsilon_1,\epsilon_2,t)$ reduces to $F^{\rm top}_{{\rm np}}(\lambda,t)$. The latter expression was studied in \cite{Alim:2021ukq} and was shown to correspond to the Borel summed free energy in a distinguished region of the Borel plane in \cite{ASTT21}. It was furthermore shown in \cite{Alim:2021ukq} that it can be written in the form\footnote{The discrepancies of the relative factor between the two summands to \cite{Alim:2021ukq,ASTT21, Hatsuda:2015owa} is due to different normalization conventions for $F_{{\rm NS}}$ used in these works.}
 \begin{equation}\label{GVNS}
F^{\rm top}_{{\rm np}}(\lambda,t)= F_{{\rm GV}}(\lambda,t) + \frac{1}{2\pi} \frac{\partial}{\partial \lambda} \big(\lambda F_{{\rm NS}} \big(1/\check{\lambda},(t-1/2)/\check{\lambda}\big)\big).
\end{equation}
This result \eqref{GVNS} matches equation~(5.6) of \cite{Hatsuda:2015owa}, which was derived using a generalized Borel resummation.

\item A similar integral representation for the refined non-perturbative topological string partition function was obtained from Chern--Simons theory in \cite{Krefl:2015vna}.
\end{itemize}
\end{Remark}

\subsection{Difference equations in the NS limit}\label{NSlimitsec}

In Section \ref{sec:freeenergyrescon}, we introduced the NS limit $F^{\rm NS}(\epsilon,t)$ of the perturbative refined free energy $F^{\rm ref}(\epsilon_1,\epsilon_2,t)$. In this section we define the NS limit $F^{\rm NS}_{{\rm np}}(\epsilon,t)$ of the non-perturbative refined free energy $F^{\rm ref}_\mathrm{np}(\epsilon_1,\epsilon_2,t)$. We find an integral representation for $F^{\rm NS}_{{\rm np}}(\epsilon,t)$ and write down the difference equations this free energy satisfies.

Recall the definition \eqref{def:FNS} of the NS limit
\begin{gather*}
 F^{\rm NS}(\epsilon,t) := - \frac{1}{2} \sum_{k=1}^{\infty} \frac{Q^k}{k^2 \sin(k\epsilon/2)}
 \end{gather*}
of $F^{\rm ref}(\epsilon_1,\epsilon_2,t)$.
We similarly define the NS-limit of $F_{{\rm np}}^{\rm ref}(\epsilon_1,\epsilon_2,t)$ as
\begin{equation}\label{eqn:FNSnpfirst}
 F^{\rm NS}_{{\rm np}}(\epsilon,t):= \lim_{\epsilon_2\rightarrow 0} \epsilon_2\cdot F^{{\rm ref}}_{{\rm np}}(\epsilon_1,\epsilon_2,t)|_{\epsilon_1=\epsilon}.
\end{equation}

From (\ref{eq:intreprref}), we find that $F^{\rm NS}_{{\rm np}}(\epsilon,t)$ has the integral representation
\begin{equation}\label{eq:FNSintegral}
F^{\rm NS}_{{\rm np}}(\epsilon,t) = 2\pi \int_{\mathbb{R}+\mathrm{i} 0^{+}} \frac{{\rm e}^{(t+\frac{\check\epsilon}{2} )s}}{({\rm e}^s-1)({\rm e}^{\check{\epsilon}s}-1) } \frac{{\rm d}s}{s^2},
\end{equation}
which is valid for $\operatorname{Re} (\check{\epsilon}) >0,$ and $-\operatorname{Re}\left(\frac{\check{\epsilon}}{2} \right) < \operatorname{Re} t < \operatorname{Re} \left(\left(\frac{\check{\epsilon}}{2}\right)+1\right)$. The contour $\mathbb{R}+\mathrm{i} 0^{+}$ is following the real axis from $-\infty$ to $\infty$, avoiding the origin by a small detour in the upper half plane.

\begin{Proposition} Assume that $\operatorname{Re} (\check{\epsilon}) >0$ and $-\operatorname{Re}\big(\frac{\check{\epsilon}}{2} \big) < \operatorname{Re} t < \operatorname{Re} \big(\big(\frac{\check{\epsilon}}{2}\big)+1\big)$. Then for $\operatorname{Im}(t)>0$ and $\check\epsilon \notin \mathbb{Q}$, the non-perturbative NS free energy
$F^{\rm NS}_{{\rm np}}(\epsilon,t)$ can be expressed as
\begin{equation}\label{FnpFNS}
 F^{\rm NS}_{{\rm np}}(\epsilon,t)=F^{\rm NS}(\epsilon,t) + \check{\epsilon} F^{\rm NS}(2\pi (1/\check{\epsilon}+1),(t-1/2)/\check{\epsilon}).
\end{equation}
\end{Proposition}

\begin{proof}
In the upper half plane without zero, the integrand of the integral representation \eqref{eq:FNSintegral} has two infinite sets of poles given by
\begin{equation*}
s= 2\pi \mathrm{i} k, \qquad s=\frac{2\pi \mathrm{i} k }{\check{\epsilon}},
\end{equation*}
for $k \in \mathbb{N}\setminus\{0\}$. These poles are simple if we assume that $\check{\epsilon}\notin \mathbb{Q}$.

Consider a sequence $R_n>0$, such that $R_n \to \infty$ and such that the semicircle $C_{R_n}$, centered at $0$ with radius $R_n$, in the upper half-plane does not intersect the above sets of poles.
By analyticity in $t$ and $\epsilon$, it is enough to show equation \eqref{FnpFNS} for $\operatorname{Re}(t)=0$. In this case, by an application of Jordan's lemma (which requires $\operatorname{Im}(t)>0$), we have that
\begin{equation*}
 \left|\int_{C_{R_n}} \frac{{\rm e}^{(t+\frac{\check\epsilon}{2} )s}}{({\rm e}^s-1)({\rm e}^{\check{\epsilon}s}-1) } \frac{{\rm d}s}{s^2}\right| \le \frac{\pi}{-\mathrm{i} t}M_{R_n}(\epsilon),
\end{equation*}
where
\begin{gather*}
 M_{R_n}(\epsilon)=\operatorname{Max}_{\theta\in [0,\pi]} |g(\epsilon, R_n{\rm e}^{i\theta})|,\qquad
 \mathrm{with} \quad g(\epsilon,s)=\frac{1}{2s^2({\rm e}^s-1)\sinh(s\check\epsilon/2)}.
\end{gather*}

With the above choice of the $C_{R_{n}}$, we can ensure that $M_{R_{n}}(\epsilon)\to 0$ as $n\to \infty$. Hence, we can compute \eqref{eq:FNSintegral} by summing the residues and find the desired expression using \eqref{def:FNS}.
\end{proof}

Furthermore, $F^{\rm NS}_{{\rm np}}(\epsilon,t)$ is invariant under an S-duality-like transformation. Indeed, let us define
\begin{equation*}
W^{\rm NS}_\mathrm{np}(\epsilon,t) = F^{\rm NS}_{\rm np}\left(\epsilon,t-\frac{\check{\epsilon}}{2}\right).
\end{equation*}
Then we have
\begin{Proposition}\label{prop:Sduality}
Under the change of variables
\begin{equation*}
 \epsilon'=\frac{4\pi^2}{\epsilon}, \qquad t'=\frac{2\pi t}{\epsilon},
\end{equation*}
we have
\begin{equation}\label{eq:Sduality}
 W^{\rm NS}_{\rm np}(\epsilon',t')= \frac{1}{\check{\epsilon}} W^{\rm NS}_{\rm np}(\epsilon,t).
\end{equation}
\end{Proposition}

\begin{proof}
Equation~\eqref{eq:Sduality} follows directly from the integral equation~\eqref{eq:FNSintegral}.
\end{proof}

Moreover,
\begin{Proposition}
$F^{\rm NS}(\epsilon,t)$ satisfies the difference equation
\begin{equation}\label{eq:NSdiffeq}
 F^{\rm NS}(\epsilon,t+\check{\epsilon}/2)-F^{\rm NS}(\epsilon,t-\check\epsilon/2)= -\mathrm{i} {\rm Li}_2(Q).
\end{equation}
$F^{\rm NS}_{{\rm np}}(\epsilon,t)$ also obeys this equation, and moreover satisfies the additional difference equation
\begin{equation}\label{tdiffeqNS}
 F^{\rm NS}_{{\rm np}}(\epsilon,t+1)-F^{\rm NS}_{{\rm np}}(\epsilon,t)=-\mathrm{i}\check\epsilon \mathrm{Li}_2\big({\rm e}^{2\pi \mathrm{i}(t+\check\epsilon/2)/\check\epsilon}\big).
\end{equation}
\end{Proposition}

Note that the latter difference equation is invisible to its perturbative expansion.

\begin{proof}
The first part follows from an explicit computation using the formal expansion (\ref{eq:FNSasymp}):
\begin{align*}
 F^{\rm NS}(\epsilon,t+\check\epsilon/2)-F^{\rm NS}(\epsilon,t-\check\epsilon/2)&=\frac{1}{2\pi} \sum_{n=0}^{\infty} \partial_t^n {\rm Li}_{3}(Q) \frac{B_n}{n!} \check{\epsilon}^{n-1}\big((-1)^{n-1}+1\big),\\
 &=\frac{2 B_1}{2\pi}\partial_{t}\mathrm{Li}_3(Q) =-\mathrm{i}\mathrm{Li}_2(Q),
 \end{align*}
where we used that $B_{2n+1}=0$ for $n>0$.

For the second part of the proposition, we use the integral representation \eqref{eq:FNSintegral} and note that
\begin{align*}
 F^{\rm NS}_{{\rm np}}(\epsilon,t+\check{\epsilon}/2)-F_{{\rm np}}^{\rm NS}(\epsilon,t-\check{\epsilon}/2) &=2\pi \int_{\mathbb{R}+\mathrm{i} 0^{+}} \frac{{\rm e}^{ts}}{({\rm e}^s-1) } \frac{{\rm d}s}{s^2}
 =(2\pi)^2 \sum_{k=1}^{\infty}\frac{\mathrm{i} {\rm e}^{ts}}{s^2}\Big|_{s=2\pi \mathrm{i} k}\\
 &= -\frac{\mathrm{i} {\rm e}^{2\pi \mathrm{i} kt}}{k^2}
 = -\mathrm{i} {\rm Li}_2 \left(Q\right),
\end{align*}
where we have computed the integral by the sum over the residues in the upper half plane.

Similarly, we find
\begin{align*}
 F^{\rm NS}_{{\rm np}}(\epsilon,t+1)-F_{{\rm np}}^{\rm NS}(\epsilon,t) &=2\pi \int_{\mathbb{R}+\mathrm{i} 0^{+}} \frac{{\rm e}^{\left(t+\frac{\check\epsilon}{2} \right)s}}{({\rm e}^{\check\epsilon s}-1) } \frac{{\rm d}s}{s^2}
 =\frac{(2\pi)^2}{\check\epsilon} \sum_{k=1}^{\infty}\frac{\mathrm{i} {\rm e}^{(t+\check\epsilon/2)s}}{s^2}\Big|_{s=2\pi \mathrm{i} k/\check\epsilon}\\
 &= -\mathrm{i}\check\epsilon \sum_{k=1}^{\infty}\frac{{\rm e}^{2\pi \mathrm{i} k(t+\check\epsilon/2)/\check\epsilon}}{k^2}
 = -\mathrm{i}\check\epsilon \mathrm{Li}_2\big({\rm e}^{2\pi \mathrm{i}(t+\check\epsilon/2)/\check\epsilon}\big).\tag*{\qed}
\end{align*}
\renewcommand{\qed}{}
\end{proof}

\subsection{Unrefined topological string}
The limit in which the parameters of the refined topological string are set to $\epsilon_1=-\epsilon_2=\lambda$ corresponds to the unrefined topological string. The difference equation analogous to \eqref{eq:refdiffeq} was derived from the asymptotic expansion of the Gromov--Witten potential of the resolved conifold in \cite{alim2020difference} following methods of \cite{Iwaki2}. It is given by
\begin{equation*}
F^{\rm top}\big(\lambda,t+\check{\lambda}\big) + F^{\rm top}\big(\lambda,t-\check{\lambda}\big) - 2 F^{\rm top}\left(\lambda,t\right)=-{\rm Li}_1(Q), \qquad \check{\lambda}=\frac{\lambda}{2\pi}.
\end{equation*}

Furthermore a solution in terms of the triple sine function of this difference equation was considered, which was given by \cite{Alim:2021lld}
\begin{gather*}
F^{\rm top}_{{\rm np}}(\lambda,t) := \left(\frac{\pi \mathrm{i}}{6} B_{3,3}\left(t+\check\lambda | \check{\lambda},\check{\lambda},1 \right)\right) + \log\big(\sin_3\big(t+\check{\lambda} | \check{\lambda},\check{\lambda},1\big)\big).
\end{gather*}
The non-perturbative content of this solution was analyzed in \cite{Alim:2021ukq} and in \cite{ASTT21}. It was identified as the Borel summation of the asymptotic series along a distinguished ray on the real axis in the Borel plane.

Here we add a further difference equation, which is the unrefined limit of the third difference equation in~\eqref{eps1shift}. Since the shift is integral, this difference equation is invisible to the periodic asymptotic series. We have that:

\begin{Proposition}
$F^{\rm top}_{{\rm np}}(\lambda,t)$ satisfies the difference equation
\begin{gather}\label{eq:intshifttopstring}
 F_{{\rm np}}^{\rm top}(\lambda,t+1)-F_{{\rm np}}^{\rm top}(\lambda,t)=\frac{1}{2\pi \mathrm{i}} \frac{\partial}{\partial \check{\lambda}} \big( \check{\lambda} {\rm Li}_2\big({\rm e}^{2\pi \mathrm{i} t/\check{\lambda}}\big)\big).
\end{gather}
\end{Proposition}

\begin{proof}
To prove this proposition we take the third difference equation in \eqref{eps1shift} and set $\epsilon_1=-\epsilon_2=\lambda$. The right-hand side is then given in terms of the quantum dilogarithm with repeated argument, for which we can use the expression \eqref{eq:qdilogrepeat} in Appendix \ref{appendix:qdilog}. That is,
\begin{equation*}
 - \log\big( \mathcal{S}_2\big(t+\check{\lambda} | \check{\lambda}, \check{\lambda} \big) \big)= \frac{1}{2\pi \mathrm{i}} \frac{\partial}{\partial \check{\lambda}} \big( \check{\lambda} {\rm Li}_2\big({\rm e}^{2\pi \mathrm{i} t/\check{\lambda}}\big)\big).\tag*{\qed}
\end{equation*}\renewcommand{\qed}{}
\end{proof}

\begin{Remark}
The right-hand side of the difference equation \eqref{eq:intshifttopstring} equals Stokes jump of the Borel resummation of $F^\textrm{top}(\lambda,t)$ obtained in \cite{ASTT21}. The difference equation can be given the interpretation as a relation between the Borel resummations in the different Stokes sectors. Indeed, as was shown in \cite[Corollary~3.12]{ASTT21},
\begin{equation*}
 F^{\rm top}_{{\rho}_{k-n}} (\lambda,t+n)=F_{\rho_k}^{\rm top} (\lambda,t),
\end{equation*}
where $F_{\rho}^{\rm top}$ denotes the Borel sum along the ray $\rho$, and $\rho_k$ is a ray between $l_k$ and $l_{k-1}$, where $l_k=\mathbb{R}_{<0}\cdot 2\pi \mathrm{i} (t+k)$. This relation together with the jumps \eqref{eq:topstringjumps} is equivalent to the difference equation \eqref{eq:intshifttopstring}. We will find a similar interpretations for the analogous difference equation in the NS limit.
\end{Remark}


\subsection{Relations between free energies through difference equations}\label{sec:relfreeenergies}

In this section we derive further relations between the refined topological string free energy and its various limits. These will be useful in the rest of the paper.

\begin{Proposition}
The two difference equations
\begin{gather}
 F^{{\rm ref}}\left(\epsilon_1,\epsilon_2,t- \frac{\check\epsilon_1+\check\epsilon_2}{2} + \check{\epsilon}_1\right)- F^{{\rm ref}}\left(\epsilon_1,\epsilon_2,t- \frac{\check\epsilon_1+\check\epsilon_2}{2} \right) = \frac{1}{2\pi} \partial_t F^{\mathrm{NS}} \left(\epsilon_2,t-\frac{\check\epsilon_2}{2}\right), \nonumber\\
 F^{{\rm ref}}\left(\epsilon_1,\epsilon_2,t- \frac{\check\epsilon_1+\check\epsilon_2}{2} + \check{\epsilon}_2\right)- F^{{\rm ref}}\left(\epsilon_1,\epsilon_2,t- \frac{\check\epsilon_1+\check\epsilon_2}{2} \right) = \frac{1}{2\pi} \partial_t F^{\mathrm{NS}} \left(\epsilon_1,t-\frac{\check\epsilon_1}{2}\right)\label{eq:relfreffns}
\end{gather}
give a relation between $F^{{\rm ref}}$ and $F^{\mathrm{NS}}$.
\end{Proposition}
\begin{proof}
This proposition can again be verified by an explicit computation. We prove the first equation here, the second one follows analogously. Since the above shifts of $F^{\rm ref}$ can all be expanded in $\epsilon_1$ and $\epsilon_2$ with $Q$-dependent coefficients as in~(\ref{refshifexp1}) and~(\ref{refshifexp2}), we find
\begin{gather*}
 F^{{\rm ref}}\left(\epsilon_1,\epsilon_2,t+\frac{\check{\epsilon}_1}{2}-\frac{\check{\epsilon}_2}{2}\right)+ F^{{\rm ref}}\left(\epsilon_1,\epsilon_2,t-\frac{\check{\epsilon}_1}{2}-\frac{\check{\epsilon}_2}{2}\right) \\
\qquad{} =\sum_{m,n=0}^{\infty}{\rm Li}_{3-m-n}(Q) \frac{B_m B_n}{m! n!} \mathrm{i}^{m+n} \epsilon_1^{m-1} \epsilon_2^{n-1}\left((-1)^{m-1}+1\right).
 \end{gather*}
The last factor in the summand only contributes when $m=1$, since all other odd Bernoulli numbers vanish. The right-hand side then becomes
\begin{equation*}
 \sum_{n=0}^{\infty}{\rm Li}_{2-n}(Q) \frac{B_n}{ n!} (\mathrm{i}\epsilon_2)^{n-1},
 \end{equation*}
which can be matched with the right-hand side of the proposition using the asymptotic expan\-sion~\eqref{eq:FNSasymp}.
\end{proof}

\begin{Corollary}
Similarly, the unrefined topological string free energies obey the difference equation
\begin{align}\label{eq:ftopfnsrel}
 F^{\rm top}\big(\lambda,t +\check{\lambda}\big)- F^{\rm top}(\lambda,t) = \frac{1}{2\pi} \partial_t F^{\mathrm{NS}} \left(\lambda,t+\frac{\check{\lambda}}{2}\right).
\end{align}
\end{Corollary}

\begin{Remark}\quad
\begin{itemize}
 \item Similar relations also hold for the free energies with the subscript ${\rm np}$. These can be easily derived from the corresponding integral representations.
 \item Equation \eqref{eq:ftopfnsrel} is the difference equation relating the topological string free energy to certain Darboux coordinates, as was observed in \cite{CLT20} and \cite{AST21}. This will become clearer in Section~\ref{sec:NRS}.
\end{itemize}
\end{Remark}

\subsection{A quantum curve for closed string moduli}\label{sec:qcclosedm}

In the following, we interpret the finite difference equations found above as the quantization of an algebraic curve that should be associated with the closed string moduli. The motivation for this is the work of \cite{ADKMV}, in which it was proposed that the mirror curves arising in the mirror constructions of non-compact CY threefolds can be interpreted as the analogs of the Hamiltonians of a quantum mechanical problem. The curve which is quantized in the case of~\cite{ADKMV} parametrizes the open string moduli, as was put forward in \cite{Aganagic:2000gs}.

\begin{Proposition}\label{prop:closedwavefunction}
Define
\begin{gather*}
 \Psi^c_1 (\epsilon_1,\epsilon_2,t) := \frac{Z^{{\rm ref}}(\epsilon_1,\epsilon_2,t- \frac{\check\epsilon_1+\check\epsilon_2}{2})}{Z^{{\rm ref}}(\epsilon_1,\epsilon_2,t- \frac{\check\epsilon_1+\check\epsilon_2}{2}+\check{\epsilon}_2)}, \\
 \Psi^c_{2} (\epsilon_1,\epsilon_2,t) := \frac{Z^{{\rm ref}}(\epsilon_1,\epsilon_2,t- \frac{\check\epsilon_1+\check\epsilon_2}{2})}{Z^{{\rm ref}}(\epsilon_1,\epsilon_2,t- \frac{\check\epsilon_1+\check\epsilon_2}{2}+\check{\epsilon}_1)},
\end{gather*}
where
$ Z^{{\rm ref}}(\epsilon_1,\epsilon_2,t) =\exp\big(F^{{\rm ref}}(\epsilon_1,\epsilon_2,t)\big)$.
Using
$ Z_{{\rm np}}^{{\rm ref}}(\epsilon_1,\epsilon_2,t) =\exp\big(F_{{\rm np}}^{{\rm ref}}(\epsilon_1,\epsilon_2,t)\big)$
we obtain analogous definitions for $\Psi^c_{i,{\rm np}} (\epsilon_1,\epsilon_2,t)$ for $i\in \{1,2\}$.
Then the following holds for $i\in \{1,2\}$:
\begin{enumerate}\itemsep=0pt
\item[$1.$] $\Psi^c_{i} (\epsilon_1,\epsilon_2,t)$ only depends on $\epsilon_i$ and
\begin{equation*}
 \Psi^c_{i} (\epsilon_i,t)= \exp\left( -\frac{1}{2\pi} \partial_t F^{\rm NS} \left(\epsilon_i,t-\frac{\epsilon_i}{4\pi}\right) \right). \end{equation*}
The same is true for $\Psi^c_{i,{\rm np}} (\epsilon_1,\epsilon_2,t)$.
\item[$2.$] $\Psi^c_{i} (\epsilon_i,t)$ and $\Psi^c_{i,{\rm np}} (\epsilon_i,t) $ satisfy the equation
\begin{equation*}
 \big(1-Q- {\rm e}^{\check{\epsilon}_i \partial_t} \big) \Psi^c_{i} (\epsilon_i,t)=0.
\end{equation*}
\item[$3.$] We have
\begin{equation*}
 \Psi^{c}_{i,{\rm np}}=\mathcal{S}_2(t\mid \check\epsilon_i,1)^{-1}.
\end{equation*}
\end{enumerate}
\end{Proposition}

\begin{proof}
The first part of the proposition follows from \eqref{eq:relfreffns}, the second part follows by exponentiating the difference equation (\ref{eq:refdiffeq}) and using the definitions. The last point follows, from example, from the integral representation (\ref{eq:intreprref}) of $F_{{\rm np}}^{\rm ref}$, the definition of $\Psi_{i,{\rm np}}^c$, and the integral representation (\ref{S2int}) of $\mathcal{S}_2$.
\end{proof}

\begin{Remark} We note that the equation satisfied by $\Psi^c_{i}$ can be interpreted as the quantization of the curve
 \begin{equation*}
 1-\exp(2\pi \mathrm{i} t) - \exp(\mathrm{i} s/2\pi) =0
 \end{equation*}
 in $(\mathbb{C}^*)^2$,
 where the variables $t,s$ are promoted to operators $\hat{t},\hat{s}$ acting on a Hilbert space and obeying the commutation relations
 \begin{equation*} \left[\hat{t},\hat{s} \right]= \mathrm{i} \epsilon_i, \qquad i=1,2.\end{equation*}
 In particular, one can choose a polarization where $\hat{t}$ acts as multiplication by $t$ and $\hat{s}$ acts as $-\mathrm{i} \epsilon_i \partial_t$.
\end{Remark}

\subsection{Quantum Picard--Fuchs operator}\label{sec:quantumPF}

In \cite{Mironov:2009uv} the NS limit of Nekrasov's partition function was studied in the case of Seiberg--Witten theory and it was argued that the NS limit corresponds to a quantization of the classical periods of the Seiberg--Witten differential. These classical periods were computed as explicit period integrals in the original work \cite{Seiberg:1994rs} and were also obtained as solutions of Picard--Fuchs equations (see \cite{Lerche:1996xu} and references therein). A study of a quantum deformation of the Picard--Fuchs equation governing the classical periods was initiated in \cite{Mironov:2009uv}. Similar ideas came up in \cite{ACDKV}, where the WKB analysis of quantum curves was used to argue for differential Picard--Fuchs type operators that determine the order by order corrections to the classical periods to obtain the quantum periods. Although the arguments leading to these operators are clear it is perhaps less clear what the interpretation of such an all-order operator should be and how these should be derived systematically. We show in the following that the difference equations which we have studied here provide a clear path towards such all-order operator.

We want to derive a quantum Picard--Fuchs operator which annihilates in particular the quantum period that we expect to be related to the derivative of the topological string free energy in the NS limit. The asymptotic expansion of the latter is
\begin{equation*}
F^{\rm NS}(\epsilon,t-\check{\epsilon}/2)= -\mathrm{i} \sum_{n=0}^{\infty} {\rm Li}_{3-n}(Q) \frac{B_n}{n!} (\mathrm{i}\epsilon)^{n-1},
\end{equation*}
as given in \eqref{eq:FNSasymp}. This expansion was shown to satisfy the difference equation \eqref{eq:NSdiffeq}
\begin{equation*}
 F^{\rm NS}(\epsilon,t+\check{\epsilon}/2)-F^{\rm NS}(\epsilon,t-\check\epsilon/2)= -\mathrm{i} {\rm Li}_2(Q).
\end{equation*}
It follows that $\partial_t F^{\rm NS}(\epsilon,t-\check{\epsilon}/2)$ has the asymptotic expansion
\begin{equation*}
\partial_t F^{\rm NS}(\epsilon,t-\check{\epsilon}/2)= 2\pi \sum_{n=0}^{\infty} {\rm Li}_{2-n}(Q) \frac{B_n}{n!} (\mathrm{i}\epsilon)^{n-1},
\end{equation*}
and satisfies the difference equation
\begin{equation*}
 \partial_t F^{\rm NS}(\epsilon,t+\check{\epsilon}/2)-\partial_t F^{\rm NS}(\epsilon,t-\check\epsilon/2)= 2\pi {\rm Li}_1(Q).
\end{equation*}

Before deriving the quantum Picard--Fuchs operator from this difference equation we would like to include the classical terms discussed in Section~\ref{sec:classicalPF} in this quantum period, so that we in particular reproduce \eqref{eq:classicatFt} as the leading contribution of the quantum period in the limit $\epsilon \to 0$. The leading contribution of $ \partial_t F^{\rm NS}(\epsilon,t-\check{\epsilon}/2)$ is given by
\begin{equation*}
 \partial_t F^{\rm NS}(\epsilon,t-\check{\epsilon}/2)= - \frac{2\pi \mathrm{i}}{\epsilon} {\rm Li}_{2}(Q) +\mathcal{O}(1).
\end{equation*}
We would like to modify $F^{\rm NS}$ such that the leading piece includes the classical piece of \eqref{eq:classicatFt} and hence is of the form
\begin{equation*} - \frac{(2\pi \mathrm{i})^3}{\epsilon} \left( \frac{t^2}{2} + \frac{1}{(2\pi \mathrm{i})^2} {\rm Li}_{2}(Q)\right).\end{equation*}
Instead of only adding the term $t^2/2$ with the correct prefactor, we take advantage of the ambiguity of the linear and constant terms in $t$ in the period expansion as well as the ambiguity in the constant term in $\epsilon$ in the quantum period, and propose to add the classical term in the guise of a generalized Bernoulli polynomial which satisfies a difference equation on its own. We thus use
\begin{equation*}
 B_{2,2}(t\mid \check{\epsilon},1)= \frac{1}{\check{\epsilon}} \left(t^2 - t + \frac{1}{6}\right) - t+\frac{1}{2}+ \frac{\check{\epsilon}}{6},
 \end{equation*}
to obtain:

\begin{Proposition}
Define
\begin{equation*}
 \partial_t F^{{\rm NS},\sharp}(\epsilon,t-\check{\epsilon}/2) := 2 \pi^2 \mathrm{i}\cdot B_{2,2}(t\mid \check{\epsilon},1) + \partial_t F^{\rm NS}(\epsilon,t-\check{\epsilon}/2).
\end{equation*}
Then we find that:
\begin{itemize}\itemsep=0pt
 \item $\partial_t F^{{\rm NS},\sharp}(\epsilon,t-\check{\epsilon}/2)$ satisfies the difference equation
\begin{equation}\label{eq:FNSsharpdiffeq}
 \partial_t F^{{\rm NS},\sharp}(\epsilon,t+\check{\epsilon}/2)-\partial_t F^{{\rm NS},\sharp}(\epsilon,t-\check\epsilon/2)= (2\pi) \log \frac{Q}{Q-1}.
\end{equation}
\item $\partial_t F^{{\rm NS},\sharp}(\epsilon,t-\check{\epsilon}/2)$ is a solution of the quantum Picard--Fuchs operator
\begin{equation*}
 L_{\epsilon} := \partial_t \circ (1-Q) \circ \partial_t \circ \sum_{n=1}^{\infty} \frac{\check{\epsilon}^n}{n!} \partial_t^n.
\end{equation*}
\end{itemize}
\end{Proposition}

\begin{proof}
The difference equation follows from \eqref{tdiffeqNS} as well as the difference equation for the generalized Bernoulli polynomial of Proposition~\ref{prop:Bernoullidiff}:
\begin{equation*} B_{2,2}(t+\check{\epsilon}\mid \check{\epsilon},1)-B_{2,2}(t\mid \check{\epsilon},1)=2 B_{1,1}(t\mid 1)=2t-1.\end{equation*}
To show that $\partial_t F^{{\rm NS},\sharp}(\epsilon,t-\check{\epsilon}/2)$ is a solution of the quantum Picard--Fuchs operator, we write the difference equation (\ref{eq:FNSsharpdiffeq}) as:
\begin{equation*}
 \big({\rm e}^{\check{\epsilon}\partial_t}-1\big) \partial_t F^{{\rm NS},\sharp}(\epsilon,t-\check{\epsilon}/2)= (2\pi) \log \frac{Q}{Q-1},
\end{equation*}
expand in $\check{\epsilon}$ and further act with $\partial_t \circ (1-Q) \circ \partial_t $ on the result.
\end{proof}

\begin{Remark}\quad
\begin{itemize}\itemsep=0pt
 \item We remark that the leading contribution in $\epsilon$ to the quantum Picard--Fuchs operator is the so-called extended classical Picard--Fuchs operator $L_{\text{cl}}$, defined as
\begin{gather*}
 \partial_t \circ (1-Q) \circ \partial_t \circ \sum_{n=1}^{\infty} \frac{\check{\epsilon}^n}{n!} \partial_t^n = \check{\epsilon} (\partial_t \circ (1-Q)\circ \partial_t \circ \partial_t ) + \mathcal{O}\big(\epsilon^2\big)
 =: \check{\epsilon} L_{\text{cl}} + \mathcal{O}\big(\epsilon^2\big).
\end{gather*}
This operator is one order less than the Picard--Fuchs operator given in \eqref{eq:classicalextPF} and only has the three independent solutions $\varpi^0$, $\varpi^1$ and $\varpi^2$ of Section~\ref{sec:classicalPF}.

\item The quantum periods $\varpi^0_{\epsilon}$ and $\varpi^1_{\epsilon}$ are equal to
\begin{equation*}
 \varpi^0_{\epsilon}=\frac{1}{\check{\epsilon}}, \qquad \varpi^1_{\epsilon} = \frac{t}{\check{\epsilon}},
 \end{equation*}
up to rescaling. This will be further discussed later.

\item We have from Proposition~\ref{prop:closedwavefunction}\,(ii) and (iv) that
\begin{equation*}
 \exp\left( \frac{1}{2\pi} \partial_t F^{\rm NS}_{{\rm np}} \left(\epsilon_i,t-\frac{\epsilon_i}{4\pi}\right) \right) = \mathcal{S}_2(t\mid \check{\epsilon},1),
\end{equation*}
which suggests the non-perturbative extension
\begin{equation}\label{eq:nonpertquantumperiodwithclassical}
 \partial_t F^{{\rm NS},\sharp}_{{\rm np}}(\epsilon,t-\check{\epsilon}/2) = 2\pi (\pi \mathrm{i} B_{2,2}(t\mid \check{\epsilon},1) + \log(\mathcal{S}_2(t\mid \check{\epsilon},1)) ).
 \end{equation}
One can easily check that this is also a solution of the quantum Picard--Fuchs operator $L_\epsilon$.
\end{itemize}
\end{Remark}

\section{Borel sums and Stokes phenomena}\label{sec:Borelsums}

In this section we study the Borel summation and associated Stokes phenomena of the following two objects: on one hand we consider the following shift of $F^{\rm NS}(\epsilon,t)$:
\begin{equation*}
 W(\epsilon,t):=F^{\rm NS}(\epsilon,t-\check{\epsilon}/2)=-\mathrm{i} \sum_{k=1}^{\infty} \frac{1}{k^2} \frac{{\rm e}^{2\pi \mathrm{i} t k}}{({\rm e}^{\mathrm{i} k \epsilon}-1)} \qquad \mathrm{with} \quad \check{\epsilon}=\frac{\epsilon}{2\pi}.
\end{equation*}

One of the main reasons to study this shifted free energy is that we will be able to relate its Borel sums to the Borel sums and Stokes phenomena of the topological free energy $F^{{\rm top}}(\lambda,t)$ studied in \cite{ASTT21}, where $\lambda$ denotes the topological string coupling. More specifically, we will see that the Borel sums of the $\epsilon$-expansion of $W(\epsilon,t)$ give an $\epsilon$-potential for the Borel sums of the $\lambda$-expansion of $F^{{\rm top}}(\lambda,t)$, provided we set $\lambda=\epsilon$.

Furthermore, we also study the Borel sums and Stokes phenomena of the formal solution $S(\epsilon,x,t)$ of the Schr\"odinger equation (\ref{scheq}), given in (\ref{eq:Sasymp}); and relate the Borel sums of $S(\epsilon,x,t)$ with those of $W(\epsilon,t)$.

\begin{Remark}
We remark that a similar study of the Borel summability of $S(\epsilon,x,t)$ is performed in the recent work \cite[Section 4]{GHN}. An analogous discussion of the Borel transform of $F^{\rm NS}$ is written down in \cite[Section 4.5.1]{GHN}. We became aware of this work while this paper was being finalized.
\end{Remark}

\subsection{Main results}

We will first state the main results, and then prove them in the following subsections. The main results we wish to prove in this section regarding $W(\epsilon,t)$ is contained in the following theorem:

\begin{Theorem}\label{maintheorem}
Let $W(\epsilon,t)$ be as before. Then:
\begin{itemize}\itemsep=0pt
 \item {\rm (Borel sum)} Let $t\in \mathbb{C}^{\times}-\mathbb{Z}$ and consider the rays in the $\epsilon$-plane defined by $l_k:=\mathbb{R}_{<0}\cdot 2\pi \mathrm{i} (t+k)$ for $k\in \mathbb{Z}$ and $l_{\infty}:=\mathrm{i}\mathbb{R}_{<0}$ $($see Figure~{\rm \ref{fig:rays}} on the left$)$. Then given any ray $\rho \in \mathbb{C}^{\times}$ from $0$ to $\infty$ different from $\{\pm l_k\}_{k\in \mathbb{Z}}\cup \{\pm l_{\infty}\}$; and $\epsilon \in \mathbb{H}_{\rho}$, where $\mathbb{H}_{\rho}$ denotes the half-plane centered at $\rho$, the formal $\epsilon$-expansion of $W(\epsilon,t)$ is Borel summable with Borel sum:
 \begin{equation*}
 W_{\rho}(\epsilon,t) := -\frac{1}{\epsilon} \mathrm{Li}_{3}(Q)+ \frac{\mathrm{i}}{2}\mathrm{Li}_{2}(Q)+\frac{\epsilon}{12}\mathrm{Li_1}(Q) + \check{\epsilon}\int_{\rho}\mathrm{d}\xi\, {\rm e}^{-\xi/\check{\epsilon}}\widetilde{G}(\xi,t),
 \end{equation*}
 where
\begin{equation*}
 \widetilde{G}(\xi,t)=\frac{1}{2\pi}\sum_{m\geq 1}\frac{1}{m^3}\left(\frac{1}{1-{\rm e}^{-2\pi \mathrm{i} t+\xi/m}}-\frac{1}{1-{\rm e}^{-2\pi \mathrm{i} t-\xi/m}}\right).
\end{equation*}
In particular, one finds that if $\operatorname{Im}(t)>0$ and $0<\operatorname{Re}(t)<1$, then for $\operatorname{Re}(t)<\operatorname{Re}(\check\epsilon+1)$:
\begin{equation}\label{WBSR}
 W_{\mathbb{R}_{>0}}(\epsilon,t)=F_{{\rm np}}^{\rm NS}(\epsilon,t-\check\epsilon/2),
\end{equation}
where $F_{{\rm np}}^{\rm NS}(\epsilon,t)$ was defined in \eqref{eq:FNSintegral}.
\item {\rm (Stokes jumps)} Assume that $\operatorname{Im}(t)> 0$. Furthermore let $\rho$ be a ray in the sector determined by the Stokes rays $l_{k+1}$ and $l_{k}$, and $\rho'$ a ray in the sector determined by $l_{k}$ and $l_{k-1}$. Then for $\epsilon \in \mathbb{H}_{\rho}\cap \mathbb{H}_{\rho'}$ $($resp.\ $\epsilon \in \mathbb{H}_{-\rho}\cap \mathbb{H}_{-\rho'})$ we have
\begin{equation}\label{eq:StokesjumpsW}
 W_{\pm \rho}(\epsilon,t)
 -W_{\pm \rho'}(\epsilon,t) =- \mathrm{i}\check{\epsilon} \mathrm{Li}_2\big({\rm e}^{\pm 2\pi \mathrm{i}(t+k)/\check \epsilon}\big).
\end{equation}
If $\operatorname{Im}(t)<0$, then the previous jumps also hold provided $\rho$ is interchanged with $\rho'$ in the above formulas.

\item {\rm (Limits to $\pm l_{\infty})$} Let $\rho_k$ denote any ray between the Stokes rays $l_{k}$ and $l_{k-1}$. Furthermore, assume that $0<\operatorname{Re}(t)<1$, $\operatorname{Im}(t)>0$, $\operatorname{Re}(\epsilon)>0$, $\operatorname{Im}(\epsilon)<0$, and $\operatorname{Re}(t) < \operatorname{Re} (\check{\epsilon}+1)$. Then%
\begin{equation}\label{eq:limitImW}
 \lim_{k\to \infty}W_{\rho_k}(\epsilon,t)=W(\epsilon,t)
\end{equation}
On the other hand, assume that $0<\operatorname{Re}(t)<1$, $\operatorname{Im}(t)>0$, $\operatorname{Re}(\epsilon)<0$, $\operatorname{Im}(\epsilon)<0$, $\operatorname{Re} t < \operatorname{Re} (-\check{\epsilon}+1)$ and that $\big|{\rm e}^{-2\pi \mathrm{i}t/\check{\epsilon}}\big|<1$. Then
\begin{equation*}
 \begin{split}
\lim_{k\to -\infty}W_{-\rho_k}(\epsilon,t)=W(\epsilon,t).
 \end{split}
\end{equation*}
In other words, the limits to $l_{\infty}$ from the right and left give $W(\epsilon,t)$.

The other limits corresponding to $-l_{\infty}$ follow from the previous limits and the relations
\begin{equation}\label{refid}
 W_{\rho}(\epsilon,t)+W_{-\rho}(-\epsilon,t)=\mathrm{i}\mathrm{Li}_2(Q).
\end{equation}
and
\begin{equation*}
 W(\epsilon,t)+W(-\epsilon,t)=\mathrm{i}\mathrm{Li}_2(Q).
\end{equation*}
In particular, the limits to $-l_{\infty}$ from the right and left give $W(\epsilon,t)$.

\item {\rm (Potential for the Borel sum of the topological free energy)} Given $W_{\rho}(\epsilon,t)$, we have the identity
\begin{equation*}
 \partial_{\epsilon}W_{\rho}(\epsilon,t)=F_{\rho}(\epsilon,t),
\end{equation*}
where $F_{\rho}$ denotes the Borel sum of the non-constant map contribution to the topological free energy studied in~{\rm \cite{ASTT21}}.
\end{itemize}
\end{Theorem}

\begin{figure}[!ht]\centering
\includegraphics[width=7cm]{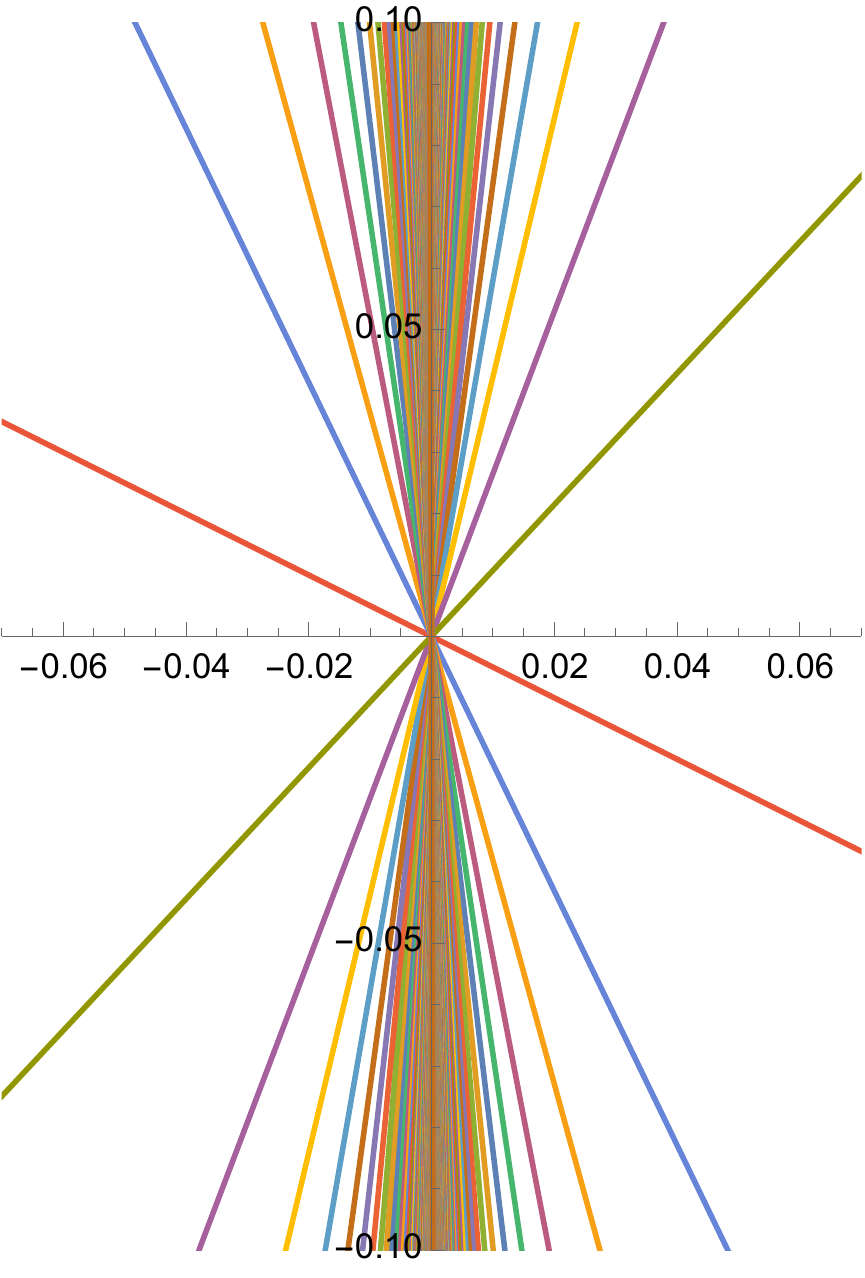}\qquad \includegraphics[width=7cm]{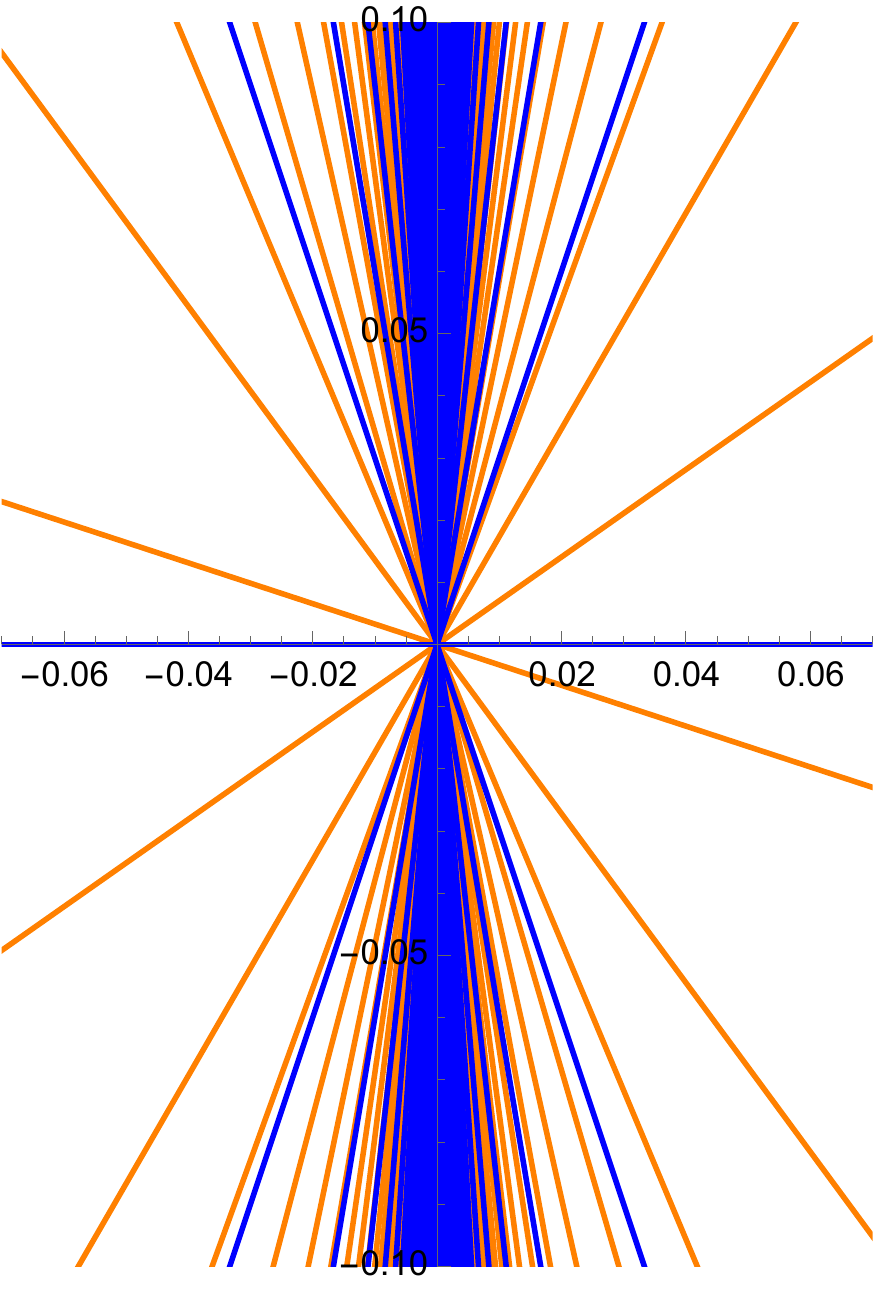}

\caption{On the left: plot of the rays $l_k = \mathrm{i} \mathbb{R}_{<0}(t+k)$ in the Borel-plane with $t=1/\pi (1 + 2 \mathrm{i})$ for $k \in (-250,250)$. On the right: plot of the rays $\widetilde{l}_k = \mathrm{i} \mathbb{R}_{<0}(x+k)$ (in blue) as well as $\widetilde{l}_{k,t} = \mathrm{i} \mathbb{R}_{<0}(t+x+k)$ (in orange) with $t=1/\pi (1 + 2 \mathrm{i})$ and $z= 2 + \mathrm{i}/3$ for $k \in (-250,250)$.}\label{fig:rays}
\end{figure}

\begin{Remark}\quad
\begin{itemize}\itemsep=0pt
 \item Notice that the Stokes jumps (\ref{eq:StokesjumpsW}) have the same form as the right-hand side of the difference equation (\ref{tdiffeqNS}) satisfied by $F_{{\rm np}}^{\rm NS}(\epsilon,t-\check\epsilon/2)$. We will use this below to establish a relation between $W_{\rho}$ and certain shifts of $F_{{\rm np}}^{\rm NS}(\epsilon,t-\check\epsilon/2)$
 \item In the previous results, we have restricted to the case $\operatorname{Im}(t)>0$ when discussing $W_{\mathbb{R}_{>0}}$ and the limits to $\pm l_{\infty}$. To do an analogous treatment for the limits to $\pm l_{\infty}$ in the case $\operatorname{Im}(t)<0$, one would need to find an analogous relation of the form (\ref{WBSR}) that holds for $\operatorname{Im}(t)<0$, since the study of the limits to $\pm l_{\infty}$ makes use of (\ref{WBSR}).
\end{itemize}

\end{Remark}

From the Stokes jumps of $W_{\rho}(\epsilon,t)$, together with (\ref{WBSR}) and the difference equation (\ref{tdiffeqNS}) satisfied by $F^{\rm NS}_{{\rm np}}(\epsilon,t)$, we obtain the following:

\begin{Corollary}\label{cortheorem1}
Assume that $\operatorname{Im}(t)>0$ and $0<\operatorname{Re}(t)<1$, and let $\rho_k$ be a ray between~$l_k$ and~$l_{k-1}$. Then
\begin{equation*}
 W_{\rho_k}(\epsilon,t)=F_{{\rm np}}^{\rm NS}(\epsilon,t-\check\epsilon/2 +k),
\end{equation*}
and
\begin{equation*}
 W_{-\rho_k}(\epsilon,t)=\mathrm{i}\mathrm{Li}_2(Q)-F_{{\rm np}}^{\rm NS}(-\epsilon,t-\check\epsilon/2+k)
\end{equation*}
\end{Corollary}
\begin{proof}
We have from $(\ref{WBSR})$ and our assumption on $t$ that
\begin{equation*}
 W_{\rho_0}(\epsilon,t)=W_{\mathbb{R}_{>0}}(\epsilon,t)=F_{{\rm np}}^{\rm NS}(\epsilon,t-\check\epsilon/2).
\end{equation*}
On the other hand, using the Stokes jumps of $W_{\rho}$ and the difference equation (\ref{tdiffeqNS}), we find
\begin{equation*}
 W_{\rho_1}(\epsilon,t)-W_{\rho_0}(\epsilon,t)=-\mathrm{i}\check\epsilon \mathrm{Li}_2\big({\rm e}^{2\pi {\rm i}t/\check\epsilon}\big)=F_{{\rm np}}^{\rm NS}(\epsilon,t-\check\epsilon/2+1)-F_{{\rm np}}^{\rm NS}(\epsilon,t-\check\epsilon/2),
\end{equation*}
so that $W_{\rho_1}(\epsilon,t)=F_{{\rm np}}^{\rm NS}(\epsilon,t-\check\epsilon/2+1)$. The result then follows from general $k\in \mathbb{Z}$ by induction. Finally, the identity for $W_{-\rho_k}$ follows from (\ref{refid}).
\end{proof}

On the other hand, by applying the same techniques used to compute the Borel sum and Stokes jumps of $W(\epsilon,t)$, we can compute the Borel sum and Stokes jumps of the formal series $S(\epsilon,x,t)$ giving a formal solution to \eqref{scheq}. More precisely, we have the following theorem

\begin{Theorem}\label{th:BorelS}
Let $S(\epsilon,x,t)$ be the formal series defined by
\begin{equation*}
 S(\epsilon,x,t)= - \sum_{n=0}^{\infty} \frac{B_n}{n!} (\mathrm{i}\epsilon)^{n-1} ( {\rm Li}_{2-n}(Q X)-{\rm Li}_{2-n}(X) ), \qquad X=\exp(2\pi \mathrm{i} x),
\end{equation*}
giving a formal solution to
\begin{equation*}
 \mathbf{D}_\Sigma (\epsilon,x,t) \Psi(\epsilon,x,t)= \big(\big(1-{\rm e}^{2 \pi \mathrm{i} x}\big) {\rm e}^{\check{\epsilon} \partial_x} - \big(1 - Q {\rm e}^{2 \pi \mathrm{i} x}\big)\big)\Psi(\epsilon,x,t)=0
\end{equation*}
via
\begin{equation*}
\Psi(\epsilon,x,t)= \exp ( S(\epsilon,x,t)).
\end{equation*}

Then:
\begin{itemize}\itemsep=0pt
 \item {\rm (Borel sum)} Let $t,x\in \mathbb{C}^{\times}$ such that $x, t+x \notin \mathbb{Z}$ and consider the rays in the $\epsilon$-plane defined by $\widetilde{l}_{k,t}:=\mathbb{R}_{<0}\cdot 2\pi \mathrm{i} (t+x+k)$ and $\widetilde{l}_{k}:=\mathbb{R}_{<0}\cdot 2\pi \mathrm{i} (x+k)$ for $k\in \mathbb{Z}$ and $l_{\infty}:=\mathrm{i}\mathbb{R}_{<0}$ $($see Figure~{\rm \ref{fig:rays}} on the right$)$. Then given any ray $\rho \in \mathbb{C}^{\times}$ from $0$ to $\infty$ different from $\big\{\pm \widetilde{l}_{k,t}\big\}_{k\in \mathbb{Z}}\cup\big\{\pm \widetilde{l}_k\big\}_{k\in \mathbb{Z}}\cup \{\pm l_{\infty}\}$; and $\epsilon \in \mathbb{H}_{\rho}$, where $\mathbb{H}_{\rho}$ denotes the half-plane centered at~$\rho$, the formal series $S(\epsilon,x,t)$ is Borel summable with Borel sum:
 \begin{gather*}
 S_{\rho}(\epsilon,x,t) := -\frac{1}{\mathrm{i}\epsilon} (\mathrm{Li}_{2}(QX)-\mathrm{Li}_{2}(X) )+ \frac{1}{2} (\mathrm{Li}_{1}(QX)-\mathrm{Li}_{1}(X) ) \nonumber\\
\hphantom{S_{\rho}(\epsilon,x,t) :=}{} + \int_{\rho}\mathrm{d}\xi\, {\rm e}^{-\xi/\check{\epsilon}}G_S(\xi,x,t),
 \end{gather*}
 where
 \begin{gather*}
 G_S(\xi,x,t)
 =\frac{1}{2\pi \mathrm{i}}\sum_{m\geq 1}\frac{1}{m^2}\bigg( \frac{1}{{\rm e}^{-2\pi \mathrm{i}(t+x)-\xi/m}-1}+\frac{1}{{\rm e}^{-2\pi \mathrm{i}(t+x)+\xi/m}-1} \nonumber\\
\hphantom{G_S(\xi,x,t)=}{}
-\frac{1}{{\rm e}^{-2\pi \mathrm{i} x-\xi/m}-1}-\frac{1}{{\rm e}^{-2\pi \mathrm{i} x+\xi/m}-1}\bigg).
 \end{gather*}
 We denote $\Psi_{\rho}(\epsilon,x,t):=\exp(S_{\rho}(\epsilon,x,t))$.

 Furthermore, we have that for $\operatorname{Min}\{\operatorname{Im}(x),\operatorname{Im}(x+t)\}>0$, $\operatorname{Min}\{\operatorname{Re}(x),\operatorname{Re}(x+t)\}>0$, $\operatorname{Max}\{\operatorname{Re}(x),\operatorname{Re}(x+t)\}<1$; the following holds for $ \operatorname{Max}\{\operatorname{Re}(x),\operatorname{Re}(x+t)\}<\operatorname{Re}(\check\epsilon+1)$:
 \begin{equation}\label{SBSR}
 \Psi_{\mathbb{R}_{>0}}=\exp (S_{\mathbb{R}_{>0}}(\epsilon,x,t) )=\frac{\mathcal{S}_2(x\mid \check\epsilon,1)}{\mathcal{S}_2(t+x\mid \check\epsilon,1)}=\Psi_{{\rm np}}(\epsilon,x,t).
 \end{equation}
 \item {\rm (Stokes jumps)} Assume that $t$ and $x$ are such that $\widetilde{l}_{k,t}$ does not overlap $\pm \widetilde{l}_r$ for any $r\in \mathbb{Z}$. Furthermore assume that $\rho$ and $\rho'$ are two rays from $0$ to $\infty$ such that traversing $\rho'$ with the opposite orientation and then $\rho$ gives a positively oriented Hankel contour. Finally, assume the sector determined by $\rho$ and $\rho'$ contains only $\widetilde{l}_{k,t}$. Then
\begin{equation*}
 S_{\pm \rho}(\epsilon,x,t)-S_{\pm \rho'}(\epsilon,x,t)=-\mathrm{Li}_1\big({\rm e}^{\pm 2\pi \mathrm{i}(k+t+x)/\check{\epsilon}}\big).
\end{equation*}
Similarly, if $\widetilde{l}_k$ does not overlap with $\pm \widetilde{l}_{r,t}$ for any $r\in \mathbb{Z}$ and $\rho$ and $\rho'$ are given as above but now the sector they define contains only $\widetilde{l}_k$, then
\begin{equation*}
 S_{\pm \rho}(\epsilon,x,t)-S_{\pm \rho'}(\epsilon,x,t)=\mathrm{Li}_1\big({\rm e}^{\pm 2\pi \mathrm{i}(k+x)/\check{\epsilon}}\big).
\end{equation*}
In the case one rays overlaps another, the jump on the overlap is given by the sum of the jumps of the corresponding rays.

\item {\rm (Limits to $\pm l_{\infty}$)} Assume that $\text{min}\{\operatorname{Im}(x),\operatorname{Im}(x+t)\}>0$, $\text{min}\{\operatorname{Re}(x),\operatorname{Re}(x+t)\}>0$, and $\max \{\operatorname{Re}(x),\operatorname{Re}(x+t)\}<1$. Furthermore, let $\rho_k$ for $k\in \mathbb{Z}$ be a sequence of rays, one for each of the sectors defined by $\widetilde{l}_n$ and $\widetilde{l}_{m,t}$, and ordered in a clockwise manner. Then for $\operatorname{Re}(\epsilon)>0$, $\operatorname{Im}(\epsilon)<0$ and $\operatorname{Max}\{\operatorname{Re}(x),\operatorname{Re}(x+t)\}<\operatorname{Re}(\check\epsilon+1)$ we have
\begin{equation}\label{Psilinfty}
\lim_{k\to \infty}\Psi_{\rho_k}(\epsilon,x,t)=\frac{L(x,\epsilon)}{L(x+t,\epsilon)}=:\Psi_{{\rm GV}}(x,\epsilon),
\end{equation}
where $L(x,\epsilon)$ is given in \eqref{LGV}.
On the other hand, if $\operatorname{Re}(\epsilon)<0$ and $\operatorname{Im}(\epsilon)<0$, while $\max \{\operatorname{Re}(x),\operatorname{Re}(x+t)\} < \operatorname{Re} (-\check{\epsilon}+1)$ and $\max \big\{\big|{\rm e}^{-2\pi \mathrm ix/\check{\epsilon}}\big|,\big|{\rm e}^{-2\pi \mathrm i(x+t)/\check{\epsilon}}\big|\big\}<1$ then
\begin{equation*}
\lim_{k\to -\infty}\Psi_{-\rho_k}(\epsilon,x,t)=\Psi_{{\rm GV}}(x,\epsilon).
\end{equation*}
The limit to $-l_{\infty}$ follows from the identities
\begin{gather*}
 S_{\rho}(\epsilon,x,t)+S_{-\rho}(-\epsilon,x,t)=\mathrm{Li}_1(QX)-\mathrm{Li}_1(X),\\
 S(\epsilon,x,t)+S(-\epsilon,x,t)=\mathrm{Li}_1(QX)-\mathrm{Li}_1(X).
\end{gather*}
In particular, this limit is also given by $\Psi_{{\rm GV}}$.
\end{itemize}
\end{Theorem}

Similarly to $W_{\rho}(\epsilon,t)$, we can use (\ref{SBSR}) together with the difference equation (\ref{qddiffeq}) satisfied by $\mathcal{S}_2$ to conclude the following:

\begin{Corollary}\label{cortheorem2}
Assume that $\operatorname{Min}\{\operatorname{Im}(x),\operatorname{Im}(x+t)\}\!>0$, $\operatorname{Min}\{\operatorname{Re}(x),\operatorname{Re}(x+t)\}\!>0$, $\operatorname{Max}\{\operatorname{Re}(x),\allowbreak \operatorname{Re}(x+t)\}<1$. Then if $\tilde{\rho}_{m,n}$ is a ray between the rays $\widetilde{l}_{m,t}$ and $\widetilde{l}_{m-1,t}$, and the rays $\widetilde{l}_{n}$ and $\widetilde{l}_{n-1}$, then
\begin{equation*}
 \Psi_{\tilde{\rho}_{m,n}}(\epsilon,x,t)=\frac{\mathcal{S}_2(x+n|\check\epsilon,1)}{\mathcal{S}_2(t+x+m|\check\epsilon,1)}=\Psi_{{\rm np}}(\epsilon,x+n,t+m-n)
\end{equation*}
and
\begin{equation*}
 \Psi_{-{\tilde{\rho}_{m,n}}}(\epsilon,x,t) =\frac{1-X}{1-QX}\left(\Psi_{{\rm np}}(-\epsilon,x+n,t+m-n)\right)^{-1}.
\end{equation*}
\end{Corollary}
\begin{proof}
The proof for the first equality follows the same reasoning as Corollary \ref{cortheorem1}, but now using (\ref{SBSR}) together with the difference equation satisfied by $\mathcal{S}_2$, given in (\ref{qddiffeq}), and the Stokes jumps of $S_{\rho}$. For example, under our hypotheses, we have
\begin{equation*}
 \Psi_{\tilde{\rho}_{0,0}}(\epsilon,x,t)=\Psi_{\mathbb{R}_{>0}}(\epsilon,x,t)=\Psi_{{\rm np}}(\epsilon,x,t).
\end{equation*}

If the next ray to $\rho_{0,0}$ in clockwise order is $\widetilde{l}_{0,t}$ then
\begin{align*}
 \frac{\Psi_{\tilde{\rho}_{1,0}}(\epsilon,x,t)}{\Psi_{\tilde{\rho}_{0,0}}(\epsilon,x,t)}&=\exp\bigl(-\mathrm{Li}_1\big({\rm e}^{2\pi \mathrm{i}(t+x)/\check\epsilon}\big)\bigr)=1-{\rm e}^{2\pi \mathrm{i}(t+x)/\check\epsilon} \nonumber\\
 & =\frac{\mathcal{S}_2(x+t|\check\epsilon,1)}{\mathcal{S}_2(x+t+1|\check\epsilon,1)}=\frac{\Psi_{{\rm np}}(\epsilon,x,t+1)}{\Psi_{{\rm np}}(\epsilon,x,t)},
\end{align*}
so that
\begin{equation*}
 \Psi_{\tilde{\rho}_{1,0}}(\epsilon,x,t)=\Psi_{{\rm np}}(\epsilon,x,t+1).
\end{equation*}
If, on the other hand, the next ray to $\tilde{\rho}_{0,0}$ in clockwise order is $\widetilde{l}_{0}$, then
\begin{gather*}
 \frac{\Psi_{\tilde{\rho}_{0,1}}(\epsilon,x,t)}{\Psi_{\tilde{\rho}_{0,0}}(\epsilon,x,t)} =\exp\big(\mathrm{Li}_1\big({\rm e}^{2\pi \mathrm{i} x/\check\epsilon}\big)\big)=\frac{1}{1-{\rm e}^{2\pi \mathrm{i} x/\check\epsilon}}
 =\frac{\mathcal{S}_2(x+1|\check\epsilon,1)}{\mathcal{S}_2(x|\check\epsilon,1)}=\frac{\Psi_{{\rm np}}(\epsilon,x+1,t-1)}{\Psi_{{\rm np}}(\epsilon,x,t)},\textbf{}
\end{gather*}
so that
\begin{equation*}
 \Psi_{\tilde{\rho}_{0,1}}(\epsilon,x,t)=\Psi_{{\rm np}}(\epsilon,x+1,t-1).
\end{equation*}
The general identity then follows.

Last, the identity for $-\tilde{\rho}_{m,n}$ follows from the first identity together with
\begin{equation*}
 S_{\rho}(\epsilon,x,t)+S_{-\rho}(-\epsilon,x,t)=\mathrm{Li}_1(QX)-\mathrm{Li}_1(X).\tag*{\qed}
\end{equation*}\renewcommand{\qed}{}
\end{proof}

Finally, we have the following theorem, relating $W_{\rho}$ with $\Psi_{\rho}=\exp(S_{\rho})$.
To state the result, we consider the paths $X(s)={\rm e}^{2\pi \mathrm{i} x(s)}$ where $x(s)=\mathrm{i} s{\rm e}^{\mathrm{i}\theta}$ for $s\in \mathbb{R}$ and $\theta$ fixed. These paths connect the points $X=0$ and $X=1$ and will be important in Section \ref{5dWKBsec}.

\begin{Theorem}\label{theorem3}
Assume that $\operatorname{Im}(t)>0$, while $0<\operatorname{Re}(t)<1$, and let $\rho_k$ be a ray between $l_k$ and~$l_{k-1}$. Furthermore, pick $x=\mathrm{i} s_*{\rm e}^{\mathrm{i} \theta_*}$ such that $\operatorname{Im}(x)>0$, while $0<\operatorname{Re}(x)<1$, and~$\rho_k$ is of the form $\rho_{k,0}$ $($recall the notation of Corollary~{\rm \ref{cortheorem2})}. Finally, consider the trajectory $x(s)=\mathrm{i} s{\rm e}^{\mathrm{i}\theta_*}$ from $s=0$ to $s=s_*$. Then, provided that $t+k\neq a\check{\epsilon}+b$ for $a,b\in \mathbb{Z}_{\leq 0}$ or~$\mathbb{Z}_{> 0}$, we can analytically continue $\mathcal{S}_2(x(s)\mid \check\epsilon,1)^{-1}\cdot\Psi_{\rho_k}(\epsilon,x(s),t)$ along the trajectory $x(s)$ from $x(s_*)=\mathrm{i} s_*{\rm e}^{\mathrm{i}\theta_*}$ to $x(0)=0$, and
\begin{equation*}
 \exp\left(-\frac{1}{2\pi}\partial_tW_{\rho_k}(\epsilon,t)\right)=\big(\mathcal{S}_2(x\mid \check\epsilon,1)^{-1}\cdot \Psi_{\rho_k}(\epsilon,x,t)\big)\big|_{x=0}.
\end{equation*}
\end{Theorem}

In the following subsections, together with Appendix \ref{appA}, we give the proofs of Theorems \ref{maintheorem}, \ref{th:BorelS} and \ref{theorem3}.

\subsection{Borel transform}\label{BTsect}

Before studying the Borel summability of $W(\epsilon,t)$, we must start by studying the Borel transform of the $\epsilon$-expansion of $W(\epsilon,t)$. Since $W(\epsilon,t)=F^{\rm NS}(\epsilon,t-\check\epsilon/2)$, we have by (\ref{eq:FNSasymp}) that we can write the formal expansion
\begin{equation*}
 W(\epsilon,t)= -\frac{1}{2\pi} \sum_{n=0}^{\infty} \partial_t^n {\rm Li}_{3}(Q) \frac{B_n}{n!} \check{\epsilon}^{n-1}
=-\frac{1}{\epsilon} {\rm Li}_3(Q) + \frac{\mathrm{i}}{2} {\rm Li}_2(Q)+ \frac{\epsilon}{12} {\rm Li}_1(Q) +\Phi(\check{\epsilon},t),
\end{equation*}
where (using that $B_{2n+1}=0$ for $n>0$):
\begin{equation*}
 \Phi(\check{\epsilon},t):=-\frac{1}{2\pi} \sum_{n=3}^{\infty} \partial_t^n {\rm Li}_{3}(Q) \frac{B_n}{n!} \check{\epsilon}^{n-1}=-\frac{1}{2\pi} \sum_{n=2}^{\infty} \partial_t^{2n} {\rm Li}_{3}(Q) \frac{B_{2n}}{(2n)!} \check{\epsilon}^{2n-1}.
\end{equation*}
We now wish to compute the Borel transform of $\Phi(\check{\epsilon},t)$
and specify its domain of convergence.
The Borel transform is defined as the formal power series $G(\xi,t):=\mathcal{B}(\Phi(-,t))(\xi)$, where
\begin{equation*}
 \mathcal{B}\colon \ \check{\epsilon} \mathbb{C}[[\check{\epsilon}]]\to \mathbb{C}[[\xi]], \mathcal{B}(\check{\epsilon}^{n+1})=\frac{\xi^n}{n!}.
\end{equation*}
Namely, we wish to study
\begin{equation*}
 G(\xi,t)= -\frac{1}{2\pi}\sum_{n=2}^{\infty} \frac{ B_{2n}}{ (2n)! (2n-2)!} \xi^{2n-2} \partial_t^{2n} {\rm Li}_3(Q).
\end{equation*}

The main result about $G(\xi,t)$ that we will prove and use is the following:

\begin{Proposition}\label{Boreltransprop}
Take $t\in \mathbb{C}^{\times}$ and $\epsilon \in \mathbb{C}^{\times}$ such that $0<|\operatorname{Re}(t)|<1/2$ and $|\epsilon/t|\in (0,2\pi)$, then $G(\xi,t)$ converges. Furthermore, we have the alternate representation
\begin{equation}\label{Boreltrans2}
 G(\xi,t)=\frac{1}{4\pi}\sum_{m\in \mathbb{Z}-\{0\}} \frac{1}{m^2}\big(\mathrm{Li}_1\big({\rm e}^{2\pi \mathrm{i} t +\xi/m}\big) +\mathrm{Li}_1\big({\rm e}^{2\pi \mathrm{i} t -\xi/m}\big)-2\mathrm{Li}_1(Q)\big),
 \end{equation}
 where $\mathrm{Li}_1(z)=-\log(1-z)$ uses the principal branch. In particular $G(\xi,t)$ admits a continuation for all $\xi> 0$.
\end{Proposition}

\begin{proof} The proof of this proposition is given in Appendix \ref{BTapp}. The techniques and ideas are the same as the ones used in \cite{garoufalidis2020resurgence}, which were also used in~\cite{ASTT21} to study the Borel transform of the topological free energy of the resolved conifold.
\end{proof}

On the other hand, the study of the Borel transform of $S(\epsilon,x,t)$ follows from the same arguments in Appendix~\ref{BTapp} used to study $G(\xi,t)$. One first writes
\begin{equation*}
 S(\epsilon,x,t)=-\frac{1}{\mathrm{i}\epsilon}(\mathrm{Li}_2(QX)-\mathrm{Li}_2(X)) +\frac{1}{2}(\mathrm{Li}_1(QX)-\mathrm{Li}_1(X)) + \Phi_S(\check\epsilon,x,t),
\end{equation*}
where
\begin{align*}
 \Phi_S(\check\epsilon,x,t):={}&{-}\sum_{n=1}^{\infty}\frac{B_{2n}}{(2n)!}(\mathrm{i}\epsilon)^{2n-1}(\mathrm{Li}_{2-2n}(QX)-\mathrm{Li}_{2-2n}(X))\\
={}&{-}\frac{1}{2\pi \mathrm{i}}\sum_{n=1}^{\infty}\frac{B_{2n}}{(2n)!}\check{\epsilon}^{2n-1}\partial_x^{2n}(\mathrm{Li}_{2}(QX)-\mathrm{Li}_{2}(X)).
 \end{align*}
We wish to study $G_S(\xi,x,t):=\mathcal{B}(\Phi_S(-,x,t))(\xi)$ given by
\begin{equation}\label{BoreltransS}
 G_S(\xi,x,t)=-\frac{1}{2\pi \mathrm{i}}\sum_{n=1}^{\infty}\frac{B_{2n}}{(2n)!}\frac{\xi^{2n-2}}{(2n-2)!}\partial_x^{2n}(\mathrm{Li}_{2}(QX)-\mathrm{Li}_{2}(X)).
\end{equation}

By using the arguments of Appendix \ref{BTapp} with $G_S(\xi,x,t)=(f_1\oast f_2(-,x,t))(\xi)$, and
\begin{gather*}
 f_1(\xi):=-\frac{1}{2\pi \mathrm{i}}\sum_{n=1}^{\infty}\frac{B_{2n}}{(2n)!}\xi^{2n-2},\\
 f_2(\xi,x,t):=\sum_{n=1}^{\infty}\frac{\xi^{2n-2}}{(2n-2)!}\partial_x^{2n}(\mathrm{Li}_{2}(QX)-\mathrm{Li}_{2}(X)),
 \end{gather*}
we find the following:

\begin{Proposition}
Assume $t,x\in \mathbb{C}^{\times}$ with $t\neq -x$ and $\operatorname{Max}\{|\operatorname{Re}(x)|,|\operatorname{Re}(t+x)|\}<1/2$. Then for $|\xi|<\operatorname{Min}\{2\pi |t+x|,2\pi |x|\}$ we have that $G(\xi,x,t)$ in \eqref{BoreltransS} converges. Furthermore, using the integral representation \eqref{Hadprodint} for $G_S(\xi,x,t)=(f_1\oast f_2(-,x,t))(\xi)$ we find the alternate expression
\begin{gather*}
 G_S(\xi,x,t)
 =\frac{1}{2\pi \mathrm{i}}\sum_{m\geq 1}\frac{1}{m^2}\big(\mathrm{Li}_0\big(Q{\rm e}^{2\pi \mathrm{i} x +\xi/m}\big)+\mathrm{Li}_0\big(Q{\rm e}^{2\pi \mathrm{i} x -\xi/m}\big)\\
 \hphantom{G_S(\xi,x,t)=\frac{1}{2\pi \mathrm{i}}\sum_{m\geq 1}}{}
 -\mathrm{Li}_0\big({\rm e}^{2\pi \mathrm{i} x +\xi/m}\big)-\mathrm{Li}_0\big({\rm e}^{2\pi \mathrm{i} x -\xi/m}\big)\big)\\
\hphantom{G_S(\xi,x,t)}{}
 =\frac{1}{2\pi \mathrm{i}}\sum_{m\geq 1}\frac{1}{m^2}\bigg(\frac{1}{{\rm e}^{-2\pi \mathrm{i}(t+x)-\xi/m}-1}+\frac{1}{{\rm e}^{-2\pi \mathrm{i}(t+x)+\xi/m}-1}\\
 \hphantom{G_S(\xi,x,t)=\frac{1}{2\pi \mathrm{i}}\sum_{m\geq 1}}{}
 -\frac{1}{{\rm e}^{-2\pi \mathrm{i} x-\xi/m}-1}-\frac{1}{{\rm e}^{-2\pi \mathrm{i} x+\xi/m}-1}\bigg),
 \end{gather*}
so that $G_S(\xi,x,t)$ admits an analytic continuation in $\xi$ to a meromorphic function with poles at $2\pi \mathrm{i} (t+x+k)m$ and $2\pi \mathrm{i}(x+k)m$ for $k\in \mathbb{Z}$ and $m\in \mathbb{Z}-\{0\}$.
\end{Proposition}

\subsection{Borel sum}

We now study the Borel sum of the $\epsilon$-expansion of $W(\epsilon,t)$ along $\mathbb{R}_{>0}$. As before, we assume $t\in \mathbb{C}^{\times}$ such that $0<|\operatorname{Re}(t)|<1/2$, so that $G(\xi,t)$ admits a continuation for $\xi>0$ with the expression (\ref{Boreltrans2}), and we can consider the Borel sum along $\mathbb{R}_{>0}$.

In the following, we assume that $\epsilon \in \mathbb{H}_{\mathbb{R}_{>0}}$, where $\mathbb{H}_{\mathbb{R}_{>0}}$ denotes the half-plane centered at $\mathbb{R}_{>0}$. We then integrate by parts and use the fact that the boundary terms vanish to write
\begin{equation*}
 \int_{\mathbb{R}_{>0}} {\rm d} \xi\, {\rm e}^{-\xi/\check{\epsilon}} G(\xi,t)= \frac{\check{\epsilon}}{4\pi}\int_{\mathbb{R}_{>0}}\mathrm{d}\xi {\rm e}^{-\xi/\check{\epsilon}}\sum_{m\in \mathbb{Z}-\{0\}}\frac{1}{m^3}\left(\frac{1}{1-{\rm e}^{-2\pi \mathrm{i} t+\xi/m}}-\frac{1}{1-{\rm e}^{-2\pi \mathrm{i} t-\xi/m}}\right).
\end{equation*}

The resulting expression in the integrand has poles at the points $\xi=2\pi {\rm i}(t+k)m$. We define
\begin{align*}
 \widetilde{G}(\xi,t):={}& \frac{1}{4\pi}\sum_{m\in \mathbb{Z}-\{0\}}\frac{1}{m^3}\left(\frac{1}{1-{\rm e}^{-2\pi \mathrm{i} t+\xi/m}}-\frac{1}{1-{\rm e}^{-2\pi \mathrm{i} t-\xi/m}}\right)\\
 ={}& \frac{1}{2\pi}\sum_{m\geq 1}\frac{1}{m^3}\left(\frac{1}{1-{\rm e}^{-2\pi \mathrm{i} t+\xi/m}}-\frac{1}{1-{\rm e}^{-2\pi \mathrm{i} t-\xi/m}}\right).
 \end{align*}

One advantage of the expression of the Borel sum with $\widetilde{G}(\xi,t)$ is that we can now integrate freely along rays avoiding the poles, provided $\epsilon$ is in the appropriate range. More precisely, we define:

\begin{Definition}
For $t\in \mathbb{C}^{\times}-\mathbb{Z}$, for $\rho$ a ray from $0$ to $\infty$ avoiding the rays $\pm l_k=\pm \mathbb{R}_{<0}\cdot 2\pi \mathrm{i} (t+k)$ for $k\in \mathbb{Z}$ and $\pm l_{\infty}=\pm \mathrm{i} \mathbb{R}_{<0}$, and for $\check{\epsilon}\in \mathbb{H}_{\rho}$ where $\mathbb{H}_{\rho}$ denotes the half-plane centered at $\rho$, we denote the Borel sum of $W(\epsilon,t)$ along $\rho$ by
\begin{equation*}
 W_{\rho}(\epsilon,t) := -\frac{1}{\epsilon} \mathrm{Li}_{3}(Q)+ \frac{\mathrm{i}}{2}\mathrm{Li}_{2}(Q)+\frac{\epsilon}{12}\mathrm{Li_1}(Q) + \check{\epsilon}\int_{\rho}\mathrm{d}\xi\, {\rm e}^{-\xi/\check{\epsilon}}\widetilde{G}(\xi,t).
\end{equation*}
\end{Definition}

The arguments from Section \ref{NSlimitsec} then motivate the consideration of the following function:

\begin{Definition} For $\operatorname{Re}(\check{\epsilon}) >0$ and $-\operatorname{Re}(\check{\epsilon}) < \operatorname{Re}(t) < \operatorname{Re} (\check{\epsilon}+1)$, we define
\begin{equation*}
W_{\rm np}(\epsilon,t) :=F_{{\rm np}}^{\rm ref}(\epsilon,t-\check\epsilon/2)=2\pi \int_{\mathbb{R}+\mathrm i0^+} \frac{{\rm d}s}{s^2}\frac{{\rm e}^{s t}}{({\rm e}^s-1)({\rm e}^{\check{\epsilon}s}-1)},
\end{equation*}
\end{Definition}

This function generates the quantum dilogarithm $\mathcal{S}_2(z\mid \omega_1,\omega_2)$ (see Appendix \ref{specialfunctapp} or \cite{BridgelandCon}), in the sense that
\begin{equation*}
 \partial_t W_{{\rm np}}(\epsilon,t)=2\pi \log \mathcal{S}_2(t|\check\epsilon,1).
\end{equation*}

Indeed, we have
\begin{align}
 \partial_{t}W_{{\rm np}}(\epsilon,t)& =2\pi \partial_t\int_{\mathbb{R}+\mathrm{i} 0^+}\frac{{\rm d}s}{s^2} \frac{{\rm e}^{st}}{({\rm e}^s-1)({\rm e}^{\check\epsilon s}-1)} \nonumber\\
 & =2\pi \int_{\mathbb{R}+\mathrm{i} 0^+}\frac{{\rm d}s}{s} \frac{{\rm e}^{st}}{({\rm e}^s-1)({\rm e}^{\check\epsilon s}-1)}=2\pi \log \mathcal{S}_2(t\mid \check\epsilon,1).\label{tderWnp}
\end{align}

The function $W_{{\rm np}}$ has the following relation to $W_{\mathbb{R}_{>0}}$:
\begin{Proposition}\label{WRWnp}
If $\operatorname{Im}(t)>0$ and $0<\operatorname{Re}(t)<1$, then for $\operatorname{Re}(t)<\operatorname{Re}(\check\epsilon+1)$ we have
\begin{equation*}
 W_{\mathbb{R}_{>0}}(\epsilon,t)= W_{{\rm np}}(\epsilon,t).
\end{equation*}
\end{Proposition}
\begin{proof}The proof of this is given in Appendix \ref{BorelsumRapp}. The proof follows the same lines of the corresponding proofs in \cite{ASTT21}, showing the relation between the Borel sum along $\mathbb{R}_{>0}$ of the non-constant map contribution of the topological free energy, and the triple sine function.
\end{proof}

On the other hand, we have the corresponding result for $S_{\mathbb{R}_{>0}}$:
\begin{Proposition}
For $\operatorname{Min}\{\operatorname{Im}(x),\operatorname{Im}(x+t)\}>0$, as well as $\operatorname{Min}\{\operatorname{Re}(x),\operatorname{Re}(x+t)\}>0$, and $\operatorname{Max}\{\operatorname{Re}(x),\operatorname{Re}(x+t)\}<1$, we have that
\begin{equation*}
 S_{\mathbb{R}_{>0}}(\epsilon,x,t)=\log(\mathcal{S}_2(x\mid \check\epsilon,1))-\log(\mathcal{S}_2(x+t\mid \check\epsilon,1)),
\end{equation*}
for $ \operatorname{Max}\{\operatorname{Re}(x),\operatorname{Re}(x+t)\}<\operatorname{Re}(\check\epsilon+1)$. Here, $\mathcal{S}_2$ denotes the Faddeev quantum dilogarithm defined in Appendix~{\rm \ref{specialfunctapp}}.
\end{Proposition}
\begin{proof}
The proof follows easier versions of the same argument given in Appendix \ref{BorelsumRapp} for $W_{\mathbb{R}_{>0}}=W_{{\rm np}}(\epsilon,t)$. More specifically, by following a similar computation to Proposition~\ref{WorformBorelsum}, one shows that for $t,x\in \mathbb{C}^{\times}$ with $\operatorname{Im}(t)>0$, $\operatorname{Im}(t+x)>0$, and $\epsilon >0$, we have
\begin{equation*}
 S_{\mathbb{R}_{>0}}(\epsilon,x,t)=\frac{\mathrm{i}}{2\pi}\int_{\mathbb{R}+\mathrm{i}0^+}\mathrm{d}s \frac{1}{1-{\rm e}^s} \big(\log\big(1-{\rm e}^{\check\epsilon s+2\pi \mathrm{i} x}\big)-\log\big(1-{\rm e}^{\check\epsilon s+2\pi \mathrm{i}(x+t)}\big)\big),
\end{equation*}
 and then one shows that
\begin{equation*}
 \log(\mathcal{S}_2(x|\check{\epsilon},1))=\frac{\mathrm{i}}{2\pi}\int_{\mathbb{R}+i0^+}\mathrm{d}s \frac{1}{1-{\rm e}^s} \log\big(1-{\rm e}^{\check\epsilon s+2\pi \mathrm{i} x}\big)
\end{equation*}
by following the argument of Proposition \ref{woronowiczlemma}, and by using the integral representation of $\log(\mathcal{S}_2(x|\check\epsilon,1))$ given in~(\ref{S2int}).
\end{proof}

\subsection{Stokes jumps}

In this section, we study the dependence of the Borel sum $W_{\rho}(\epsilon,t)$ on the choice of $\rho$. We start with the following result:

\begin{Proposition}\label{Stokesjumps} Assume that $\operatorname{Im}(t)> 0$ and for $k \in \mathbb{Z}$ let $l_k=\mathbb{R}_{<0}\cdot 2\pi \mathrm{i}(t+k)$. Furthermore let $\rho$ be a ray in the sector determined by the Stokes rays $l_{k+1}$ and $l_{k}$, and $\rho'$ a ray in the sector determined by $l_{k}$ and $l_{k-1}$. Then for $\epsilon \in \mathbb{H}_{\rho}\cap \mathbb{H}_{\rho'}$ $($resp.\ $\epsilon \in \mathbb{H}_{-\rho}\cap \mathbb{H}_{-\rho'})$ we have
\begin{equation*}
 W_{\pm \rho}(\epsilon,t)
 -W_{\pm \rho'}(\epsilon,t) =-\mathrm{i}\check{\epsilon}\mathrm{Li}_2\big({\rm e}^{\pm 2\pi \mathrm{i}(t+k)/\check \epsilon}\big).
\end{equation*}
If $\operatorname{Im}(t)<0$, then the previous jumps also hold provided $\rho$ is interchanged with $\rho'$ in the above formulas.
\end{Proposition}

\begin{proof}
Note that
\begin{equation*}
 W_{\rho}(\epsilon,t)-W_{\rho'}(\epsilon, t)=\check{\epsilon}\int_{\mathcal{H}(l_k)}{\rm d}\xi \, {\rm e}^{-\xi/\check \epsilon}\widetilde{G}(\xi,t),
\end{equation*}
where $\mathcal{H}(l_k)$ is a Hankel contour around $l_k =\mathbb{R}_{<0}\cdot 2\pi \mathrm{i} (t + k)$.

Furthermore, note that we can write $\widetilde{G}(\xi,t)$ as a sum over $m>0$ as
\begin{equation*}
 \widetilde{G}(\xi,t):=\frac{1}{2\pi}\sum_{m>0}\frac{1}{m^3}\left(\frac{1}{1-{\rm e}^{-2\pi \mathrm{i} t+\xi/m}}-\frac{1}{1-{\rm e}^{-2\pi \mathrm{i} t-\xi/m}}\right).
\end{equation*}
We then have
\begin{align*}
 \check{\epsilon}\int_{\mathcal{H}(l_k)}{\rm d}\xi\, {\rm e}^{-\xi/\check \epsilon}\widetilde{G}(\xi,t)&=2\pi \mathrm{i} \check{\epsilon} \sum_{m=-1}^{-\infty}\big({\rm e}^{-\xi/\check \epsilon}(\xi-2\pi \mathrm im(t+k))\widetilde{G}(\xi,t)\big)\big|_{\xi=2\pi \mathrm{i} m(t+k)}\\
 &=-2\pi \mathrm{i} \check{\epsilon} \sum_{m=-1}^{-\infty}\frac{{\rm e}^{-2\pi {\rm i}m(t+k)/\check\epsilon}}{2\pi m^2}
 =-\mathrm{i} \check{\epsilon} \mathrm{Li}_2\big({\rm e}^{2\pi \mathrm{i}(t+k)/\check \epsilon}\big).
 \end{align*}
A similar computation follows for $-l_k$.
\end{proof}

The computation for the Stokes jumps of $S_{\rho}(\epsilon,x,t)$ follows exactly the argument as above.

\subsection[Limits to $\pm l_{\infty}$]{Limits to $\boldsymbol{\pm l_{\infty}}$}

On the other hand, the jumps along $\pm l_{\infty}$ will follow from the following proposition, which discusses the limits to $l_{\infty}$.

\begin{Proposition}\label{rblimit1} Let $\rho_k$ denote any ray between the Stokes rays $l_{k}$ and $l_{k-1}$. Furthermore, assume that $0<\operatorname{Re}(t)<1$, $\operatorname{Im}(t)>0$, $\operatorname{Re}(\epsilon)>0$, $\operatorname{Im}(\epsilon)<0$, and $\operatorname{Re}(t) < \operatorname{Re} (\check{\epsilon}+1)$. Then
\begin{equation*}
 \lim_{k\to \infty}W_{\rho_k}(\epsilon,t)=W(\epsilon,t).
\end{equation*}

On the other hand, assume that $0<\operatorname{Re}(t)<1$, $\operatorname{Im}(t)>0$, $\operatorname{Re}(\epsilon)<0$, $\operatorname{Im}(\epsilon)<0$, $\operatorname{Re} t < \operatorname{Re} (-\check{\epsilon}+1)$ and that $\big|w^{-1}\big|<1$. Then
\begin{equation*}
 \begin{split}
\lim_{k\to -\infty}W_{-\rho_k}(\epsilon,t)=W(\epsilon,t).
 \end{split}
\end{equation*}
\end{Proposition}
\begin{proof}
 This is proven in Appendix \ref{linftylimit}, since the computation is rather lengthy. The main strategy for the first limit is to express it as a sum over the Stokes jumps in the corresponding quadrant and $W_{\mathbb{R}_{>0}}$, and to compute $W_{\mathbb{R}_{>0}}=W_{{\rm np}}$ in terms of a sum over residues. For the second limit, we again sum over the Stokes jumps over the corresponding quadrant, and use the relation $ W_{\rho}(\epsilon,t)+W_{-\rho}(-\epsilon,t)=-\mathrm{i}\mathrm{Li}_2(Q) $ to relate $W_{-\mathbb{R}_{>0}}$ to $W_{\mathbb{R}_{>0}}$ and hence express it as a~sum over residues.
\end{proof}

To compute the other limits, notice that we have the easy to check relations
\begin{equation}\label{reflection1}
 W_{\rho}(\epsilon,t)+W_{-\rho}(-\epsilon,t)=\mathrm{i}\mathrm{Li}_2(Q),
\end{equation}
and
\begin{equation}\label{reflection2}
 W(\epsilon,t)+W(-\epsilon,t)=\mathrm{i}\mathrm{Li}_2(Q).
\end{equation}
From this, the next corollary follows, discussing the limits to $-l_{\infty}$:

\begin{Corollary}
For $0<\operatorname{Re}(t)<1$ and $\operatorname{Im}(t)>0$, whereas $\operatorname{Re}(\epsilon)<0$, $\operatorname{Im}(\epsilon)>0$ and $\operatorname{Re}(t) < \operatorname{Re} (-\check{\epsilon}+1)$, we have
\begin{equation*}
 \lim_{k\to \infty}W_{-\rho_k}(\epsilon,t)=W(\epsilon,t)-\mathrm{i}\mathrm{Li}_2(Q).
\end{equation*}

On the other hand, assume that $0<\operatorname{Re}(t)<1$, $\operatorname{Im}(t)>0$, $\operatorname{Re}(\epsilon)>0$, $\operatorname{Im}(\epsilon)>0$ and $\operatorname{Re} (t) < \operatorname{Re} (\check{\epsilon}+1)$ while $\big|{\rm e}^{2\pi \mathrm{i} t/\check\epsilon}\big|<1$. Then
\begin{equation*}
\lim_{k\to -\infty}W_{\rho_k}(\epsilon,t)=W(\epsilon,t)+\mathrm{i}\mathrm{Li}_2(Q).
\end{equation*}
\end{Corollary}
\begin{proof}
By using Proposition \ref{rblimit1} and the identities (\ref{reflection1}) and (\ref{reflection2}), we get
\begin{equation*}
 \lim_{k\to \infty}W_{-\rho_k}(\epsilon,t)=- \lim_{k\to \infty}W_{\rho_k}(-\epsilon,t)+\mathrm{i}\mathrm{Li}_2(Q)
 =-W(-\epsilon,t)+\mathrm{i}\mathrm{Li}_2(Q)
 =W(\epsilon,t).
\end{equation*}
Similarly,
\begin{equation*}
 \lim_{k\to -\infty}W_{\rho_k}(\epsilon,t)=- \lim_{k\to -\infty}W_{-\rho_k}(-\epsilon,t)+\mathrm{i}\mathrm{Li}_2(Q)
 =-W(-\epsilon,t)+\mathrm{i}\mathrm{Li}_2(Q)
 =W(\epsilon,t).\tag*{\qed}
\end{equation*}\renewcommand{\qed}{}
\end{proof}

Finally, using the same idea as in the proof of Proposition \ref{rblimit1}, one can prove the limit (\ref{Psilinfty}) for $\Psi_{\rho}$. In fact, due to the relation
\begin{equation*}
 \Psi_{\mathbb{R}_{>0}}(\epsilon,x,t)=\frac{\mathcal{S}_2(x\mid \check\epsilon,1)}{\mathcal{S}_2(x+t\mid \check\epsilon,1)},
\end{equation*}
it is enough to use the integral representation of $\log(\mathcal{S}_2(x|\check\epsilon,1))$ and the jumps across $\widetilde{l}_k$, and show that
\begin{equation}\label{S2l}
 \log(\mathcal{S}_2(x\mid \check\epsilon,1))+\sum_{k=0}^{\infty}\mathrm{Li}_1(\exp(2\pi \mathrm{i}(k+x)/\check\epsilon))=\sum_{l=1}^{\infty}\frac{{\rm e}^{2\pi \mathrm{i} lx}}{l({\rm e}^{2\pi \mathrm{i} l\check\epsilon}-1)}=\log(L(x,\epsilon)),
\end{equation}
since then the jumps along $\widetilde{l}_{k,t}$ combine with $-\log(\mathcal{S}_2(x+t\mid \check\epsilon,1))$ to give $-\log(L(x+t,\epsilon))$.

\subsection{A potential for the Borel sums}

The main result we wish to prove in the section is the following, showing that $W_{\rho}(\epsilon,t)$ serves as a potential for the Borel sums of the non-constant map contribution of the topological free energy:

\begin{Proposition} The $\epsilon$-derivative of the Stokes jumps of $W_{\rho}$ equal the Stokes jumps of $F_{\rho}$. Furthermore, we have
\begin{equation}\label{pot}
 \partial_{\epsilon}W_{\rho}(\epsilon,t)=F_{\rho}(\epsilon,t),
\end{equation}
where $F_{\rho}$ denotes the Borel sum of the non-constant map contribution to the topological free energy studied in~{\rm \cite{ASTT21}}.
\end{Proposition}

\begin{proof}
The first statement on the $\epsilon$-derivative is clear. Indeed, we have
\begin{equation*}
 \partial_{\epsilon}\bigl(-\mathrm{i} \check{\epsilon} \mathrm{Li}_2\big({\rm e}^{2\pi \mathrm{i}(t+k)/\check \epsilon}\big)\bigr)=\frac{1}{2\pi \mathrm{i}}\partial_{\check{\epsilon}}\big( \check{\epsilon} \mathrm{Li}_2\big({\rm e}^{2\pi \mathrm{i}(t+k)/\check \epsilon}\big)\big),
\end{equation*}
matching the jumps of $F_{\rho}$ of~\cite{ASTT21} under the identification $\epsilon=\lambda$, where $\lambda$ denotes the topological string coupling.

Now recall that by Proposition~\ref{WRWnp}, we have that on their common domains of definition
\begin{equation*}
 W_{\mathbb{R}_{>0}}(\epsilon,t)= W_{{\rm np}}(\epsilon,t).
\end{equation*}
In particular, we obtain that
\begin{align*}
 \partial_{\epsilon}W_{\mathbb{R}_{>0}}(\epsilon,t)&=\partial_{\epsilon}W_{{\rm np}}(\epsilon,t)
 =\partial_{\epsilon}\left(2\pi \int_{\mathbb{R}+\mathrm i0^+} \frac{{\rm d}s}{s^2}\frac{{\rm e}^{s t}}{({\rm e}^s-1)({\rm e}^{\check{\epsilon}s}-1)}\right)\\
 &=-\int_{\mathbb{R}+\mathrm i0^+} \frac{{\rm d}s}{s}\frac{{\rm e}^{s (t+\check\epsilon)}}{({\rm e}^s-1)({\rm e}^{\check{\epsilon}s}-1)^2}=F_{\mathbb{R}_{>0}}(\epsilon,t).
 \end{align*}
where in the last equality we have used \cite[Theorem~2]{ASTT21}.

Now let $l_k=\mathbb{R}_{<0}\cdot 2\pi \mathrm{i}(t+k)$ as before, and let $\rho_k$ be a ray between $l_k$ and $l_{k-1}$. The fact that
\begin{equation*}
 \partial_{\epsilon}W_{\rho_k}(\epsilon,t)=F_{\rho_k}(\epsilon,t).
\end{equation*}
follows from the fact that $\partial_{\epsilon}W_{\rho_0}=\partial_{\epsilon}W_{\mathbb{R}_{>0}}=F_{{\rm np}}=F_{\mathbb{R}_{>0}}$, together with the fact that the $\epsilon$-derivatives of the Stokes jumps of $W_{\rho}$ equal those of $F_{\rho}$.

On the other hand, to check this for $-\rho_k$, we use the following easily verifiable relations
\begin{equation*}
 W_{-\rho}(\epsilon,t)=-W_{\rho}(-\epsilon,t)+i\mathrm{Li}_2(Q), \qquad F_{-\rho}(\epsilon,t)=F_{\rho}(-\epsilon,t).
\end{equation*}
It then follows that
\begin{equation*}
 \partial_{\epsilon}W_{-\rho_k}(\epsilon,t)=-\partial_{\epsilon}(W_{\rho_k}(-\epsilon,t)-\mathrm{i}\mathrm{Li}_2(Q))=F_{\rho_k}(-\epsilon,t)=F_{-\rho_k}(\epsilon,t).
\end{equation*}
The result then follows.
\end{proof}

\begin{Remark}
Analogously, one can prove (\ref{pot}) by computing the $\epsilon$-derivative of the Borel sum expression of $W_{\rho}(\epsilon,t)$, and show that it matches (after an integration by parts in one of the summands) the expression of the Borel sum $F_{\rho}(\epsilon,t)$ of \cite{ASTT21}. The statement of the jumps then follows automatically.
\end{Remark}

\subsection[Relation between $\Psi_{\rho}$ and $W_{\rho}$]{Relation between $\boldsymbol{\Psi_{\rho}}$ and $\boldsymbol{W_{\rho}}$}\label{sec:relationPsiPeriod}

In this subsection we prove Theorem \ref{theorem3}:

\medskip

\noindent
\textbf{Theorem \ref{theorem3}.}
{\it Assume that $\operatorname{Im}(t)>0$, while $0<\operatorname{Re}(t)<1$, and let $\rho_k$ be a ray between~$l_k$ and~$l_{k-1}$. Furthermore, pick $x=\mathrm{i} s_*{\rm e}^{\mathrm{i}\theta_*}$ such that $\operatorname{Im}(x)>0$, while $0<\operatorname{Re}(x)<1$, and~$\rho_k$ is of the form $\tilde{\rho}_{k,0}$ $($recall the notation of Corollary~{\rm \ref{cortheorem2})}.\footnote{For example, we can pick $0<s_*$ with $s_*$ sufficiently close to $0$ and $-\pi/2<\theta_*<0$ with $\theta_*$ sufficiently close to $-\pi/2$.\label{footnote9}}
Finally, consider the trajectory $x(s)=\mathrm{i} s{\rm e}^{\mathrm{i} \theta_*}$ from $s=0$ to $s=s_*$. Then, provided that $t+k\neq a\check{\epsilon}+b$ for $a,b\in \mathbb{Z}_{\leq 0}$ or~$\mathbb{Z}_{> 0}$, we can analytically continue $\mathcal{S}_2(x(s)\mid \check\epsilon,1)^{-1} \cdot\Psi_{\rho_k}(\epsilon,x(s),t)$ along the trajectory $x(s)$ from $x(s_*)=\mathrm{i} s_*{\rm e}^{\mathrm{i}\theta_*}$ to $x(0)=0$, and}
\begin{equation}\label{relWpsi}
 \exp\left(-\frac{1}{2\pi}\partial_tW_{\rho_k}(\epsilon,t)\right)=\big(\mathcal{S}_2(x\mid \check\epsilon,1)^{-1}\cdot \Psi_{\rho_k}(\epsilon,x,t)\big)\big|_{x=0}.
\end{equation}

\begin{proof}
Recall that by Corollary \ref{cortheorem1} and equation (\ref{tderWnp}),
\begin{equation*}
 \partial_{t}W_{\rho_k}(\epsilon,t)=\partial_{t}F_{{\rm np}}^{\rm NS}(\epsilon,t-\check\epsilon/2+k)=2\pi \log\mathcal{S}_2(t+k\mid \check\epsilon,1).
\end{equation*}
On the other hand, note that for all $s$ such that $0<s\leq s_*$ we continue to satisfy that $\rho_k=\tilde{\rho}_{k,0}$. We can therefore apply Corollary \ref{cortheorem2} and write
\begin{equation*}
 \Psi_{\rho_k}(\epsilon,x(s),t)=\Psi_{\rho_{k,0}}(\epsilon,x(s),t) =\Psi_{{\rm np}}(\epsilon,x(s),t+k)=\frac{\mathcal{S}_2(x(s)\mid \check\epsilon,1)}{\mathcal{S}_2(x(s)+t+k\mid \check\epsilon,1)},
\end{equation*}
so that by taking the limit to $s\to 0$ we obtain
\begin{equation*}
 \mathcal{S}_2(x\mid \check\epsilon,1)^{-1}\cdot \Psi_{\rho_k}(\epsilon,x,t)|_{x=0}=\mathcal{S}_2(t+k\mid \check\epsilon,1)^{-1}.
\end{equation*}
The previous quantity is well defined due to the contraint $t+k\neq a\check{\epsilon}+b$ for $a,b\in \mathbb{Z}_{\leq 0}$ or $\mathbb{Z}_{> 0}$ (recall point $1$ of Appendix \ref{appendix:qdilog}).
The desired equation (\ref{relWpsi}) then follows.
\end{proof}

\begin{Remark}\label{remark:thetastar}
It may be helpful to note that the rays $\tilde{l}_{k,t}$ approximate the rays $l_{k}$ when $x \to 0$, while the rays $\tilde{l}_m$ for $m<0$ (or $m>0$) collapse into the positive (or negative) imaginary axis when $x \to 0$, and the phase of the ray $\tilde{l}_0$ is equal to $\vartheta_*$ along the trajectory $x(s)$. It is thus impossible to choose the fixed ray $\rho_k$ of the form $\tilde{\rho}_{k,m}$ for $m \neq \{0,1\}$ when $x \to 0$.

With the restrictions on $x$ as in the assumptions of the above theorem, the phase $\vartheta_*$ along the trajectory $x(s)$ should be in between 0 and $-\pi/2$. This means that we can only choose $\rho_k$ of the form $\tilde{\rho}_{k,1}$ when the phase of $\rho_k$ is in between 0 and $-\pi/2$ (in which case $\vartheta_*$ should be chosen in between 0 and the phase of $\rho_k$, and equation \eqref{relWpsi} would need to be modified with a~small factor). On the other hand, if we choose $\rho_k$ to be of the form $\tilde{\rho}_{k,0}$, this means $\vartheta_*$ should be chosen smaller than 0 while in between the phase of $\rho_k$ and $-\pi/2$. This can be achieved by choosing $\vartheta_*$ sufficiently close to $-\pi/2$ as in footnote~\eqref{footnote9}.

We remark that the restrictions on $x$ in the above theorem are mainly there to be able to apply Corollary \ref{cortheorem1}, which in turn uses these restrictions to be able to say that $\Psi_{\mathbb{R}_{>0}}=\Psi_{{\rm np}}$. The restrictions on $x$ in the above theorem are however not stringent and can be relaxed at the cost of modifying the statement \eqref{relWpsi} slightly.
\end{Remark}

\begin{Corollary}\label{corthm3}
Under the assumption of Theorem~{\rm \ref{theorem3}} and the assumptions on $\epsilon$ from Proposition~{\rm \ref{rblimit1}}, we can take the limit $k \to \infty$ and obtain
\begin{equation*}
 \exp\left(-\frac{1}{2\pi}\partial_tW_{\mathrm{GV}}(\epsilon,t)\right)=\big(L(x\mid \check\epsilon,1)^{-1}\cdot \Psi_{\mathrm{GV}}(\epsilon,x,t)\big)\big|_{x=0}.
\end{equation*}
Here we have introduced the notation $W_{{\rm GV}}(\epsilon,t):=W(\epsilon,t)$.
\end{Corollary}

\begin{proof}
Using our previous results, we have
\begin{align*}
 \lim_{k \to \infty} \mathcal{S}_2(x\mid \check{\epsilon},1)^{-1} \cdot \Psi_{\rho_k}(\epsilon,x,t)|_{x=0} &= \lim_{k \to \infty} \mathcal{S}_2(t+k\mid \check{\epsilon},1)^{-1} \\
 & = L(t,\check{\epsilon})^{-1} = \big(L(x,\check{\epsilon})^{-1}\cdot\Psi_{\mathrm{GV}}(\epsilon,x,t)\big)\big|_{x=0}.
 \end{align*}
In the computation of the limit above, we have used equation~\eqref{S2l} together with
\begin{equation*}
 \frac{\mathcal{S}_2(x+1\mid \check\epsilon,1)}{\mathcal{S}_2(x\mid \check\epsilon,1)}=\frac{1}{1-{\rm e}^{2\pi \mathrm{i} x/\check\epsilon}}.
\end{equation*}
The result then follows from Proposition \ref{rblimit1}.
\end{proof}

\begin{Remark}\label{remper}
In the setting of the previous theorem, one has
\begin{equation*}
 \Psi_{\rho_{k}}(\epsilon,x(s),t)=\Psi_{{\rm np}}(\epsilon,x(s),t+k)=\frac{\mathcal{S}_2(x(s)\mid \check\epsilon,1)}{\mathcal{S}_2(x(s)+t+k\mid \check\epsilon,1)}
\end{equation*}
for $0<s\leq s_*$, since $\rho_k$ is of the form $\rho_{k,0}$.
We can use the right-hand side of the above equation to analytically continue $\Psi_{\rho_{k}}(\epsilon,x,t)$ by taking $s\to \infty$. If we further assume that $\operatorname{Im}(\check\epsilon)>0$, we can use the product formula of $\mathcal{S}_2$ given in~(\ref{S2prod}) to conclude that
\begin{equation*}
 \Psi_{\rho_{k}}(\epsilon,x(s),t)|_{s=\infty}=1.
\end{equation*}

In Section \ref{5dWKBsec}, we will consider exponentiated trajectories of the form $X(s)={\rm e}^{2\pi {\rm i}x(s)}={\rm e}^{-2\pi s{\rm e}^{{\rm i}\theta_*}}$ connecting $X=0$ and $X=1$. The previous result motivates the definition of the following ``regularized" period:

\begin{Definition}\label{def:regBperiod}
In the setting of Theorem \ref{theorem3} and Corollary \ref{corthm3}, we define the regularized quantum period
\begin{equation*}
\Pi_{B,\rho_k}^{\rm reg} := \log \frac{\big(\mathcal{S}_2(x\mid \check\epsilon,1)^{-1}\cdot \Psi_{\rho_k}(\epsilon,x,t)\big)\big|_{x=0}}{\Psi_{\rho_k}(\epsilon,x,t)|_{x=\mathrm{i}\infty}},
\end{equation*}
and its $k \to \infty$ limit
\begin{equation*}
\Pi_{B,\mathrm{GV}}^{\rm reg} := \lim_{k \to \infty} \Pi_{B,\rho_k}^{\rm reg} = \log \frac{\big(L(x\mid \check\epsilon,1)^{-1}\cdot \Psi_{\mathrm{GV}}(\epsilon,x,t)\big)\big|_{x=0}}{\Psi_{\mathrm{GV}}(\epsilon,x,t)|_{x=\mathrm{i}\infty}}.
\end{equation*}
\end{Definition}

One can then interpret equation~\eqref{relWpsi} as a relation between the NS free energies $W$ and the corresponding regularized exact quantum $B$-periods $\Pi_{B}^{\rm reg}$. Namely,
\begin{equation*}
 \exp\left(-\frac{1}{2\pi}\partial_tW_{\rho_k}(\epsilon,t)\right)=\exp\big(\Pi_{B,\rho_k}^{\rm reg}\big),
\end{equation*}
and its $k \to \infty$ limit
\begin{equation*}
 \exp\left(-\frac{1}{2\pi}\partial_tW_{\mathrm{GV}}(\epsilon,t)\right)=\exp\big(\Pi_{B,\mathrm{GV}}^{\rm reg}\big),
\end{equation*}
with again $W_{{\rm GV}}(\epsilon,t):=W(\epsilon,t)$.
This perspective will be useful later in Section~\ref{5dWKBsec}.
\end{Remark}

\section{DT invariants and line bundles}\label{sec:DT}

In \cite{ASTT21}, the Borel summability of the Gromov--Witten (GW) potential for the resolved conifold is studied, giving the following result for the Borel sum along a ray $\rho$ \cite[equation~(4.35)]{ASTT21}:
\begin{equation}\label{normalizedpartfunc}
 \widehat{F}_{\rho}(\lambda,t)=F_{\rho}(\lambda,t)-F_{\rho}(\lambda,0)-\frac{\log(\check{\lambda})}{12}.
\end{equation}
Here $t$ denotes the K\"{a}hler parameter as before, and $\lambda$ ($\check{\lambda}=\lambda/2\pi$) denotes the topological string coupling. $F_{\rho}(\lambda,t)$ denotes the Borel sum of the non-constant map contribution of the GW potential, while $-F_{\rho}(\lambda,0)-\log\big(\check{\lambda}\big)/12$ is shown to be (up to the addition of a constant) the Borel sum of constant map contribution of the GW potential. $F_{\rho}(\lambda,0)$ is defined in terms of~$F_{\rho}(\lambda,t)$ by a certain limit in $t\to 0$ (see \cite[Section~4.2.1]{ASTT21} for more details).

The Borel sum $\hat{F}_{\rho}$ experiences Stokes jumps in the $\lambda$-plane along the rays $\pm l_{k}$ and $\pm l_{\infty}$ from before, with the jumps containing the information of the DT invariants of the resolved conifold in the following way (we assume below that $\operatorname{Im}(t)>0$):

\begin{itemize}
 \item Let $M$ be the space parametrized by the $t$-parameter and consider $\Gamma\to M$ a trivial rank-$2$ local system of lattices spanned by $\beta$ and $\delta$. The BPS indices of the resolved conifold are then expressed as
 \begin{equation*}
 \Omega(\gamma) =
 \begin{cases}
 \hphantom{-}1 & \quad \text{if } \gamma = \pm \beta + n \delta\quad \text{for } n\in \mathbb{Z},\\
 -2 & \quad \text{if } \gamma= k \delta \quad \text{for } k\in \mathbb{Z}\setminus \{ 0\},\\
 \hphantom{-}0 & \quad \text{otherwise}.
 \end{cases}
 \end{equation*}
 We further consider the central charge function $Z$ (a holomorphic section of ${\Gamma^*\otimes \mathbb{C} \to M}$) defined by
 \begin{equation*}
 Z_{n\beta+m\delta}=2\pi \mathrm{i}(nt+m), \qquad n,m \in \mathbb{Z}.
 \end{equation*}
 \item Along $\pm l_k= \pm \mathbb{R}_{<0}\cdot 2\pi \mathrm{i}(t+k)=\pm \mathbb{R}_{<0}\cdot Z_{\beta +k\delta}$ the jump of $\hat{F}_{\rho}$ encodes $\Omega(\beta+k\delta)$ as
 \begin{equation}\label{jumpsDT1}
 \frac{1}{2\pi \mathrm{i}}\partial_{\check{\lambda}}\big(\check{\lambda}\mathrm{Li}_2\big({\rm e}^{\pm 2\pi \mathrm{i}(t+k)/\check{\lambda}}\big)\big)=\frac{\Omega(\beta+k\delta)}{2\pi \mathrm{i}}\partial_{\check{\lambda}}\big(\check{\lambda}\mathrm{Li}_2\big({\rm e}^{\pm Z_{\beta+k\delta}/\check{\lambda}}\big)\big).
 \end{equation}
 \item Along $\pm l_{\infty}=\pm \mathrm{i}\mathbb{R}_{<0}= \mathbb{R}_{<0}\cdot Z_{\pm k\delta}$ the jumps of $\hat{F}_{\rho}$ encode the $\Omega(k\delta)$ via
 \begin{equation}\label{jumpsDT2}
 \sum_{k\geq 1}-\frac{1}{\pi \mathrm{i}}\partial_{\check{\lambda}}\big(\check{\lambda}\mathrm{Li}_2\big({\rm e}^{\pm 2\pi \mathrm{i} k/\check{\lambda}}\big)\big)\mp \frac{\pi \mathrm{i}}{12}=\sum_{k\geq 1}\frac{\Omega(k\delta)}{2\pi \mathrm{i}}\partial_{\check{\lambda}}\big(\check{\lambda}\mathrm{Li}_2\big({\rm e}^{\pm Z_{k\delta}/\check{\lambda}}\big)\big)\mp \frac{\pi \mathrm{i}}{12}.
 \end{equation}
\end{itemize}
We remark that the interpretation of the coefficients in front of the $\mathrm{Li}_2$ summands as the DT invariants comes from the relation done in \cite[Section~4.2]{ASTT21} to the Riemann--Hilbert (RH) problem associated to the resolved conifold \cite{BridgelandCon}. It is shown that the jumps of $\hat{F}_{\rho}$ serve as potentials for the jumps of the RH problem, involving the DT invariants \cite[Corollary~4.13]{ASTT21}.

We would like to show that a properly normalized $W_{\rho}(\epsilon,t)$ also encodes the DT invariants of the resolved conifold, in a similar way. Given the relation
\begin{equation*}
 \partial_{\epsilon}W_{\rho}(\epsilon,t)=F_{\rho}(\epsilon,t)
\end{equation*}
shown in \eqref{pot}, together with \eqref{normalizedpartfunc},
a natural object to consider is
\begin{equation}\label{normalizedW}
 \widehat{W}_{\rho}(\epsilon,t):=W_{\rho}(\epsilon,t)-W_{\rho}(\epsilon,0) -\frac{\epsilon}{12}\log(\check\epsilon)+\frac{\pi}{6}\epsilon,
\end{equation}
where $W_{\rho}(\epsilon,0)$ is defined by a certain limit $t\to 0$. We discuss this limit and the Stokes jumps of~$\hat{W}_{\rho}(\epsilon,t)$ below. We then establish the link between the jumps of $\hat{W}_{\rho}(\epsilon,t)$ and the DT invariant of the resolved conifold.

On the other hand, in \cite[Section 4]{ASTT21}, a certain projectivized version of $\hat{F}_{\rho}$ is shown to define a~section of a~line bundle, related to a conformal limit of a hyperholomorphic line bundle discussed in \cite{APP, Neitzke_hyperhol}. We show below that (a projectived version of)~$\hat{W}_{\rho}(\epsilon,t)$ can be thought as defining a~section of the same line bundle as the one associated to $\hat{F}_{\rho}(\epsilon,t)$.

\subsection[Normalized $W_{\rho}$ and the DT invariants]{Normalized $\boldsymbol{W_{\rho}}$ and the DT invariants}\label{sec:normalizedW}

In this section, we wish to show how the DT invariants of the resolved conifold as encoded in the Stokes jumps of a properly normalized $W_{\rho}(\epsilon,t)$. To show how the $-2$ BPS index is contained in the Borel sums $W_{\rho}(\epsilon,t)$, we will first study a limit of the form
\begin{equation*}
 W_{\rho}(\epsilon,0):= \lim_{t\to 0}W_{\rho}(\epsilon,t),
\end{equation*}
where $t$ is taken to satisfy $\operatorname{Re}(t)>0$, $\operatorname{Im}(t)>0$; and such that along the limit, $\rho$ is always between $l_{-1}$ and $l_0$ (resp.\ $-l_{-1}$ and $-l_0$) if $\rho$ is on the right (resp.\ left) Borel half-plane. This is analogous to the way $F_{\rho}(\lambda,0)$ was defined in \cite{ASTT21}.

To see that this gives a well defined limit, notice that when $\rho$ is on the right Borel plane we can write
\begin{align*}
 W_{\rho}(\epsilon,t)&=-\frac{1}{\epsilon} \mathrm{Li}_{3}(Q)+ \frac{\mathrm{i}}{2}\mathrm{Li}_{2}(Q)+\frac{\epsilon}{12}\mathrm{Li_1}(Q) + \check{\epsilon}\int_{\rho}\mathrm{d}\xi\, {\rm e}^{-\xi/\check{\epsilon}}\widetilde{G}(\xi,t)\\
 &=-\frac{1}{\epsilon} \mathrm{Li}_{3}(Q)+ \frac{\mathrm{i}}{2}\mathrm{Li}_{2}(Q) +\check{\epsilon}\int_{\rho}\mathrm{d}\xi \left({\rm e}^{-\xi/\check{\epsilon}}\widetilde{G}(\xi,t)+\frac{\pi}{6}\frac{1}{{\rm e}^{\xi-2\pi {\rm i}t}-1}\right).
 \end{align*}
Now when $t\to 0$ as above, we see that ${\rm e}^{-\xi/\check\epsilon}\widetilde{G}(\xi,0)$ develops a simple pole at $\xi=0$ of the from $-\frac{\pi}{6\xi}$, which gets canceled with the simple pole arising from the second term. Furthermore, no Stokes rays $l_k$ cross $\rho$ in the above limit in $t$, so we get a well defined limit. When $\rho$ is on the left Borel plane, we can reduce to the previous case by using the identity $W_{\rho}(\epsilon,t)=-W_{-\rho}(-\epsilon,t)+\mathrm{i}\mathrm{Li}_2(Q)$. Hence, we have a well defined limit
\begin{equation*}
 W_{\rho}(\epsilon,0)=-\frac{1}{\epsilon} \mathrm{Li}_{3}(1)+ \frac{\mathrm{i}}{2}\mathrm{Li}_{2}(1) +\lim_{t\to 0}\left(\check{\epsilon}\int_{\rho}\mathrm{d}\xi \left({\rm e}^{-\xi/\check{\epsilon}}\widetilde{G}(\xi,t)+\frac{\pi}{6}\frac{1}{{\rm e}^{\xi-2\pi \mathrm{i} t}-1}\right)\right),\\
\end{equation*}
where $t$ satisfies the constraints specified above.

The following proposition suggests that one can obtain the appropriate Stokes jumps at $\pm l_{\infty}$ by considering a normalization of $W_{\rho}(\epsilon,t)$ involving $W_{\rho}(\epsilon,0)$:
\begin{Proposition}\label{linftyjumps}
Let $\rho$ $($resp.\ $\rho')$ be a ray close to $l_{\infty}=\mathrm{i} \mathbb{R}_{<0}$ from the left $($resp.\ right$)$. Then for $\epsilon$ in their common domain of definition
\begin{equation*}
 W_{\pm \rho}(\epsilon,0)-W_{\pm \rho'}(\epsilon,0)=-2\mathrm{i}\check{\epsilon}\sum_{k\geq 1} \mathrm{Li}_2\big({\rm e}^{\pm 2\pi \mathrm{i} k/\check\epsilon}\big) -\frac{\pi \mathrm{i} }{12}\epsilon.
\end{equation*}

Furthermore $W_{\rho}(\epsilon,0)$ only has Stokes jumps along $\pm l_{\infty}$.
\end{Proposition}
\begin{proof}
First, notice that by our definition of the limit in $t$, we have
\begin{equation*}
\begin{split}
 W_{\rho}(\epsilon,0)-W_{\rho'}(\epsilon,0)
 & =\lim_{t\to 0}\bigg(\int_{\rho}{\rm d}\xi\, {\rm e}^{-\xi/\check{\epsilon}} \widetilde{G}(\xi,t)-\int_{\rho'}{\rm d}\xi\, {\rm e}^{-\xi/\check{\epsilon}} \widetilde{G}(\xi,t)\bigg)\\
& =\lim_{t\to 0}\int_{\mathcal{H}}{\rm d}\xi \,{\rm e}^{-\xi/\check{\epsilon}}\widetilde{G}(\xi,t),
\end{split}
 \end{equation*}
where $\mathcal{H}=\rho-\rho'$ denotes a Hankel contour along $\mathrm{i}\mathbb{R}_{<0}$, containing $l_k$ for $k\geq 0$ and $-l_{k}$ for $k<0$. Hence, for $\epsilon$ close to $l_{\infty}$, the Hankel contour just gives the contribution of these rays that we previously computed:
\begin{align*}
 W_{\rho}(\epsilon,0)-W_{\rho'}(\epsilon,0)
 &=\lim_{t\to 0}\int_{\mathcal{H}}{\rm d}\xi \, {\rm e}^{-\xi/\check{\epsilon}}\widetilde{G}(\xi,t)\\
 &=\lim_{t\to 0} \biggl(-\mathrm{i}\sum_{k\geq 1}\big[ \check\epsilon \mathrm{Li}_2\big({\rm e}^{2\pi \mathrm{i}(t+k)/\check\epsilon}\big)+\check\epsilon \mathrm{Li}_2\big({\rm e}^{-2\pi \mathrm{i}(t-k)/\check\epsilon}\big)\big]
 -\mathrm{i}\check\epsilon \mathrm{Li}_2({\rm e}^{2\pi \mathrm{i} t/\check\epsilon})\biggr)\\
 &=-2\mathrm{i}\check{\epsilon}\sum_{k\geq 1} \mathrm{Li}_2\big({\rm e}^{2\pi \mathrm{i} k/\check\epsilon}\big) -\mathrm{i}\check\epsilon \mathrm{Li}_2(1)
 =-2\mathrm{i}\check{\epsilon}\sum_{k\geq 1} \mathrm{Li}_2\big({\rm e}^{2\pi \mathrm{i} k/\check\epsilon}\big) -\frac{\pi \mathrm{i} }{12}\epsilon.
 \end{align*}

A similar argument follows for $-l_{\infty}=\mathrm{i}\mathbb{R}_{>0}$. Furthermore, the fact that there are no other Stokes jumps follows from the way we have defined the limit $W_{\rho}(\epsilon,0)$.
\end{proof}

By the previous arguments, we would like to now consider:
\begin{equation*}
 \widehat{W}_{\rho}(\epsilon,t)=W_{\rho}(\epsilon,t)-W_{\rho}(\epsilon,0) -\frac{\epsilon}{12}\log(\check\epsilon)+\frac{\pi}{6}\epsilon,
\end{equation*}
where the $\log$-term has a branch cut at $l_{\infty}$, and study its Stokes jumps.

\begin{Corollary}
The Stokes jumps of $\widehat{W}_{\rho}(\epsilon,t)$ along $\pm l_{k}$ are given by the same jumps as $W_{\rho}(\epsilon,t)$, while the Stokes jumps at $\pm l_{\infty}$ are given by
\begin{equation*}
 2\mathrm{i}\check{\epsilon}\sum_{k\geq 1} \mathrm{Li}_2\big({\rm e}^{\pm 2\pi \mathrm{i} k/\check\epsilon}\big) \mp \frac{\pi \mathrm{i} }{12}\epsilon.
\end{equation*}
In particular, the $\epsilon$-derivatives of the Stokes jumps of $\hat{W}_{\rho}$ give the Stokes jumps of $\hat{F}_{\rho}$.
\end{Corollary}

Joining the results of the previous corollary, together with the way the DT invariants are encoded in terms of the Stokes jumps of $\hat{F}_{\rho}$, given in (\ref{jumpsDT1}) and (\ref{jumpsDT2}), we find that:

\begin{Corollary}
The DT invariants of the resolved conifold are encoded in the Stokes jumps of~$\hat{W}_{\rho}$ in the following way:

\begin{itemize}
\item Along $\pm l_k= \pm \mathbb{R}_{<0}\cdot 2\pi \mathrm{i}(t+k)=\pm \mathbb{R}_{<0}\cdot Z_{\beta +k\delta}$ the jump of $\hat{W}_{\rho}$ encodes $\Omega(\beta+k\delta)$ as
 \begin{equation*}
 -\mathrm{i}\check{\epsilon}\mathrm{Li}_2\big({\rm e}^{\pm 2\pi \mathrm{i}(t+k)/\check{\epsilon}}\big)=\frac{ \Omega(\beta+k\delta)}{2\pi \mathrm{i}}\epsilon \mathrm{Li}_2\big({\rm e}^{\pm Z_{\beta+k\delta}/\check{\epsilon}}\big).
 \end{equation*}
 \item Along $\pm l_{\infty}=\pm \mathrm{i}\mathbb{R}_{<0}= \mathbb{R}_{<0}\cdot Z_{\pm k\delta}$ the jumps of $\hat{W}_{\rho}$ encode the $\Omega(k\delta)$ via
 \begin{equation*}
 2\mathrm{i}\check{\epsilon}\sum_{k\geq 1} \mathrm{Li}_2\big({\rm e}^{\pm 2\pi \mathrm{i} k/\check\epsilon}\big) \mp \frac{\pi \mathrm{i} }{12}\epsilon =\sum_{k\geq 1}\frac{\Omega(k\delta)}{2\pi \mathrm{i}}\epsilon \mathrm{Li}_2\big({\rm e}^{\pm Z_{k\delta}/\check{\epsilon}}\big)\mp \frac{\pi \mathrm{i} }{12}\epsilon.
 \end{equation*}
\end{itemize}
\end{Corollary}

\subsection[Line bundle defined by $\hat{W}_{\rho}$]{Line bundle defined by $\boldsymbol{\hat{W}_{\rho}}$}

In \cite{ASTT21}, it turned out to be convenient to projectize the K\"ahler parameter $t$ and the topological string coupling $\lambda$ as
\begin{equation*}
 \lambda = 2\pi \lambda_B/w, \qquad t=v/w,
\end{equation*}
in order to relate to the Riemann--Hilbert problem considered in \cite{Bridgeland1,BridgelandCon}. Furthermore, the projectivized partition function
\begin{equation*}
 Z^{{\rm top}}_{\rho}(v,w,\lambda_B):=\exp\big(\widehat{F}_{w^{-1}\cdot \rho}(2\pi\lambda_B/w, v/w)\big)
\end{equation*}
was considered, which reduces to the usual partition function when $w=1$. That is,
\begin{equation*}
 Z^{{\rm top}}_{\rho}(v,1,\lambda_B):=\exp\big(\widehat{F}_{\rho}(2\pi\lambda_B, v)\big)=\exp\big(\widehat{F}_{\rho}(\lambda, t)\big).
\end{equation*}

The projectivized free energy $\widehat{F}_{w^{-1}\cdot \rho}(2\pi\lambda_B/w, v/w)$ now has Stokes jumps in the $\lambda_B$ variable along $\pm \mathcal{L}_k:= \pm w\cdot l_k$ and $\pm \mathcal{L}_{\infty}:= \pm w\cdot l_{\infty}$. Provided $\operatorname{Im}(v/w)>0$, these are given by
\begin{itemize}\itemsep=0pt
\item Along $\pm \mathcal{L}_k$ the jump is
 \begin{equation*}
 \frac{1}{2\pi \mathrm{i}}\partial_{\lambda_B}\big(\lambda_B\mathrm{Li}_2\big({\rm e}^{\pm 2\pi \mathrm{i}(v+kw)/\lambda_B}\big)\big).
 \end{equation*}
 \item Along $\pm \mathcal{L}_{\infty}$ the jump is
 \begin{equation*}
 \sum_{k\geq 1}-\frac{1}{\pi \mathrm{i}}\partial_{\lambda_B}\big(\lambda_B\mathrm{Li}_2\big({\rm e}^{\pm 2\pi {\rm i}kw/\lambda_B}\big)\big)\mp \frac{\pi \mathrm{i}}{12}.
 \end{equation*}
\end{itemize}

The exponentials of such jumps were then interpreted as specifying the transition functions of a line bundle $\mathcal{L}^{{\rm top}}\to \mathbb{C}^{\times}\times M_{+}$ having a global section specified by the Borel sums $Z_{\rho}^{{\rm top}}(v,w,\lambda_B)$, and where
\begin{equation*}
 M_{+}:=\{(v,w) \mid \operatorname{Im}(v/w)>0\}.
\end{equation*}

The line bundle $\mathcal{L}^{{\rm top}}$ was then shown to correspond to a certain conformal limit of hyperholomorphic line bundles previously considered in \cite{APP,Neitzke_hyperhol}.

In this section we wish to consider an analogous projectivized version of $\hat{W}_{\rho}(\epsilon,t)$, and show that it defines a section of the same line bundle from before.

We start by recalling more specifically how $\mathcal{L}^{{\rm top}}\to \mathbb{C}^{\times}\times M_+$ is defined. Let $\rho_k$ be the middle ray between $\mathcal{L}_k$ and $\mathcal{L}_{k-1}$, and consider the open subsets
\begin{equation*}
 U_k^{\pm}:=\{(\lambda_B,v,w)\in \mathbb{C}^{\times}\times M_{+} \mid \lambda_B\in \mathbb{H}_{\pm \rho_k}\}.
\end{equation*}

Then $\{U_k^+\}_{k\in \mathbb{Z}}\cup \{U_k^-\}_{k\in \mathbb{Z}}$ form an open covering of $\mathbb{C}^{\times}\times M_{+}$ and $Z^{{\rm top}}_{\pm \rho_k}(\lambda_B,v,w)$ is defined on $U_k^{\pm}$. We then define a 1-\v{C}ech cocycle specifying $\mathcal{L}^{{\rm top}}$ and associated to the previous cover as follows:
\begin{itemize}\itemsep=0pt
\item If $U_{k_1}^{\pm}\cap U_{k_2}^{\pm }\neq \varnothing$ for $k_1<k_2$, we then define for $(\lambda_{\mathrm B},v,w)\in U_{k_1}^{\pm}\cap U_{k_2}^{\pm }$,
\begin{equation*}
 g^{\pm}_{k_1,k_2}(\lambda_{\mathrm B},v,w):=\prod_{k_1\leq k < k_2}\Xi_{\pm \mathcal{L}_k}(\lambda_{\mathrm B},v,w),
\end{equation*}
where
\begin{equation*}
 \Xi_{\pm \mathcal{L}_k}(\lambda_{\mathrm B},v,w):=\exp \left(\frac{1}{2\pi \mathrm{i}}\partial_{\lambda_B}\big(\lambda_B\mathrm{Li}_2\big({\rm e}^{\pm 2\pi \mathrm{i}(v+kw)/\lambda_B}\big)\big)\right).
\end{equation*}

\item On the other hand, if for some $k_1,k_2\in \mathbb{Z}$ we have $U_{k_1}^{+}\cap U_{k_2}^{-}\neq \varnothing$, then $\rho_{k_1}\neq \rho_{k_2}$ and hence out of the two sectors determined by $\rho_{k_1}$ and $-\rho_{k_2}$ there is a smallest one, which we denote by $[\rho_{k_1},-\rho_{k_2}]$. For all $(\lambda_{\mathrm B},v,w)\in U_{k_1}^{+}\cap U_{k_2}^{-}$ we must either have that $\mathcal{L}_{\infty} \subset [\rho_{k_1},-\rho_{k_2}]$ or $- \mathcal{L}_{\infty} \subset [\rho_{k_1},-\rho_{k_2}]$. In the first case we define
\begin{equation*}
 g_{k_1,k_2}^{\infty}(\lambda_{\mathrm B},v,w):={\rm e}^{-\pi\mathrm i/12}\prod_{k\geq k_1}\Xi_{\mathcal{L}_k}(\lambda_{\mathrm B},v,w)\prod_{k< k_2}\Xi_{-\mathcal{L}_{k}}(\lambda_{\mathrm B},v,w)\prod_{k\geq 1}\Xi_{\mathcal{L}_{\infty},k}(\lambda_{\mathrm B},v,w),
\end{equation*}
and $g_{k_2,k_1}^{\infty}:=\big(g_{k_1,k_2}^{\infty}\big)^{-1}$, where
\begin{equation*}
 \Xi_{\pm \mathcal{L}_{\infty},k}(v,w,\lambda_\mathrm B)=\exp\left(-\frac{1}{\pi \mathrm{i}}\partial_{\lambda_{\mathrm B}}\big(\lambda_{\mathrm B} \mathrm{Li}_2\big({\rm e}^{\pm kw/ \lambda_{\mathrm B}}\big)\big)\right) ;
\end{equation*}
while in the second case we define
\begin{equation*}
 g_{k_2,k_1}^{-\infty}(\lambda_{\mathrm B},v,w):={\rm e}^{\pi \mathrm{i}/12}\prod_{k\geq k_2}\Xi_{-\mathcal{L}_{k}}(\lambda_{\mathrm B},v,w)\prod_{k< k_1}\Xi_{\mathcal{L}_{k}}(\lambda_{\mathrm B},v,w)\prod_{k\geq 1}\Xi_{-\mathcal{L}_{\infty},k}(\lambda_{\mathrm B},v,w),
\end{equation*}
and $g_{k_1,k_2}^{-\infty}:=\big(g_{k_2,k_1}^{-\infty}\big)^{-1}$.
\end{itemize}

We now define a projectived version of $W_{\rho}(\epsilon,t)$. As before, we consider the projectivized variables
\begin{equation*}
 \epsilon=2\pi \epsilon_B/w, \qquad t=v/w,
\end{equation*}
and
\begin{equation*}
 Z_{\rho}^{\rm NS}(v,w,\epsilon_B):=\exp\big(\widehat{W}_{w^{-1}\cdot \rho}(2\pi\epsilon/w, v/w)\big).
\end{equation*}

From the same analysis as before, we obtain that the jumps of $\widehat{W}_{w^{-1}\cdot \rho}(2\pi\epsilon/w, v/w)$ along~$\pm \mathcal{L}_k$ and~$\mathcal{L}_{\infty}$ are given by
\begin{itemize}\itemsep=0pt
\item Along $\pm \mathcal{L}_k$ the jump is
 \begin{equation*}
 -\mathrm{i}\check{\epsilon}\mathrm{Li}_2\big({\rm e}^{\pm 2\pi \mathrm{i}(v+kw)/\epsilon_B}\big).
 \end{equation*}
 \item Along $\pm \mathcal{L}_{\infty}$ the jump is
 \begin{equation*}
 \sum_{k\geq 1} 2\mathrm{i}\check{\epsilon}\mathrm{Li}_2\big({\rm e}^{\pm 2\pi {\rm i}kw/\epsilon_B}\big)\mp \frac{\pi \mathrm{i} }{12}\frac{2\pi \epsilon_B}{w}.
 \end{equation*}
\end{itemize}

Under the identification $\epsilon_B=\lambda_B$, we have that $Z_{\pm \rho_k}^{\rm NS}$ is defined on $U_k^{\pm}$, and by using the exponentials of the Stokes jumps of $Z_{\pm \rho_k}^{\rm NS}$ we can as before construct a line bundle $\mathcal{L}^{\rm NS}\to \mathbb{C}^{\times}\times M_{+}$ such that $Z_{\pm \rho_k}^{\rm NS}$ define a section of $\mathcal{L}^{\rm NS}$.

To show that $\mathcal{L}^{\rm NS}=\mathcal{L}^{{\rm top}}$ it is enough to show that we can scale each section $Z_{\pm \rho_k}^{\rm NS}$ over $U_{k}^{\pm}$ by a function defined on $U_k^{\pm}$ such that we recover the same transition functions of $\mathcal{L}^{{\rm top}}$.

For this we will use the functions $f_{\pm \rho_k}(\epsilon_B,v,w)\colon U_k^{\pm} \to \mathbb{C}$ defined in \cite[Lemma 4.18]{ASTT21}, which satisfy
\begin{equation*}
 f_{\pm \rho_{k+1}}(\epsilon_B,v,w)-f_{\pm \rho_k}(\epsilon_B,v,w)=\pm \frac{2\pi \mathrm{i}(v+kw)}{\epsilon_B}\log\big(1-{\rm e}^{\pm 2\pi \mathrm{i}(v+kw)/\epsilon_B}\big).
\end{equation*}
We then scale
\begin{equation*}
 Z_{\pm \rho_k}^{\rm NS}\to \widetilde{Z}_{\pm \rho_k}^{\rm NS}:=s_{\pm \rho_k}Z_{\pm \rho_k}^{\rm NS},
\end{equation*}
where $s_{\pm \rho_k}\colon U_k^{\pm} \to \mathbb{C}^{\times}$ is defined by
\begin{equation*}
 s_{\pm \rho_k}(\epsilon_B,v,w):=\exp\big((1-\epsilon)\cdot\widehat{F}_{\pm w^{-1}\cdot \rho_k}(2\pi\lambda_B/w, v/w)\big) \exp\left(\frac{\epsilon}{2\pi \mathrm{i}}f_{\pm \rho_k}(\epsilon_B,v,w) \right),
\end{equation*}
and where we recall that $\epsilon=2\pi \epsilon_B/w$.

One then finds that
\begin{gather*}
 \log\big(\widetilde{Z}_{\pm \rho_{k+1}}^{\rm NS}\big) -\log\big(\widetilde{Z}_{\pm \rho_k}^{\rm NS}\big) \\
 \qquad{} = \frac{(1-\epsilon)}{2\pi \mathrm{i}}\left(\mathrm{Li}_2\big({\rm e}^{ \pm 2\pi \mathrm{i}(v+kw)/\epsilon_B}\big)\pm \frac{ 2\pi \mathrm{i}(v+kw)}{\epsilon_B}\log\big(1-{\rm e}^{ \pm 2\pi \mathrm{i}(v+kw)/\epsilon_B}\big)\right)\\
 \qquad\quad{} \pm \frac{\epsilon}{2\pi \mathrm{i}}\left(\frac{2\pi \mathrm{i}(v+kw)}{\epsilon_B}\log\big(1-{\rm e}^{\pm 2\pi \mathrm{i}(v+kw)/\epsilon_B}\big)\right)
 +\frac{\epsilon}{2\pi \mathrm{i}}\mathrm{Li}_2\big({\rm e}^{\pm 2\pi \mathrm{i}(v+kw)/\epsilon_B}\big)\\
\qquad{} = \frac{1}{2\pi \mathrm{i}}\left(\mathrm{Li}_2\big({\rm e}^{\pm 2\pi \mathrm{i}(v+kw)/\epsilon_B}\big)\pm \frac{2\pi \mathrm{i}(v+kw)}{\epsilon_B}\log\big(1-{\rm e}^{\pm 2\pi \mathrm{i}(v+kw)/\epsilon_B}\big)\right)\\
 \qquad{} = \frac{1}{2\pi \mathrm{i}} \partial_{\epsilon_B}\big(\epsilon_B \mathrm{Li}_2\big({\rm e}^{\pm 2\pi \mathrm{i}(v+kw)/\epsilon_B}\big)\big).
 \end{gather*}
This implies that $\widetilde{Z}_{\pm \rho_{k+1}}^{\rm NS}$ and $\widetilde{Z}_{\pm \rho_{k}}^{\rm NS}$ are related by the corresponding transition function of $\mathcal{L}^{{\rm top}}$ on $U_{k+1}^{\pm}\cap U_{k}^{\pm}$. The remaining transition functions (i.e. the ones corresponding to $U_{k_1}^{+}\cap U_{k_2}^{-}$) follow similarly, showing the following:

\begin{Proposition}
The line bundle $\mathcal{L}^{\rm NS}\to \mathbb{C}^{\times}\times M_+$ defined by the Borel sums $Z_{\rho}^{\rm NS}$ is the same as the line bundle $\mathcal{L}^{{\rm top}}\to \mathbb{C}^{\times}\times M_+$ defined by the Borel sums $Z_{\rho}^{{\rm top}}$.
\end{Proposition}

\section{Gauge theory, exact WKB and integrable systems}\label{5dWKBsec}

In this section, we interpret the previously obtained results in terms of exact WKB analysis, five-dimensional gauge theory in the $\frac{1}{2}\Omega$-background and quantum integrable systems. In Section~\ref{sec:BPS}, we summarize the resurgence results from Section~\ref{sec:Borelsums} in terms of BPS states in M-theory and gauge theory. In Section~\ref{sec:spectralnetwork}, we see how these BPS states are encoded in exponential spectral networks. Here we also introduce two special examples, called $\mathcal{W}_\mathrm{GV}$ and $\mathcal{W}_\mathrm{np}$. In Section~\ref{sec:exactWKB}, we explain how the quantum vevs $\Psi_\rho$ may be interpreted as local sections in an extension of the exact WKB analysis to difference operators. In Section~\ref{sec:abelianization}, we define two types of spectral coordinates by abelianizing with respect to the networks $\mathcal{W}_\mathrm{GV}$ and $\mathcal{W}_\mathrm{np}$. In Section~\ref{sec:NRS}, we interpret the relation between $\Psi_\rho$ and $F^{\rm NS}_\rho$ stated in Theorem~\ref{theorem3} in terms of a five-dimensional analogue of the Nekrasov--Rosly--Shatahvili proposal \cite{Nekrasov:NRS}, while in Section~\ref{sec:spectralproblem}, we formulate two spectral problems associated to the two special networks $\mathcal{W}_\mathrm{GV}$ and $\mathcal{W}_\mathrm{np}$.

\subsection{Resurgence and BPS states}\label{sec:BPS}

In Section~\ref{sec:Borelsums}, we calculated and analyzed the Borel sum
\begin{equation*}
W_\rho(\epsilon,t) = F_\rho^{\rm NS}(\epsilon, t - \check{\epsilon}/2)
\end{equation*}
of the NS free energy in the resolved conifold geometry, along any ray $\rho$ in the $\epsilon$-plane where its Borel transform does not have any singularities. As stated in Theorem~\ref{maintheorem}, these singularities lie along an infinite set of rays
\begin{equation*}
\pm l_k = \pm \mathbb{R}_{< 0} \cdot 2 \pi \mathrm{i} (t+k),
\end{equation*}
together with and $\pm l_\infty = \pm \mathrm{i} \mathbb{R}_{< 0}$, and $W_\rho$ jumps across the ray $\pm l_k$ with a Stokes factor{\samepage
\begin{equation*}
\Delta W_\rho(\epsilon,t) = -\mathrm{i}\check{\epsilon} \mathrm{Li}_2\big({\rm e}^{\pm 2\pi \mathrm{i}(t + k)/\check \epsilon}\big),
\end{equation*}
when $\operatorname{Im}(t)>0$.}

Furthermore, when additionally $0 < \operatorname{Re}(t) <1$, the Borel sum of $W_\rho$ along any ray $\rho_k$ between~$l_k$ and~$l_{k-1}$ can be written in terms of the non-perturbative free energy $F_{\rm np}^{\rm NS}$ (first encountered in equation~\eqref{eqn:FNSnpfirst}) as
\begin{equation*}
 W_{\rho_k}(\epsilon,t) = F_{\rm np}^{\rm NS}(\epsilon, t - \check{\epsilon}/2 + k),
\end{equation*}
as stated in Corollary~\ref{cortheorem1}.

Note that all $W_{\rho_k}(\epsilon,t)$ are solutions to the difference equation~\eqref{eq:NSdiffeq}, and that in the limit $k \to \infty$ we retrieve the Gopakumar--Vafa expression (see (\ref{eq:limitImW}))
\begin{align*}
 W_\mathrm{GV}(\epsilon,t) := W(\epsilon,t) = \lim_{k \to \infty} W_\mathrm{np}(\epsilon,t+k).
\end{align*}

In Section~\ref{sec:DT}, we interpreted the jumps geometrically in terms of DT invariants and a distinguished line bundle. In particular, we found that each ray $l_k$ is associated to a bound state of one D2-brane and $k$ D0-branes in the resolved conifold geometry, with central charge
\begin{equation*}
Z_{k,t} = \frac{2 \pi \mathrm{i}}{R} (t+k),
\end{equation*}
where we have re-introduced the radius $R$. In the dual five-dimensional ${\rm U}(1)$ gauge theory in the ``space-time" background $\mathbb{R}^4 \times S^1_R$, geometrically engineered by the resolved conifold geometry, the DT invariants correspond to a Kaluza--Klein tower of 5d BPS particles with electro-magnetic charge $\gamma_k = \gamma_A + k \gamma_0$.

In Section~\ref{sec:Borelsums}, we also computed the Borel sum
\begin{equation*}
 \Psi_\rho(\epsilon,x,t) = \exp ( S_\rho(\epsilon,x,t) )
\end{equation*}
of the quantum vev of a brane in the resolved conifold geometry, along any ray in the $\epsilon$-plane where its Borel transform does not have any singularities. As summarized in Theorem~\ref{th:BorelS}, these singularities lie along two sets of rays
\begin{align*}
\pm \widetilde{l}_k =\pm \mathbb{R}_{< 0} \cdot 2 \pi \mathrm{i} (x+k) \qquad \mathrm{and} \qquad
\pm \widetilde{l}_{k,t} = \pm \mathbb{R}_{< 0} \cdot 2 \pi \mathrm{i} (t+x+k),
\end{align*}
for $k \in \mathbb{Z}$ and $\pm l_\infty = \pm \mathrm{i} \mathbb{R}_{<0}$, and $S_\rho$ jumps across the ray $\pm \widetilde{l}_k$ (resp.~$\pm \widetilde{l}_{k,t}$) with a Stokes factor
\begin{align*}
\Delta S_{\rho}(\epsilon,x,t) = \mathrm{Li}_1 \big( {\rm e}^{\pm 2\pi \mathrm{i}(x+k)/\check\epsilon} \big) \qquad \mathrm{resp.} \quad \Delta S_{\rho}(\epsilon,x,t) = - \mathrm{Li}_1 \big( {\rm e}^{\pm 2\pi \mathrm{i}(t+x+k)/\check \epsilon} \big).
\end{align*}
Again, the $\Psi_\rho$ can be written in terms of the non-perturbative quantum vev $\Psi_\mathrm{np}$ (first encountered in equation~\eqref{eq:psiopennp}) with shifted arguments as in Corollary~\ref{cortheorem2}.

Geometrically, to each such Stokes ray we may associate an open DT invariant in the resolved conifold geometry with central charge
\begin{equation*}
\pm \widetilde{Z}_{k} = \mp \frac{2 \pi \mathrm{i}}{R} (x+k) \qquad \mathrm{vs.} \quad \pm \widetilde{Z}_{k,t} = \mp \frac{2 \pi \mathrm{i}}{R} (t+x+k),
\end{equation*}
that counts the bound states of one D2-brane and $k$ D0-branes (or the bound states of the corresponding anti-branes) ending on a D4-brane that ``wraps'' the Lagrangian A-brane corresponding to $\Psi_\rho$. Perhaps, these invariants are best understood in terms of the M-theory setting, where they may be lifted to open M2-branes ending on an M5-brane \cite{Ooguri:1999bv}. The M5-brane wraps the Lagrangian A-brane corresponding to $\Psi_\rho$, while the open M2-branes either wrap the disc passing through the north pole or the south pole of the compact $\mathbb{P}^1$ \cite{Aganagic:2000gs}. Similar jumps in this context have been found previously in \cite{Aganagic:2009cg}.
Aside from wrapping a Lagrangian A-brane in the resolved conifold geometry, the M5-brane wraps a $\mathbb{R}^2 \times S^1_R$ in the ``space-time" directions. In the dual five-dimensional ${\rm U}(1)$ gauge theory the M5-brane thus engineers a BPS surface defect. In the three-dimensional worldvolume theory on the surface defect, compactified on $S^1_R$, the open DT invariants correspond to two Kaluza--Klein towers of 3d BPS particles \cite{Dimofte:2010tz}.

\subsection{Exponential spectral networks}\label{sec:spectralnetwork}

Similar to 2d-4d BPS states in 4d $\mathcal{N}=2$ theories of class S (see \cite{Gaiotto:2009hg} and follow-ups), the 3d-5d BPS states of their five-dimensional lifts may be visualized using the technology of spectral networks (also known as exponential BPS graphs in this context) \cite{Banerjee:2018syt, Eager:2016yxd}.

Exponential spectral networks are generally more complex than their 4d analogs because of the logarithm of $Y$ in the classical Liouville form. It is therefore convenient to choose a~trivialization of the $\log$-covering $\widetilde{\Sigma} \to \Sigma$, and label the $\vartheta$-trajectories by an additional index when crossing the logarithmic branch-cuts.

Say we are given a covering $\Sigma \to \mathbb{C}^*_X$ with classical Liouville form
\begin{equation*}
\lambda^\mathrm{cl} = \frac{1}{2 \pi \mathrm{i}} \log Y(x) \,{\rm d}x.
\end{equation*}
Choose a local trivialization for the spectral covering $\Sigma \to \mathbb{C}^*_X$ as well as for the $\log$-covering $\widetilde{\Sigma} \to \Sigma$. Fix a phase $\vartheta \in \mathbb{R}/2 \pi \mathbb{Z}$, and denote the additional $\log$-index by $N \in \mathbb{Z}$. Then the exponential spectral network $\mathcal{W}_\vartheta(\mathbf{t})$ is defined in terms of the differential equation
\begin{align}\label{eq:network_trajectory}
 \frac{1}{2 \pi \mathrm{i}} \big( \log Y_j(x) - \log Y_i(x) + 2 \pi \mathrm{i} N \big) \frac{{\rm d} x}{{\rm d} s} \in {\rm e}^{\mathrm{i} \vartheta} \mathbb{R}^\times,
\end{align}
where the subscripts $i$ and $j$ correspond to the restriction of $\lambda^\mathrm{cl}$ to the corresponding sheet in the local trivialization of $\Sigma \to \mathbb{C}^*_X$.

More precisely, the exponential spectral network $\mathcal{W}_\vartheta(\mathbf{t})$ is a collection of $\vartheta$-trajectories with the labels $(ij,N)$. Each trajectory is a path $\exp (2 \pi {\rm i} \gamma(s))$ in $\mathbb{C}^*_X$ such that
\begin{align}\label{eqn:thetatrajectoryRC}
 \big( y_j ( \exp(2 \pi \mathrm{i} \gamma(s)) ) - y_i ( \exp(2 \pi \mathrm{i} \gamma(s)) ) + N \big) \frac{ \partial \gamma(s)}{\partial s} \in {\rm e}^{\mathrm{i} \vartheta} \mathbb{R}^\times.
\end{align}
Note that if a path parametrized by $\gamma(s)$ is a solution to the above constraint, then the opposite path will be a solution as well with the opposite sign. We say that $\gamma(s)$ is positively oriented if the constraint \eqref{eqn:thetatrajectoryRC} is valued in ${\rm e}^{\mathrm{i} \vartheta} \mathbb{R}_{>0}$.

A 3d BPS particle with central charge $\widetilde{Z}$, bound to the surface defect inserted at position $X=X_* \in \mathbb{C}^*_X$, appears as a finite web of $\vartheta$-trajectories with endpoint at $X=X_*$ where
\begin{equation*}
\vartheta = \arg\big(\widetilde{Z}\big).
\end{equation*}
If this 3d BPS particle arises as a three-dimensional field configuration that interpolates between the vacua labeled by $(i,N)$ and $(j,M)$ of the worldvolume theory on this surface defect, the $\vartheta$-trajectory ending on the position $X$ carries the label $(ij,n)$ with $n = M-N$.
A 5d BPS particle with central charge~$Z$, on the other hand, shows up as a topology change of the exponential network $\mathcal{W}_{\vartheta}(u)$ at phase
\begin{equation*}
 \vartheta= \arg(Z).
\end{equation*}
Such a 5d BPS particle generally appears as a finite web of $\vartheta$-trajectories with endpoints on the branch-points of the covering $\Sigma \to \mathbb{C}^*_X$.

Exponential spectral networks for the resolved conifold geometry and the corresponding spectrum of 5d BPS particles have been extensively studied in \cite{BLR}, albeit with the mirror curve $\Sigma$ in a different framing than the one we are considering in equation~\eqref{eqn:mirrorcurve1p}. The resulting 5d BPS spectrum should however be independent of this choice. A brief discussion of the closed spectrum, in our choice of framing, as well as an comprehensive analysis of the open spectrum of 3d BPS states has furthermore appeared recently in \cite{GHN}. We will therefore keep our discussion concise and merely emphasize key aspects, as well as introduce two distinguished networks that play an important role later.

\begin{figure}
\begin{center}
\includegraphics[width=8cm]{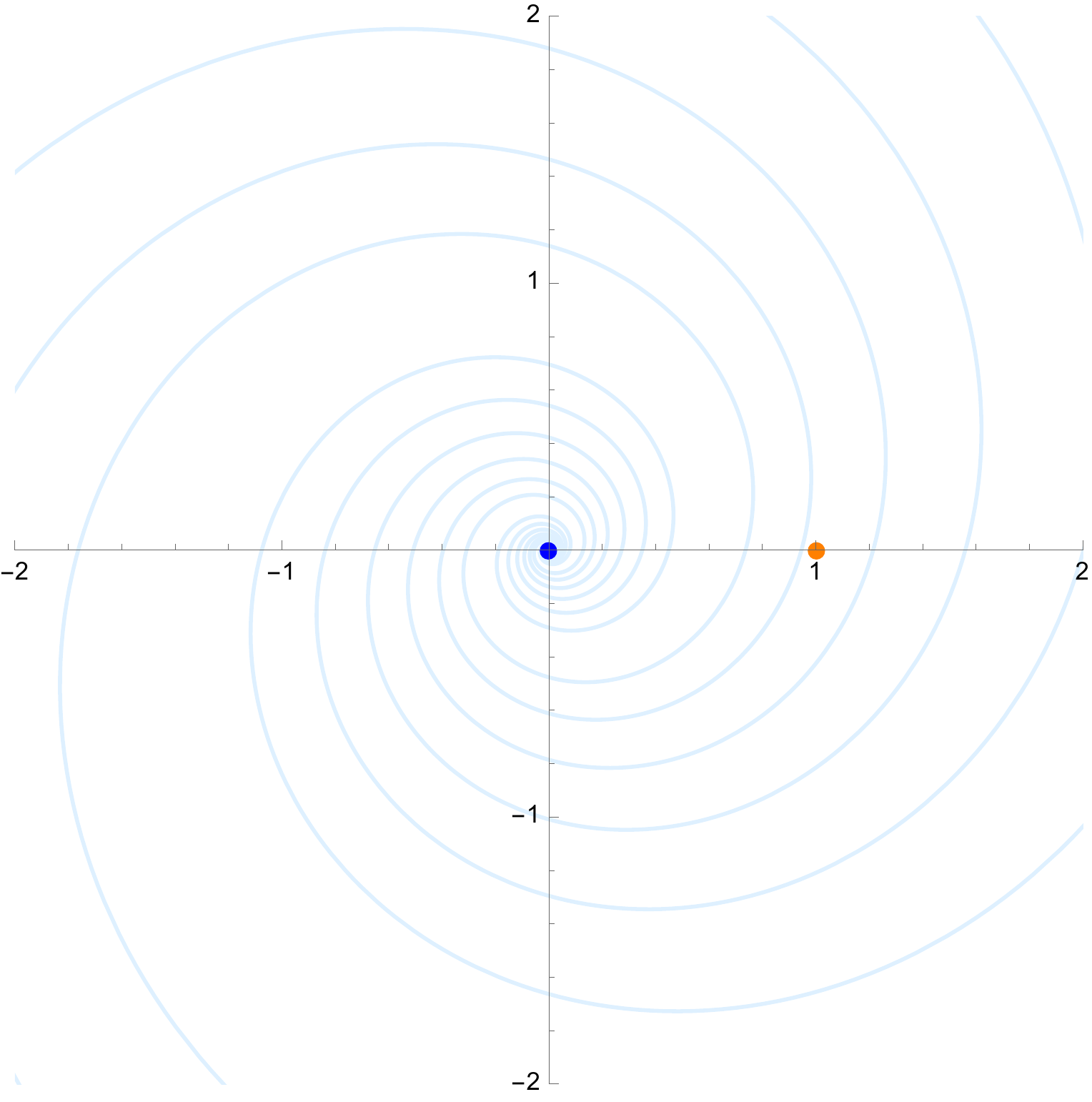}
\end{center}
\caption{A generic family of $\vartheta$-trajectories in the neighborhood of $X=0$ (the blue dot), drawn for $\vartheta = 1.2$. }
\label{fig:traj_close0}
\end{figure}

For our choice of covering $\Sigma \to \mathbb{C}^*_X$, the classical part of the Liouville form is given by
\begin{equation*}
 \lambda^\mathrm{cl} = \frac{1}{2\pi \mathrm{i}} \log \left( \frac{1-Q {\rm e}^{2 \pi \mathrm{i} x}} {1-{\rm e}^{2 \pi \mathrm{i} x}}\right) {\rm d}x.
\end{equation*}
There are no branch-points of $\Sigma \to \mathbb{C}^*_X$, because the covering only has a single sheet. Trajectories may therefore either be compact, or have their end-points on the punctures of $\Sigma$. We choose a trivialization of the $\log$-covering $\widetilde{\Sigma} \to \Sigma$ by fixing a logarithmic branch-cut between the punctures at $X=1$ and $X=Q^{-1}$, and label the trajectories $\gamma(s)$ accordingly with an extra index $N \in \mathbb{Z}$.

Since our covering $\Sigma \to \mathbb{C}^*$ has a single sheet, the $\vartheta$-trajectories simply carry the label $(11,n) = n$ and are parametrized by paths $\exp (2 \pi \mathrm{i} \gamma(s))$ in the $X$-plane such that
\begin{equation*}
 n \gamma'(s) \in \mathrm{i} {\rm e}^{\mathrm{i} \vartheta} \mathbb{R}^\times.
\end{equation*}
This implies that $\gamma(s) = i s {\rm e}^{i \vartheta} + C$ for some constant $C$ and it follows that the $\vartheta$-trajectories are parametrized by the paths
\begin{equation}\label{eqn:3d-theta-traj}
 X(s) = {\rm e}^{2 \pi {\rm i} \gamma(s)} = {\rm e}^{-2 \pi s {\rm e}^{\mathrm{i} \vartheta} } {\rm e}^{2 \pi {\rm i} C}
\end{equation}
for some constant $C$.

Figure~\ref{fig:traj_close0} shows a generic family of $\vartheta$-trajectories for $\vartheta = 1.2$. Indeed, if $-\pi/2 < \vartheta < \pi/2$, the $\vartheta$-trajectories spiral away from the puncture at $X=\infty$ and into the puncture at $X=0$ as~$s$ increases, whereas if $\pi/2 < \vartheta < 3\pi/2$ they have the opposite orientation. Yet, when $\vartheta=0$ (or $\vartheta=\pi$) the $\vartheta$-trajectories run along rays $X(r) = r {\rm e}^{2 \pi {\rm i} C}$ with $r \ge 0$, and if $\vartheta = \pm \pi/2$ the $\vartheta$-trajectories form a family of concentric circles interpolating between $X=0$ and $X=\infty$.

We only need to analyze a small selection of the $\vartheta$-trajectories~\eqref{eqn:3d-theta-traj} to locate the 3d and 5d BPS states. This is similar to the four-dimensional story, where we would draw the spectral network as the collection of $\vartheta$-trajectories that have at least one end-point on a branch-point of the covering $\Sigma \to C$. In this example, the (finite-mass) 3d BPS states bound to a surface defect at position $X_*$ must be encoded in the $\vartheta$-trajectories that have one end-point at $X_*$ and one end-point at a log-puncture (either at $X=1$ or $X=Q^{-1}$). Hence we may draw the exponential network $\mathcal{W}_\vartheta(Q)$ as merely the collection of $\vartheta$-trajectories that pass through either the puncture at $X=1$ or $X=Q^{-1}$. Generically, there are two of these, see Figure~\ref{fig:expnetwork} for an example.

\begin{figure}
\begin{center}
\includegraphics[width=6cm]{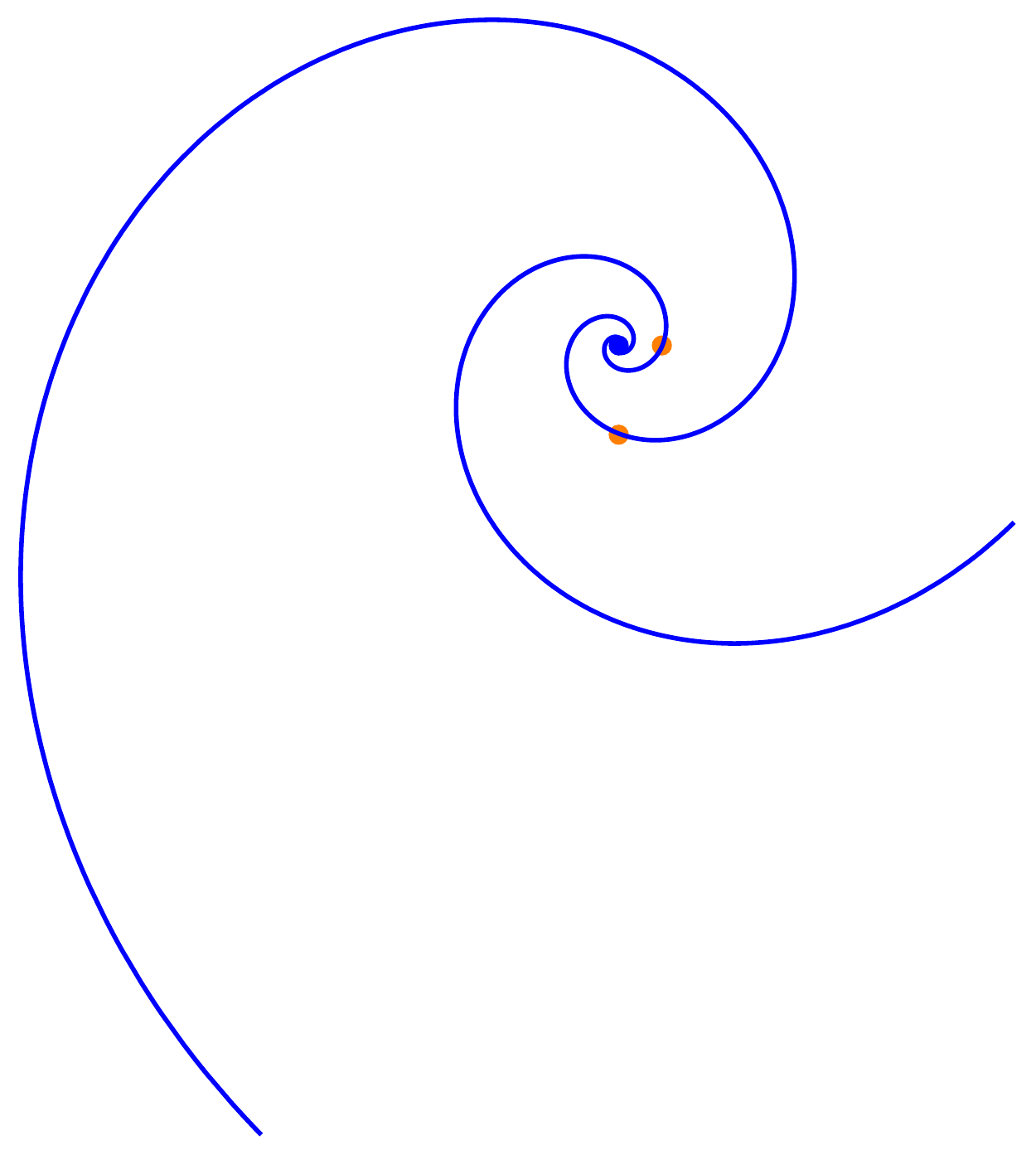}
\end{center}
\caption{Example of an exponential spectral network $\mathcal{W}_\vartheta(Q)$ for $Q=\frac{\mathrm{i}}{2}$ and $\vartheta=1.2$. The log-punctures at $X=1$ and $X=Q^{-1}$ are colored orange, while the two $\vartheta$-trajectories and the puncture at $X=0$ are colored blue.}
\label{fig:expnetwork}
\end{figure}

\begin{figure}
\begin{center}
\includegraphics[width=5.4cm]{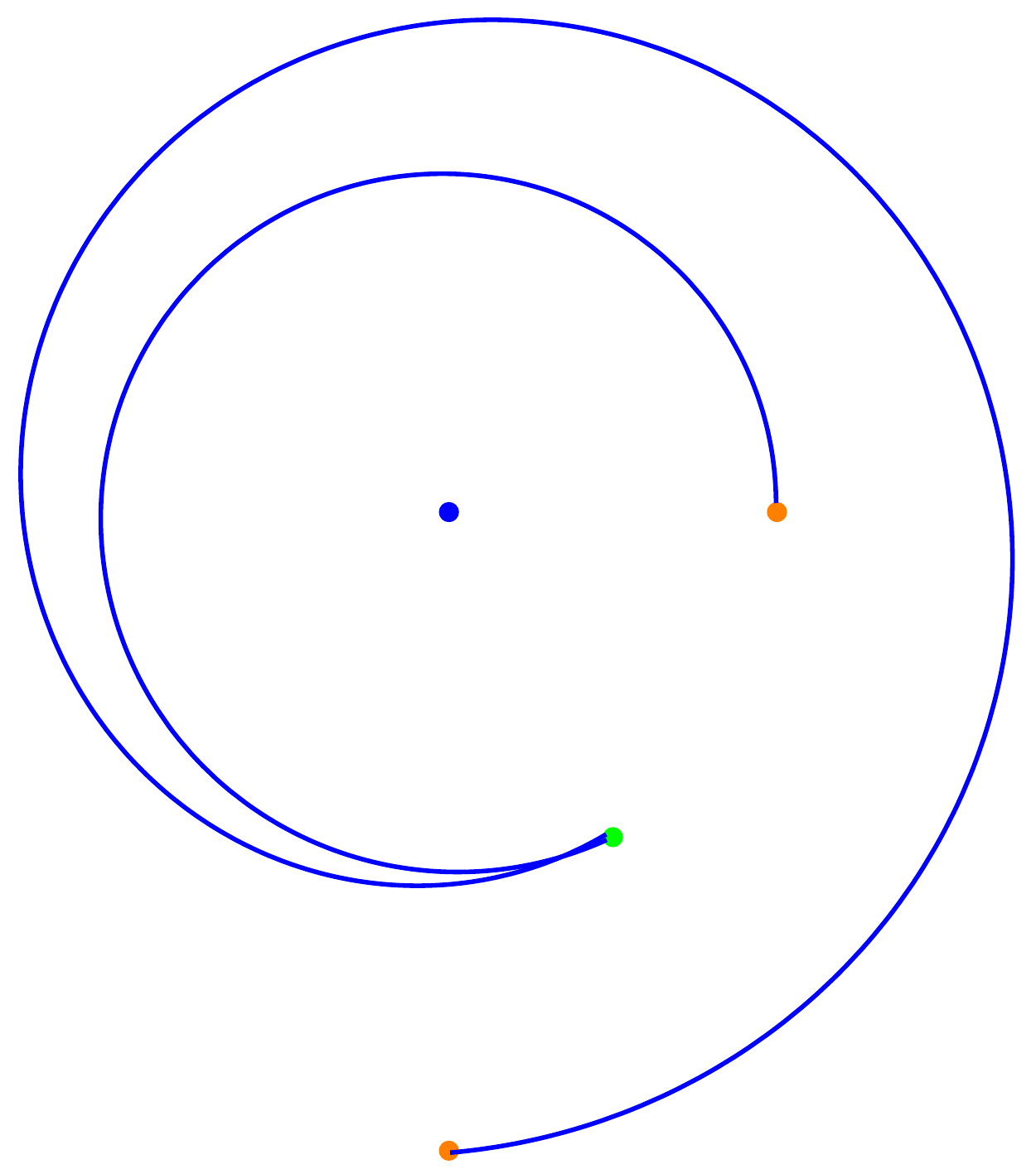}
\end{center}
\caption{Two $\vartheta$-trajectories (in blue) encoding the 3d BPS states with central charges $\widetilde{Z}_1$ and $\widetilde{Z}_{1,t}$ at $Q=\frac{\mathrm{i}}{2}$ and $X_*=\frac{1}{2}-\mathrm{i}$. The surface defect is inserted at $X=X_*$ (the green dot), the log-punctures are at $X=1$ and $X=Q^{-1}$ (the orange dots). }
\label{fig:Expnetwork3dBPS}
\end{figure}

Suppose we insert a surface defect at the position $X= X_*=\exp (2 \pi {\rm i} x_*)$. Then there are two kinds of $\vartheta$-trajectories that are relevant to describe the corresponding 3d BPS states. The first kind has an end-point at $X=1$ and the second kind has an end-point at $X=Q^{-1}$.

The first kind of trajectory may be parameterized by the path $\gamma(s)$ with
\begin{equation*}
 2 \pi {\rm i} \gamma(s) = - 2 \pi s {\rm e}^{{\rm i} \vartheta},
\end{equation*}
where $\vartheta$ is chosen such that the trajectory $X(s) = \exp (2 \pi {\rm i} \gamma(s))$ passes through the point ${X=X_*}$. Using that $\exp (2 \pi {\rm i} x_*) = \exp (2 \pi {\rm i} (x_*+k))$ for any $k\in \mathbb{Z}$, we find that these trajectories appear at the phases
\begin{equation*}
 \vartheta = \arg(- 2 \pi {\rm i} (x_*+k)) = \arg \widetilde{Z}_k.
\end{equation*}

The second kind of trajectory may similarly be parameterized by the path with
\begin{equation*}
 2 \pi {\rm i} (\gamma(s)+t) = - 2 \pi s {\rm e}^{{\rm i} \vartheta},
\end{equation*}
where $\vartheta$ is chosen such that the trajectory $X(s) = \exp (2 \pi {\rm i} \gamma(s))$ passes through the point ${X=X_*}$. Using that $\exp (2 \pi {\rm i} (x_*+t)) = \exp (2 \pi {\rm i} (x_*+t+k))$ for any $k\in \mathbb{Z}$, we find that these trajectories appear at the phases
\begin{equation*}
 \vartheta = \arg(- 2 \pi {\rm i} (x_*+t+k)) = \arg \widetilde{Z}_{k,t}.
\end{equation*}
We conclude that the 3d BPS particles with central charges $\widetilde{Z}_k$ and $\widetilde{Z}_{k,t}$ are indeed realized as the anticipated $\vartheta$-trajectories in the exponential network. An example is shown in Figure~\ref{fig:Expnetwork3dBPS}.

\begin{figure}[t]\centering
\includegraphics[width=4.55cm]{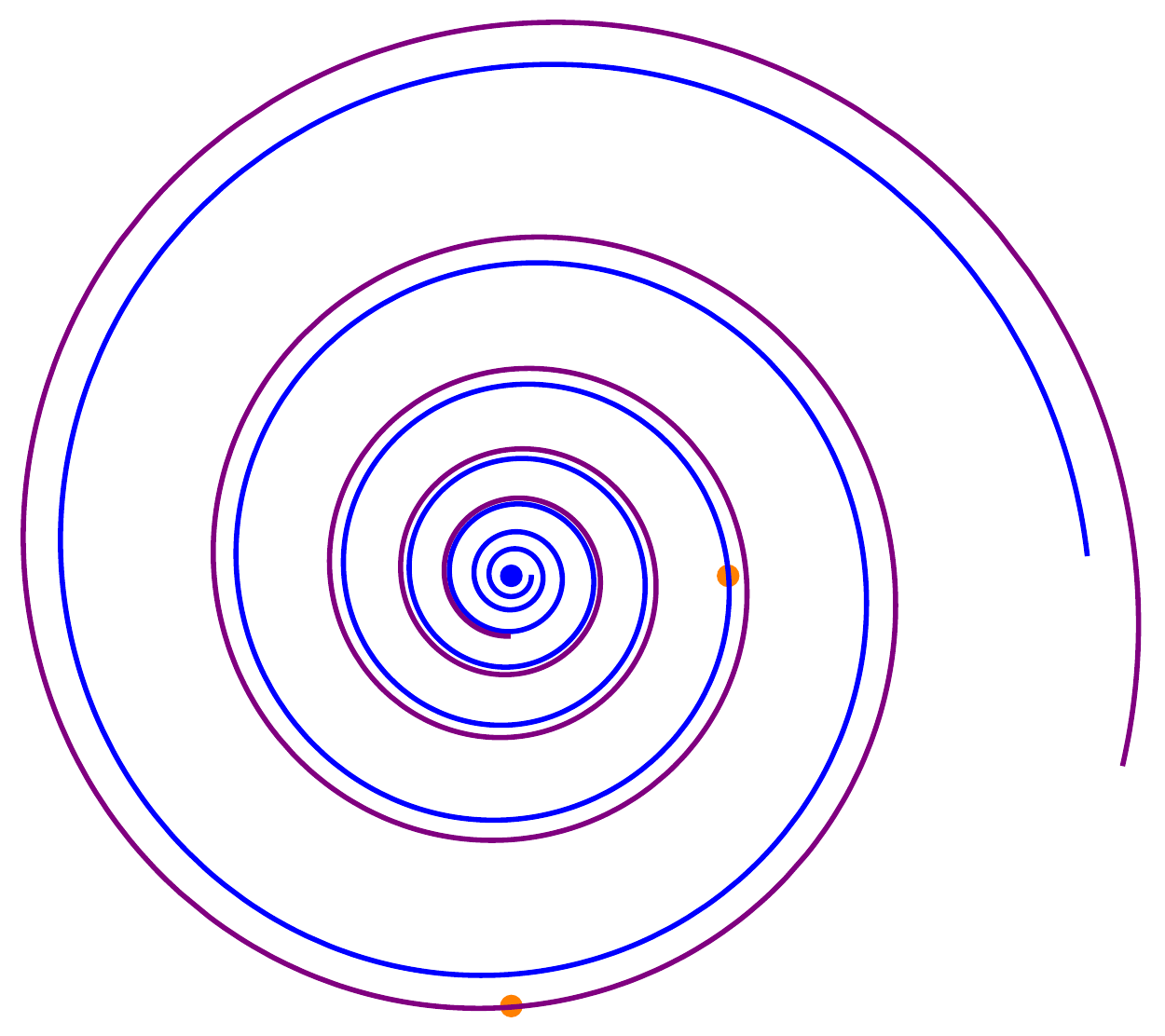}\quad \includegraphics[width=4.55cm]{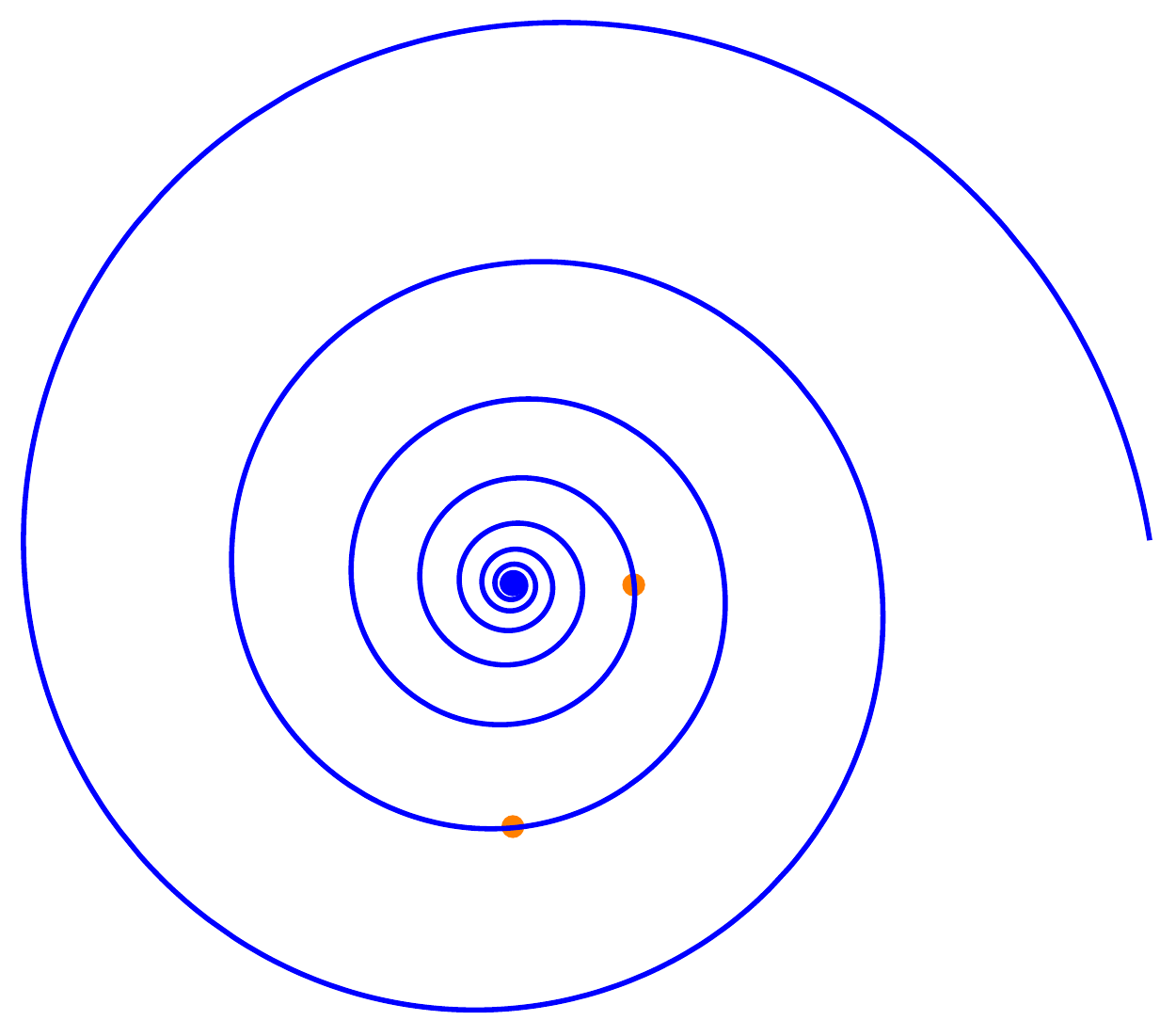}\quad \includegraphics[width=4.55cm]{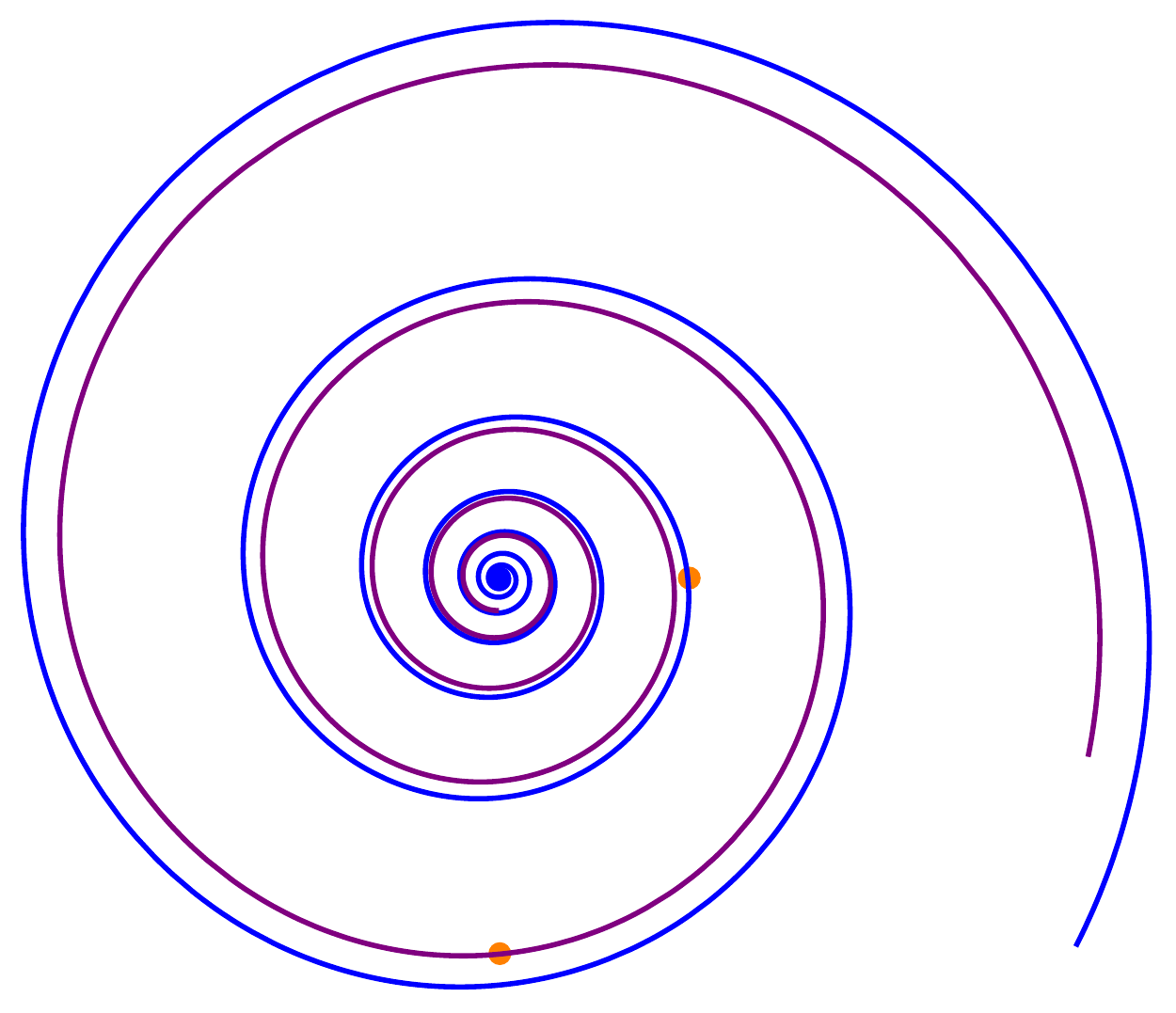}
\caption{Topology change in the network $\mathcal{W}_\vartheta(Q)$ for $Q=\frac{\mathrm{i}}{2}$ and with $\vartheta$ varying from $\vartheta_- = \arg(-2 \pi \mathrm{i} (t + 1))-0.01$ (on the left) to $\vartheta_+ = \arg(-2 \pi \mathrm{i} (t + 1))+0.01$ (on the right).}\label{fig:expnetwork5dBPStopologychange}
\end{figure}

Generically a $\vartheta$-trajectory in the exponential network $\mathcal{W}_\vartheta$ has one end-point on the regular puncture at $X=0$ or $X=\infty$ and one end-point on a log-puncture at $X=1$ or $X=\infty$. There are however special phases for which $\vartheta$-trajectories appear with both end-points on a~log-puncture. As is illustrated in Figure~\ref{fig:expnetwork5dBPStopologychange}, such trajectories correspond to topology changes in the exponential network when the phase $\vartheta$ is varied. They thus ought to realize 5d BPS states. Indeed, with the same argument as before, these trajectories appear at the phases
\begin{equation}\label{eqn:phase5dBPS}
 \vartheta = \arg(2 \pi {\rm i} (t+k)) = \arg Z_{k,t}.
\end{equation}
All these 5d BPS trajectories are illustrated in the example with $Q={\rm i}/2$ in Figure~\ref{fig:expnetwork5dBPS}.

Note that equation~\eqref{eqn:phase5dBPS} may be written as
\begin{equation*}
\tan(\vartheta) = - \frac{\operatorname{Re}(t)+k}{\operatorname{Im}(t)}.
\end{equation*}
For instance, if $\operatorname{Re}(t)=0$ the 5d BPS state with central charge $Z_{0,t}$ is encoded at $\vartheta=0$ as a~straight line segment between the log-punctures at $X=1$ and $X=Q^{-1}$. Indeed, an auxiliary path around this trajectory lifts to the 1-cycle~$\gamma_A$ on~$\Sigma$. When $\operatorname{Re}(t)$ is varied, the trajectories of the network change accordingly. In particular, a $\vartheta$-trajectory corresponding to the same 1-cycle~$\gamma_A$ now appears at the phase with ${\rm e}^{{\rm i} \vartheta} = 2 \pi {\rm i} t$.
More generally, the 5d BPS states with central charge $Z_{k,t}$ correspond to $\vartheta$-trajectories between the log-punctures at $X=1$ and $X=Q^{-1}$ that cross the logarithmic branch-cut $k$ times. Indeed, an auxiliary path around such a trajectory lifts to the 1-cycle $\gamma_A + k \gamma_0$ on $\Sigma$.

Suppose that $\operatorname{Re}(t)=0$. Note that there are two special limits. In the limit $\operatorname{Im}(t) \to \infty$ (i.e.,~$Q \to 0)$ all 5d BPS states of central charge $Z_{k,t}$ appear as a single trajectory along the positive $X$-axis at the phase $\vartheta=0$. In the opposite limit $\operatorname{Im}(t) \to 0$ (i.e.,~$Q \to 1)$ all BPS states of charge $Z_{k,t}$ appear as a single concentric $\vartheta$-trajectory at the phase $\vartheta = \pi/2$.

\begin{figure}[t]\centering
\includegraphics[width=4.2cm]{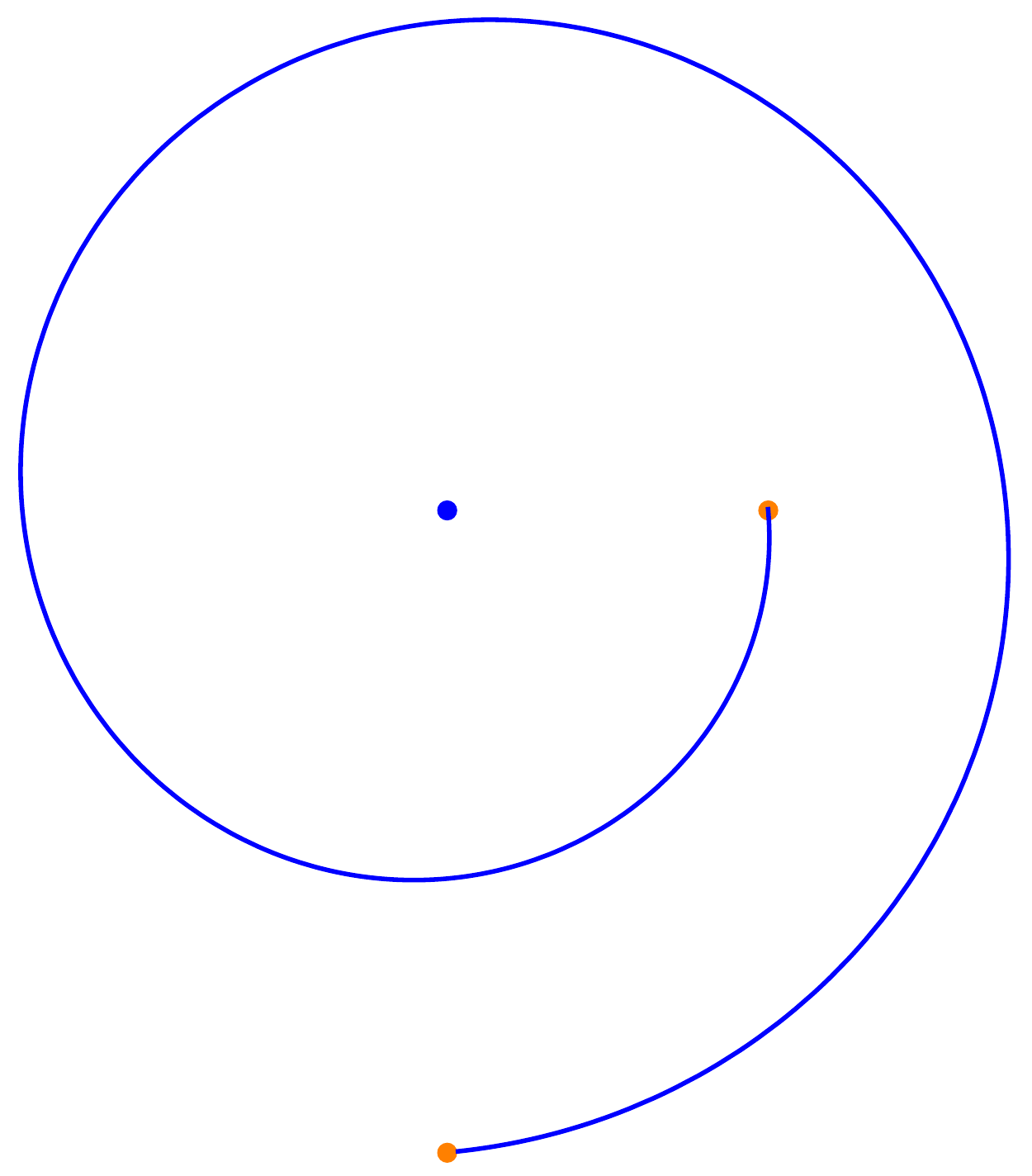} \quad \includegraphics[width=4.2cm]{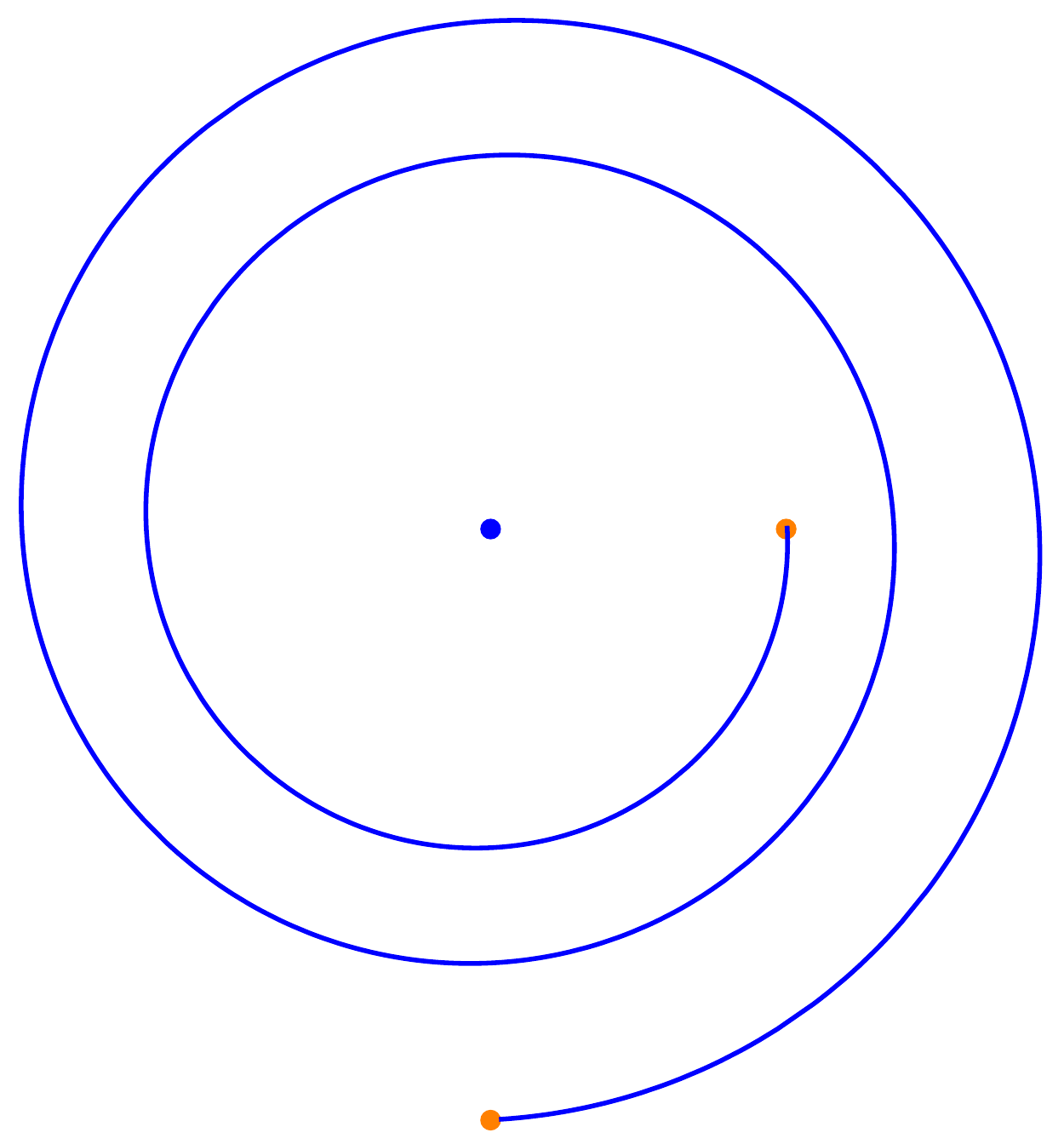} \quad
\includegraphics[width=4.2cm]{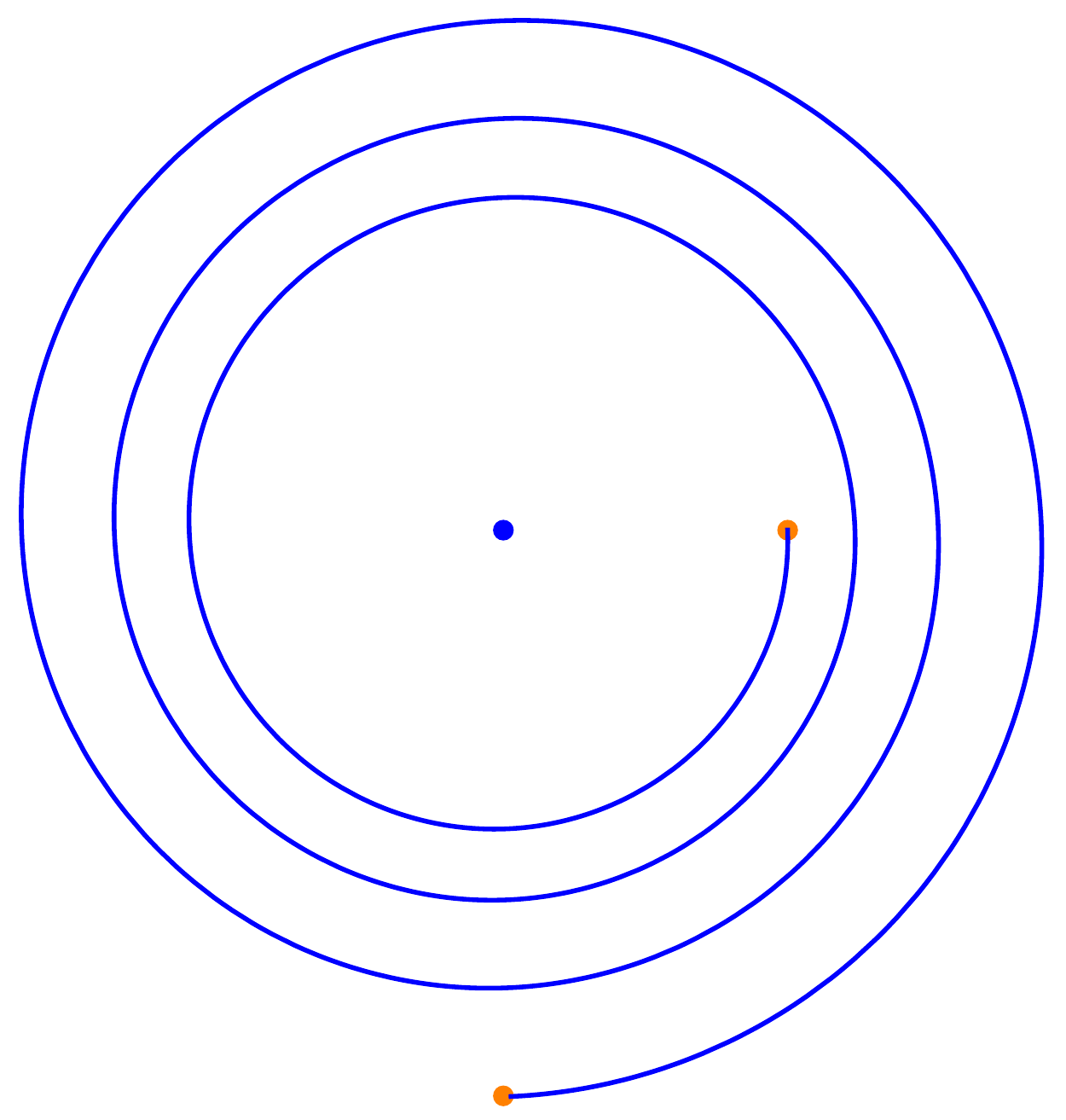}
\caption{A few $\vartheta$-trajectories encoding 5d BPS states for $Q=\frac{\mathrm{i}}{2}$. 5d BPS trajectories appear at phases $\vartheta = \arg(- 2 \pi {\rm i} (t+k))$ for $k \in \mathbb{Z}$, and are illustrated here for $k=1,2,3$, from left to right.}\label{fig:expnetwork5dBPS}
\end{figure}

\begin{figure}[t]\centering
\includegraphics[width=6.5cm]{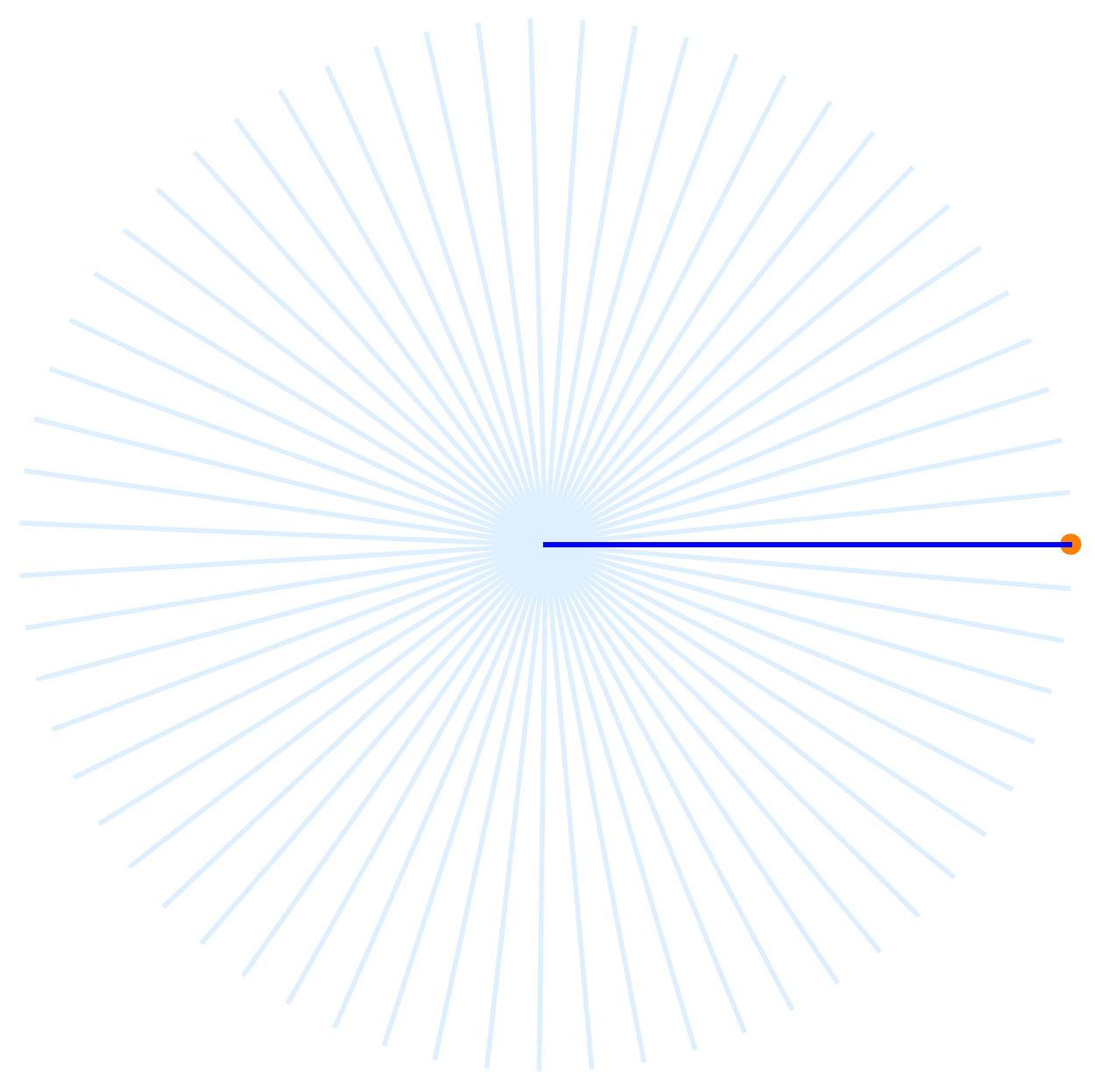}\quad
\includegraphics[width=6.5cm]{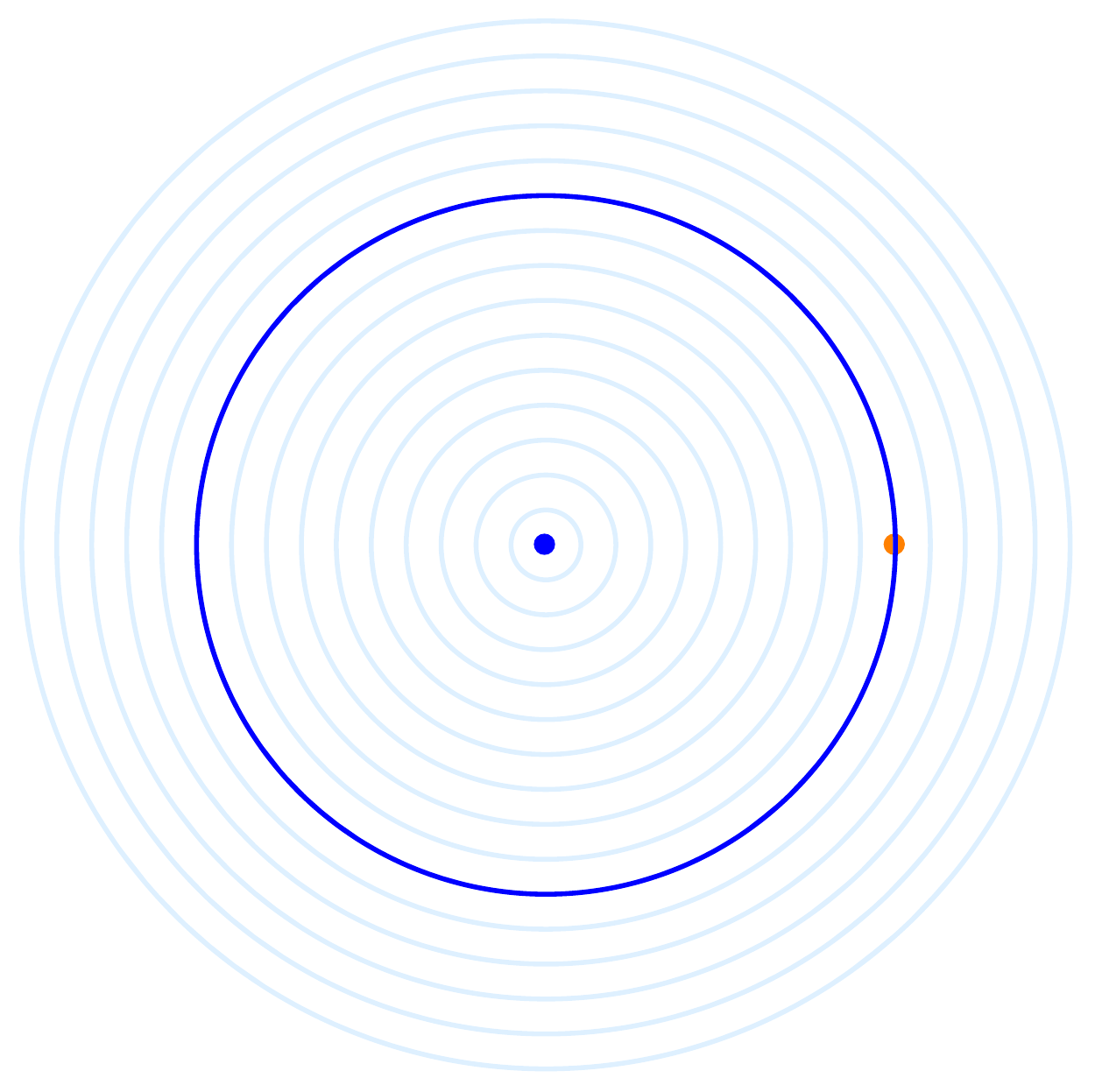}
\caption{Two distinguished networks: $\mathcal{W}_\mathrm{np}$ on the left and $\mathcal{W}_\mathrm{GV}$ on the right. The $\vartheta$-trajectories through the log-punctures are colored blue, whereas all other $\vartheta$-trajectories are included as well in light-blue.}\label{fig:WGVandWnp}
\end{figure}

In the following we refer to the exponential network $\mathcal{W}_{\vartheta =0}$ as $\mathcal{W}_\mathrm{np}$, and to the network $\mathcal{W}_{\vartheta=\pi/2}$ as $\mathcal{W}_\mathrm{GV}$. These two exponential networks (along with all their $\vartheta$-trajectories) are illustrated in Figure~\ref{fig:WGVandWnp}. Note that the network $\mathcal{W}_\mathrm{GV}$ is resemblant to a Fenchel--Nielsen network in the terminology of \cite{Hollands:2013qza}.

\subsection{Exact WKB for difference operators}\label{sec:exactWKB}

Spectral networks are also known as Stokes graphs in the exact WKB analysis, where they encode the Stokes phenomena of linear differential operators $\mathbf{d}(\epsilon)$ (also known as Schr\"odinger operators or more generally $\mathfrak{g}$-opers) on (punctured) Riemann surfaces. See for instance \cite[Section 4.4]{Hollands:recipe} for a brief introduction and references.

Local solutions $\psi$ of the Schr\"odinger equation are written as the exponential of a formal series expansion in the Planck constant $\epsilon$, starting with
\begin{equation}\label{eq:formalepspsi}
 \psi^\mathrm{for} (\epsilon, z)= \exp \left( \frac{4 \pi^2 \mathrm{i}}{\epsilon} \int^z \lambda^\mathrm{cl} + \mathcal{O}\big(\epsilon^0\big) \right),
\end{equation}
where $z$ is a local coordinate on the Riemann surface $C$. Similar to~\eqref{eq:network_trajectory}, the meromorphic differential $\lambda^\mathrm{cl}$ defines the $\vartheta$-trajectories of the Stokes graph $\mathcal{W}_\vartheta$ through the constraint
\begin{equation*}
 \big(\lambda^\mathrm{cl}_j(z) - \lambda^\mathrm{cl}_i(z)\big) (v) \in {\rm e}^{\mathrm{i} \vartheta} \mathbb{R}^\times,
\end{equation*}
where $v$ is a tangent vector to the trajectory.

The Borel sum $\psi_\rho(\epsilon,z)$ of the local solution $\psi(\epsilon,z)$ along the ray $\rho = {\rm e}^{\mathrm{i} \vartheta} \mathbb{R}_{< 0}$ plays an important role in computing the monodromies, bound states and Stokes phenomena of the differential operator $\mathbf{d}(\epsilon)$. In particular, note that the Borel sum $\psi_\rho(\epsilon,z)$ has good WKB asymptotics: the function $\psi_\rho(\epsilon,z)$ has the expansion $\psi^\mathrm{for} (\epsilon, z)$ in~\eqref{eq:formalepspsi} when $\epsilon \to 0$ while remaining in the closed half-plane with $\operatorname{Re}\big({\rm e}^{- \mathrm{i} \vartheta} \epsilon\big) \ge 0$, while $\psi_\rho(\epsilon,z)$ decreases fastest along any $\vartheta$-trajectory with $\arg(\epsilon) = \vartheta$. In~\cite{Hollands:t3abel}, the exact WKB analysis is reformulated in terms of the $\mathcal{W}$-abelianization of flat $\mathfrak{g}$-connections on $C$ for $\mathfrak{g}= \mathfrak{su}(K)$.

Our aim in this section is to interpret the results in this paper as a generalization of $\mathcal{W}$-abelianization (as well as the exact WKB method) to linear $q$-difference operators $\mathbf{D}(\epsilon)$, such as the Schr\"odinger operator~\eqref{eqn:diffDX1} for the resolved conifold geometry.

The very simplest examples occur when the difference operator $\mathbf{D}(\epsilon)$ is defined with respect to a 1-fold covering $\Sigma \to \mathbb{C}^*$. For instance, consider the Schr\"odinger operator~\eqref{eqn:diffDX1}
\begin{equation*}
 \mathbf{D} (\epsilon) = (1-X) {\rm e}^{\check{\epsilon} \partial_x} - (1 - Q X)
\end{equation*}
 for the resolved conifold geometry, with formal solution
$\Psi(\epsilon,x,t)$ as in equation~\eqref{eq:Sasymp}. Its Borel sum
\begin{equation*}
\Psi_\rho(\epsilon,x,t) = \exp S_\rho(\epsilon,x,t)
\end{equation*}
was computed in Theorem~\ref{th:BorelS}. The exact quantum periods (or Voros periods) are defined by
 \begin{equation*}
 \Pi_{\gamma,\rho} = \mathcal{B}_\rho \oint_\gamma \lambda^{\mathrm{qu}}(\epsilon,x,t) = \frac{\epsilon}{4 \pi^2 \mathrm{i}} \oint_\gamma \frac{\partial}{\partial z} S_\rho(\epsilon,x,t),
 \end{equation*}
where $\gamma$ is an open or closed 1-cycle on $\Sigma$ and $\mathcal{B}_\rho$ stands for Borel sum along the ray $\rho$. They characterize the monodromies of $\mathbf{D} (\epsilon)$ and can be expressed in terms of the Borel sums $\Psi_\rho$.

For example, using \eqref{eqn:Sformal} we may compute the quantum A-period to all orders in $\epsilon$. We find:

\begin{Proposition}\label{prop:A-period}
The quantum A-period has the $\epsilon$-expansion
\begin{equation*}
 \Pi^\rho_A = \oint_{\gamma_A} \lambda^{\mathrm{qu}}(\epsilon,x,t) = t.
\end{equation*}
\end{Proposition}

\begin{proof}
As illustrated in Figure~\ref{fig:cyclesSigma}, we have chosen the A-cycle $\gamma_A$ such that it loops around the two punctures at $X=0$ and $X=\infty$. The classical Liouville form
\begin{equation*}\lambda^\mathrm{cl} = \frac{1}{2 \pi\mathrm{i}} \log \left( \frac{1-Q{\rm e}^{2 \pi \mathrm{i} x}}{1-{\rm e}^{2 \pi \mathrm{i} x}}\right) {\rm d} x\end{equation*}
can be expanded close to the punctures at $X=0$ and $X=\infty$ as
 \begin{alignat*}{3}
 & X=0\colon \quad && \lambda^\mathrm{cl} = -\frac{1}{4 \pi^2} (1-Q) \,{\rm d} X + \cdots, &\\
 &X=\infty\colon \quad && \lambda^\mathrm{cl} = \frac{1}{4 \pi^2} \log Q \,{\rm d} \log \widetilde{X}_\infty + \cdots &
\end{alignat*}
with $\widetilde{X}_\infty = 1/X$ local coordinates near $X=\infty$. We thus find that
\begin{align*}
 \oint_{\gamma_A} \lambda^\mathrm{cl} = \frac{1}{2 \pi {\rm i}} \log Q = t.
\end{align*}

To show that this classical result does not obtain any quantum corrections, define
\begin{align*}
\lambda^s := \frac{1}{2\pi \mathrm{i}} \log \left( \frac{1-Q {\rm e}^{2 \pi \mathrm{i} (x+s)}}{1-{\rm e}^{2 \pi \mathrm{i} (x+s)}}\right) {\rm d} x.
\end{align*}
This may be expanded similarly as
\begin{alignat*}{3}
 & X=0\colon \quad && \lambda^s = -\frac{1}{4 \pi^2} {\rm e}^{2 \pi \mathrm{i} s} (1-Q) \,{\rm d} X + \cdots, & \\
& X=\infty\colon \quad && \lambda^s = \frac{1}{4\pi^2} \log Q \,{\rm d} \log \widetilde{X}_\infty + \cdots.&
\end{alignat*}
close to the punctures at $X=0$ and $X=\infty$, respectively.
It follows that
\begin{gather*}
 \oint_{\gamma_A} \lambda^s = \frac{1}{2 \pi \mathrm{i}} \log Q = t
\end{gather*}
is not dependent on $s$. Using Proposition~\eqref{prop:Sasymp}, we have
\begin{gather*}
 \lambda^{\mathrm{qu}}_n = \frac{B_n}{n!} \check{\epsilon}^n \partial_x^n \lambda
 = \frac{B_n}{n!} \check{\epsilon}^n \partial_s^n \lambda^s \big|_{s=0},
\end{gather*}
so that we may conclude that
\begin{gather*}
 \oint_{\gamma_A} \lambda^{\mathrm{qu}}_n = \frac{B_n}{n!} \check{\epsilon}^n \left( \partial_s^n \oint_{\gamma_A} \lambda^s \right)\Big|_{s=0} = 0
\end{gather*}
for $n>0$.
 \end{proof}

Note that this implies that the exact quantum A-period does not receive any non-perturbative corrections and hence does not have any dependence on the ray $\rho$. This agrees with the fact that the underlying BPS problem is uncoupled.

On the other hand, since the non-compact B-cycle runs from the logarithmic puncture at $X=1$ to the puncture at $X=0$ as in Figure~\ref{fig:cyclesSigma}, the exact quantum B-period along the ray $\rho$ is naively given by
\begin{equation*}
\Pi_{B,\rho} = \mathcal{B}_\rho \int_{\gamma_B} \lambda^{\mathrm{qu}}(\epsilon,z,t) = \frac{ \epsilon} {4 \pi^2 \mathrm{i}} \oint_{\gamma_B} \frac{\partial}{\partial z} S_\rho(\epsilon,x,t)
= \frac{ \epsilon} {4 \pi^2 \mathrm{i}} \log \frac{\Psi_\rho(\epsilon,x,t)|_{x=0}}{\Psi_\rho(\epsilon,x,t)|_{x=\mathrm{i}\infty}},
\end{equation*}
The numerator of this ratio is ill-defined, however, which prompted us to define a regularized version of the exact quantum B-periods in Definition~\ref{def:regBperiod}.

That is, with the assumptions of Theorem~\ref{theorem3}, we defined
\begin{equation}\label{prop:B-period}
\Pi_{B,\rho_k}^{\rm reg}(\epsilon,t) = \frac{ \epsilon} {4 \pi^2 \mathrm{i}} \log \frac{\mathcal{S}_2(x\mid \check\epsilon,1)^{-1} \Psi_{\rho_k}(\epsilon,x,t)|_{x=0}}{\Psi_{\rho_k}(\epsilon,x,t)|_{x=\mathrm{i}\infty}},
\end{equation}
where $\mathcal{S}_2$ is the Fadeev quantum dilogarithm, defined in Appendix \ref{appendix:qdilog}, and
\begin{equation}\label{prop:B-periodGV}
\Pi_{B,\mathrm{GV}}^{\rm reg} = \lim_{k \to \infty} \Pi_{B,\rho_k}^{\rm reg}(\epsilon,t) = \frac{ \epsilon} {4 \pi^2 \mathrm{i}} \log \frac{L(x\mid \check\epsilon,1)^{-1} \Psi_{\mathrm{GV}}(\epsilon,x,t)|_{x=0}}{\Psi_{\mathrm{GV}}(\epsilon,x,t)|_{x=\mathrm{i}\infty}},
\end{equation}
where $L$ is the quantum dilogarithm, defined in equation~\eqref{eqn:quandilog}.

Note that these regularized periods are defined with respect to the path $X(s)= \exp \big(\mathrm{i} s {\rm e}^{\mathrm{i} \vartheta_*} \big)$ in Theorem~\ref{theorem3}. This implies that we want to deform the B-cycle into the path $X(s)$. The assumptions of Theorem~\ref{theorem3} enforce some restrictions on $\vartheta_*$. However, as noted at the end of Remark~\ref{remark:thetastar}, these may be lifted in return for some small modifications to the definitions~\eqref{prop:B-period} and~\eqref{prop:B-periodGV}.

\subsection{Spectral coordinates}\label{sec:abelianization}

The (regularized) exact quantum periods above may be interpreted as spectral coordinates in the context of $\mathcal{W}$-abelianization. Analogous to the 4d set-up, we thus want to consider the $q$-difference operators $\mathbf{D}_\epsilon(t)$ as $q$-opers $\nabla_\epsilon(t)$ on $\mathbb{C}^*_X$ and study the complex 1-dimensional family~$\nabla_\epsilon(t)$, with complex parameter $t$, as a half-dimensional subspace of some other hyperk\"ahler moduli space. This other moduli space turns out to be the moduli space $\mathcal{M}_\mathrm{mon}$ of solutions of the Bogomolny equations on $\mathbb{C}^* \times S^1$ (with corresponding singularities), which is also known as the moduli space of periodic monopoles (see \cite{Elliott:2018yqm} for more details and references).

Indeed, it is known $\mathcal{M}_\mathrm{mon}$ has a hyperk\"ahler structure.\footnote{Although similar to the moduli space of solutions to the Hitchin equations, one important difference is that the moduli space of solutions of the Bogomolny equations on $\mathbb{C}^* \times S^1$ does not have an ${\rm U}(1)$ isometry rotating the complex structures.}
In complex structure $I$ it is equivalent to the moduli space of multiplicative Higgs bundles, while for a generic complex structure it is equivalent to the moduli space $\mathcal{M}_\mathrm{diff}$ of difference connections. We consider $\mathcal{M}_\mathrm{mon}$ in complex structure $J$, where the family of $q$-opers $\nabla_\epsilon(t)$ forms a half-dimensional complex Lagrangian subspace of $\mathcal{M}_\mathrm{diff}$.

After fixing a choice of spectral network $\mathcal{W}$, we may construct spectral coordinates $\mathcal{X}^\mathcal{W}_\gamma$ on~$\mathcal{M}_\mathrm{diff}$ using the technology of \cite{Hollands:2013qza}. These spectral coordinates have the special property that, if the spectral network $\mathcal{W}$ is in the same isotopy class as the WKB spectral network~$\mathcal{W}_\vartheta(\mathbf{\mathbf{t}})$,
\begin{equation*}
 \log \mathcal{X}^\mathcal{W}_\gamma (\nabla_\epsilon(\mathbf{t})) = \oint_\gamma \frac{\partial}{\partial x} S_\rho(\epsilon,x,\mathbf{t}) \,{\rm d}x = \frac{4 \pi^2 \mathrm{i}}{\epsilon} \Pi_{\gamma,\rho},
\end{equation*}
where $\mathrm{arg}(\epsilon) = \vartheta$ and $\rho = \exp {\rm i} \vartheta$.
That is, the log of the spectral coordinate $\mathcal{X}^\mathcal{W}_\gamma$, evaluated on the $q$-difference opers $\nabla_\epsilon(\mathbf{t})$, is proportional to the exact quantum period $\Pi_{\gamma,\rho}$.

To be precise, the spectral coordinates $\mathcal{X}^\mathcal{W}_\gamma$ are defined on the slightly different moduli space of $\mathcal{W}$-framed difference connections. The framing on the difference connections is generically given by a choice of local sections at the punctures. The spectral coordinates may then be computed in terms of this framing data. Here, we are mostly interested in determining the spectral coordinates $\mathcal{X}^\mathrm{np}_\gamma$ and $\mathcal{X}^\mathrm{GV}_\gamma$ corresponding to the exponential networks $\mathcal{W}_\mathrm{np}$ and $\mathcal{W}_\mathrm{GV}$, respectively.

Let us start with the degenerate network $\mathcal{W}_\mathrm{np}$. In this example the framing at $X=0$ ($X=\infty$) is given by a choice of local section at $X=0$ ($X=\infty$) that decays fastest when approaching the puncture at $X=0$ along the positive $X$-axis. For the family of $q$-difference opers $\nabla_\epsilon(\mathbf{t})$ this implies that the framing is given by the Borel sum $\Psi_\mathrm{np}(\epsilon,t)$ at $X=0$ (and its analytic continuation at $X=\infty$). The spectral A-coordinate is thus given by the holonomy of~$\Psi_\mathrm{np}$ along the 1-cycle $\gamma_A$, whereas the spectral B-coordinate is given by the (regularized) ratio of $\Psi_\mathrm{np}$ between $X=1$ and $X=0$.
 In particular, the spectral coordinates $\mathcal{X}^{\mathrm{np}}_\gamma$, evaluated at~$\nabla_\epsilon(\mathbf{t})$ with $\arg(\epsilon) = 0$, are equal to the (regularized) exact quantum periods~$\Pi_{A}$ and~$\Pi^{\rm reg}_{B,\mathrm{np}}$.

On the other hand, the framing for the degenate network $\mathcal{W}_\mathrm{GV}$ at $X=0$ (as well as $X=\infty$) is given by a local section invariant under the monodromy $x \mapsto x+1$. For the family of $q$-difference opers $\nabla_\epsilon(\mathbf{t})$ the framing is thus given by the Borel sum $\Psi_\mathrm{GV}(\epsilon,t)$ at $X=0$ (and its analytic continuation at~$X=\infty$). The canonical spectral A-coordinate is then given by the holonomy of~$
\Psi_\mathrm{GV}$ around $X=0$, whereas the spectral B-coordinate is given by the (regularized) ratio of~$\Psi_\mathrm{GV}$ between $X=1$ and $X=0$. We thus find that the spectral coordinates $\mathcal{X}^{\mathrm{GV}}_\gamma$, evaluated at $\nabla_\epsilon(\mathbf{t})$ with $\arg(\epsilon) = \pi/2$, are equivalent to the (regularized) exact quantum periods~$\Pi_{A}$ and~$\Pi^{\rm reg}_{B,\mathrm{GV}}$.

\subsection{5d NRS proposal}\label{sec:NRS}

As reviewed in Section~\ref{sec:quantummirrorcurve}, motivated by \cite{ADKMV} and \cite{NS} it was first proposed in \cite{ACDKV} that the topological string partition function $Z^{\rm NS}$ in the NS limit may be recovered from the monodromies of the open partition function $\Psi$. This was checked in an $\epsilon$-expansion for various toric geometries.

In Theorem~\ref{theorem3}, we proved its non-perturbative generalization for the resolved conifold geometry. That is, with the assumptions on the ray $\rho_k$ and parameters as in Theorem \ref{theorem3}, we found that the exact NS free energy $W_{\rho_k}$ and the exact quantum vev $\Psi_{\rho_k}$ are related as
\begin{equation}\label{eqn:relationWandPsi}
 \exp\left(-\frac{1}{2\pi}\partial_tW_{\rho_k}(\epsilon,t)\right)=\big(\mathcal{S}_2(x\mid \check\epsilon,1)^{-1}\cdot \Psi_{\rho_k}(\epsilon,x,t)\big)\big|_{x=0},
\end{equation}
and in the limit $k \to \infty$ as
\begin{equation}\label{eqn:relationWandPsiGV}
 \exp\left(-\frac{1}{2\pi}\partial_t W_{\mathrm{GV}}(\epsilon,t)\right)=\big(L(x\mid \check\epsilon,1)^{-1}\cdot \Psi_{\mathrm{GV}}(\epsilon,x,t)\big)\big|_{x=0},
\end{equation}
where $W_{\mathrm{GV}}(\epsilon,t)$ was called $W(\epsilon,t)$ in Section~\ref{sec:Borelsums}.

The aim of this section is to interpret equation~\eqref{eqn:relationWandPsiGV} as a five-dimensional lift of the Nekrasov--Rosly--Shatashvili conjecture \cite{Nekrasov:NRS}, and equation~\eqref{eqn:relationWandPsi} as its extension to more general boundary conditions in the $\frac{1}{2}\Omega$-background.

Remember that for any four-dimensional $\mathcal{N}=2$ gauge theory $T$ of rank $K$, with a ramified Seiberg--Witten covering $\Sigma \to C$ of degree $K$, it was realized in \cite{Hollands:Heun, Hollands:2013qza} that the corresponding spectral network takes a special form at a certain phase $\vartheta_{\mathrm{FN}}$ (as well as $\vartheta_{\mathrm{FN}}+\pi$) in any weak-coupling region. This network is characterized by the appearance of families of compact trajectories, and called a Fenchel--Nielsen (FN) type network. The spectral coordinates corresponding to such a FN-type network, defined using abelianization, agree with (a possibly higher-rank version of) complexified FN length-twist coordinates $\mathcal{X}_\gamma^{\mathrm{FN}}$.

Furthermore, in \cite{Hollands:Heun, Hollands:recipe} the generating function $W^\mathrm{oper}$ of the relevant family of oper connections $\nabla_\epsilon(u)$ was computed with respect to these complexified FN coordinates in several examples. This was accomplished by evaluating the spectral coordinates $\mathcal{X}_\gamma^{\mathrm{FN}}$ on the family~$\nabla_\epsilon(u)$. As predicted by \cite{Nekrasov:NRS}, the generating function $W^\mathrm{oper}$ in terms of the spectral coordinates~$\mathcal{X}_\gamma^\mathrm{FN}$ was found to agree with the
NS free energy~$F^{\rm NS}$. That is, it was found that
\begin{equation*}
y^i = \frac{1}{\epsilon} \frac{\partial F^{\mathrm{NS}}(x,\epsilon)}{\partial x_i},
\end{equation*}
where the $x_i$ and $y^i$ are defined in terms of the spectral coordinates $\mathcal{X}_\gamma^{\mathrm{FN}}$ as
\begin{gather*}
x_i = \frac{\epsilon}{\pi \mathrm{i}} \log \mathcal{X}^{\mathrm{FN}}_{A^i} (\nabla_\epsilon(u)) = \frac{\epsilon}{\pi \mathrm{i}} \log \bigg(\frac{1}{\epsilon} \mathcal{B}_{\vartheta^{\mathrm{FN}}} \oint_{A^i} \lambda^{\mathrm{qu}}(\epsilon,u) \bigg),\\
y^i = \frac{1}{2 \epsilon} \log \mathcal{X}^{\mathrm{FN}}_{B_i} (\nabla_\epsilon(u)) =\frac{1}{2 \epsilon} \log \bigg(\frac{1}{\epsilon} \mathcal{B}_{\vartheta^{\mathrm{FN}}} \oint_{B_i} \lambda^{\mathrm{qu}}(\epsilon,u) \bigg).
\end{gather*}

Our goal in this section is to interpret the free energy $F^{\mathrm{NS}}_\mathrm{GV}(\epsilon,t)$ for the resolved conifold geometry similarly as a generating function of difference opers with respect to a special choice of Darboux coordinates on the associated moduli space of difference connections $\mathcal{M}_\mathrm{diff}$. The relevant Darboux coordinates in this set-up are the spectral coordinates $\mathcal{X}_\gamma^\mathrm{GV}$ that were defined in the previous section.

Indeed, equation~\eqref{eqn:relationWandPsiGV}, together with the limit
\begin{equation*}
 \lim_{x\to \mathrm{i} \infty} \Psi_{\mathrm{GV}}(\epsilon,x,t) = 1,
\end{equation*}
implies that the NS free energy $F^{\rm NS}_\mathrm{GV}(\epsilon,t) = W_\mathrm{GV} \big(\epsilon, t- \frac{\check{\epsilon}}{2}\big)$ for the resolved conifold geometry may be obtained as a solution from the relation
\begin{gather}\label{eqn:5dNRSconifold}
v = - \frac{1}{2 \pi } \frac{\partial W_\mathrm{GV}(\epsilon,t)}{\partial t},
\end{gather}
where
\begin{equation*}
t = \frac{\epsilon}{4 \pi^2 \mathrm{i}} \log \mathcal{X}^{\mathrm{GV}}_{A} (\nabla_\epsilon(t)),\qquad
v = \log \mathcal{X}^{\mathrm{GV}}_{B} (\nabla_\epsilon(t)).
\end{equation*}
These three equations together generalize the Nekrasov--Rosly--Shatashvili conjecture to five dimensions.\footnote{Note that in this paper we have mostly restricted to Re$(t)>0$ as well as Im$(t)>0$. Whereas we expect the free energy $W_{\rm GV}$ to be continuous at Re$(t)=0$, this does not seem to be true at Im$(t)=0$. In the 5d gauge theory this corresponds to the statement that the 1-loop contribution to the Nekrasov partition function is not uniquely defined. The 1-loop contribution that we wrote down in equation~\eqref{eqn:one-loop} is also known as the 1-loop contribution to the partition function in the topological string scheme. We expect that the average partition function is given by the 1-loop contribution in the so-called Liouville scheme. See \cite{Hollands:recipe} for 4d examples and references. }

More generally, the function $W_\rho(\epsilon,t)$ is the generating function of $q$-difference opers with respect to the spectral coordinates $\mathcal{X}^\rho_\gamma$ defined by Borel summing the quantum periods $\Pi_\gamma(\epsilon)$ with respect to any ray $\rho$.
Similarly as in four dimensions, the exponent of the free energy $F^{\mathrm{NS}}_\rho(\epsilon,t)= W_\rho \big(\epsilon, t+ \frac{\check{\epsilon}}{2}\big)$ has an interpretation as a five-dimensional partition function as well, but now with respect to a boundary condition labeled by $\rho$.

Whereas the four-dimensional boundary conditions can be realized explicitly by coupling a~3d $\mathcal{N}=2$ theory of class $\mathcal{R}$ to the boundary $\mathbb{R}^2 \times S^1$ of the $\frac{1}{2}\Omega$-background $ \mathbb{R}^2 \times \mathbb{R}_\epsilon^2$ (see \cite{Hollands:t3abel} and in particular \cite[Section 5.7]{Hollands:recipe}), in the five-dimensional set-up we should consider the lift of the 3d theory to $S^1_R \times \mathbb{R}^2 \times S^1$. In particular, the Stokes jumps
\begin{equation*}
\Delta W_\rho(\epsilon,t) = -\mathrm{i} \check{\epsilon} \mathrm{Li}_2\big({\rm e}^{\pm 2\pi \mathrm{i} ( t + k)/\check \epsilon R }\big),
\end{equation*}
can be realized by coupling a lift of the tetrahedron theory $\Delta$ (analyzed in \cite{Dimofte:2011ju}) to the boundary of the five-dimensional $\frac{1}{2}\Omega$-background.

Note that equation~\eqref{eqn:5dNRSconifold} does not fix the $t$-independent part of $F^{\mathrm{NS}}_\rho(\epsilon,t)$, i.e.,~the constant map contribution to the free energy. The additional equation~(\ref{normalizedW}),
\begin{gather}\label{eqn:Wpot}
\partial_\epsilon \hat{W}_\rho (\epsilon,t) = \hat{F}^\mathrm{top}_\rho (\epsilon,t),
\end{gather}
however, does fix the normalization of $\hat{W}_\rho$ in terms of the normalization of $\hat{F}^\mathrm{top}_\rho (\epsilon,t)$. Since this relation contains a derivative with respect to merely $\epsilon$ (and no other moduli), it might have an interpretation in terms of the BPS states with central charge~$Z_k$, corresponding to bound states of D0-branes.

In fact, it is tempting to interpret the equations~\eqref{eqn:5dNRSconifold} and equation~\eqref{eqn:Wpot} together as a~system of relations on the extended moduli space of local special geometry, parametrized by all four 1-cycles $\gamma_0$, $\gamma_A$, $\gamma_B$ and $\gamma_6$, that determine $F^{\rm NS}$ (including constant map contribution) as a~generating function of the family of difference equations in suitable Darboux coordinates.

\subsection{Spectral problems and exact quantization conditions}\label{sec:spectralproblem}

In the work of Nekrasov and Shatashvili \cite{NS} it was proposed that the NS limit of the Nekrasov partition function provides a link between 4d $\mathcal{N}=2$ gauge theories in the $\frac{1}{2}\Omega$-background and quantum integrable systems with quantization parameter $\hbar = \epsilon$. In particular, it was proposed that the equation
\begin{equation}\label{eq:exactquantization}
 \exp\big(\partial_t F^{\rm NS}(\epsilon,t)\big) = 1,
\end{equation}
whose solutions determine the vacua of the effective 2d gauge theory, corresponds to the Bethe ansatz equation for the dual quantum integrable system. A similar equation was considered in~\cite{ACDKV} as a condition for the open topological string wave function to be well defined.

The equation~\eqref{eq:exactquantization} was interpreted in \cite{Mironov:2009uv} as an analogue of a Bohr--Sommerfeld quantization condition, where the volume of the phase space given by the quantum periods is quantized in units of $\hbar$. The solutions of this equation determine a spectral problem: the values of $t$ for which this equation is satisfied determine the energy eigenvalues of the corresponding Schr\"odinger operator. Analogous spectral problems in the context of topological string theory and quantum mirror curves have been analyzed in great depth in~\cite{GHM} and many follow-ups (see~\cite{MarinoSpectral} for a~review and references). Here, it was realized that in order to recover the correct energy spectrum, the free energy~$F^{\rm NS}$ in equation~\eqref{eq:exactquantization} needs to be replaced with its non-perturbative completion~$F^{\rm NS}_\mathrm{np}$. Simplifying expressions were found at special values of~$\hbar$, such as the maximally symmetric value $\hbar=2\pi$. Indeed, the latter value is the fixed point of the S-duality transformation $\hbar \rightarrow 2\pi/\hbar$, which was observed to be a symmetry of the Grassi--Hatsuda--Marino quantization condition in~\cite{Hatsuda:2015fxa,Wang:2015wdy}.

In this section, we formulate and study two novel spectral problems associated to the quantum mirror curve for the resolved conifold. But before getting there, let us explain in more detail what a spectral problem is and how one can extract spectral problems from a Schr\"odinger operator~$\mathbf{d}$ (or~$\mathbf{D}$) associated to a~4d $\mathcal{N}=2$ field theory (or its 5d lift).

To define a spectral problem, it is important to keep in mind that a Schr\"odinger operator alone does not yet specify a quantum mechanics problem. Indeed, this also requires a choice of Hilbert space containing the wave-functions on which the Schr\"odinger operator acts. Together, a Schr\"odinger operator and a suitable Hilbert space $\mathcal{H}$ define a spectral problem: namely, to determine the energy eigenfunctions $\psi$ (or $\Psi$) and eigenvalues $E$ of the operator $\mathbf{d}$ (or $\mathbf{D}$) in the space of functions~$\mathcal{H}$. The energy eigenfunctions are often referred to as bound states.

The most famous example of a spectral problem is the quantum harmonic oscillator, i.e., determining the $L^2(\mathbb{R})$ spectrum of the Schr\"odinger equation
\begin{equation*}
 \hbar^2 \psi''(x) + \big(x^2 - E \big) \psi(x) = 0,
\end{equation*}
with $E>0$ and $x \in \mathbb{R}$. Its energy spectrum can of course be exactly determined as the discrete set of eigenvalues $E = \hbar \left( 2 n + 1 \right)$ for $n>0$.

Given a differential operator $\mathbf{d}(\epsilon,z)$ on a (punctured) Riemann surface $C$ with $\epsilon \in \mathbb{C}$, there are often multiple interesting choices of Hilbert space $\mathcal{H}$. Consider for instance the Schr\"odinger operator
\begin{equation}\label{eq:Mathieu}
 \mathbf{d}(\epsilon,z) = \epsilon^2 \partial_z^2 + \left( \frac{1}{z^3} + \frac{2E+\frac{1}{4} \epsilon^2} {z^2} + \frac{1}{z} \right)
\end{equation}
which plays an important role in the analysis of the 4d $\mathcal{N}=2$ pure ${\rm SU}(2)$ theory.

If we assume $\epsilon = \hbar >0$ and write $z={\rm e}^{i x}$ (as well as redefine the wave-function $\psi$ slightly), we recover the Mathieu differential equation
\begin{equation*}
 \hbar^2 \psi''(x) - \left( 2 \cos(x) + 2 E \right) \psi(x)=0.
\end{equation*}
If we consider $x \in \mathbb{R}$, the associated Mathieu spectral problem is to find quasi-periodic solutions
\begin{equation*}
 \psi(x + 2 \pi) = {\rm e}^{{\rm i} \nu} \psi(x),
\end{equation*}
with $\nu \in \mathbb{R}/2 \pi \mathbb{Z}$. It turns out that there is a countable set of solutions labeled by $n \in \mathbb{N}$ with energies $E_n < 1$.

On the other hand, if we make the substitution $x = \mathrm{i} x'+ \pi$, then the Schr\"odinger operator~\eqref{eq:Mathieu} defines the modified Mathieu equation
\begin{equation*}
 \hbar^2 \psi''(x') - \left( \cosh(x') - 2 E \right) \psi(x')=0.
\end{equation*}
If we consider $x' \in \mathbb{R}$, the associated modified Mathieu spectral problem is to find $L^2(\mathbb{R})$ solutions for this confining potential. Again, there is a discrete set of solutions labeled by $n \in \mathbb{N}$ with energies $E_n > 1$.

It turns out to be very helpful to consider these spectral problems from the perspective of spectral networks. Indeed, the condition for bound states can be written in terms of spectral coordinates, and the solutions may be analysed using the exact WKB analysis \cite{Hollands:t3abel}.

For instance, bound states for the modified Mathieu spectral problem correspond to solutions of the exact quantization condition
\begin{equation*}
 \mathcal{X}^\mathrm{FN}_B = 1
\end{equation*}
in terms of the Fenchel--Nielsen coordinates $\mathcal{X}^\mathrm{FN}_{A,B}$ associated to the Fenchel--Nielsen spectral network (which appears at $E>1$ and $\vartheta = 0$). Because the energy eigenvalues for this spectral problem have values~$E_n>1$ as well, this means the WKB asymptotics of the spectral coordinates~$\mathcal{X}^\mathrm{FN}_{A,B}$ may be used to find good approximations of the energies $E_n$.

On the other hand, bound states for the original Mathieu spectral problem correspond to solutions of
\begin{equation*}
 \mathcal{X}^\mathrm{FN}_A = {\rm e}^{{\rm i} \nu}.
\end{equation*}
Since the eigenvalues for this spectral problem have values $E_n<1$, the WKB asymptotics of the spectral coordinates $\mathcal{X}^\mathrm{FN}_{A,B}$ do not give a good approximation in this case. Rather, this problem should be analysed in terms of a different set of spectral coordinates corresponding to the spectral network at values $E<1$ and $\vartheta =0$. The resulting exact quantization condition is given in \cite[equation~(4.25)]{Hollands:t3abel}.

Consider now the Schr\"odinger difference operator~\eqref{eqn:diffDX1}
\begin{equation*}
 \mathbf{D}(\epsilon,x,t) = (1-X) {\rm e}^{\check{\epsilon} \partial_x} - (1 - Q X),
\end{equation*}
on $\mathbb{C}^*_X$ with $X = {\rm e}^{2 \pi {\rm i} x}$, obtained by canonically quantizing the mirror curve of the resolved conifold geometry. Similarly as in the Mathieu example, we may try to formulate two distinct spectral problems associated to this difference operator, which are in this case intrinsically linked to the two spectral networks $\mathcal{W}^\mathrm{GV}$ and $\mathcal{W}^\mathrm{np}$.

To formulate the spectral problem corresponding to the spectral network $\mathcal{W}^\mathrm{GV}$ one would suppose $\epsilon = i \hbar \in i \mathbb{R}_{>0}$, and consider periodic solutions $\Psi(X)$ to the difference equation
\begin{equation*}
 (1-X) \Psi(q X) - (1 - Q X) \Psi(X) =0,
\end{equation*}
with $q = \exp (i \epsilon)$ and $Q = \exp (-2 \pi E)$, for $x \in \mathbb{R}$. This corresponds to the exact quantization condition
\begin{equation*}
 \log \mathcal{X}_A^\mathrm{GV} = 2 \pi \mathrm{i} n
\end{equation*}
for $n \in \mathbb{Z}$. This spectral problem is rather trivial, however. It just implies that
\begin{equation*}
E = \frac{\hbar n}{2\pi},
\end{equation*}
since the quantum $A$-period does not receive any quantum corrections.

A more interesting spectral problem may be found by supposing that $\epsilon = \hbar >0$, corresponding to $\vartheta = 0$ and the spectral network $\mathcal{W}^\mathrm{np}$. The associated spectral problem is to find $L^2$ solutions to the difference equation
\begin{equation*}
 (1-X) \Psi(q X) - (1 - Q X) \Psi(X) =0,
\end{equation*}
in the interval $X \in (0,1)$ (or equivalently $x \in (0,\mathrm{i} \infty)$).
This corresponds to the exact quantization condition
\begin{equation*}
 \mathcal{X}_B^\mathrm{np} = 1
\end{equation*}
for $n \in \mathbb{N}$. Using Theorem~\ref{theorem3} and the definition of the spectral coordinate $\mathcal{X}_B^\mathrm{np}$ in Section~\ref{sec:exactWKB}, this condition may be reformulated in terms of the NS free energy as
\begin{equation*}
 \exp \left(-\frac{1}{2 \pi} \partial_t W_{\rm np} (\epsilon,t) \right) = 1.
\end{equation*}

If we include the classical contributions to $W_{\rm np} (\epsilon,t) = F^{\rm NS}_{{\rm np}}(\epsilon,t-\check{\epsilon}/2)$ as discussed in Section~\ref{sec:quantumPF}, and use the resulting non-perturbative form of the NS free energy from equation~\eqref{eq:nonpertquantumperiodwithclassical}, we find that the exact quantization condition can be rewritten as
\begin{equation}\label{eq:exactquantcond}
 \exp\left( \frac{1}{2\pi} \partial_t F^{{\rm NS},\sharp}_{{\rm np}} \left(\epsilon,t-\frac{\check{\epsilon}}{2} \right) \right) = \exp\left( \pi \mathrm{i} B_{2,2}(t\mid \check{\epsilon},1) \right) \cdot \mathcal{S}_2(t\mid \check{\epsilon},1) = 1.
\end{equation}
Note that this quantization condition is manifestly invariant under the S-duality transformation $t\rightarrow t/\check{\epsilon}$ with $\check{\epsilon}\rightarrow 1/\check{\epsilon}$, as expected in \cite{Hatsuda:2015fxa}. This follows from the fact that $\mathcal{S}_2(z \mid \omega_1,\omega_2)$ and $B_{2,2}(z \mid \omega_1,\omega_2)$ are both invariant under a simultaneous rescaling of the arguments, and symmetric in $\omega_1$ and $\omega_2$.

Although the expression~\eqref{eq:exactquantcond} combines classical pieces in $t$ with the quantum dilogarithm, the combination again has a product expansion. This follows from the product representation of $\mathcal{S}_2$ together with the identities from~\eqref{eq:multiplesineproduct}. That is,
\begin{equation*}
 \exp\left(\pi \mathrm{i} B_{2,2}(t\mid \check{\epsilon},1) \right) \cdot \mathcal{S}_2(t\mid \check{\epsilon},1)= \prod_{k=0}^{\infty} \frac{1- {\rm e}^{-\frac{2\pi \mathrm{i}}{\check{\epsilon}} (t+k)}}{1- {\rm e}^{- 2\pi \mathrm{i} (t-(k+1)\check{\epsilon})}}.
\end{equation*}
The spectral problem thus translates into finding the solutions of
\begin{equation}\label{eqn:TBA}
 \prod_{k=0}^{\infty} \frac{1- {\rm e}^{-\frac{2\pi \mathrm{i}}{\check{\epsilon}} (t+k)}}{1- {\rm e}^{- 2\pi \mathrm{i} (t-(k+1)\check{\epsilon})}} =1.
\end{equation}

This problem can be simplified in the limit $\epsilon \to 0$. Indeed, assuming $0<\operatorname{Re}(t)<1$ and $\operatorname{Im}(t)>0$, we have that as $\epsilon \to 0$ with $\operatorname{Re}(\epsilon)>0$
\begin{equation*}
 \pi \mathrm{i} B_{2,2}(t\mid \check\epsilon,1) +\log(\mathcal{S}_2(t\mid \check\epsilon,1))\sim \frac{1}{2\pi \mathrm{i}\check\epsilon}\mathrm{Li_2(Q)} + \frac{\pi \mathrm{i}}{\check\epsilon}\left(t^2-t+\frac{1}{6}\right).
\end{equation*}
In such a limit, we would then like to solve
\begin{equation}\label{eqn:TBAapprox}
 \frac{1}{2\pi \mathrm{i}\check\epsilon}\mathrm{Li_2(Q)} + \frac{\pi \mathrm{i}}{\check\epsilon}\left(t^2-t+\frac{1}{6}\right)=2\pi \mathrm{i} n,
\end{equation}
or equivalently
\begin{equation*}
 \mathrm{Li}_2(Q)=2\pi^2\big(t^2-t+1/6\big)-2\pi n \epsilon.
\end{equation*}
Numerical solutions to this equation are plotted in Figure~\ref{fig:densityplotTBA}.

By taking $\operatorname{Im}(t)$ very large, we can assume the polynomial term dominates over the $\mathrm{Li}_2(Q)$, so that we may drop the latter. Hence, an approximated solution satisfying all the constraints of $0<\operatorname{Re}(t)<1$ and $\operatorname{Im}(t)>0$ very large is given by
\begin{equation*}
 t=\frac{1}{2}+{\rm i}\frac{\sqrt{-\pi -12 n \epsilon}}{2\sqrt{3\pi}}
\end{equation*}
for all $n<N<0$ with $|N|>>0$. In particular, as $n\to -\infty$ we have $t\sim {\rm i}\sqrt{-n\epsilon/\pi}$.

\begin{figure}
\begin{center}
\includegraphics[width=8cm]{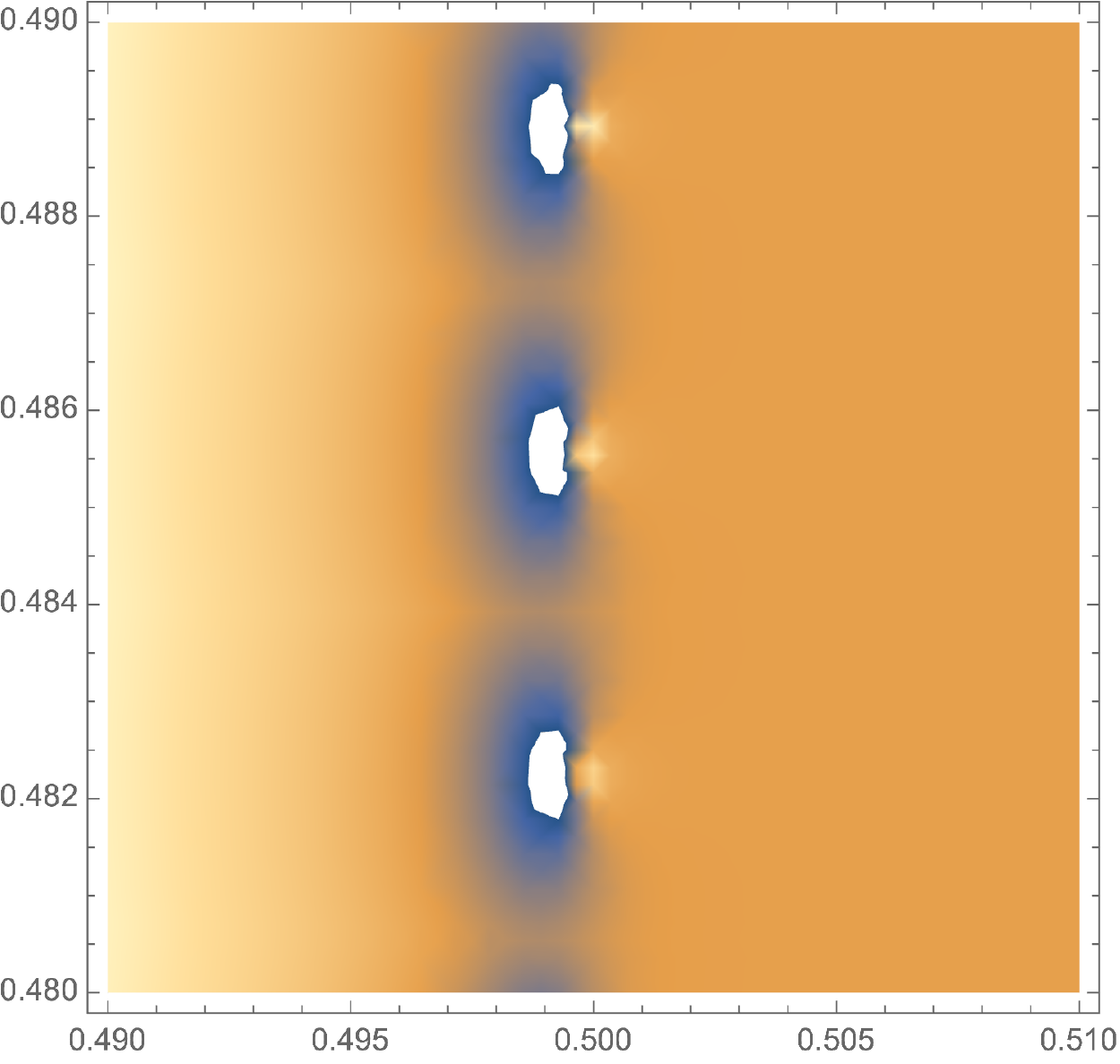}
\end{center}
\caption{Density plot showing the solutions to the TBA~\eqref{eqn:TBA} in the limit $\epsilon \to 0$ with $\operatorname{Im}(t)$ very large. This plot shows the solutions in the complex $t$-plane to equation~\eqref{eqn:TBAapprox} for $\epsilon = 0.01$.}
\label{fig:densityplotTBA}
\end{figure}

On the other hand, the spectral problem also simplifies when special values of $\check{\epsilon}$ are considered, see \cite{GHM,MarinoSpectral}. For instance when $\check{\epsilon}=1$, which is the fixed point under the S-duality transformation, the exact quantization condition is
\begin{equation}\label{qdspecialval}
 \pi \mathrm{i} B_{2,2}(t\mid 1,1) + \log ( \mathcal{S}_2(t\mid 1,1)) = 2\pi \mathrm{i} n,
\end{equation}
for $n\in \mathbb{Z}$. After using equation~\eqref{eq:qdilogrepeat}, this reads
\begin{equation*}
 \pi \mathrm{i}\left( \frac{5}{6} -2t+t^2 \right) -\frac{1}{2\pi \mathrm{i}} {\rm Li}_2(Q)+ (t-1){\rm Li}_1(Q) = 2\pi \mathrm{i} n.
\end{equation*}
Here we can, as before, assume that $\operatorname{Im}(t)>0$ is sufficiently big, so that the polynomial terms dominate over the polylog terms. We can then approximate the previous expression by dropping the polylog terms. The solutions for $\operatorname{Im}(t)\gg 0$ are therefore approximated by
\begin{equation*}
 t=1+ {\rm i}\frac{\sqrt{-1-12n}}{\sqrt{6}}
\end{equation*}
for $n<N<0$ and $|N|\gg 0$. In particular, as $n\to -\infty$ we have $t\sim {\rm i}\sqrt{-2n}$.

This type of equation \eqref{qdspecialval} that contains the quantum dilogarithm function evaluated at special rational values of $\epsilon$ is similar to equations studied in \cite{Garoufalidisevaluation}. It may also be used to find solutions numerically as in \cite{MarinoSpectral}.

Note that in terms of gauge theory, the solutions of the TBA equation~\eqref{eqn:TBA} determine the vacua of the five-dimensional ${\rm U}(1)$ theory in the $\frac{1}{2} \Omega$-background, with respect to the boundary conditions corresponding to $\rho= \mathbb{R}_{>0}$.

\section{Discussion and outlook}\label{sec:discussion}

In this paper we have developed mathematical tools and methods to study exact and non-perturbative refined topological string theory as well as its limits, guided by the example of the resolved conifold and its mirror. We were in particular able to provide the connections as well as further develop the appearance of finite difference equations in the open and closed string moduli.

The finite difference equations in the open string case are the quantum mirror curves. These were advocated in \cite{ADKMV} to capture the quantum corrections to the disk generating functions \cite{Aganagic:2000gs} and in \cite{ACDKV} to capture the NS limit of refined topological string theory via monodromies of the open topological string wave-function. We revisited the WKB method applied to a quantum mirror curve associated to the resolved conifold and proved the all-order expression for the asymptotic expansion of the open topological string wave-function for the resolved conifold. With this all-order expression we constructed the Borel transform and studied its Borel sums as well as the corresponding Stokes jumps across the singular rays in the Stokes plane.

On the closed string moduli side we derived new finite difference equations obtained from the asymptotic expansion of the refined topological string free energy (as well as its limits), generalizing the study of the unrefined topological string in~\cite{alim2020difference}. We found analytic solutions in the topological string coupling for these difference equations, with the correct asymptotics, that were shown to encode the expected non-perturbative content to the refined free energy. These analytic solutions furthermore made it possible to derive a new kind of finite difference equations that cannot be accessed from the asymptotic expansion, see for instance the third equation in~\eqref{eps1shift}. The latter difference equations correspond to integer shifts in the closed string moduli and, in the case of the NS-limit as well as the unrefined limit, were shown to encode the Stokes jumps of the Borel resummation of the initial asymptotic expansion. This establishes a novel link between finite difference equations and Borel analysis. While we have not worked out the Borel analysis for the refined case in detail, we expect $F_{{\rm np}}^{\rm ref}$ to play an analogous role as $W_{{\rm np}}$, and the refined difference equation to encode the Stokes jumps.

In the light of the discussion of the S-duality transformation of topological string theory of~\cite{ASTT21}, we are led to the interpretation that there are two dual sets of difference equations in both S-duality frames in the topological string coupling (or~$\epsilon$ in the NS limit). One kind of difference equation is derived from the asymptotic expansion, while the other kind encodes the corresponding Stokes jumps. The role of the two is exchanged when the duality frame is changed. While the asymptotic expansion and the Stokes jumps in the case of the unrefined topological string theory differ drastically, as worked out in \cite{ASTT21}, we find that in the NS limit the S-duality is almost exact. This can be seen from the two dual difference equations~\eqref{eq:NSdiffeq} and \eqref{tdiffeqNS}, as well from the explicit S-duality transformation of the non-perturbative free energy in the NS limit of Proposition~\ref{prop:Sduality}.

In terms of the closed string moduli, the finite difference equations are most naturally interpreted as being related to the integrable hierarchies underlying Gromov--Witten theory, see, e.g.,~\cite{Pandharipande:2000}. Finite difference equations of the type put forward in our work are shown in~\cite{NO} to govern the perturbative parts of certain natural representation-theoretic partition functions. Related to this, in~\cite{BGT} finite difference equations in the closed string moduli, corresponding to $q$-deformations of Painlev\'e equations, were shown to be satisfied by certain tau functions that are related to topological string theory. In the similar context of the study of tau-functions relating to topological string theory associated to class S theories, finite difference equations in the closed moduli also feature in~\cite{CLT20}.

On different grounds, the topological string partition function is expected to correspond to a quantum mechanical problem obtained from the geometric quantization of a bundle over the closed string moduli space~\cite{Witten:1993ed}. It is thus natural to wonder whether the closed string difference equations are reflections of this quantization problem. Motivated by this question, we re-interpreted the finite difference equations in the closed string moduli as resulting from the quantization of a curve in a phase space associated to the closed string moduli. The wave-function that is annihilated by this quantum curve, and thus only depends on the closed string moduli, is given by a difference quotient of the topological string partition function, but, happens to be equivalent to the refined partition function in the NS limit. The closed quantum curve therefore corresponds to a quantum mechanical problem for the refined topological string in the NS limit.

Although our guiding example in this paper is the resolved conifold and its mirror, we are certain that many aspects of our work can be generalized to a larger class of geometries, and in particular to toric Calabi--Yau threefolds.\footnote{Note that even though the GV free energy for more general geometries is generically not a convergent series in the K\"ahler parameters~$Q_i$, we can still apply Borel summation techniques to such a GV free energy when we considered it a series expansion in the string coupling $\lambda$ whose coefficients are formal series in the~$Q_i$.} In this case, the mirror curves are available and the corresponding quantum curves have been intensively studied. While the closed quantum curve in our work was obtained from the explicit knowledge of the asymptotic series of the topological string free energy in this case, we think that the path to generalization is through understanding the more general link between the open and the closed quantum curve. This link will also lead to the generalization of the quantum Picard--Fuchs operator that we derived from the difference equation for the quantum periods. We speculate that the quantum Picard--Fuchs operator for a~generic toric Calabi--Yau threefold could be obtained from the quantum curve in the same way that the classical Picard--Fuchs operator is obtained from the classical curve, namely through a~study of the changes in the open wave function under variation of the closed moduli.

Further guidance towards the generalization of the resurgence properties of more general CY geometries may be offered by understanding the links to the studies of resurgence in the context the holomorphic anomaly equations of~\cite{Bershadsky:1993cx} as in~\cite{Couso-Santamaria:2014iia}; see also~\cite{CousoSantamaria:2014xml} and references therein.
These techniques have been applied to the study of the proposal of~\cite{GHM} in~\cite{CMS}.

Another place where $q$-deformations of classical Picard--Fuchs operators have appeared recently is in the context of quantum K-theory, see, e.g.,~\cite{Jockers:2019wjh,Jockers:2018sfl}. It would be very beneficial to examine the precise links. Yet another very active area of research where similar questions on the interplay between quantization of curves, resurgence and non-perturbative structures is examined is that of the topological recursion. See, for example,~\cite{Eynard:2021sxg} for a recent study in this direction.

In Section~\ref{sec:Borelsums}, we have given the non-perturbative solutions to the open and closed difference equations (in the NS limit) an interpretation in terms of Borel summation. For instance, the Borel summation of the NS free energy $F^{\mathrm{NS}}(\epsilon)$ along any ray in the Borel plane, corresponding to a phase $\vartheta = \arg(\epsilon)$, yields an exact solution that may be analytically continued to the half-plane centered around $\vartheta$. The particular solution $F^{\mathrm{NS}}_\mathrm{np}(\epsilon)$, obtained by Borel summation along the ray $\rho = \mathbb{R}_{>0}$, is in some sense the richest of all these exact solutions. All other Borel summations may be obtained after adding a number of Stokes jumps. In particular, in the limit $\vartheta \to \pi/2$ the solution $F^{\mathrm{NS}}_\rho(\epsilon)$ limits to the Gopakumar--Vafa resummation of the NS free energy.

Borel summation is known to play an important role in the exact WKB analysis of differential operators. In Section~\ref{5dWKBsec}, we have shown that a similar story is true for $q$-difference operators, at least in the example of the resolved conifold. In particular, we show that the rays~$\rho$ along which the Borel transform of $F^{\mathrm{NS}}(\epsilon)$ has its singularities correspond to phases for which the corresponding Stokes graph (or exponential spectral network) undergoes a topology change. In the dual 5d gauge theory these topology changes correspond to 5d BPS particles. Although the spectrum for more general toric geometries may be extracted from the corresponding Stokes graphs (see~\cite{Banerjee:2020moh} for the example of the local $\mathbb{P}^1 \times \mathbb{P}^1$ geometry), it is not clear whether the corresponding phases~$\vartheta$ have a similar interpretation in terms of the Borel sum of the corresponding NS free energy (some aspects of this Borel resummation have been studied numerically in \cite{Grassi:2014cla}) and which type of invariants these BPS states correspond to in terms of the Calabi--Yau geometry (see \cite{Banerjee:2022oed} for recent progress on this).

It has been known for a while that the open and closed topological string in the NS limit are related through the quantum periods \cite{ACDKV}. By interpreting
Theorem~\ref{theorem3} as a relation between the Borel summed quantum periods and the NS free energy in Section~\ref{sec:NRS}, we found a non-perturbative generalization of this statement. We furthermore argued that the exact quantum periods may be interpreted as spectral coordinates on the corresponding moduli space of multiplicative Hitchin systems through a $q$-difference version of abelianization. Using this language we formulated a 5d lift of the Nekrasov--Rosly--Shatashvili conjecture \cite{Nekrasov:NRS}. It would be very exciting to try to generalize this story to mirror curves of higher genus. This is certainly not obvious. For instance, one obstacle is that one cannot define difference operators on an arbitrary Riemann surface---the surface has to be of type $\mathbb{C}$, $\mathbb{C}^*$ or a torus in order to be able to define the difference operator as a global object. We do think that there should be a way around this, for instance by considering a degeneration limit of the higher genus mirror curves, or by considering differential operators for affine Lie algebras.

\appendix
\section{Some technical or long proofs}\label{appA}
\subsection{Proof of the Borel transform}\label{BTapp}

In order to prove Proposition \ref{Boreltransprop}, we follow the same arguments as the ones used to compute the Borel transform of the topological free energy in \cite{ASTT21}, which in turn were based on the techniques of \cite{garoufalidis2020resurgence}.

We start by recalling the Hadamard product:

\begin{Definition}
Consider two formal power series $\sum_{n=0}^{\infty}a_nz^n,\sum_{n=0}^{\infty}b_nz^n\in \mathbb{C}[[z]]$. Then the Hadamard product $\oast\colon\mathbb{C}[[z]]\times \mathbb{C}[[z]] \to \mathbb{C}[[z]]$ is defined by
\begin{equation*}
 \bigg(\sum_{n=0}^{\infty}a_nz^n\bigg)\oast \bigg(\sum_{n=0}^{\infty}b_nz^n\bigg)=\sum_{n=0}^{\infty}a_nb_nz^n.
\end{equation*}
\end{Definition}

Whenever the power series $\sum a_nz^n$ and $\sum b_nz^n$ have a non-zero radius of convergence, their Hadamard product also has a non-zero radius of convergence, and we have the following integral representation for the product, which follows easily from the Cauchy integral formula:

\begin{Lemma} \label{lemHad} Consider two holomorphic functions near $z=0$ having series expansions
\begin{equation*}
 f_1(z)=\sum_{n=0}^{\infty}a_nz^n, \qquad f_2(z)=\sum_{n=0}^{\infty}b_nz^n
\end{equation*}
with radius of convergence $r_1>0$ and $r_2>0$, respectively. Then $(f_1\oast f_2)(z)$ converges for $|z|<r_1r_2$, and for any $\rho \in (0,r_1)$ the following holds for $|z|<\rho r_2$:
\begin{equation}\label{Hadprodint}
 (f_1\oast f_2)(z)=\frac{1}{2\pi \mathrm{i}}\int_{|s|=\rho}\frac{{\rm d} s}{s}f_1(s)f_2\Big(\frac{z}{s}\Big).
\end{equation}
\end{Lemma}

In order to compute the desired $G(\xi,t)$, we write
\begin{equation*}
 G(\xi,t)=(f_1\oast f_2(-,t))(\xi)
\end{equation*}
for two functions $f_1(\xi)$, $f_2(\xi,t)$ which are holomorphic near $\xi=0$, and then use the integral representation of Lemma \ref{lemHad}.
We will take $f_1(\xi)$, $f_2(\xi,t)$ to be the following:
\begin{gather}
 f_1(\xi) = -\frac{1}{2\pi}\sum_{n=2}^{\infty} \frac{ B_{2n}}{ (2n)!} \xi^{2n-2},\nonumber\\
 f_2(\xi,t)= \sum_{n=2}^{\infty} \frac{\xi^{2n-2}}{(2n-2)!} \partial_t^{2n} {\rm Li}_3(Q)=\sum_{n=1}^{\infty} \frac{\xi^{2n-2}}{(2n-2)!} (2\pi \mathrm{i})^{2n} {\rm Li}_{3-2n}(Q).\label{f'sdef}
 \end{gather}

\begin{Proposition}\label{boreltransconv} Let $t\in \mathbb{C}^{\times}$ with $|\operatorname{Re}(t)|<\frac12$. Then $G(\xi,t)$ converges for $|\xi|<2\pi |t|$.
\end{Proposition}

\begin{proof}
Using the fact that
\begin{equation*}
 B_{2n}\sim (-1)^{n+1}4\sqrt{\pi n}\Big(\frac{n}{\pi {\rm e}}\Big)^{2n} \qquad \text{as} \quad n\to \infty,
\end{equation*}
we find that the radius of convergence for $f_1(\xi)$ is $2\pi$. On the other hand, using the fact that for $t\in \mathbb{C}^{\times}$ with $|\operatorname{Re}(t)|<1/2$, we have
\begin{equation*}
 \mathrm{Li}_{3-2n}\big({\rm e}^{2\pi \mathrm{i} t}\big)\sim \Gamma(1-3+2n)(-2\pi \mathrm{i} t)^{3-2n-1} \qquad \text{as}\quad n\to \infty,
\end{equation*}
we find that the radius of convergence of $f_2(\xi,t)$ is $r_2(t)=|t|$.

By the use of Lemma \ref{lemHad}, we find that provided $t\in \mathbb{C}^{\times}$ satisfies $|\operatorname{Re}(t)|<\frac12$, we have that $G(\xi,t)=(f_1\oast f_2(-,t))(\xi)$ converges for $|\xi|<r_1r_2(t)=2\pi |t|$.
\end{proof}

We now prove Proposition \ref{Boreltransprop}.

\begin{proof}[Proof of Proposition \ref{Boreltransprop}]
We use the integral representation of the Hadamard product in Lemma~\ref{lemHad}, together with the convergence results of Proposition~\ref{boreltransconv}. For $t\in \mathbb{C}^{\times}$ with $|\operatorname{Re}(t)|<1/2$ and $\rho \in (0,2\pi)$, we have for $|\xi|<\rho|t|$
\begin{equation*}
 G(\xi,t)=\frac{1}{2\pi \mathrm i}\int_{|s|=\rho}\frac{{\rm d}s}{s}f_1(s)f_2\left(\frac{\xi}{s},t\right),
\end{equation*}
where $f_1(\xi)$ and $f_2(\xi,t)$ are as in (\ref{f'sdef}).

Now notice that
\begin{gather*}
 f_1(\xi) =- \frac{1}{2\pi}\sum_{n=2}^{\infty} \frac{ B_{2n}}{ (2n)!} \xi^{2n-2}= - \frac{1}{2\pi}\frac{1}{\xi^2}\sum_{n=2}^{\infty} \frac{ B_{2n}}{ (2n)!} \xi^{2n}
 =- \frac{1}{2\pi}\left( \frac{1}{\xi({\rm e}^{\xi}-1)} -\frac{1}{\xi^2}+\frac{1}{2\xi}-\frac{1}{12}\right),\!
 \end{gather*}
where we have used the expression for the generating function of the Bernoulli numbers
\begin{equation*}
\frac{w}{{\rm e}^w-1} = \sum_{n=0}^{\infty} B_n \frac{w^n}{n!},
\end{equation*}
and the fact that except $B_1=-\frac12$, all odd Bernoulli numbers vanish. From the final expression we see that $f_1(\xi)$ admits an analytic continuation to a meromorphic function with simple poles at $\xi=2\pi \mathrm{i} \mathbb{Z} \setminus \{0\}$.

On the other hand, for $f_2(\xi,t)$ we find
\begin{align*}
 f_2(\xi,t) &= \sum_{n=2}^{\infty} \frac{\xi^{2n-2}}{(2n-2)!} \partial_t^{2n} {\rm Li}_3(Q)
 = \partial_{\xi}^2 \left(\sum_{n=2}^{\infty} \frac{\partial_t^{2n} {\rm Li}_3(Q)}{(2n)!} \xi^{2n}\right) \\
 &= \partial_{\xi}^2 \left( \frac{1}{2} \big({\rm Li}_3\big({\rm e}^{2\pi \mathrm{i}(t+\xi)}\big)+{\rm Li}_3\big({\rm e}^{2\pi \mathrm{i}(t-\xi)}\big)\big) - {\rm Li}_3\big({\rm e}^{2\pi \mathrm{i} t}\big)-\partial_t^2 {\rm Li}_3\big({\rm e}^{2\pi \mathrm{i} t}\big) \frac{\xi^2}{2} \right) \\
 &=\frac{(2\pi \mathrm{i})^2}{2}\big(\mathrm{Li}_1\big({\rm e}^{2\pi \mathrm{i}(t + \xi)}\big)+\mathrm{Li}_1\big({\rm e}^{2\pi \mathrm{i}(t - \xi)}\big)-2\mathrm{Li}_1 (Q)\big),
\end{align*}
so that $f_2(\xi,t)$ admits an analytic continuation in $\xi$ with branch cuts at $k \mp t\pm x/2\pi {\rm i}$ for $x\geq 0$ and $k\in \mathbb{Z}$. Furthermore the function $f_2(\xi/s,t)$ has branch cuts in $s$ along $\pm \frac{2\pi \mathrm{i} \xi}{2\pi \mathrm{i}(k-t)+x}$ for $x\geq 0$. The integral representation then becomes
\begin{gather*}
 G(\xi,t)=\frac{1}{2\mathrm{i}}\int_{|s|=\rho} \frac{{\rm d}s}{s}\left(\frac{1}{ s({\rm e}^{s}-1)}-\frac{1}{s^2}+\frac{1}{2s}-\frac{1}{12}\right)\\
\hphantom{G(\xi,t)=}{}\times \big(\mathrm{Li}_1\big({\rm e}^{2\pi \mathrm{i}(t + \xi/s)}\big)+\mathrm{Li}_1\big({\rm e}^{2\pi \mathrm{i}(t - \xi/s)}\big) -2\mathrm{Li}_1(Q) \big).
\end{gather*}
This expression has simple poles at $s\in 2\pi {\rm i} (\mathbb{Z}-\{0\})$, while for the parameters satisfying the hypotheses, the branches of $f_2(\xi/s,t)$ are all inside $|s|=\rho$. Hence, one may deform the contour to $\infty$, capturing the simple poles of the first factor (notice that this encircles the poles clock-wise), without crossing any branch cut of the second factor:
\begin{align*}
 G(\xi,t)&=-\pi \sum_{m\in \mathbb{Z}-\{0\}} \lim_{s\rightarrow 2\pi \mathrm{i} m} \left(\frac{(s-2\pi \mathrm{i} m)}{ s^2({\rm e}^{s}-1)}\right)\big(\mathrm{Li}_1\big({\rm e}^{2\pi \mathrm{i}(t + \xi/s)}\big)+\mathrm{Li}_1\big({\rm e}^{2\pi \mathrm{i}(t - \xi/s)}\big)-2\mathrm{Li}_1(Q)\big) \\
 &=-\pi \sum_{m\in \mathbb{Z}-\{0\}} \frac{1}{(2\pi {\rm i} m)^2}\big(\mathrm{Li}_1\big({\rm e}^{2\pi \mathrm{i} t + \xi/m}\big)+\mathrm{Li}_1\big({\rm e}^{2\pi \mathrm{i} t - \xi/m}\big)-2\mathrm{Li}_1(Q)\big)\\
 &=\frac{1}{4\pi}\sum_{m\in \mathbb{Z}-\{0\}} \frac{1}{m^2}\big(\mathrm{Li}_1\big({\rm e}^{2\pi \mathrm{i} t + \xi/m}\big)+\mathrm{Li}_1\big({\rm e}^{2\pi \mathrm{i} t - \xi/m}\big)-2\mathrm{Li}_1(Q)\big).
\end{align*}
The previous expression has branch cuts at $2\pi \mathrm{i}(t+k)m +mx$ for $x\geq 0$, $k\in \mathbb{Z}$, $m\in \mathbb{Z}-\{0\}$. In particular since $0<|\operatorname{Re}(t)|<1/2$, we have that the branches are away from $\mathbb{R}_{>0}$, and we can analytically continue $G(\xi,t)$ to $\xi>0$.
\end{proof}

\subsection[The Borel sum along $\mathbb{R}_{>0}$]{The Borel sum along $\boldsymbol{\mathbb{R}_{>0}}$}\label{BorelsumRapp}
\begin{Proposition}\label{WorformBorelsum}
Let $t\in \mathbb{C}^{\times}$ with $\operatorname{Im}(t)>0$ and let $\epsilon >0$. Then $W_{\mathbb{R}_{>0}}(\epsilon,t)$
admits the following representation:
\begin{equation}\label{altform2}
 W_{\mathbb{R}_{>0}}(\epsilon,t)=-\frac{1}{2\pi}\int_{\mathbb{R}+\mathrm{i}0^+}{\rm d}v\, \frac{1}{1-{\rm e}^v}\mathrm{Li}_2\big({\rm e}^{\check\epsilon v +2\pi \mathrm{i} t}\big).
\end{equation}
where $\mathbb{R}+\mathrm{i}0^+$ denotes a contour along $\mathbb{R}$ avoiding $0$ by doing a small detour along the upper half-plane.
\end{Proposition}

\begin{proof}
We start by performing the change of variables $y=\check{\epsilon} v$ on the right-hand side of (\ref{altform2}), obtaining
\begin{align*}
 -\frac{1}{2\pi\check{\epsilon}}\int_{\epsilon(\mathbb{R}+\mathrm{i} 0^+)}{\rm d}y \, \frac{1}{1-{\rm e}^{2\pi y/\epsilon }}\mathrm{Li}_2\big({\rm e}^{y+2\pi \mathrm{i} t}\big)
 &=-\frac{1}{2\pi\check{\epsilon}}\int_{\mathbb{R}+\mathrm{i} 0^+}{\rm d}y\, \frac{1}{1-{\rm e}^{2\pi y/\epsilon }}\mathrm{Li}_2\big({\rm e}^{y+2\pi \mathrm{i} t}\big)\\
 &=\lim_{\delta \to 0^+}-\frac{1}{2\pi\check{\epsilon}}\int_{\mathbb{R}+\mathrm{i} 0^+}{\rm d}y\, \frac{1}{1-{\rm e}^{2\pi y/\epsilon -i\delta }}\mathrm{Li}_2\big({\rm e}^{y+2\pi \mathrm{i} t}\big),
 \end{align*}
where in the second equality we have used that the range of $\epsilon$ allows us to deform the contour back to $\mathbb{R}+\mathrm{i} 0^+$. Now using
\begin{equation*}
 \frac{{\rm d}}{{\rm d}y}\big({-}\log\big(1-{\rm e}^{-2\pi y/\epsilon +\mathrm{i} \delta}\big)\big)=\frac{2\pi }{\epsilon(1-{\rm e}^{2\pi y/\epsilon -\mathrm{i} \delta})}
\end{equation*}
and integration by parts, we find
\begin{gather*}
 \lim_{\delta \to 0^+} -\frac{1}{2\pi\check{\epsilon}}\int_{\mathbb{R}+\mathrm{i} 0^+}{\rm d}y \,\frac{1}{1-{\rm e}^{2\pi y/\epsilon -\mathrm{i}\delta }}\mathrm{Li}_2\big({\rm e}^{y+2\pi \mathrm{i} t}\big)\\
\qquad{}=\lim_{\delta\to 0^+}\bigg[\frac{1}{2\pi }\log\big(1-{\rm e}^{-2\pi y/\epsilon +\mathrm{i}\delta}\big)\mathrm{Li}_2\big({\rm e}^{y+2\pi {\rm i}t}\big)\Big|_{y=-\infty}^{\infty}\\
\qquad\quad{}+\frac{1}{2\pi }\int_{\mathbb{R}}{\rm d}y\, \log\big(1-{\rm e}^{-2\pi y/\epsilon +\mathrm{i}\delta}\big)\log\big(1-{\rm e}^{y+2\pi {\rm i}t}\big)\bigg].
 \end{gather*}
Because $\epsilon>0$, we obtain that the boundary terms vanish. Furthermore,
splitting the integration over the left and right half intervals, one then obtains
\begin{gather*}
 \lim_{\delta\to 0^+} \bigg[\frac{1}{2\pi }\int_{0}^{\infty}{\rm d}y\, \log\big(1-{\rm e}^{-2\pi y/\epsilon +\mathrm{i} \delta}\big)\log\big(1-{\rm e}^{y+2\pi \mathrm{i} t}\big)\\
 \qquad\quad{}+\frac{1}{2\pi }\int_{0}^{\infty}{\rm d}y \,\log\big(1-{\rm e}^{2\pi y/\epsilon +\mathrm{i}\delta}\big)\log\big(1-{\rm e}^{-y+2\pi \mathrm{i} t}\big)\bigg]\\
 \qquad{} =\lim_{\delta \to 0^+} \widetilde{H}(\epsilon,t,\delta) + \lim_{\delta\to 0^+}H(\epsilon,t,\delta),
 \end{gather*}
where we have defined
\begin{gather*}
 \widetilde{H}(\epsilon,t,\delta):=\frac{1}{2\pi}\int_{0}^{\infty}{\rm d}y \, \log\big(1-{\rm e}^{-y+2\pi \mathrm{i} t}\big)\big(\log\big(1-{\rm e}^{2\pi y/\epsilon+\mathrm{i}\delta}\big)-\log\big(1-{\rm e}^{-2\pi y/\epsilon+\mathrm{i}\delta}\big)\big),\\
 H(\epsilon,t,\delta) = \frac{1}{2\pi }\int_{0}^{\infty}{\rm d}y\, \log\big(1-{\rm e}^{-2\pi y/\epsilon +\mathrm{i}\delta}\big)\big(\log\big(1-{\rm e}^{y+2\pi \mathrm{i} t}\big)+\log\big(1-{\rm e}^{-y+2\pi \mathrm{i} t}\big)\big).
\end{gather*}

We can compute the limit $\widetilde{H}(\epsilon,t):=\lim_{\delta \to 0^+}\widetilde{H}(\epsilon,t,\delta)$, obtaining
\begin{align*}
 \widetilde{H}(\epsilon,t)&=\lim_{\delta \to 0^+}\frac{1}{2\pi}\int_{0}^{\infty}{\rm d}y\, \log\big(1-{\rm e}^{-y+2\pi \mathrm{i} t}\big)\left(\log\left(-{\rm e}^{2\pi y/\epsilon-\mathrm{i}\delta}\frac{1-{\rm e}^{2\pi y/\epsilon+\mathrm{i}\delta}}{1-{\rm e}^{2\pi y/\epsilon -\mathrm{i}\delta}}\right)\right)\\
 &=\frac{1}{2\pi \epsilon}\int_{0}^{\infty}{\rm d}y\, (2\pi y-\mathrm{i}\pi \epsilon) \log\big(1-{\rm e}^{-y+2\pi \mathrm{i} t}\big),
 \end{align*}
where in the last equality we have used that as $\delta \to 0^+$:
\begin{equation*}
 \operatorname{Im}\left(-{\rm e}^{2\pi y/\epsilon-\mathrm{i}\delta}\frac{1-{\rm e}^{2\pi y/\epsilon+\mathrm{i}\delta}}{1-{\rm e}^{2\pi y/\epsilon -\mathrm{i}\delta}}\right)<0.
\end{equation*}

We can furthermore compute $\widetilde{H}(\epsilon,t)$ explicitly by performing an integration by parts to get rid of the log term:
\begin{align}
 \widetilde{H}(\epsilon,t)&=\frac{1}{2\pi \epsilon}\left(\big(\pi y^2-\pi \mathrm{i} \epsilon y\big)\log\big(1-{\rm e}^{-y+2\pi \mathrm{i} t}\big)\Big|_{y=0}^{\infty}-\int_{0}^{\infty}{\rm d}y \, \big(\pi y^2 -\pi \mathrm{i} \epsilon y\big) \frac{-1}{1-{\rm e}^{y-2\pi \mathrm{i} t}}\right)\nonumber\\
 &=-\frac{1}{2\epsilon}\int_0^{\infty}{\rm d}y \, \frac{y^2}{{\rm e}^{y-2\pi \mathrm{i} t}-1}+\frac{{\rm i}}{2}\int_0^{\infty}{\rm d}y \, \frac{y}{{\rm e}^{y-2\pi \mathrm{i} t}-1}.\label{Gcomp}
 \end{align}
Since $\operatorname{Im}(t)>0$, we find that $\big|{\rm e}^{2\pi \mathrm{i} t}\big|<1$, so that the last integrals in (\ref{Gcomp}) corresponds to
\begin{equation*}
 \widetilde{H}(\epsilon,t)=-\frac{1}{\epsilon}\mathrm{Li}_3(Q) +\frac{\mathrm{i}}{2}\mathrm{Li}_2(Q).
\end{equation*}

On the other hand, by expanding the first log term of $H$ and applying the Fubini--Tonelli theorem, we find that
\begin{equation*}
 H(\epsilon,t,\delta)
 =-\sum_{n=1}^\infty\frac{1}{2\pi }\int_{0}^{\infty}\mathrm{d}y\, \frac{{\rm e}^{-2\pi n y/\epsilon +\mathrm{i} n\delta}}{n}\big(\log\big(1-{\rm e}^{y+2\pi \mathrm{i} t}\big)+\log\big(1-{\rm e}^{-y+2\pi \mathrm{i} t}\big)\big).
\end{equation*}
Performing a change of variables in each integral, and interchanging integral and summations again, we obtain
\begin{gather*}
 H(\epsilon,t,\delta)
 =-\frac{1}{2\pi}\int_{0}^{\infty}{\rm d}y\, \mathrm{e}^{-2\pi y/\epsilon}\sum_{n=1}^\infty\frac{{\rm e}^{in\delta}} {n^2}\big( \log\big(1-{\rm e}^{y/n+2\pi \mathrm{i} t}\big)+\log\big(1-{\rm e}^{-y/n+2\pi \mathrm{i} t}\big)\big).
 \end{gather*}
Letting $H(\epsilon,t):=\lim_{\delta \to 0}H(\epsilon,t,\delta)$, we get
\begin{gather*}
H(\epsilon,t) =-\frac{1}{2\pi}\int_{0}^{\infty}{\rm d}y \, \mathrm{e}^{-2\pi y/\epsilon}\sum_{n=1}^\infty\frac{1} {n^2}\big(\log\big(1-{\rm e}^{y/n+2\pi \mathrm{i} t}\big)+\log\big(1-{\rm e}^{-y/n+2\pi \mathrm{i} t}\big)\big).
\end{gather*}
Finally, using that $-\frac{2\pi}{\epsilon}{\rm e}^{-2\pi y/\epsilon}=\frac{{\rm d}}{{\rm d}y}{\rm e}^{-2\pi y/\epsilon}$ and integrating by parts yields
\begin{gather*}
 H(\epsilon,t) =
\left[\frac{\epsilon {\rm e}^{-2\pi y/\epsilon}}{(2\pi)^2}\sum_{n=1}^\infty\frac{1} {n^2}\big(\log\big(1-{\rm e}^{y/n+2\pi \mathrm{i} t}\big)+\log\big(1-{\rm e}^{-y/n+2\pi \mathrm{i} t}\big)\big)\right]\bigg|_{y=0}^{\infty}\\
\hphantom{H(\epsilon,t) =}{}
-\epsilon\int_{0}^{\infty}dy \frac{{\rm e}^{-2\pi y/\epsilon}}{(2\pi)^2}\frac{{\rm d}}{{\rm d}y}\left[\sum_{n=1}^\infty\frac{1} {n^2}\big(\log\big(1-{\rm e}^{y/n+2\pi \mathrm{i} t}\big)+\log\big(1-{\rm e}^{-y/n+2\pi \mathrm{i} t}\big)\big)\right].
\end{gather*}
Using that the boundary term at $\infty$ vanishes, and interchanging the derivative with the sum, we obtain
\begin{align*}
H(\epsilon,t) &=-\frac{2\epsilon}{(2\pi)^2}\log(1-Q)\sum_{n=1}^{\infty}\frac{1}{n^2}+\check{\epsilon}\int_{0}^{\infty}{\rm d}y\, \mathrm{e}^{-2\pi y/\epsilon}\widetilde{G}(y,t)\\
 &=\frac{\epsilon}{12}\mathrm{Li}_1(Q)+\check{\epsilon}\int_{0}^{\infty}{\rm d}y \,\mathrm{e}^{-2\pi y/\epsilon}\widetilde{G}(y,t).
\end{align*}

The result then follows.
\end{proof}

Finally, the result from Proposition \ref{WRWnp} follows from the previous proposition, together with the following:

\begin{Proposition}\label{woronowiczlemma} Let $t\in \mathbb{C}$ be such that $0<\operatorname{Re}(t)<1$, $\operatorname{Im}(t)>0$, and let $\epsilon>0$. Furthermore, assume that $\operatorname{Re}(t)<\operatorname{Re}(\check\epsilon +1)$. Then $W_{\rm np}(\epsilon,t)$
admits the following expression:
\begin{equation*}
 W_{\rm np}(\epsilon,t)=-\frac{1}{2\pi}\int_{\mathbb{R}+\mathrm{i}0^+}{\rm d}v\,\frac{1}{1-{\rm e}^v}\mathrm{Li}_2\big({\rm e}^{\check\epsilon v +2\pi \mathrm{i} t}\big).
\end{equation*}
In particular, on their common domains of definition we have
\begin{equation*}
 W_{\mathbb{R}_{>0}}(\epsilon,t)= W_{{\rm np}}(\epsilon,t).
\end{equation*}
\end{Proposition}

\begin{proof}
We follow again the method of \cite{garoufalidis2020resurgence}, based on the unitarity of the
Fourier transform:
\[
\langle f,g\rangle=\langle Ff,Fg\rangle,\qquad \langle f,g\rangle=\int_{\mathbb{R}} \mathrm{d}x\, f(x)\overline{g(x)}, \qquad (F\psi)(x)=\int_{\mathbb{R}} {\rm d}y\, {\rm e}^{2\pi\mathrm{i} xy}\psi(y).
\]
We define for sufficiently small $\delta>0$,
\begin{equation*}
f_\delta(x):={\rm e}^{-\delta x}{\rm Li}_2\big({\rm e}^{\check\epsilon x
+2\pi \mathrm{i} t}\big),\qquad
g_\delta(x):={\rm e}^{+\delta x}\frac{1}{1-{\rm e}^{x+\mathrm{i}\delta}}.
\end{equation*}
We then easily see that
\begin{equation*}
\lim_{\delta \to 0^+}-\frac{1}{2\pi}\langle f_\delta,g_\delta\rangle=-\frac{1}{2\pi}\int_{\mathbb{R}+\mathrm{i} 0^+}{\rm d}v\, \frac{1}{1-{\rm e}^v}\mathrm{Li}_2\big({\rm e}^{\check{\epsilon} v +2\pi \mathrm{i} t}\big).
\end{equation*}
We now compute the Fourier transform of $f_\delta$ and $g_\delta$. Setting $\zeta=2\pi x+\mathrm{i}\delta$, we find that
\begin{equation*}
 Ff_\delta(x)=\int_{\mathbb{R}}\mathrm{d}y \, {\rm e}^{\mathrm{i} y\zeta}\mathrm{Li}_2({\rm e}^{\check\epsilon y+2\pi \mathrm{i} t})=\frac{\mathrm{i}\check\epsilon}{\zeta}\int_{\mathbb{R}}{\rm d}y\, {\rm e}^{\mathrm{i} y\zeta}\mathrm{Li}_1\big({\rm e}^{\check{\epsilon} y + 2\pi \mathrm{i} t}\big)=\left(\frac{\check \epsilon}{\zeta}\right)^2\int_{\mathbb{R}}{\rm d}y\, \frac{{\rm e}^{\mathrm{i} y \zeta}}{1-{\rm e}^{-\check{\epsilon} y -2\pi \mathrm{i} t}},
\end{equation*}
where we have integrated by parts, and used that the boundary terms vanish. The last integral has simple poles at $y=2\pi \mathrm{i}(k-t)/\check{\epsilon}$, and under our assumptions for the parameters $t$ and $\epsilon$, it is easy to check that the poles on the upper half-plane correspond to $k>0$, while those in the lower half-plane correspond to $k\leq 0$. If $\operatorname{Re}(x)>0$, by an application of Jordan's lemma and the residue theorem, we can compute $Ff_\delta(x)$ by summing up the residues in the upper half-plane, obtaining
\begin{align*}
 \left(\frac{\check \epsilon}{\zeta}\right)^2\int_{\mathbb{R}}{\rm d}y \frac{{\rm e}^{\mathrm{i} y \zeta}}{1-{\rm e}^{-\check{\epsilon} y -2\pi \mathrm{i} t}}&=2\pi \mathrm{i} \left(\frac{\check \epsilon}{\zeta}\right)^2\sum_{k=1}^{\infty}\frac{{\rm e}^{\mathrm{i} y\zeta}}{\check{\epsilon}}\Big|_{y=2\pi \mathrm{i} (k-t)/\check{\epsilon}}=\frac{2\pi \mathrm{i}\check\epsilon}{\zeta^2}{\rm e}^{2\pi t\zeta /\check{\epsilon}}\sum_{k=1}^{\infty}{\rm e}^{-2\pi k\zeta/\check{\epsilon} }\\
 &=\frac{\pi \mathrm{i} \check{\epsilon}}{\zeta^2}{\rm e}^{\pi \zeta(2t-1)/\check{\epsilon}}\left(2{\rm e}^{-\pi \zeta/\check{\epsilon}}\sum_{k=0}^{\infty}{\rm e}^{-2\pi \zeta k/\check{\epsilon}}\right)=\frac{\pi \mathrm{i} \check{\epsilon}}{\zeta^2}\frac{{\rm e}^{\pi \zeta(2t-1)/\check{\epsilon}}}{\sinh(\pi \zeta/\check{\epsilon})},
 \end{align*}
where in the last equality we have used the Dirichlet series representation of $1/\sinh(z)$. Similarly, if $\operatorname{Re}(x)<0$, we can compute $Ff_{\epsilon}$ summing up the residues in the lower half-plane, obtaining
\begin{align*}
 \left(\frac{\check \epsilon}{\zeta}\right)^2\int_{\mathbb{R}}{\rm d}y \frac{{\rm e}^{\mathrm{i} y \zeta}}{1-{\rm e}^{-\check{\epsilon} y -2\pi \mathrm{i} t}}&=-2\pi \mathrm{i} \left(\frac{\check \epsilon}{\zeta}\right)^2\sum_{k=0}^{-\infty}\frac{{\rm e}^{\mathrm{i} y\zeta}}{\check{\epsilon}}\Big|_{y=2\pi \mathrm{i} (k-t)/\check{\epsilon}}=-\frac{2\pi \mathrm{i} \check\epsilon}{\zeta^2}{\rm e}^{2\pi t\zeta /\check{\epsilon}}\sum_{k=0}^{\infty}{\rm e}^{2\pi k\zeta/\check{\epsilon} }\\ &=\frac{\pi \mathrm{i} \check\epsilon}{\zeta^2}{\rm e}^{\pi \zeta(2t-1)/\check{\epsilon}}\left(-2{\rm e}^{\pi \zeta/\check{\epsilon}}\sum_{k=0}^{\infty}{\rm e}^{2\pi \zeta k/\check{\epsilon}}\right)=\frac{\pi \mathrm{i} \check\epsilon}{\zeta^2}\frac{{\rm e}^{\pi \zeta(2t-1)/\check{\epsilon}}}{\sinh(\pi \zeta/\check{\epsilon})},
\end{align*}
so that $Ff_{\delta}(x)$ exists for $x\neq 0$ and
\begin{equation*}
Ff_\delta(x)=\frac{\pi \mathrm{i} \check\epsilon}{\zeta^2}\frac{{\rm e}^{\pi \zeta(2t-1)/\check{\epsilon}}}{\sinh(\pi \zeta/\check{\epsilon})}.
\end{equation*}

Similarly, we can compute $Fg_\delta(x)$, obtaining that for $x\neq 0$,
\begin{equation*}
 \overline{Fg_\delta(x)}=\pi \mathrm{i} \frac{{\rm e}^{(\delta-\pi )\zeta}}{\sinh(\pi \zeta)}.
\end{equation*}

We then have that
\begin{equation*}
\begin{split}
 \lim_{\delta\to 0^+}-\frac{1}{2\pi}\langle f_\delta,g_\delta\rangle&=\lim_{\delta\to 0^+}-\frac{1}{2\pi}\langle Ff_\delta,Fg_\delta\rangle\\
 &=-\frac{1}{2\pi}\int_{\mathbb{R}+\mathrm{i} 0^+}{\rm d}x \left(\frac{\pi \mathrm{i} \check\epsilon}{(2\pi x)^2}\frac{{\rm e}^{\pi (2\pi x)(2t-1)/\check{\epsilon}}}{\sinh(2\pi^2x /\check{\epsilon})}\right)\left(\frac{\pi {\rm i}{\rm e}^{-2\pi^2x}}{\sinh(2\pi^2x)}\right)\\
 &=-\frac{1}{2\pi}\int_{\mathbb{R}+{\rm i}0^+}{\rm d}x \left(\frac{2\pi \mathrm{i} \check\epsilon}{(2\pi x)^2}\frac{{\rm e}^{4\pi^2xt/\check\epsilon}}{{\rm e}^{4\pi^2x/\check\epsilon}-1}\right)\left(\frac{2\pi \mathrm{i}}{{\rm e}^{4\pi^2x}-1}\right)\\
 &=-\frac{1}{2\pi}\int_{\check{\epsilon}^{-1}\cdot(\mathbb{R}+\mathrm{i} 0^+)}\frac{{\rm d}s}{(2\pi)^2}\frac{(2\pi \mathrm{i} \check\epsilon)^2}{(s\check\epsilon/2\pi)^2}\frac{{\rm e}^{st}}{({\rm e}^s-1)({\rm e}^{\check\epsilon s}-1)}, \qquad s=(2\pi)^2x/\check\epsilon\\
 &=2\pi \int_{\mathbb{R}+\mathrm{i} 0^+}\frac{{\rm d}s}{s^2} \frac{{\rm e}^{st}}{({\rm e}^s-1)({\rm e}^{\check\epsilon s}-1)},
\end{split}
\end{equation*}
where we used the fact that the range of the parameter $\epsilon$ allows us to deform the contour back to $\mathbb{R}+\mathrm{i}0^{+}$. The result then follows. In particular, from Proposition \ref{WorformBorelsum} it follows that
\begin{equation*}
 W_{\mathbb{R}_{>0}}(\epsilon,t)=W_{{\rm np}}(\epsilon,t)
\end{equation*}
on their common domains of definition.
\end{proof}

\subsection[Proof of the limits to $l_{\infty}$]{Proof of the limits to $\boldsymbol{l_{\infty}}$}\label{linftylimit}

Here we prove Proposition \ref{rblimit1}, dealing with the sided limits to $l_{\infty}$.
\begin{proof}
By Proposition \ref{Stokesjumps}, we find that
\begin{equation*}
 W_{\rho_k}-W_{\rho_{k+1}}=\mathrm{i}\check\epsilon \mathrm{Li}_2\big({\rm e}^{2\pi \mathrm{i}(t+k)/\check \epsilon}\big).
\end{equation*}
Denoting $w={\rm e}^{2\pi \mathrm it/\check{\epsilon}}$ and $\widetilde{q}={\rm e}^{2\pi \mathrm{i}/\check{\epsilon}}$, we find
\begin{gather*}
 W_{\rho_0}(\epsilon,t)-\lim_{k\to \infty}W_{\rho_k}(\epsilon,t) =\sum_{k=0}^{\infty}W_{\rho_k}(\epsilon,t)-W_{\rho_{k+1}}(\epsilon,t)
 =\mathrm{i}\check\epsilon\sum_{k=0}^{\infty}\mathrm{Li}_2\big(w\widetilde{q}^k\big).
 \end{gather*}
We now use the following identity:
\begin{equation}
\sum_{k=0}^\infty \mathrm{Li}_2\big(w\widetilde{q}^k\big)=\sum_{l=1}^\infty\frac{1}{l^2}\frac{w^l}{1-\widetilde{q}^l}.
\label{id3}
\end{equation}
 In order to verify (\ref{id3}), we use that $|\widetilde{q}|<1$ and $|w|<1$ to see that both sides converge. One can then act on both sides with $w\frac{{\rm d}}{{\rm d}w}$.
The left side of the resulting equation is easily seen to be equal to
\[
-\sum_{k=0}^\infty\log\big(1-w\widetilde{q}^k\big)=\sum_{l=1}^\infty\frac{1}{l}\frac{w^l}{1-\widetilde{q}^l}.
\]
It follows that (\ref{id3}) holds up to addition of a term which is
constant with respect to $w$. In order to fix this freedom, it
suffices to note that (\ref{id3}) holds for $w=0$.

Using the previous identities, we obtain
\begin{equation*}
 W_{\rho_0}(\epsilon,t)-\lim_{k\to \infty}W_{\rho_k}(\epsilon,t)={\rm i}\check\epsilon\sum_{l=1}^\infty\frac{w^l}{l^2(1-\widetilde{q}^l)}.
\end{equation*}
Now notice that under our assumptions on $t$ and $\epsilon$, we have $W_{\rho_0}=W_{{\rm np}}$ by Proposition \ref{woronowiczlemma}. We now show that $W_{{\rm np}}$ admits the following representation as sum over residues:
\begin{equation*}
 W_{{\rm np}}(\epsilon,t)=\mathrm{i}\check\epsilon\sum_{l=1}^\infty\frac{w^l}{l^2(1-\widetilde{q}^l)}
 +\frac{1}{\mathrm{i}} \sum_{k=1}^{\infty} \frac{1}{k^2} \frac{{\rm e}^{2\pi \mathrm{i} t k}}{({\rm e}^{\mathrm{i} k \epsilon}-1)}.
\end{equation*}
In order to see this, let us recall that by definition we have
\begin{align*}
W_{\rm np}(\epsilon,t) =2\pi \int_{\mathbb{R}+\mathrm i0^+} \frac{{\rm d}s}{s^2}\frac{{\rm e}^{s t}}{({\rm e}^s-1)({\rm e}^{\check{\epsilon}s}-1)}.
\end{align*}
The integrand has two series of poles, one
at $u=2\pi\frac{\mathrm i}{\check\epsilon}l$, $l\in \mathbb{Z}$ and the other at
$u=2\pi \mathrm i k$, $k\in \mathbb{Z}$. We can compute the previous integral by closing the contour in the upper
half-plane. The contributions from the poles at $u=2\pi\frac{\mathrm i}{\check\epsilon}l$ are calculated as
\[
\begin{aligned}
\frac{(2\pi)^2 \mathrm i}{\check{\epsilon}}\frac{{\rm e}^{st}}{s^2({\rm e}^s-1)}\Big|_{s=2\pi \mathrm{i} l/{\check\epsilon}} =\frac{(2\pi)^2 \mathrm i}{\check{\epsilon}}\frac{{\rm e}^{2\pi {\rm i}l t/\check\epsilon}}{(2\pi \mathrm{i} l/\check\epsilon)^2({\rm e}^{2\pi \mathrm{i} l/\check\epsilon}-1)}
=\mathrm{i}\check{\epsilon}\frac{w^l}{l^2(1-\widetilde{q}^l)}
\end{aligned}
\]
while the contributions of the poles at $u=2\pi \mathrm i k$ gives
\begin{equation*}
 (2\pi)^2\mathrm{i}\frac{{\rm e}^{st}}{s^2({\rm e}^{\check\epsilon s}-1)}\Big|_{s=2\pi \mathrm{i} k} =(2\pi)^2 \mathrm i\frac{{\rm e}^{2\pi \mathrm{i} k t}}{(2\pi \mathrm{i} k)^2({\rm e}^{2\pi \mathrm{i} k\check\epsilon}-1)}
=\frac{1}{\mathrm{i}}\frac{{\rm e}^{2\pi \mathrm{i} k t}}{k^2({\rm e}^{ \mathrm{i} k\epsilon}-1)}.
\end{equation*}

In particular, we conclude that
\begin{align*}
 \lim_{k\to \infty}W_{\rho_k}(\epsilon,t)&=\lim_{k\to \infty}(W_{\rho_k}(\epsilon,t)-W_{\rho_0}(\epsilon,t))+W_{\rho_0}(\epsilon,t)\\
 &=-\mathrm{i}\check\epsilon\sum_{l=1}^\infty\frac{w^l}{l^2(1-\widetilde{q}^l)} +\mathrm{i}\check\epsilon\sum_{l=1}^\infty\frac{w^l}{l^2(1-\widetilde{q}^l)}
 +\frac{1}{\mathrm{i}} \sum_{k=1}^{\infty} \frac{1}{k^2} \frac{{\rm e}^{2\pi \mathrm{i} t k}}{({\rm e}^{\mathrm{i} k \epsilon}-1)}\\
 &=\frac{1}{\mathrm{i}} \sum_{k=1}^{\infty} \frac{1}{k^2} \frac{{\rm e}^{2\pi \mathrm{i} t k}}{({\rm e}^{\mathrm{i} k \epsilon}-1)}
 =W(\epsilon,t).
 \end{align*}

Now we prove the other limit. Using the jumps along the Stokes rays $-l_{k}$ for $k<0$, we find that
\begin{align*}
 W_{-\rho_0}-\lim_{k\to -\infty}W_{-\rho_k}&=\sum_{k=-1}^{-\infty}W_{-\rho_{k+1}}-W_{-\rho_{k}}
 =- \mathrm{i}\check{\epsilon}\sum_{k=-1}^{-\infty} \mathrm{Li}_2\big({\rm e}^{-2\pi \mathrm i(t+k)/\check \epsilon}\big)\\
 &=-\mathrm{i}\check{\epsilon}\sum_{k=0}^{\infty} \mathrm{Li}_2\big(w^{-1}\widetilde{q}^k\big)+\mathrm{i}\check{\epsilon}\mathrm{Li}_2\big(w^{-1}\big).
\end{align*}
Using the constraints on $t$ and $\epsilon$, we find that $|w^{-1}|<1$ and $|\widetilde{q}|<1$, so that we can expand in series as in Proposition \ref{rblimit1} and write
\begin{equation*}
 W_{-\rho_0}-\lim_{k\to -\infty}W_{-\rho_k}=-\mathrm{i}\check\epsilon\sum_{l=1}^\infty\frac{w^{-l}}{l^2(1-\widetilde{q}^l)}+\mathrm{i}\check{\epsilon}\mathrm{Li}_2\big(w^{-1}\big).
\end{equation*}

On the other hand, using the relation $W_{-\rho}(\epsilon,t)=-W_{\rho}(-\epsilon,t)+\mathrm{i}\mathrm{Li}_2(Q)$, we find that
\begin{equation*}
 W_{-\rho_0}(\epsilon,t)=-W_{\rho_0}(-\epsilon,t)+\mathrm{i}\mathrm{Li}_2(Q)=-W_{\rm np}(-\epsilon,t)+\mathrm{i}\mathrm{Li}_2(Q).
\end{equation*}
 Using the same argument as before with the integral representation of $W_{\rm np}$, we find that
\begin{align*}
 W_{-\rho_0}(\epsilon,t)&=\mathrm{i}\check\epsilon\sum_{l=1}^\infty\frac{w^{-l}}{l^2(1-\widetilde{q}^{-l})}
 -\frac{1}{\mathrm{i}} \sum_{k=1}^{\infty} \frac{1}{k^2} \frac{{\rm e}^{2\pi \mathrm{i} t k}}{({\rm e}^{-\mathrm{i} k \epsilon}-1)}+\mathrm{i}\mathrm{Li}_2(Q)\\
 &=\mathrm{i}\check\epsilon\sum_{l=1}^\infty\frac{w^{-l}}{l^2(1-\widetilde{q}^{-l})}-
 W(-\epsilon,t)+\mathrm{i}\mathrm{Li}_2(Q)
 =\mathrm{i}\check\epsilon\sum_{l=1}^\infty\frac{w^{-l}}{l^2(1-\widetilde{q}^{-l})}+W(\epsilon,t).
 \end{align*}

Joining our results together, we conclude that
\begin{equation*}
 \begin{split}
 \lim_{k\to -\infty}W_{-\rho_k}(\epsilon,t)&=\mathrm{i}\check\epsilon\sum_{l=1}^{\infty}\bigg(\bigg(\frac{w^{-l}}{l^2(1-\widetilde{q}^{-l})}\bigg) + \bigg(\frac{w^{-l}}{l^2(1-\widetilde{q}^{l})}\bigg)\bigg)+W(\epsilon,t) -\mathrm{i}\check{\epsilon}\mathrm{Li}_2\big(w^{-1}\big).
 \end{split}
\end{equation*}
Finally, notice that
\begin{equation*}
 \begin{split}
 \sum_{l=1}^{\infty}\bigg(\bigg(\frac{w^{-l}}{l^2(1-\widetilde{q}^{-l})}\bigg) + \bigg(\frac{w^{-l}}{l^2(1-\widetilde{q}^{l})}\bigg)\bigg)&=\sum_{l=1}^{\infty}\frac{w^{-l}}{l^2}=\mathrm{Li}_2\big(w^{-1}\big),
 \end{split}
\end{equation*}
where in the last equality we used that $\big|w^{-1}\big|<1$. Hence, we conclude that
\begin{equation*}
 \lim_{k\to -\infty}W_{-\rho_k}(\epsilon,t)=W(\epsilon,t).\tag*{\qed}
\end{equation*}\renewcommand{\qed}{}
\end{proof}

\section{Special functions}\label{specialfunctapp}
\subsection{Multiple gamma and multiple sine functions}
Here we introduce the building blocks for the solutions of the refined difference equation. We will follow the exposition and notation of \cite{Narukawa}.

Suppose $\omega_1,\dots,\omega_r \in \mathbb{C}$ all lie on the same side of some straight line through the origin, then the multiple zeta function is defined by the series
\begin{equation*}
\zeta_r(s,z\mid\underline{\omega}) = \sum_{n_1,\dots,n_r=0}^{\infty} \frac{1}{(n_1\omega_1+\dots+ n_r \omega_r +z)^s},
\end{equation*}
for $z\in \mathbb{C}$ and $\operatorname{Re}(s) >r$, where the exponential is rendered one-valued. This series is holomorphic in the domain $\{ \operatorname{Re} (s)>r\}$ and is analytically continued to $s\in \mathbb{C}$.

Since it is holomorphic at $s=0$ we can define the multiple gamma function by
\begin{equation*}
\Gamma_r(z | \underline{\omega}) =\exp \left( \frac{\partial}{\partial s} \zeta_r(s,z \mid \underline{\omega})|_{s=0} \right),
\end{equation*}
where
$ |\underline{\omega}| = \omega_1 + \dots + \omega_r, $
and the multiple sine function is given by
\begin{equation*}
\sin_r(z\mid\underline{\omega}) = \Gamma_r(z\mid\underline{\omega})^{-1} \Gamma_r(|\underline{\omega}| - z \mid \underline{\omega})^{(-1)^r}.
\end{equation*}
This satisfies the functional equations
\begin{gather}
\sin_r(z+\omega_j \mid \underline{\omega}) =\sin_{r-1} \big(z \mid \underline{\omega}^{-}(j)\big)^{-1} \sin_r (z\mid\underline{\omega}),\nonumber \\
\sin_r(z\mid \underline{\omega}) \sin_r(|\underline{\omega}|-z \mid \underline{\omega})^{(-1)^r}=1,\label{eq:sinefunc}
\end{gather}
where \begin{equation*} \underline{\omega}^{-}(j)= \big(\omega_1, \dots,\overset{\vee}{\omega_j},\dots,\omega_r\big),\end{equation*}
and $\overset{\vee}{\omega_j}$ means omitting $\omega_j$.

For $z\in \mathbb{C}$, $\underline{\omega}=(\omega_1,\dots,\omega_r)$, $\omega_j \in \mathbb{C}\setminus\{0\}$ we define the multiple Bernoulli polynomials $B_{r,n}(z,\underline{\omega})$ in terms of the generating function
\begin{equation}\label{eq:Bernoullipolgen}
\mathcal{B}_r(z,x\mid \omega_1, \dots,\omega_r) = \frac{x^r {\rm e}^{z x}}{\prod_{j=1}^r ({\rm e}^{\omega_j x}-1)} = \sum_{n=0}^{\infty} B_{r,n}(z|\underline{\omega}) \frac{x^n}{n!}.
\end{equation}

\begin{Proposition}\label{prop:Bernoullidiff}
The generating function of the Bernoulli polynomials $B_{r,n}(z \mid \omega_1, \dots,\omega_r)$ obeys
\begin{equation*}
 \mathcal{B}_r(z+\omega_j,x\mid \omega_1, \dots,\omega_r)- \mathcal{B}_r(z,x\mid \omega_1, \dots,\omega_r) =x \cdot \mathcal{B}_{r-1}(z,x\mid \omega_1, \dots,\overset{\vee}{\omega_j},\dots,\omega_r)
\end{equation*}
where $\overset{\vee}{\omega_j}$ means omitting $\omega_j$.
\end{Proposition}
\begin{proof}
This follows immediately by using the definition \eqref{eq:Bernoullipolgen}:
\begin{align*}
 \mathcal{B}_r(z+\omega_j,x\mid \omega_1, \dots,\omega_r)- \mathcal{B}_r(z,x\mid \omega_1, \dots,\omega_r) &= \frac{x\cdot.x^{r-1} {\rm e}^{zx} ({\rm e}^{\omega_jx}-1)}{ \prod_{i=1}^r ({\rm e}^{\omega_i x}-1)} \\
 &=x \cdot \mathcal{B}_{r-1}\big(z,x\mid \omega_1, \dots,\overset{\vee}{\omega_j},\dots,\omega_r\big).\tag*{\qed}
\end{align*}\renewcommand{\qed}{}
\end{proof}

This proposition implies in particular that
\begin{equation}\label{eq:Bernoullipolrec}
 B_{3,3}(z+\omega_j\mid\underline{\omega})-B_{3,3}(z|\underline{\omega})= 3\cdot B_{2,2}(z\mid\underline{\omega}^{-}(j)),
\end{equation}
where $\underline{\omega}=(\omega_1,\omega_2,\omega_3)$ and where $ \underline{\omega}^{-}(j)= \big(\omega_1, \dots,\overset{\vee}{\omega_j},\dots,\omega_r\big)$.

Next, we discuss a few special functions that play an important role in this paper.

The multiple sine functions are defined using the Barnes multiple Gamma functions \cite{Barnes}. For a variable $z\in \mathbb{C}$ and parameters $\omega_1,\ldots,\omega_r \in \mathbb{C}^{*}$ these are defined by
\begin{equation*}
 \sin_r(z \mid \omega_1,\dots,\omega_r):= \Gamma_{r}(z \mid \omega_1,\dots,\omega_r)^{-1} \cdot \Gamma_{r}\left(\sum_{i=1}^r \omega_i - z \mid \omega_1,\dots,\omega_r\right)^{(-1)^r}.
\end{equation*}
For further definitions see, e.g., \cite{BridgelandCon,Ruijsenaars1} and references therein.

Another special function is (related to) Fadeev's quantum dilogarithm
\begin{equation*}
 \mathcal{S}_2(z \mid \omega_1,\omega_2) := \exp\left(-\frac{\pi \mathrm{i}}{2} \cdot B_{2,2}(z\mid \omega_1,\omega_2)\right) \cdot \sin_2(z \mid \omega_1,\omega_2),
\end{equation*}
which we have introduced in terms of the double sine function as in \cite{BridgelandCon}.

The final special function is defined by\footnote{In \cite{BridgelandCon}, $\mathcal{S}_2$ is called $F$ and $\mathcal{S}_3$ is called $G$, we changed the label and introduced the subscripts to avoid confusion with the free energies and to make the relation to the double and triple sine functions clear.}
\begin{equation*}
 \mathcal{S}_3(z | \omega_1,\omega_2) := \exp\left(\frac{\pi \mathrm{i}}{6} \cdot B_{3,3}(z+\omega_1 \mid \omega_1,\omega_1,\omega_2)\right) \cdot \sin_3(z+\omega_1 \mid \omega_1,\omega_2,\omega_3).
\end{equation*}

\subsection{Quantum dilogarithm and its properties}\label{appendix:qdilog}
In this section we summarize a few properties of the function $\mathcal{S}_2(z\mid \omega_1,\omega_2)$, that is a single-valued meromorphic function of variables $z\in\mathbb{C}$ and $\omega_1,\omega_2 \in \mathbb{C}^*$, following~\cite[Section~4]{BridgelandCon} and~\cite{Narukawa}. Say,
\begin{alignat*}{3}
& x_1= \exp (2\pi \mathrm{i} z/\omega_1), \qquad && x_2= \exp (2\pi \mathrm{i} z/\omega_2), &\\
& q_1= \exp (2\pi \mathrm{i} \omega_2/\omega_1), \qquad && q_2= \exp (2\pi \mathrm{i} \omega_1/\omega_2).&
\end{alignat*}
Under the assumption $\omega_1/\omega_2\notin \mathbb{R}_{<0}$, the function $\mathcal{S}_2$ has the following properties:
\begin{enumerate}\itemsep=0pt
 \item It is regular and non-vanishing except at the points
 \begin{equation*}
 z=a \omega_1 + b\omega_2,\qquad a,b \in \mathbb{Z},
 \end{equation*}
 which are zeroes if $a,b\le 0$, poles if $a,b>0$, and otherwise neither.
 \item It is symmetric in the arguments $\omega_1$, $\omega_2$:
 \begin{equation*} \mathcal{S}_2(z\mid \omega_1,\omega_2)=\mathcal{S}_2(z\mid \omega_2,\omega_1),
 \end{equation*}
 and is invariant under simultaneous rescaling of all three arguments.
 \item It satisfies the two difference relations:
 \begin{equation}\label{qddiffeq} \frac{\mathcal{S}_2(z+\omega_1\mid \omega_1,\omega_2)}{\mathcal{S}_2(z\mid \omega_1,\omega_2)}= \frac{1}{1-x_2}, \qquad \frac{\mathcal{S}_2(z+\omega_2\mid \omega_1,\omega_2)}{\mathcal{S}_2(z\mid \omega_1,\omega_2)}= \frac{1}{1-x_1}.
 \end{equation}
 \item It has the product expansion
\begin{equation}\label{S2prod} \mathcal{S}_2(z\mid \omega_1,\omega_2)= \prod_{k\ge 1} \big(1-x_1 q_1^{-k}\big)^{-1} \cdot \prod_{k\ge 0} \big(1-x_2 q_2^{k}\big),
 \end{equation}
valid when $\operatorname{Im}(\omega_1/\omega_2) >0.$

\item When $\operatorname{Re}(\omega_i)>0$ and $0<\operatorname{Re}(z)< \operatorname{Re}(\omega_1+\omega_2)$ it has the integral representation:
\begin{equation}\label{S2int} \mathcal{S}_2(z|\omega_1,\omega_2)= \exp \left( \int_{\mathbb{R}+\mathrm{i} 0^{+}} \frac{{\rm e}^{zs}}{({\rm e}^{\omega_1s}-1)({\rm e}^{\omega_2 s}-1)} \frac{{\rm d}s}{s}\right),
\end{equation}
where the contour $\mathbb{R}+\mathrm{i} 0^{+}$ follows the real axis from $-\infty$ to $+\infty$ avoiding the origin by a~small detour in the upper half-plane.

\item Fix $z\in\mathbb{C}$ and $\omega_2\in \mathbb{C}^*$ with $0<\operatorname{Re}(z)<\operatorname{Re}(\omega_2)$ and $\operatorname{Im}(z/\omega_2)>0$. Then it has the asymptotic expansion:
\begin{equation}\label{S2asymp} \log \mathcal{S}_2(z\mid \omega_1,\omega_2) \sim \sum_{k\ge0} \frac{B_k \cdot \omega_1^{k-1}}{k!} \cdot \left( \frac{2\pi \mathrm{i}}{\omega_2}\right)^{k-1} \cdot {\rm Li}_{2-k}\big({\rm e}^{2\pi \mathrm{i} z/\omega_2}\big),
\end{equation}
valid as $\omega_1\rightarrow 0$ in any closed subsector $\Sigma$ of the half-plane $\operatorname{Re}(\omega_1)>0$.
\end{enumerate}

\subsubsection{Relation to Faddeev's quantum dilogarithm}

Faddeev's quantum dilogarithm $\Phi_{b}(z)$, which is defined by
\begin{equation*}
 \Phi_b(z)= \exp \left( \int_{\mathbb{R}+ \mathrm{i} 0^{+}} \frac{{\rm e}^{-2 \mathrm{i} z s}}{4 \sinh(sb) \sinh(sb^{-1})} \frac{{\rm d}s}{s}\right).
\end{equation*}
The function $\mathcal{S}_2$ is related to Faddeev's quantum dilogarithm
as
\begin{equation*}
\Phi_b(z)= \mathcal{S}_2\bigl(-\mathrm{i} z + \big(b+b^{-1}\big)/2\mid b,b^{-1}\bigr).
\end{equation*}

\subsubsection{Special values}

Using the integral representation of the quantum dilogarithm we can find the expression for the quantum dilogarithm for $\omega_2=\omega_1$. We find
\begin{align}
\log\left(\mathcal{S}_2(z+\omega_1\mid \omega_1,\omega_1)\right) &= \int_{\mathbb{R}+\mathrm{i} 0^{+}} \frac{{\rm e}^{(z+\omega_1)s}}{({\rm e}^{\omega_1s}-1)({\rm e}^{\omega_1 s}-1)} \frac{{\rm d}s}{s}\nonumber \\
&= - \frac{\partial}{\partial \omega_1} \left(\frac{\omega_1}{2\pi \mathrm{i}} \cdot {\rm Li}_2\big({\rm e}^{2\pi \mathrm{i} z/\omega_1}\big) \right).\label{eq:qdilogrepeat}
\end{align}
The proof of this equation can be found for instance in the proof of Proposition~4.4 in \cite{BridgelandCon}.

\subsection{Product forms}\label{app:productform}
See \cite[Proposition~5]{Narukawa}. Let $r\ge 2$ and $\operatorname{Im} (\omega_j/\omega_k) \ne 0$, then the multiple sine function $\sin_r(z|\underline{\omega})$ has the following infinite product representation:
\begin{align}
 \sin_r(z\mid \underline{\omega}) &= \exp \left\{ (-1)^r \frac{\pi \mathrm{i}}{r} B_{rr}(z\mid \underline{\omega})\right\} \prod_{k=1}^r (x_k;\underline{q_k})_{\infty}^{(r-2)} \nonumber\\
 &= \exp \left\{ (-1)^{r-1} \frac{\pi \mathrm{i}}{r} B_{rr}(z\mid \underline{\omega})\right\} \prod_{k=1}^r \big(x_k^{-1};\underline{q_k}^{-1}\big)_{\infty}^{(r-2)},\label{eq:multiplesineproduct}
 \end{align}
where we have set
\begin{gather*}
x_k={\rm e}^{2\pi \mathrm{i} z/\omega_k}, \qquad
q_{jk}={\rm e}^{2\pi \mathrm{i} \omega_j/\omega_k},\\ \underline{q_k}=\big(q_{1k},\dots,\overset{\vee}{q_{kk}},\dots,q_{rk}\big), \qquad \underline{q_{k}}^{-1}= \big(q_{1k}^{-1},\dots,\overset{\vee}{q_{kk}^{-1}},\dots,q_{rk}^{-1}\big),
\end{gather*}
and where
\begin{equation*}
 (x;\underline{q})_{\infty}^{(r)}= \prod_{j_0,\dots,j_r=0}^{\infty} \big(1-x q_0^{j_0} \cdots q_r^{j_r}\big),
\end{equation*}
with $\underline{q}= (q_0,\dots,q_r)$.

See \cite[Proposition~1]{Narukawa}. It is also useful to know that
\begin{equation}\label{eq:productinvpow}
 (x;\underline{q})_{\infty}^{(r)} = \frac{1}{\big(q_j^{-1}x;\underline{q}[j]\big)_{\infty}^{(r)}},
\end{equation}
where
$\underline{q}[j]= \big(q_0,\dots,q_j^{-1}, \dots,q_r\big)$.

\subsection*{Acknowledgements}

We would like to thank the organizers of the Western Hemisphere Colloquium on Geometry and Physics for bringing together our community at least virtually in times when personal interactions remained sparse. The second author, L.H.~presented her ideas on the relations of the Borel analysis of~\cite{Hollands:recipe} to that of~\cite{ASTT21} at this colloquium in November 2021, triggering our joint work. We would like to thank Alba Grassi, Qianyu Hao and Andy Neitzke for discussions and for informing us about their related work~\cite{GHN}. We would like to furthermore thank Sergei Alexandrov, Tom Bridgeland, Boris Pioline, Arpan Saha and J\"org Teschner for discussions and correspondence. Finally, we would like to thank the anonymous referees for helpful comments and suggestions. The work of I.T.~is funded by the Deutsche Forschungsgemeinschaft (DFG, German Research Foundation) under Germany's Excellence Strategy EXC 2121 Quantum Universe 390833306. The work of L.H.~is supported by a Royal Society Dorothy Hodgkin Fellowship. The work of M.A.~is supported through the DFG Emmy Noether grant AL 1407/2-1.

\LastPageEnding


\addcontentsline{toc}{section}{References}
\begin{thebibliography}{99}
\footnotesize\itemsep=0pt

\bibitem{ACDKV}
Aganagic M., Cheng M.C.N., Dijkgraaf R., Krefl D., Vafa C., Quantum geometry of
 refined topological strings, \href{https://doi.org/10.1007/JHEP11(2012)019}{\textit{J.~High Energy Phys.}} \textbf{2012}
 (2012), no.~11, 019, 53~pages, \href{https://arxiv.org/abs/1105.0630}{arXiv:1105.0630}.

\bibitem{ADKMV}
Aganagic M., Dijkgraaf R., Klemm A., Mari\~no M., Vafa C., Topological strings
 and integrable hierarchies, \href{https://doi.org/10.1007/s00220-005-1448-9}{\textit{Comm. Math. Phys.}} \textbf{261} (2006),
 451--516, \href{https://arxiv.org/abs/hep-th/0312085}{arXiv:hep-th/0312085}.

\bibitem{Aganagic:2000gs}
Aganagic M., Vafa C., Mirror symmetry, {D}-branes and counting holomorphic
 discs, \href{https://arxiv.org/abs/hep-th/0012041}{arXiv:hep-th/0012041}.

\bibitem{Aganagic:2009cg}
Aganagic M., Yamazaki M., Open {BPS} wall crossing and {M}-theory,
 \href{https://doi.org/10.1016/j.nuclphysb.2010.03.019}{\textit{Nuclear Phys.~B}} \textbf{834} (2010), 258--272, \href{https://arxiv.org/abs/0911.5342}{arXiv:0911.5342}.

\bibitem{APP}
Alexandrov S., Persson D., Pioline B., Wall-crossing, {R}ogers dilogarithm, and
 the {QK}/{HK} correspondence, \href{https://doi.org/10.1007/JHEP12(2011)027}{\textit{J.~High Energy Phys.}} \textbf{2011}
 (2011), no.~12, 027, 65~pages, \href{https://arxiv.org/abs/1110.0466}{arXiv:1110.0466}.

\bibitem{alim2020difference}
Alim M., Difference equation for the {G}romov--{W}itten potential of the
 resolved conifold, \href{https://doi.org/10.1016/j.geomphys.2022.104688}{\textit{J.~Geom. Phys.}} \textbf{183} (2023), 104688,
 4~pages, \href{https://arxiv.org/abs/2011.12759}{arXiv:2011.12759}.

\bibitem{Alim:2021ukq}
Alim M., Intrinsic non-perturbative topological strings, \href{https://arxiv.org/abs/2102.07776}{arXiv:2102.07776}.

\bibitem{Alim:2021lld}
Alim M., Saha A., On the integrable hierarchy for the resolved conifold,
 \href{https://doi.org/10.1112/blms.12676}{\textit{Bull. Lond. Math. Soc.}} \textbf{54} (2022), 2014--2031,
 \href{https://arxiv.org/abs/2101.11672}{arXiv:2101.11672}.

\bibitem{ASTT21}
Alim M., Saha A., Teschner J., Tulli I., Mathematical structures of
 non-perturbative topological string theory: from {GW} to {DT} invariants,
 \href{https://doi.org/10.1007/s00220-022-04571-y}{\textit{Comm. Math. Phys.}}, {t}o appear, \href{https://arxiv.org/abs/2109.06878}{arXiv:2109.06878}.

\bibitem{AST21}
Alim M., Saha A., Tulli I., A hyperk\"ahler geometry associated to the {BPS}
 structure of the resolved conifold, \href{https://doi.org/10.1016/j.geomphys.2022.104618}{\textit{J.~Geom. Phys.}} \textbf{180}
 (2022), 104618, 32~pages, \href{https://arxiv.org/abs/2106.11976}{arXiv:2106.11976}.

\bibitem{Aniceto:2011nu}
Aniceto I., Schiappa R., Vonk M., The resurgence of instantons in string
 theory, \href{https://doi.org/10.4310/CNTP.2012.v6.n2.a3}{\textit{Commun. Number Theory Phys.}} \textbf{6} (2012), 339--496,
 \href{https://arxiv.org/abs/1106.5922}{arXiv:1106.5922}.

\bibitem{Banerjee:2018syt}
Banerjee S., Longhi P., Romo M., Exploring 5d {BPS} spectra with exponential
 networks, \href{https://doi.org/10.1007/s00023-019-00851-x}{\textit{Ann. Henri Poincar\'e}} \textbf{20} (2019), 4055--4162,
 \href{https://arxiv.org/abs/1811.02875}{arXiv:1811.02875}.

\bibitem{BLR}
Banerjee S., Longhi P., Romo M., Exponential {BPS} graphs and {D} brane
 counting on toric {C}alabi--{Y}au threefolds: {P}art~{I}, \href{https://doi.org/10.1007/s00220-021-04242-4}{\textit{Comm. Math.
 Phys.}} \textbf{388} (2021), 893--945, \href{https://arxiv.org/abs/1910.05296}{arXiv:1910.05296}.

\bibitem{Banerjee:2020moh}
Banerjee S., Longhi P., Romo M., Exponential {BPS} graphs and {D} brane
 counting on toric {C}alabi--{Y}au threefolds: {P}art~{II},
 \href{https://arxiv.org/abs/2012.09769}{arXiv:2012.09769}.

\bibitem{Banerjee:2022oed}
Banerjee S., Longhi P., Romo M., {A}-branes, foliations and localization,
 \href{https://doi.org/10.1007/s00023-022-01231-8}{\textit{Ann. Henri Poincar\'e}}, {t}o appear, \href{https://arxiv.org/abs/2201.12223}{arXiv:2201.12223}.

\bibitem{Barnes}
Barnes E., On the theory of multiple gamma function, \textit{Trans. Cambridge
 Philos. Soc.} \textbf{19} (1904), 374--425.

\bibitem{Benini:2009gi}
Benini F., Benvenuti S., Tachikawa Y., Webs of five-branes and {${\mathcal
 N}=2$} superconformal field theories, \href{https://doi.org/10.1088/1126-6708/2009/09/052}{\textit{J.~High Energy Phys.}}
 \textbf{2009} (2009), no.~9, 052, 33~pages, \href{https://arxiv.org/abs/0906.0359}{arXiv:0906.0359}.

\bibitem{Bershadsky:1993cx}
Bershadsky M., Cecotti S., Ooguri H., Vafa C., Kodaira--{S}pencer theory of
 gravity and exact results for quantum string amplitudes, \href{https://doi.org/10.1007/BF02099774}{\textit{Comm. Math.
 Phys.}} \textbf{165} (1994), 311--427, \href{https://arxiv.org/abs/hep-th/9309140}{arXiv:hep-th/9309140}.

\bibitem{BGT}
Bonelli G., Grassi A., Tanzini A., Quantum curves and {$q$}-deformed
 {P}ainlev\'e equations, \href{https://doi.org/10.1007/s11005-019-01174-y}{\textit{Lett. Math. Phys.}} \textbf{109} (2019),
 1961--2001, \href{https://arxiv.org/abs/1710.11603}{arXiv:1710.11603}.

\bibitem{Bridgeland1}
Bridgeland T., Riemann--{H}ilbert problems from {D}onaldson--{T}homas theory,
 \href{https://doi.org/10.1007/s00222-018-0843-8}{\textit{Invent. Math.}} \textbf{216} (2019), 69--124, \href{https://arxiv.org/abs/1611.03697}{arXiv:1611.03697}.

\bibitem{BridgelandCon}
Bridgeland T., Riemann--{H}ilbert problems for the resolved conifold,
 \href{https://doi.org/10.4310/jdg/1594260015}{\textit{J.~Differential Geom.}} \textbf{115} (2020), 395--435,
 \href{https://arxiv.org/abs/1703.02776}{arXiv:1703.02776}.

\bibitem{Bullimore:2014awa}
Bullimore M., Kim H.-C., Koroteev P., Defects and quantum {S}eiberg--{W}itten
 geometry, \href{https://doi.org/10.1007/JHEP05(2015)095}{\textit{J.~High Energy Phys.}} \textbf{2015} (2015), no.~5, 095,
 78~pages, \href{https://arxiv.org/abs/1412.6081}{arXiv:1412.6081}.

\bibitem{Ceresole:1992su}
Ceresole A., D'Auria R., Ferrara S., Lerche W., Louis J., Picard--{F}uchs
 equations and special geometry, \href{https://doi.org/10.1142/S0217751X93000047}{\textit{Internat.~J.~Modern Phys.~A}}
 \textbf{8} (1993), 79--113, \href{https://arxiv.org/abs/hep-th/9204035}{arXiv:hep-th/9204035}.

\bibitem{Chianglocal}
Chiang T.-M., Klemm A., Yau S.-T., Zaslow E., Local mirror symmetry: calculations
 and interpretations, \href{https://doi.org/10.4310/ATMP.1999.v3.n3.a3}{\textit{Adv. Theor. Math. Phys.}} \textbf{3} (1999),
 495--565, \href{https://arxiv.org/abs/hep-th/9903053}{arXiv:hep-th/9903053}.

\bibitem{Chuang:2022uey}
Chuang W.-Y., Quantum {R}iemann--{H}ilbert problems for the resolved conifold,
 \href{https://arxiv.org/abs/2203.00294}{arXiv:2203.00294}.

\bibitem{Coman:2018uwk}
Coman I., Longhi P., Teschner J., From quantum curves to topological string
 partition functions, \href{https://doi.org/10.1007/s00220-022-04579-4}{\textit{Comm. Math. Phys.}}, {t}o appear,
 \href{https://arxiv.org/abs/1811.01978}{arXiv:1811.01978}.

\bibitem{CLT20}
Coman I., Longhi P., Teschner J., From quantum curves to topological string
 partition functions~{II}, \href{https://arxiv.org/abs/2004.04585}{arXiv:2004.04585}.

\bibitem{CousoSantamaria:2014xml}
Couso-Santamar\'{\i}a R., Resurgence in topological string theory, Ph.D.~Thesis, {U}niversidade de Santiago de Compostela, 2014, available at
 \url{http://hdl.handle.net/10347/11868}.

\bibitem{Couso-Santamaria:2014iia}
Couso-Santamar\'{\i}a R., Edelstein J.D., Schiappa R., Vonk M., Resurgent
 transseries and the holomorphic anomaly: nonperturbative closed strings in
 local {${\mathbb{CP}}^2$}, \href{https://doi.org/10.1007/s00220-015-2358-0}{\textit{Comm. Math. Phys.}} \textbf{338} (2015),
 285--346, \href{https://arxiv.org/abs/1407.4821}{arXiv:1407.4821}.

\bibitem{CMS}
Couso-Santamar\'{\i}a R., Mari\~no M., Schiappa R., Resurgence matches
 quantization, \href{https://doi.org/10.1088/1751-8121/aa5e01}{\textit{J.~Phys.~A}} \textbf{50} (2017), 145402, 34~pages,
 \href{https://arxiv.org/abs/1610.06782}{arXiv:1610.06782}.

\bibitem{Couso-Santamaria:2016vcc}
Couso-Santamar\'{\i}a R., Schiappa R., Vaz R., On asymptotics and resurgent
 structures of enumerative {G}romov--{W}itten invariants, \href{https://doi.org/10.4310/CNTP.2017.v11.n4.a1}{\textit{Commun.
 Number Theory Phys.}} \textbf{11} (2017), 707--790, \href{https://arxiv.org/abs/1605.07473}{arXiv:1605.07473}.

\bibitem{coxkatz}
Cox D.A., Katz S., Mirror symmetry and algebraic geometry, \textit{Math.
 Surveys Monogr.}, Vol.~68, \href{https://doi.org/10.1090/surv/068}{Amer. Math. Soc.}, Providence, RI, 1999.

\bibitem{Dimofte:2011ju}
Dimofte T., Gaiotto D., Gukov S., Gauge theories labelled by three-manifolds,
 \href{https://doi.org/10.1007/s00220-013-1863-2}{\textit{Comm. Math. Phys.}} \textbf{325} (2014), 367--419, \href{https://arxiv.org/abs/1108.4389}{arXiv:1108.4389}.

\bibitem{Dimofte:2010tz}
Dimofte T., Gukov S., Hollands L., Vortex counting and {L}agrangian
 3-manifolds, \href{https://doi.org/10.1007/s11005-011-0531-8}{\textit{Lett. Math. Phys.}} \textbf{98} (2011), 225--287,
 \href{https://arxiv.org/abs/1006.0977}{arXiv:1006.0977}.

\bibitem{DingleMorganI}
Dingle R.B., Morgan G.J., {${\rm WKB}$} methods for difference equations.~{I},
 \href{https://doi.org/10.1007/BF00382348}{\textit{Appl. Sci. Res.}} \textbf{18} (1968), 221--237.

\bibitem{DingleMorganII}
Dingle R.B., Morgan G.J., {${\rm WKB}$} methods for difference equations.~{II},
 \href{https://doi.org/10.1007/BF00382349}{\textit{Appl. Sci. Res.}} \textbf{18} (1968), 238--245.

\bibitem{Eager:2016yxd}
Eager R., Selmani S.A., Walcher J., Exponential networks and representations of
 quivers, \href{https://doi.org/10.1007/jhep08(2017)063}{\textit{J.~High Energy Phys.}} \textbf{2017} (2017), no.~8, 063,
 68~pages, \href{https://arxiv.org/abs/1611.06177}{arXiv:1611.06177}.

\bibitem{Elliott:2018yqm}
Elliott C., Pestun V., Multiplicative {H}itchin systems and supersymmetric
 gauge theory, \href{https://doi.org/10.1007/s00029-019-0510-y}{\textit{Selecta Math. (N.S.)}} \textbf{25} (2019), 64, 82~pages,
 \href{https://arxiv.org/abs/1812.05516}{arXiv:1812.05516}.

\bibitem{Eynard:2021sxg}
Eynard B., Garcia-Failde E., Marchal O., Orantin N., Quantization of classical
 spectral curves via topological recursion, \href{https://arxiv.org/abs/2106.04339}{arXiv:2106.04339}.

\bibitem{Faber}
Faber C., Pandharipande R., Hodge integrals and {G}romov--{W}itten theory,
 \href{https://doi.org/10.1007/s002229900028}{\textit{Invent. Math.}} \textbf{139} (2000), 173--199,
 \href{https://arxiv.org/abs/math.AG/9810173}{arXiv:math.AG/9810173}.

\bibitem{Forbes}
Forbes B., Jinzenji M., Extending the {P}icard--{F}uchs system of local mirror
 symmetry, \href{https://doi.org/10.1063/1.1996441}{\textit{J.~Math. Phys.}} \textbf{46} (2005), 082302, 39~pages,
 \href{https://arxiv.org/abs/hep-th/0503098}{arXiv:hep-th/0503098}.

\bibitem{Gaiotto:2014ina}
Gaiotto D., Kim H.-C., Surface defects and instanton partition functions,
 \href{https://doi.org/10.1007/JHEP10(2016)012}{\textit{J.~High Energy Phys.}} \textbf{2016} (2016), no.~10, 012, 49~pages,
 \href{https://arxiv.org/abs/1412.2781}{arXiv:1412.2781}.

\bibitem{Gaiotto:2009hg}
Gaiotto D., Moore G.W., Neitzke A., Wall-crossing, {H}itchin systems, and the
 {WKB} approximation, \href{https://doi.org/10.1016/j.aim.2012.09.027}{\textit{Adv. Math.}} \textbf{234} (2013), 239--403,
 \href{https://arxiv.org/abs/0907.3987}{arXiv:0907.3987}.

\bibitem{Garoufalidis:2020nut}
Garoufalidis S., Gu J., Mari\~no M., The resurgent structure of quantum knot
 invariants, \href{https://doi.org/10.1007/s00220-021-04076-0}{\textit{Comm. Math. Phys.}} \textbf{386} (2021), 469--493,
 \href{https://arxiv.org/abs/2007.10190}{arXiv:2007.10190}.

\bibitem{Garoufalidis:2020xec}
Garoufalidis S., Gu J., Mari\~no M., Peacock patterns and resurgence in complex
 {C}hern--{S}imons theory, \href{https://arxiv.org/abs/2012.00062}{arXiv:2012.00062}.

\bibitem{Garoufalidisevaluation}
Garoufalidis S., Kashaev R., Evaluation of state integrals at rational points,
 \href{https://doi.org/10.4310/CNTP.2015.v9.n3.a3}{\textit{Commun. Number Theory Phys.}} \textbf{9} (2015), 549--582,
 \href{https://arxiv.org/abs/1411.6062}{arXiv:1411.6062}.

\bibitem{garoufalidis2020resurgence}
Garoufalidis S., Kashaev R., Resurgence of {F}addeev's quantum dilogarithm, in
 Topology and Geometry~-- a~Collection of Essays Dedicated to {V}ladimir
 {G}.~{T}uraev, \textit{IRMA Lect. Math. Theor. Phys.}, Vol.~33, \href{https://doi.org/10.4171/IRMA/33-1/14}{Eur. Math.
 Soc.}, Z\"urich, 2021, 257--271, \href{https://arxiv.org/abs/2008.12465}{arXiv:2008.12465}.

\bibitem{GKZ}
Gel'fand I.M., Zelevinskii A.V., Kapranov M.M., Hypergeometric functions and
 toric varieties, \href{https://doi.org/10.1007/BF01078777}{\textit{Funct. Anal. Appl.}} \textbf{23} (1989), 94--106.

\bibitem{Gopakumar:1998ii}
Gopakumar R., Vafa C., {M}-theory and topological strings.~{I},
 \href{https://arxiv.org/abs/hep-th/9809187}{arXiv:hep-th/9809187}.

\bibitem{Gopakumar:1998jq}
Gopakumar R., Vafa C., {M}-theory and topological strings.~{II},
 \href{https://arxiv.org/abs/hep-th/9812127}{arXiv:hep-th/9812127}.

\bibitem{GV}
Gopakumar R., Vafa C., On the gauge theory/geometry correspondence,
 \href{https://doi.org/10.4310/ATMP.1999.v3.n5.a5}{\textit{Adv. Theor. Math. Phys.}} \textbf{3} (1999), 1415--1443,
 \href{https://arxiv.org/abs/hep-th/9811131}{arXiv:hep-th/9811131}.

\bibitem{GHN}
Grassi A., Hao Q., Neitzke A., Exponential NETWOrks, {WKB} and the topological
 string, \href{https://arxiv.org/abs/2201.11594}{arXiv:2201.11594}.

\bibitem{GHM}
Grassi A., Hatsuda Y., Mari\~no M., Topological strings from quantum mechanics,
 \href{https://doi.org/10.1007/s00023-016-0479-4}{\textit{Ann. Henri Poincar\'e}} \textbf{17} (2016), 3177--3235,
 \href{https://arxiv.org/abs/1410.3382}{arXiv:1410.3382}.

\bibitem{Grassi:2014cla}
Grassi A., Mari\~no M., Zakany S., Resumming the string perturbation series,
 \href{https://doi.org/10.1007/JHEP05(2015)038}{\textit{J.~High Energy Phys.}} \textbf{2015} (2015), no.~5, 038, 35~pages,
 \href{https://arxiv.org/abs/1405.4214}{arXiv:1405.4214}.

\bibitem{Gu:2021ize}
Gu J., Mari\~no M., Peacock patterns and new integer invariants in topological
 string theory, \href{https://doi.org/10.21468/scipostphys.12.2.058}{\textit{SciPost Phys.}} \textbf{12} (2022), 058, 51~pages,
 \href{https://arxiv.org/abs/2104.07437}{arXiv:2104.07437}.

\bibitem{Hatsuda:2015fxa}
Hatsuda Y., Comments on exact quantization conditions and non-perturbative
 topological strings, \href{https://arxiv.org/abs/1507.04799}{arXiv:1507.04799}.

\bibitem{Hatsuda:2015oaa}
Hatsuda Y., Spectral zeta function and non-perturbative effects in {ABJM}
 {F}ermi-gas, \href{https://doi.org/10.1007/JHEP11(2015)086}{\textit{J.~High Energy Phys.}} \textbf{2015} (2015), no.~11, 086,
 33~pages, \href{https://arxiv.org/abs/1503.07883}{arXiv:1503.07883}.

\bibitem{Hatsuda:2013oxa}
Hatsuda Y., Mari\~no M., Moriyama S., Okuyama K., Non-perturbative effects and
 the refined topological string, \href{https://doi.org/10.1007/JHEP09(2014)168}{\textit{J.~High Energy Phys.}} \textbf{2014}
 (2014), no.~9, 168, 42~pages, \href{https://arxiv.org/abs/1306.1734}{arXiv:1306.1734}.

\bibitem{Hatsuda:2015owa}
Hatsuda Y., Okuyama K., Resummations and non-perturbative corrections,
 \href{https://doi.org/10.1007/JHEP09(2015)051}{\textit{J.~High Energy Phys.}} \textbf{2015} (2015), no.~9, 051, 29~pages,
 \href{https://arxiv.org/abs/1505.07460}{arXiv:1505.07460}.

\bibitem{Hollands:Heun}
Hollands L., Kidwai O., Higher length-twist coordinates, generalized {H}eun's
 opers, and twisted superpotentials, \href{https://doi.org/10.4310/ATMP.2018.v22.n7.a2}{\textit{Adv. Theor. Math. Phys.}}
 \textbf{22} (2018), 1713--1822, \href{https://arxiv.org/abs/1710.04438}{arXiv:1710.04438}.

\bibitem{Hollands:2013qza}
Hollands L., Neitzke A., Spectral networks and {F}enchel--{N}ielsen
 coordinates, \href{https://doi.org/10.1007/s11005-016-0842-x}{\textit{Lett. Math. Phys.}} \textbf{106} (2016), 811--877,
 \href{https://arxiv.org/abs/1312.2979}{arXiv:1312.2979}.

\bibitem{Hollands:t3abel}
Hollands L., Neitzke A., Exact {WKB} and abelianization for the {$T_3$}
 equation, \href{https://doi.org/10.1007/s00220-020-03875-1}{\textit{Comm. Math. Phys.}} \textbf{380} (2020), 131--186,
 \href{https://arxiv.org/abs/1906.04271}{arXiv:1906.04271}.

\bibitem{Hollands:recipe}
Hollands L., R\"uter P., Szabo R.J., A geometric recipe for twisted
 superpotentials, \href{https://doi.org/10.1007/jhep12(2021)164}{\textit{J.~High Energy Phys.}} \textbf{2021} (2021), no.~12,
 164, 90~pages, \href{https://arxiv.org/abs/2109.14699}{arXiv:2109.14699}.

\bibitem{Hori:2000kt}
Hori K., Vafa C., Mirror symmetry, \href{https://arxiv.org/abs/hep-th/0002222}{arXiv:hep-th/0002222}.

\bibitem{Hosonolocal}
Hosono S., Central charges, symplectic forms, and hypergeometric series in
 local mirror symmetry, in Mirror Symmetry.~{V}, \textit{AMS/IP Stud. Adv.
 Math.}, Vol.~38, \href{https://doi.org/10.1090/amsip/038/18}{Amer. Math. Soc.}, Providence, RI, 2006, 405--439,
 \href{https://arxiv.org/abs/hep-th/0404043}{arXiv:hep-th/0404043}.

\bibitem{Huang:2010kf}
Huang M.-X., Klemm A., Direct integration for general {$\Omega$} backgrounds,
 \href{https://dx.doi.org/10.4310/ATMP.2012.v16.n3.a2}{\textit{Adv. Theor. Math. Phys.}} \textbf{16} (2012), 805--849,
 \href{https://arxiv.org/abs/1009.1126}{arXiv:1009.1126}.

\bibitem{IKV}
Iqbal A., Koz\c{c}az C., Vafa C., The refined topological vertex,
 \href{https://doi.org/10.1088/1126-6708/2009/10/069}{\textit{J.~High Energy Phys.}} \textbf{2009} (2009), no.~10, 069, 58~pages,
 \href{https://arxiv.org/abs/hep-th/0701156}{arXiv:hep-th/0701156}.

\bibitem{Iwaki2}
Iwaki K., Koike T., Takei Y., Voros coefficients for the hypergeometric
 differential equations and {E}ynard--{O}rantin's topological recursion:
 {P}art~{II}: {F}or confluent family of hypergeometric equations,
 \href{https://doi.org/10.1093/integr/xyz004}{\textit{J.~Integrable Syst.}} \textbf{4} (2019), xyz004, 46~pages,
 \href{https://arxiv.org/abs/1810.02946}{arXiv:1810.02946}.

\bibitem{NI}
Iwaki K., Nakanishi T., Exact {WKB} analysis and cluster algebras,
 \href{https://doi.org/10.1088/1751-8113/47/47/474009}{\textit{J.~Phys.~A}} \textbf{47} (2014), 474009, 98~pages, \href{https://arxiv.org/abs/1401.7094}{arXiv:1401.7094}.

\bibitem{Jockers:2019wjh}
Jockers H., Mayr P., Quantum {K}-theory of {C}alabi--{Y}au manifolds,
 \href{https://doi.org/10.1007/jhep11(2019)011}{\textit{J.~High Energy Phys.}} \textbf{2019} (2019), no.~11, 011, 20~pages,
 \href{https://arxiv.org/abs/1905.03548}{arXiv:1905.03548}.

\bibitem{Jockers:2018sfl}
Jockers H., Mayr P., A 3d gauge theory/quantum {K}-theory correspondence,
 \href{https://doi.org/10.4310/ATMP.2020.v24.n2.a4}{\textit{Adv. Theor. Math. Phys.}} \textbf{24} (2020), 327--457, \href{https://arxiv.org/abs/1808.02040}{arXiv:1808.02040}.

\bibitem{Kashani-Poor:2016edc}
Kashani-Poor A.-K., Quantization condition from exact {WKB} for difference
 equations, \href{https://doi.org/10.1007/JHEP06(2016)180}{\textit{J.~High Energy Phys.}} \textbf{2016} (2016), no.~6, 180,
 34~pages, \href{https://arxiv.org/abs/1604.01690}{arXiv:1604.01690}.

\bibitem{KKV}
Katz S., Klemm A., Vafa C., Geometric engineering of quantum field theories,
 \href{https://doi.org/10.1016/S0550-3213(97)00282-4}{\textit{Nuclear Phys.~B}} \textbf{497} (1997), 173--195,
 \href{https://arxiv.org/abs/hep-th/9609239}{arXiv:hep-th/9609239}.

\bibitem{Katz:1997eq}
Katz S., Mayr P., Vafa C., Mirror symmetry and exact solution of {$4$}{D}
 {$N=2$} gauge theories.~{I}, \href{https://doi.org/10.4310/ATMP.1997.v1.n1.a2}{\textit{Adv. Theor. Math. Phys.}} \textbf{1}
 (1997), 53--114, \href{https://arxiv.org/abs/hep-th/9706110}{arXiv:hep-th/9706110}.

\bibitem{Krefl:2015vna}
Krefl D., Mkrtchyan R.L., Exact {C}hern--{S}imons/topological string duality,
 \href{https://doi.org/10.1007/JHEP10(2015)045}{\textit{J.~High Energy Phys.}} \textbf{2015} (2015), no.~10, 045, 27~pages,
 \href{https://arxiv.org/abs/1506.03907}{arXiv:1506.03907}.

\bibitem{Lerche:1996xu}
Lerche W., Introduction to {S}eiberg--{W}itten theory and its stringy origin,
 \href{https://doi.org/10.1016/S0920-5632(97)00073-X}{\textit{Nuclear Phys.~B Proc. Suppl.}} \textbf{55} (1997), 83--117,
 \href{https://arxiv.org/abs/hep-th/9611190}{arXiv:hep-th/9611190}.

\bibitem{Lockhart:2012vp}
Lockhart G., Vafa C., Superconformal partition functions and non-perturbative
 topological strings, \href{https://doi.org/10.1007/jhep10(2018)051}{\textit{J.~High Energy Phys.}} \textbf{2018} (2018),
 no.~10, 051, 43~pages, \href{https://arxiv.org/abs/1210.5909}{arXiv:1210.5909}.

\bibitem{Marinolecture}
Mari\~no M., Lectures on non-perturbative effects in large {$N$} gauge
 theories, matrix models and strings, \href{https://doi.org/10.1002/prop.201400005}{\textit{Fortschr. Phys.}} \textbf{62}
 (2014), 455--540, \href{https://arxiv.org/abs/1206.6272}{arXiv:1206.6272}.

\bibitem{MarinoSpectral}
Mari\~no M., Spectral theory and mirror symmetry, in String-{M}ath 2016,
 \textit{Proc. Sympos. Pure Math.}, Vol.~98, Amer. Math. Soc., Providence, RI,
 2018, 259--294, \href{https://arxiv.org/abs/1506.07757}{arXiv:1506.07757}.

\bibitem{MM}
Mari\~no M., Moore G., Counting higher genus curves in a {C}alabi--{Y}au
 manifold, \href{https://doi.org/10.1016/S0550-3213(98)00847-5}{\textit{Nuclear Phys.~B}} \textbf{543} (1999), 592--614,
 \href{https://arxiv.org/abs/hep-th/9808131}{arXiv:hep-th/9808131}.

\bibitem{Marino:2016rsq}
Mari\~no M., Zakany S., Exact eigenfunctions and the open topological string,
 \href{https://doi.org/10.1088/1751-8121/aa791e}{\textit{J.~Phys.~A}} \textbf{50} (2017), 325401, 50~pages,
 \href{https://arxiv.org/abs/1606.05297}{arXiv:1606.05297}.

\bibitem{Mironov:2009uv}
Mironov A., Morosov A., Nekrasov functions and exact {B}ohr--{S}ommerfeld
 integrals, \href{https://doi.org/10.1007/JHEP04(2010)040}{\textit{J.~High Energy Phys.}} \textbf{2010} (2010), no.~4, 040,
 15~pages, \href{https://arxiv.org/abs/0910.5670}{arXiv:0910.5670}.

\bibitem{Narukawa}
Narukawa A., The modular properties and the integral representations of the
 multiple elliptic gamma functions, \href{https://doi.org/10.1016/j.aim.2003.11.009}{\textit{Adv. Math.}} \textbf{189} (2004),
 247--267, \href{https://arxiv.org/abs/math.QA/0306164}{arXiv:math.QA/0306164}.

\bibitem{Neitzke_hyperhol}
Neitzke A., On a hyperholomorphic line bundle over the {C}oulomb branch,
 \href{https://arxiv.org/abs/1110.1619}{arXiv:1110.1619}.

\bibitem{Nekrasov:2002qd}
Nekrasov N.A., Seiberg--{W}itten prepotential from instanton counting,
 \href{https://dx.doi.org/10.4310/ATMP.2003.v7.n5.a4}{\textit{Adv. Theor. Math. Phys.}} \textbf{7} (2003), 831--864,
 \href{https://arxiv.org/abs/hep-th/0206161}{arXiv:hep-th/0206161}.

\bibitem{NO}
Nekrasov N.A., Okounkov A., Seiberg--{W}itten theory and random partitions, in
 The Unity of Mathematics, \textit{Progr. Math.}, Vol.~244, \href{https://doi.org/10.1007/0-8176-4467-9_15}{Birkh\"auser
 Boston}, Boston, MA, 2006, 525--596, \href{https://arxiv.org/abs/hep-th/0306238}{arXiv:hep-th/0306238}.

\bibitem{Nekrasov:NRS}
Nekrasov N.A, Rosly A., Shatashvili S.L., Darboux coordinates, {Y}ang--{Y}ang
 functional, and gauge theory, \href{https://doi.org/10.1016/j.nuclphysbps.2011.04.150}{\textit{Nuclear Phys.~B Proc. Suppl.}}
 \textbf{216} (2011), 69--93, \href{https://arxiv.org/abs/1103.3919}{arXiv:1103.3919}.

\bibitem{NS}
Nekrasov N.A., Shatashvili S.L., Quantization of integrable systems and four
 dimensional gauge theories, in X{VI}th {I}nternational {C}ongress on
 {M}athematical {P}hysics, \href{https://doi.org/10.1142/9789814304634_0015}{World Sci. Publ.}, Hackensack, NJ, 2010, 265--289,
 \href{https://arxiv.org/abs/0908.4052}{arXiv:0908.4052}.

\bibitem{Ooguri:1999bv}
Ooguri H., Vafa C., Knot invariants and topological strings, \href{https://doi.org/10.1016/S0550-3213(00)00118-8}{\textit{Nuclear
 Phys.~B}} \textbf{577} (2000), 419--438, \href{https://arxiv.org/abs/hep-th/9912123}{arXiv:hep-th/9912123}.

\bibitem{Pandharipande:2000}
Pandharipande R., The {T}oda equations and the {G}romov--{W}itten theory of the
 {R}iemann sphere, \href{https://doi.org/10.1023/A:1026571018707}{\textit{Lett. Math. Phys.}} \textbf{53} (2000), 59--74,
 \href{https://arxiv.org/abs/math.AG/9912166}{arXiv:math.AG/9912166}.

\bibitem{Pasquetti:2009jg}
Pasquetti S., Schiappa R., Borel and {S}tokes nonperturbative phenomena in
 topological string theory and {$c=1$} matrix models, \href{https://doi.org/10.1007/s00023-010-0044-5}{\textit{Ann. Henri
 Poincar\'e}} \textbf{11} (2010), 351--431, \href{https://arxiv.org/abs/0907.4082}{arXiv:0907.4082}.

\bibitem{Ruijsenaars1}
Ruijsenaars S.N.M., On {B}arnes' multiple zeta and gamma functions,
 \href{https://doi.org/10.1006/aima.2000.1946}{\textit{Adv. Math.}} \textbf{156} (2000), 107--132.

\bibitem{Seiberg:1994rs}
Seiberg N., Witten E., Electric-magnetic duality, monopole condensation, and
 confinement in {$N=2$} supersymmetric {Y}ang--{M}ills theory, \href{https://doi.org/10.1016/0550-3213(94)90124-4}{\textit{Nuclear
 Phys.~B}} \textbf{426} (1994), 19--52, {E}rratum,
 \href{https://doi.org/10.1016/0550-3213(94)00449-8}{\textit{Nuclear Phys.~B}}
 \textbf{430} (1994), 485--486, \href{https://arxiv.org/abs/hep-th/9407087}{arXiv:hep-th/9407087}.

\bibitem{Takeireview}
Takei Y., W{KB} analysis and {S}tokes geometry of differential equations, in
 Analytic, Algebraic and Geometric Aspects of Differential Equations, \textit{Trends
 Math.}, \href{https://doi.org/10.1007/978-3-319-52842-7_5}{Birkh\"auser/Springer}, Cham, 2017, 263--304.

\bibitem{Wang:2015wdy}
Wang X., Zhang G., Huang M.-X., New exact quantization condition for toric
 {C}alabi--{Y}au geometries, \href{https://doi.org/10.1103/PhysRevLett.115.121601}{\textit{Phys. Rev. Lett.}} \textbf{115} (2015),
 121601, 5~pages, \href{https://arxiv.org/abs/1505.05360}{arXiv:1505.05360}.

\bibitem{Witten:1993ed}
Witten E., Quantum background independence in string theory,
 \href{https://arxiv.org/abs/hep-th/9306122}{arXiv:hep-th/9306122}.

\end{thebibliography}
\end{document}